\documentclass[11pt, preprint]{imsart}
\RequirePackage[OT1]{fontenc}
\RequirePackage{amsthm,amsmath}
\RequirePackage[colorlinks,citecolor=blue,urlcolor=blue]{hyperref}
\usepackage{graphicx}
\usepackage{enumerate}
\usepackage{latexsym}
\usepackage{amssymb}
\usepackage{multirow}
\usepackage{float}
\usepackage{bm}
\usepackage{paralist}
\usepackage{natbib}
\usepackage{commath}
\usepackage{soul}
\usepackage{algorithm,algorithmic}
\usepackage[inline]{enumitem}
\usepackage[figuresright]{rotating}
\usepackage{pifont}
\usepackage{threeparttable}
{\end{quotation}}
\usepackage[normalem]{ulem}

\DeclareMathOperator*{\argmin}{argmin}

\numberwithin{equation}{section}

\def \mcF{\mathcal{F}}
\def \mcH{\mathcal{H}}
\def \mcX{\mathcal{Z}}
\def \mcB{\mathcal{B}}
\def \mcN{\mathcal{N}}
\def \mcV{\mathcal{V}}

\setstcolor{red}

\newcommand{\Argmin}{\mathop{\mathrm{argmin}}}

\def \bfx{\mathbf{z}}

\def \bfg{\mathbf{g}}
\def \bfX{\mathbf{Z}}
\def \bfZ{\mathbf{Z}}

\def \bfB{\mathbf{B}}
\def \bfY{\mathbf{Y}}
\def \bff{\mathbf{f}}

\def \bfP{\mathbf{P}}
\def \bfepsilon{\bm{\epsilon}}
\def \ev{\mathbb{E}}
\def \pr{\mathbb{P}}
\def \bfW{\mathbf{W}}
\def \bfw{\mathbf{w}}
\def \bfPsi{\mathbf{\Psi}}
\def \bft{\mathbf{t}}
\def \bfv{\mathbf{v}}
\def \bfu{\mathbf{u}}
\def \bfP{\mathbf{P}}
\def \bfR{\mathbf{R}}
\def \cid{\xrightarrow[\text{}]{\text{$\mathcal{D}$}}}
\newcommand{\vo}{\vec{o}\@ifnextchar{^}{\,}{}}

\newcommand{\floor}[1]{\left\lfloor #1 \right\rfloor}
\startlocaldefs
\theoremstyle{plain}

\newtheorem{theorem}{\indent Theorem}
\newtheorem*{theorem*}{\indent Theorem}
\newtheorem{corollary}{\indent Corollary}
\newtheorem{Assumption}{Assumption}

\newtheorem{Condition}{Condition}

\newtheorem{lemma}{\indent Lemma}
\newtheorem{proposition}{\indent Proposition}
\newtheorem{Remark}{Remark}
\theoremstyle{definition}

\DeclareSymbolFont{largesymbolsA}{U}{txexa}{m}{n}
\DeclareMathSymbol{\varprod}{\mathop}{largesymbolsA}{16}
\usepackage{xparse}
\NewDocumentCommand{\ceil}{s O{} m}{%
  \IfBooleanTF{#1} 
    {\left\lceil#3\right\rceil} 
    {#2\lceil#3#2\rceil} 
}
\DeclareFontFamily{U}{mathx}{\hyphenchar\font45}
\DeclareFontShape{U}{mathx}{m}{n}{
      <5> <6> <7> <8> <9> <10>
      <10.95> <12> <14.4> <17.28> <20.74> <24.88>
      mathx10
      }{}
\DeclareSymbolFont{mathx}{U}{mathx}{m}{n}
\DeclareMathSymbol{\bigtimes}{1}{mathx}{"91}

\makeatletter\def\theequation{\arabic{section}.\arabic{equation}}
\@addtoreset{equation}{section}\makeatother

\endlocaldefs
\begin{document}
\title{Statistical Inference on Partially Linear Panel Model under Unobserved Linearity}

\runtitle{Statistical Inference under Unobserved Linearity}

	\begin{aug}
		\author{\fnms{Ruiqi} \snm{Liu}\thanksref{m1}\ead[label=e4]{liuruiq@iu.edu}},
		\author{\fnms{Ben} \snm{Boukai}\thanksref{m1}\ead[label=e1]{bboukai@iupui.edu}},
		\and
		\author{\fnms{Zuofeng} \snm{Shang}\thanksref{m2,t2}
			\ead[label=e3]{shangzf@iu.edu}}
         \thankstext{t1}{Corresponding author. Email: liuruiq@iu.edu }
         \runauthor{Liu et al.}
         \thankstext{m1}{Department of Mathematical Sciences, Indiana University - Purdue University Indianapolis
, IN 46202, USA.}
          \thankstext{m2}{Department of Mathematical Sciences, New Jersey Institute of Technology, NJ 07102, USA.}
          \thankstext{t2}{Sponsored by NSF DMS-1764280 and NSF DMS-1821157}
	\end{aug}
%
%
%
%
%
%

\begin{center}
\textit{This version}  \today
\end{center}
\maketitle

 \begin{center}
\textbf{Abstract}
\end{center}
A new statistical procedure,  based on a modified spline basis, is proposed to identify the linear components in the panel data model with fixed effects. Under some mild assumptions, the proposed procedure is shown to consistently estimate the underlying regression function,  correctly select the linear components, and effectively conduct the statistical inference. When compared to existing methods for detection of linearity in the panel model, our approach is demonstrated to be theoretically justified as well as practically convenient. We provide a computational algorithm that implements the proposed procedure along with a path-based solution method for linearity detection, which avoids the burden of selecting the tuning parameter for the penalty term. Monte Carlo simulations are conducted to examine the finite sample performance of our proposed procedure with detailed findings that confirm our theoretical results in the paper. Applications to Aggregate Production and Environmental Kuznets Curve data also illustrate the necessity for detecting linearity in the partially linear panel model.


\begin{center}
\textbf{\newpage }
\end{center}

\section{Introduction}\label{sec:introduction}
Panel models have attracted much attention from economists and econometricians, especially for their flexibility in modeling homogeneity while preserving individual-level heterogeneity. With the rapid increase in availability of panel data in the past two decades or so, panel models in  both parametric and nonparametric frameworks have been well studied in the literature; see  \cite{rwc00}, \cite{hcl08}, \cite{f17}, \cite{sz15}, \cite{Hsiao1997}, \cite{lr15}, \cite{ht08}, \cite{Hsiao1997}, and \cite{h14}. Still, either framework, parametric or nonparametric,  is not fully satisfactory in modeling panel data, as each has its own advantages and drawbacks. In light of its simplicity and interpretability,  the parametric model becomes  a prominent  tool for panel data analysis; see \cite{bg83}, \cite{dj00}, \cite{kt04}, \cite{s57}, \cite{g64}, \cite{bcs04}. However, when compared to the nonparametric model, it appears to be more sensitive to model misspecification, which is often the case in empirical applications. Based on fewer model assumptions, the nonparametric model can lead to a more robust estimator, especially when dealing with relatively large panel data sets. On the other hand, with a larger dimension of input data, a purely nonparametric model is usually not preferred in empirical applications due to the infamous ``Curse of Dimensionality'' issue and the poor model interpretability.  To address these noted drawbacks and make the
best use of the apparent advantages, the partially linear panel model strikes a balance between parametric and nonparametric frameworks.  For instance, \cite{hcl08} studied both nonparametric and partially linear
panel models with fixed effects and proposed a kernel estimator with a corresponding linearity specification test.
Combining the works by \cite{hcl08} and \cite{mbT09}, \cite{ll15} proposed a two-step estimator in partially linear panel; \cite{bd02} considered the problem of estimating a partially linear fixed effects panel model with possible endogeneity and lagged dependent variables in the linear part; \cite{sz16} proposed  estimation and specification testing procedures for partially linear dynamic panel model with fixed effects, with either exogenous or endogenous variables or both in the linear part and the lagged dependent variables, together with some other exogenous variables entering  nonparametrically in the model.  Following \cite{sz16}, \cite{sz15} extended their work to the panel model with interactive fixed effects.

In practice, however, when considering the partially linear model, the researchers need to consider the following two questions: (a) which variables should be included in the model? (b) what is the functional form of each variable?  Various statistical variable selection techniques, such as \cite{ll09}, \cite{x09}, \cite{hhw10}, are available to  address the first question.  Nevertheless, in the context of economic modeling, one would prefer to select the dependent variables also by economic theory, as relying on purely statistical variable selection procedures may fit a model which is lacking in its economic justification and interpretability, see \cite{bvw18}, \cite{bpmv15}, \cite{kp09}, \cite{d08}. Even though the economic theory can explain which variables should be included in the model, it fails to specify the functional forms of the variables. Therefore,  the second question is of more practical importance than the first one. Misspecification of the functional forms of the regressors can either (a) result in inconsistent estimation if fitting nonlinear functions by linear forms, (b)  or reduce the model interpretability and estimation efficiency if the linear functions are estimated nonparametrically. Thus, correct specification of the linear components, if any,  is essential to improve estimation and model explainability. However, to the best of our knowledge, all the linearity detection methods advocated in the literature on partially linear panel model, are all based on specification tests; see \cite{hcl08}, \cite{sz16} and \cite{sz15}. One primary drawback to this approach is that the test statistics is often difficult to construct and may be deficient in its power when the number of dependent variables is large, which may lead to incorrect model specification. Under cross-sectional data settings,  \cite{zcl11} propose a smoothing-spline-type estimator which is able to estimate the underlying regression function and discover the linear regressors simultaneously.  However, how to conduct valid statistical inference using the approach in \cite{zcl11} is still unknown.

The main purpose of this paper is to propose a unified statistical procedure capable of simultaneously estimating underlying regression function, detecting linear components, and conducting inference in the partially linear panel. This paper is organized as follows.  In Section \ref{sec:model},  we mathematically formulate the linearity detection problem in the partially linear panel model. In Section \ref{sec:estimation}, we  propose a penalized estimator for linearity detection, and  provide {the corresponding computational} algorithm.  The asymptotic properties of the proposed estimator and the corresponding linearity detection procedure are established in Section \ref{sec:asympotics} for both short and large panels. In Section \ref{sec:choice:turning:parameter}, we {discuss} how to determine the tuning parameters involved in the proposed procedure. Section \ref{sec:simulation} carries out a set of Monte Carlo simulations to investigate the finite sample performance of our {proposed}  method. Applications to two real-world datasets are provided in Section \ref{sec:empirical:study}. {Technical details and proofs of the main theorems and auxiliary results are deferred in the Appendix. Throughout this paper, we use the following notation}. 

\noindent\textbf{Notation:} Define $\otimes$ as the tensor product operator. For positive real number $m$, let $\floor{m}$ be the largest integer that is strictly less than $m$ and $\ceil{m}=\floor{m}+1$. Denote $(x)_+=\max(x,0)$ for  $x\in \mathbb{R}$.

\section{Partially Linear Panel Model with Unknown Structure}\label{sec:model}
Suppose that the observations $\{(Y_{it}, \bfX_{it}), i=1,\ldots, N, t=1,\ldots, T\}$ are generated from the following model
\begin{equation}\label{eq:model}
	Y_{it}=f_0(\bfX_{it})+\alpha_i^0+\epsilon_{it}, 
\end{equation}
where $Y_{it}$ is the response variable,
$\bfX_{it}=(Z_{it1}, Z_{it2}, \ldots, Z_{itp})^\top\in \mcX:=[0, 1]^p$ are explanatory variables, both observed
for individual $i$ at time period $t$, $\alpha_i^0\in \mathbb{R}$ are
unobservable individual-level fixed effects, $\epsilon_{it}\in \mathbb{R}$ is unobservable errors.
Assume that the unknown regression function $f_0:\mcX \to \mathbb{R}$ has the following semiparametric expression:
\begin{equation}\label{semi:function}
	f_0(\bfx)=\sum_{j\in J_{\textrm{lin}}} f_{j,0}(z_j)+\sum_{j\in J_{\textrm{lin}}^c} f_{j,0}(z_j), \bfx=(z_1,\ldots, z_p)\in \mcX,\nonumber
\end{equation}
where $J_{\textrm{lin}}$ is a (unknown) subset of $\{1,\ldots, p\}$ and $J_{\textrm{lin}}^c$ denotes its complement,
$f_{j,0}$ for $j\in J_{\textrm{lin}}$ are linear functions and  $f_{j,0}$ for $j\in J_{\textrm{lin}}^c$ are nonlinear. 
Our aim is to identify $J_{\textrm{lin}}$ as well as to conduct statistical inference about $f_0$ based on the observations. 
Without loss of generality, we may assume that $J_{\textrm{lin}}=\{1,2,\ldots, d\}$ for some nonnegative integer $d\leq p$, therefore,
$f_{1,0},\ldots, f_{d,0}$ are linear and  $f_{d+1,0},\ldots,f_{p,0}$ are nonlinear. For convenience, define $\{1,\ldots, d\}$ as the empty set when $d=0$. 

\section{Penalized Estimation}\label{sec:estimation}
In this section, we propose a penalized sieve estimator based on a modified spline basis, which can consistently estimate the underlying regression function $f_0$, effectively identify the linear components, and validly conduct statistical inference. 
\subsection{Centralized Spline}
To estimate $f_0=\sum_{j=1}^pf_{j,0}$, we follow the idea of sieve estimation, i.e., estimating each $f_{j,0}$ by a linear combination of basis functions. The common basis function used in literature includes B-spline basis, wavelet basis, etc. (see \citealp{c07} for an excellent review of sieve basis). However, for linearity detection purpose, the existing bases are not adequate. Thus, we will propose a modified spline space and the corresponding basis to address this issue. Given $M+1$ strictly increasing knots $\bft_M=\{t_0, t_1, \ldots, t_M\}$ with $t_0=0, t_M=1$  and integer $r\geq 1$, define $r$-th degree \textit{Centralized Spline Space} 
\begin{align*}
	\textrm{CSpl}(r, \bft_M)=\bigg\{\sum_{k=1}^{r}c_k\psi_k(z)+\sum_{k=1}^{M-1}\widetilde{c}_k\widetilde{\psi}_k(z): z\in [0, 1], c_k, \widetilde{c}_k\in \mathbb{R} \bigg\},
\end{align*}
with
\begin{align}
	\psi_1(z)&=z-\frac{1}{2},\;\;\psi_k(z)=\bigg(z^k-\frac{1}{k+1}\bigg)-\frac{6k}{(k+1)(k+2)}\bigg(z-\frac{1}{2}\bigg), \textrm{ for } k=2, \ldots, r,\nonumber\\
	\widetilde{\psi}_k(z)&=\bigg((z-t_k)^r_+-\frac{1}{r+1}(1-t_k)^{r+1}\bigg)-\bigg\{\frac{6(1-t_k)^{r+1}}{(r+1)}-\frac{12(1-t_k)^{r+2}}{(r+1)(r+2)}\bigg\}\bigg(z-\frac{1}{2}\bigg),\nonumber\\
&\textrm{ for } k=1,\ldots, M-1,\nonumber
\end{align}
being the corresponding \textit{Centralized Spline Basis}. Expressed by centralized spline basis, any function in centralized spline space can be decomposed two orthogonal parts. To be more specific, for any $f=\sum_{k=1}^{r}c_k\psi_k+\sum_{k=1}^{M-1}\widetilde{c}_k\widetilde{\psi}_k\in \textrm{CSpl}(r, \bft_M)$, we decompose $f=f_{-}+f_{\sim}$, with
\begin{align}
	f_{-}(z)=c_1\psi_1(z)=c_1(z-1/2)\;\; \textrm{ and }\;\; f_\sim(z)=\sum_{k=2}^{r}c_k\psi_k(z)+\sum_{k=1}^{M-1}\widetilde{c}_k\widetilde{\psi}_k(z),\label{eq:linear:nonlinear:definition}
\end{align}
which are corresponding to the linear {and nonlinear components}. It can be verified that
\begin{align}
	\int_0^1f_{-}(z)f_\sim(z)dz=0.\label{eq:orthogonality}
\end{align}

\begin{Remark}
The centralized spline basis essentially is an orthogonal version of the polynomial spline basis $\{z, z^2,\ldots, z^r, (z-t_1)^r_+,\ldots, (z-t_{M-1})^r_+\}$. However, compared to the classical  polynomial splines or B-splines, centralized spline basis is able to effectively separate the linear part from the nonlinear component due to (\ref{eq:orthogonality}). Even though, all the bases generate similar function spaces and the difference is only up to a constant.
\end{Remark}

\subsection{Penalized Estimator}
We begin by introduing the following function spaces
\begin{align*}
\mathcal{H}=\bigg\{f:\mcX \to \mathbb{R} \;|\;\int_\mcX f^2(\bfx)d\bfx<\infty \bigg\} \textrm{,\;\;\; } \mcH_0=\bigg\{f\in H \;|\; \int_\mcX f(\bfx)d\bfx=0\bigg\},
\end{align*}
and
\begin{equation}\label{eq:Sieve:space:ThetaNT}
	\Theta_{NT}=\bigg\{f(\bfx)=\sum_{j=1}^p f_j(z_j) \;|\; f_j \in \Theta_{NT, j}, \textrm{ for } j\in [p]\bigg\},
\end{equation}
where $\Theta_{NT,  j}=\textrm{CSpl}(r_j, \bft_{j, M_j})$ for some integers  $M_j, r_j\geq 1$ and knots $\bft_{j,M_j}=(t_{j,0}, \ldots, t_{j, M_j})$ with $t_{j,0}=0, t_{j, M_j}=1$ for $j\in [p]$. Clearly, $\Theta_{NT}$ is a linear subspace of $\mcH_0$ and in the following it will be the sieve space to estimate the underlying regression function $f_0$. Moreover, for $g, f\in \mcH$, we {introduce the} following notation when the corresponding values exist,  
\begin{align*}
\zeta_i(g, f)&=\frac{1}{T}\sum_{t=1}^T\bigg(g(\bfX_{it})-\frac{1}{T}\sum_{s=1}^Tg(\bfX_{is})\bigg)\bigg(f(\bfX_{it})-\frac{1}{T}\sum_{s=1}^Tf(\bfX_{is})\bigg),\\
\langle g, f\rangle_{NT}&=\frac{1}{N}\sum_{i=1}^N\zeta_i(g,f),\;\;\langle g, f\rangle=\ev(\langle g, f\rangle_{NT}),\;\;\|g\|_{NT}^2=\langle g, g\rangle_{NT},\;\;\|g\|^2=\langle g, g\rangle.
\end{align*}
{We can show that under mild assumptions}, $\langle \cdot , \cdot\rangle$ is a valid inner product on $\Theta_{NT}$ {(see Lemma}  \ref{lemma:expectation:sample:variance} and Lemma \ref{lemma:expectation:sample:variance:large:T}  for details). By above notation, we define a penalized objective function on $\Theta_{NT}$ as follows. For $f(\bfx)=\sum_{j=1}^pf_j(z_j) \in \Theta_{NT}$ with $f_j\in \Theta_{NT,j}$, let
\begin{equation}\label{eq:objective:function}
l_{NT}(f)=\frac{1}{NT}\sum_{i=1}^N\sum_{t=1}^T\bigg[Y_{it}-f(\bfX_{it})-\frac{1}{T}\sum_{s=1}^T\bigg(Y_{is}-f(\bfX_{is})\bigg)\bigg]^2+\sum_{j=1}^pp_{\lambda_{NT}}\bigg(\|f_{j,\sim}\|_{NT}\bigg),
\end{equation}
where $f_{j,\sim}$ is the nonlinear component of $f_j$ as defined in (\ref{eq:linear:nonlinear:definition}), and  $p_{\lambda_{NT}}$ is a given penalty function with tuning parameter $\lambda_{NT}$. The penalized estimator is defined as the minimizer of (\ref{eq:objective:function}), {namely}, 
\begin{equation}
	\widehat{f}=\Argmin_{f\in \Theta_{NT}}l_{NT}(f).\label{eq:penalized:estimator}
\end{equation}
There are {several possible} choices {for the functional form of the penalty term} $p_{\lambda_{NT}}$.  To name a few,  Ridge penalty for $p_{\lambda_{NT}}(z)=\lambda_{NT}z^2$,  Lasso penalty \citep{t96} for $p_{\lambda_{NT}}(z)=\lambda_{NT}|z|$, and Smoothly Clipped Absolute Deviation (SCAD) penalty \citep{fl01} for $p_{\lambda_{NT}}$ with first order derivative
\begin{align}
	p'_{\lambda_{NT}}(z)=\lambda_{NT}I(z\leq \lambda_{NT})+\frac{(\kappa\lambda_{NT}-z)_+}{\kappa-1}I(z>\lambda_{NT}),\label{eq:derivative:scad}
\end{align} 
where $\kappa>2$ is some predetermined constant. In general, with larger $\lambda_{NT}$, the penalty function $p_{\lambda_{NT}}$ will be larger and thus (\ref{eq:objective:function}) will  {tend}  to shrink the nonlinear components $f_{j, \sim}$'s. When compared  to other penalties, the solution via SCAD penalty simultaneously enjoys three desirable properties, i.e., unbiasedness, sparsity, and continuity, see \cite{fl01} for {a} detailed discussion. Therefore, throughout this paper, we will consider $p_{\lambda_{NT}}$ as SCAD penalty, and extension to other types of penalties are left as future work.


\subsection{Computational Algorithm}
{In this section} we propose a local quadratic approximation algorithm to solve optimization problem in (\ref{eq:penalized:estimator}). For each $j=1,\ldots, p$, let $\psi_{j, 1}, \psi_{j, 2}, \ldots, \psi_{j, r_j}, \widetilde{\psi}_{j, 1}, \ldots, \widetilde{\psi}_{j, M_j-1}$ be the centralized spline basis and for any $f_j \in \Theta_{NT,j}$, it follows that $f_j(z)=f_{j,-}(z)+f_{j, \sim}(z)$, with $f_{j,-}(z)=v_j\psi_{j,1}(z) \textrm{ and } f_{j,\sim}=u_j^\top\bfPsi_{j,\sim}(z),$
for some $v_j \in \mathbb{R}$, $u_{j}\in \mathbb{R}^{M_j+r_j-2}$, and all $z\in [0,1]$. Here $\bfPsi_{j,\sim}(z)=( \psi_{j, 2}(z), \ldots, \psi_{j, r_j}(z), \widetilde{\psi}_{j, 1}(z), \ldots, \widetilde{\psi}_{j, M_j-1}(z))^\top$ is a $(M_j+r_j-2)$-dimensional vector of functions. Furthermore, for each $j \in [p]$, we define vectors 
\begin{align*}
	\bfB_{j,\sim}&=(\bfPsi_{j,\sim}(Z_{11j}), \bfPsi_{j,\sim}(Z_{12j}),\ldots, \bfPsi_{j,\sim}(Z_{itj}), \ldots, \bfPsi_{j,\sim}(Z_{NTj}))^\top\in \mathbb{R}^{NT\times(M_j+r-2)},\\
	v&=(v_1, \ldots, v_p)^\top \in \mathbb{R}^p, \;\bfY=(Y_{11}, Y_{12},\ldots, Y_{it},\ldots, Y_{NT})^\top \in \mathbb{R}^{NT}
\end{align*}
and matrices
\begin{align*}
	\bfB_{-}&=\begin{pmatrix}
	\psi_{1,1}(Z_{111})& \psi_{2,1}(Z_{112})&\ldots & \psi_{p,1}(Z_{11p})\\
	\psi_{1,1}(Z_{121})& \psi_{2,1}(Z_{122})&\ldots & \psi_{p,1}(Z_{12p})\\
	\vdots&\vdots&\vdots&\vdots\\
	\psi_{1,1}(Z_{it1})& \psi_{2,1}(Z_{it2})&\ldots & \psi_{p,1}(Z_{itp})\\
	\vdots&\vdots&\vdots&\vdots\\
	\psi_{1,1}(Z_{NT1})& \psi_{2,1}(Z_{NT2})&\ldots & \psi_{p,1}(Z_{NTp})\\
	\end{pmatrix}\in \mathbb{R}^{NT\times p},\\
	H&=I_T-\frac{1}{T}uu^\top \in \mathbb{R}^{T\times T},\; \textrm{ with } u=(1,1,\ldots, 1)^\top\in \mathbb{R}^T,\;\; M_H=I_N\otimes H\in \mathbb{R}^{NT \times NT}.
\end{align*}
By {using the} above notation, it is not difficult to verify the  following equalities,
\begin{align}
	\|f_{j,\sim}\|_{NT}^2=\frac{1}{NT}u_j^\top \bfB_{j,\sim}^\top M_H\bfB_{j,\sim}u_j\nonumber
\end{align}
and
\begin{align}
	l_{NT}(f)=\frac{1}{NT}\bigg(\bfY-\bfB_{-}v-\sum_{j=1}^p\bfB_{j,\sim}u_j\bigg)^\top M_H \bigg(\bfY-\bfB_{-}v-\sum_{j=1}^p\bfB_{j,\sim}u_j\bigg)\nonumber\\
	+\sum_{j=1}^pp_{\lambda_{NT}}\bigg(\sqrt{\frac{1}{NT}u_j^\top \bfB_{j,\sim}^\top M_H\bfB_{j,\sim}u_j}\bigg).\label{eq:adapted:optimization}
\end{align}
Therefore, the optimization problem in (\ref{eq:penalized:estimator}) is adapted to the optimization problem in (\ref{eq:adapted:optimization}), which {is reduced to finding the}  corresponding minimizer $v$ and $u_j$'s. {As in}  \cite{fl01}, {we will also use quadratic functions}  to approximate the penalty terms in (\ref{eq:adapted:optimization}). Note that 
\begin{align}
	\frac{\partial p_{\lambda}\bigg(\sqrt{\frac{1}{NT}u^\top \bfB_{j,\sim}^\top M_H\bfB_{j,\sim}u}\bigg)}{\partial u}=\frac{\sqrt{NT}p'_{\lambda}\bigg(\sqrt{\frac{1}{NT}u^\top \bfB_{j,\sim}^\top M_H\bfB_{j,\sim}u}\bigg)\bfB_{j,\sim}^\top M_H\bfB_{j,\sim}u}{\sqrt{u^\top \bfB_{j,\sim}^\top M_H\bfB_{j,\sim}u}},\nonumber
\end{align}
provided $u^\top \bfB_{j,\sim}^\top M_H\bfB_{j,\sim}u>0$. Therefore,  if $u \approx u^0$, Taylor expansion leads to
\begin{align*}
	&p_{\lambda}\bigg(\sqrt{\frac{1}{NT}u^\top \bfB_{j,\sim}^\top M_H\bfB_{j,\sim}u}\bigg)\\
	\approx& p_{\lambda}\bigg(\sqrt{\frac{1}{NT}u^{0\top} \bfB_{j,\sim}^\top M_H\bfB_{j,\sim}u^0}\bigg)+D_j(u^0)u^\top\bfB_{j,\sim}^\top M_H\bfB_{j,\sim}(u-u^0)\\
	\approx& p_{\lambda}\bigg(\sqrt{\frac{1}{NT}u^{0\top} \bfB_{j,\sim}^\top M_H\bfB_{j,\sim}u^0}\bigg)+D_j(u^0)\bigg[u^{\top}\bfB_{j,\sim}^\top M_H\bfB_{j,\sim}u-u^{0\top}\bfB_{j,\sim}^\top M_H\bfB_{j,\sim}u^0 \bigg],
\end{align*}
with $D_j(u^0)=\sqrt{NT}p'_{\lambda}\left(\sqrt{\frac{1}{NT}\smash[b]{u^{0\top} \bfB_{j,\sim}^\top M_H\bfB_{j,\sim}u^0}}\right)\left(u^{0\top} \bfB_{j,\sim}^\top M_H\bfB_{j,\sim}u^0\right)^{-1/2}$ and provided $D_j(u^0)$ exists. As a consequence, if $u_j\approx u_j^0$ for all $j=1,\ldots, p$, (\ref{eq:adapted:optimization}) can be locally approximated, up to a constant, by 
\begin{align}
\frac{1}{NT}\bigg(\bfY-\bfB_{-}v-\sum_{j=1}^p\bfB_{j,\sim}u_j\bigg)^\top M_H \bigg(\bfY-\bfB_{-}v-\sum_{j=1}^p\bfB_{j,\sim}u_j\bigg)+\sum_{j=1}^pD_j(u^0_j)u^{\top}_j\bfB_{j,\sim}^\top M_H\bfB_{j,\sim}u_j.\nonumber
\end{align}
From above equation, we summarize the proposed algorithm below.
\begin{enumerate}[label=(\alph*),ref=(\alph*)]
\item \label{algm:step:1} Choose initial values $(fv^{(0)}, u_1^{(0)}, \ldots,  u_p^{(0)})$.
\item \label{algm:step:2} In the $s$-th iteration, solve following optimization problem:
\begin{align}
	(v^{(s+1)}, u_1^{(s+1)}, \ldots,  u_p^{(s+1)})=&\argmin_{v, u_1, \ldots, u_p} \bigg(\bfY-\bfB_{-}v-\sum_{j=1}^p\bfB_{j,\sim}u_j\bigg)^\top M_H \bigg(\bfY-\bfB_{-}v-\sum_{j=1}^p\bfB_{j,\sim}u_j\bigg)\nonumber\\
	&+NT\sum_{j=1}^pD_j(u^{(s)}_j)u^{\top}_j\bfB_{j,\sim}^\top M_H\bfB_{j,\sim}u_j. \label{eq:quadratic:approximation}
\end{align}
\item \label{algm:step:3}  Repeat \ref{algm:step:2} until the difference between $(v^{(s)}, u_1^{(s)}, \ldots,  u_p^{(s)})$ and $(v^{(s+1)}, u_1^{(s+1)}, \ldots,  u_p^{(s+1)})$ is small enough.
\end{enumerate}
\begin{Remark}
It is  {worthwhile} mentioning that the optimization {problem}  in (\ref{eq:quadratic:approximation}) is a ridge-type regression problem, which can significantly reduces the {computatioal}  complexity. For convergence analysis {of the}  proposed algorithm, we refer the readers to \cite{x09} and \cite{hl05}.
\end{Remark}

\section{Asymptotic Theory}\label{sec:asympotics}
In this section we present several asymptotic results concerning our proposed procedure for both short panel (fixed $T$) and large panel (diverging $T$). However, before proceeding further, we {remind the readers the Holder-smoothness notion of a function}. An univariate function $f:[0,1]\to \mathbb{R}$ is said to be \textit{$m$-smooth}, if $m=r+\delta$, for some $0<\delta\leq 1$ and integer $r$ such that $f$ is $r$-times continuously differentiable and $|f^{(r)}(u)-f^{(r)}(v)|\leq c|u-v|^\delta$ for some $c>0$ and all $u, v\in [0,1]$. {Additionally, in the sequel, we use the following notation.} {We let} $q_i(\bfw)$ be the density function of $\bfW_i=(\bfX_{i1}, \ldots, \bfX_{iT})$ and $\pi_i(\bfx)$ be the density function of $\bfX_{i1}$. For a function $g: [0,1]^k \to \mathbb{R}$, we define $\|g\|_2^2=\int g^2(\bfu)d\bfu-[\int g(\bfu)d\bfu]^2$ whenever the integrals exist. {Finally we set} $\mathbb{Z}=(\bfW_1, \bfW_2,\ldots, \bfW_N)$ and $\bfepsilon=(\epsilon_{11},\epsilon_{12},\ldots, \epsilon_{it},\ldots, \epsilon_{NT})^\top \in \mathbb{R}^{NT}$.

\subsection{Consistency}
The main results of this section show that  the proposed penalized estimator is consistent in terms of both estimation and linearity detection. However, these results require some  technical conditions, which are stated as follows.
\begin{Assumption}\label{Assumption:A1}
\begin{enumerate}[label={(\roman*}),ref={(\roman*})]
\item \label{A1:a} $T$ is a fixed constant.
\item \label{A1:b} For some $a_1>1$ , it satisfies that $a_1^{-1}\leq q_i(\bfw)\leq a_1$ for all $i=1,\ldots, N$ and all $\bfw \in [0, 1]^{pT}$. 
\end{enumerate}
\end{Assumption}
\begin{Assumption}\label{Assumption:A2}
\begin{enumerate}[label={(\roman*}),ref={(\roman*})]
\item  \label{A2:a}$T$ is diverging.
\item \label{A2:b} For some $a_3>1$ and $0\leq a_4<1$, it satisfies that $a_3^{-1}\leq \pi_i(\bfx)\leq a_3$ for all $i=1,\ldots, N$ and all $\bfx \in \mcX$. For each $i$,  $\{\bfX_{i1},\ldots, \bfX_{iT}\}$ is a stationary alpha-mixing  sequence with alpha mixing coefficient $\alpha_{[i]}(t)\leq a_4^t$ for all $t\geq 0$.
\end{enumerate}
\end{Assumption}
\begin{Assumption}\label{Assumption:common}
\begin{enumerate}[label={(\roman*}),ref={(\roman*})]
\item \label{Ac:a1} $\{\bfW_i, i=1,\ldots, N\}$ are independent across $i$.
\item \label{Ac:a2} There exist $a_2>1$ such that the eigenvalues $\ev(\bfepsilon\bfepsilon^\top|\mathbb{X})$ are in $[a_2^{-1}, a_2]$ and $\ev(\epsilon_{it}|\bfX_{it})=0$ for all $i=1,\ldots, N$ and $t=1,\ldots, T$.
\item \label{Ac:a} $f_0(\bfx)=\sum_{j=1}^p f_{j,0}(z_j)$ such that
\begin{enumerate}[label={(\alph*})]
\item  $\int_0^1 f_{j,0}(z)dz=0, \textrm{ for } j=1,\ldots, p$. 
\item  For some constant $a_6>0$ and $\beta_{1,0}, \ldots, \beta_{d,0}\in \mathbb{R}$ that
\begin{eqnarray*}
	&& f_{j,0}(z)=\beta_{j,0}(z-1/2) \;\;\textrm{ for } j=1,2,\ldots, d,\\
	&&\int_0^1|f_{j,0}(z)-\beta(z-1/2)|^2dz\geq a_6\;\; \textrm{ for all } \beta \in \mathbb{R} \;\textrm{ and for }\; j=d+1,\ldots, p.
\end{eqnarray*}
\item For each $j=d+1, \ldots, p$, $f_{j,0}$ is  $m_j$-smooth for some constant $m_j>1$.
\end{enumerate}

\item \label{Ac:b} There exists $a_7>0$ such that, for all $j\in [p]$, the bandwidth of knots $\bft_{j, M_j}$ satisfies 
\begin{align*}
\frac{\max_{1\leq i\leq M_j}(t_{j, i}-t_{j, i-1})}{\min_{1\leq i\leq M_j}(t_{j, i}-t_{j, i-1})}\leq a_7.
\end{align*}
\item \label{Ac:d} The degree of centralized spline space $\textrm{CSpl}(r_j, \bft_{j, M_j})$  satisfies that
\begin{align*}
r_j \geq \begin{cases}
1 & \textrm{ for } j=1,\ldots, d\\
\floor{m_j} & \textrm{ for } j=d+1,\ldots, p
\end{cases}.
\end{align*} 
\end{enumerate}
\end{Assumption}
\begin{Remark}
Assumption \ref{Assumption:A1}.\ref{A1:a} is the classical setting for short panel. \ref{Assumption:A1}.\ref{A1:b} imposes a quasi-uniformity condition on the density $q_i$, with the correlation among explanatory variables $Z_{it1},\ldots, Z_{itp}$ and  the dependence among $\bfX_{i1}, \ldots, \bfX_{iT}$ along the time dimension being jointly controlled by $a_1$. Similar assumptions are also proposed by \cite{h98} and \cite{h03}. Assumption \ref{Assumption:A2}.\ref{A2:a} allows $T$ is diverging, which is the standard setting for large panel. In the case of diverging $T$, Assumption \ref{Assumption:A2}.\ref{A2:b} requires the sequence  $\bfX_{i1}, \ldots, \bfX_{iT}$ is stationary for each $i$. Moreover, the correlation among explanatory variables $Z_{it1},\ldots, Z_{itp}$ is characterized by the quasi-uniform assumption on $\pi_i$, while the weak dependence for the observations along the time dimension is controlled by a geometric $\alpha$-mixing coefficient sequence. A similar $\alpha$-mixing condition can be found in \cite{ssp16}, \cite{sj17}, and \cite{sc13}. The stationarity assumption in Assumption \ref{Assumption:A2}.\ref{A2:b} can be relaxed at a  cost of introducing more notation. 
\end{Remark}
\begin{Remark}
Assumption \ref{Assumption:common}.\ref{Ac:a1} requires the explanatory variables to be independent across $i$. This is only for mathematical convenience, and we  can relax this assumption to conditional independence given fixed effects $\alpha_1, \ldots, \alpha_N$.  Assumption \ref{Assumption:common}.\ref{Ac:a2} assumes that $\bfX_{it}$ is exogenous and allows cross-sectional dependence on the error terms. Our method also can be extended to the case when $a_2$ tends to infinity slowly. In particular,  if for each $i$, $\{\epsilon_{i1}, \ldots, \epsilon_{iT}\}$ is a martingale difference sequence and $(\epsilon_{i1}, \ldots, \epsilon_{iT})$'s are mutually independent across $i$, then the eigenvalues condition  will be satisfied provided $\textrm{Var}(\epsilon_{it}) \in [a_2^{-1}, a_2]$ for all $i$ and $t$. Assumption \ref{Assumption:common}.\ref{Ac:a} imposes three conditions on the underlying regression function $f_0$, (a) Identification conditions of $f_{j,0}$'s; (b) Identification conditions of linearity; {and} (c) Smoothness conditions on  $f_{j,0}$'s. The identification conditions  of $f_{j,0}$'s are different from the classical ones in \cite{h98} for sectional data and \cite{sj12} for panel data. However, its validity can be guaranteed by mild conditions, see Lemmas \ref{lemma:expectation:sample:variance} and \ref{lemma:expectation:sample:variance:large:T} in Appendix. The identification conditions of linearity specifies the function form of each $f_{j,0}$. In particular, we requires the difference between nonlinear component and arbitrary linear function has a fixed and strictly positive lower bound $a_6$. With more cumbersome calculation, this lower bound is allowed {to decrease}  slowly to zero. The $m_j$-smoothness assumption is standard for nonparametric regression problem to reduce the model complexity, see \cite{c07}, \cite{s94}. Assumption \ref{Assumption:common}.\ref{Ac:b} and \ref{Assumption:common}.\ref{Ac:d} are common regular conditions on knots and degree in spline regression literature, which provide theoretical assurances for a good approximation to smooth functions, see \cite{zsw98} and \cite{h98}. It is worth mentioning that for $j=1,\ldots, d$, each $f_{j,0}$ is exactly a linear function, and a spline with degree $r_j\geq 1$ will be adequate to perform good approximation.
\end{Remark}

For each $j=1,\ldots, p$, let $h_j$ be the maximal length between two successive points of knots $\bft_{j, M_j}$, i.e., $h_j=\max_{1\leq i\leq M_j}(t_{j, i}-t_{j, i-1})$. Under Assumption \ref{Assumption:common}.\ref{Ac:b}, it follows that $h_j \asymp M_j^{-1}$. Theorem \ref{thm:rate:of:convergence:together} below proves that $m_j$'s and $h_j$'s play critical roles in 
the rate of convergence of the proposed estimator $\widehat{f}$.
\begin{theorem}\label{thm:rate:of:convergence:together}
Suppose $\lambda_{NT}\to 0$ and either one of the following conditions holds:
\begin{enumerate}[label=(\alph*)]
\item Assumptions \ref{Assumption:A1}, \ref{Assumption:common} are valid and  $\sum_{j=d+1}^p h_j=o(1)$, $\sum_{j=1}^p h_j^{-2}=o(N)$;
\item Assumptions \ref{Assumption:A2}, \ref{Assumption:common} are valid and $\sum_{j=d+1}^p h_j=o(1)$, $\sum_{j=1}^ph_j^{-2}=o(N)$,  $\sum_{j=1}^ph_j^{-1}=o(T)$.
\end{enumerate}
Then it follows that
\begin{align*}
	\|\widehat{f}-f_0\|^2=O_P\bigg(\sum_{j=1}^p\frac{1}{NTh_j}+\sum_{j=d+1}^ph_j^{2m_j}\bigg)\; \textrm{ and }\;\;\|\widehat{f}-f_0\|_2^2=O_P\bigg(\sum_{j=1}^p\frac{1}{NTh_j}+\sum_{j=d+1}^ph_j^{2m_j}\bigg).
\end{align*}
\end{theorem}
Theorem \ref{thm:rate:of:convergence:together} states that the rate of convergence $\widehat{f}$ consists of two parts, namely, estimation error $\sum_{j=1}^p(NTh_j)^{-1}$ and approximation error $\sum_{j=d+1}^ph_j^{2m_j}$, which coincides with standard result in  \cite{h98} and \cite{h03}. It should be noted that for linear components, namely $j=1,\ldots, d$, the approximation error does not involve in the $O_P$ term. On the other hand, for the nonparametric parts,  the rate of convergence can benefit from balancing the estimation and the  approximation errors. Specifically,  if $h_j \asymp N^{-\frac{1}{2m_j+1}}$ for $j=d+1, \ldots, p$, the rate of convergence improves. It should be observed that the convergence still holds even if $h_j\asymp 1$ for $j=1,\ldots, d$ and by doing so, the rate of convergence can be further improved. The choice of $h_j$ with constant order means the number of knots $M_j$ is not diverging. Since the first $d$ components are linear, setting the corresponding $h_j$'s to be constant does not ruin the estimation consistency. However, this is usually infeasible in practice, as the prior information about the linearity of the explanatory variables is typically unavailable. Furthermore, Theorem \ref{thm:rate:of:convergence:together} directly shows that the global minimizer $\widehat{f}$ is consistent, while previous work about SCAD penalized regression only establishes  the existence of a consistent local minimizer, e.g., see \cite{fl01} and \cite{x09}. 

Theorem \ref{thm:rate:of:convergence:together} only addresses the issue for estimation, which is not adequate to distinguish the linear components from the nonlinear ones. While with appropriate choice of tuning parameter $\lambda_{NT}$, Theorem \ref{thm:selection:consistency:together} below proves that the estimator $\widehat{f}$ will automatically recover the linearity in the underlying regression function $f_0$.
\begin{theorem}\label{thm:selection:consistency:together}
Suppose $\lambda_{NT}\to 0$ and either one of the following conditions is satisfied:
\begin{enumerate}[label=(\alph*)]
\item Assumptions \ref{Assumption:A1}, \ref{Assumption:common} hold and  $\;\sum_{j=d+1}^p h_j^{2m_j}=o(\lambda_{NT}^2)$, $\sum_{j=1}^p h_j^{-1}=o(N\lambda_{NT}^2)$, $\sum_{j=1}^ph_j^{-2}=o(N)$;
\item Assumptions \ref{Assumption:A2}, \ref{Assumption:common} hold and $\;\sum_{j=d+1}^ph_j^{2m_j}=o(\lambda_{NT}^2)$, $\sum_{j=1}^p h_j^{-1}=o(NT\lambda_{NT}^2)$, $\sum_{j=1}^ph_j^{-2}=o(N),  \sum_{j=1}^ph_j^{-1}=o(T)$.
\end{enumerate}
Then with probability approaching one, the following holds:
\begin{align*}
	\widehat{f}_{j,\sim}=0\textrm{ for } j=1,2,\ldots,d, \quad\textrm{ and }\quad \widehat{f}_{j,\sim}\neq 0\textrm{ for } j=d+1,\ldots, p.
\end{align*}
\end{theorem}
The tuning parameter $\lambda_{NT}$ in Theorem \ref{thm:selection:consistency:together} (unlike in Theorem \ref{thm:rate:of:convergence:together}) can neither be too large nor too small. With suitable choices of $\lambda_{NT}$ and $h_j$'s, the proposed estimator $\widehat{f}=\sum_{j=1}^p\widehat{f}_j$ will automatically and correctly specify the  linear and nonlinear forms with probability approaching one. Since in Theorem \ref{thm:selection:consistency:together}, the tuning parameters  $h_j$'s and $\lambda_{NT}$ play important roles in selection consistency, a fundamental  issue in practice  is the choice of these parameters. The discussion of this issue is deferred to Section \ref{sec:choice:turning:parameter}.

\subsection{Solution Path}
{In this section we define the} solution path of $\widehat{f}$ and provide its theoretical properties and practical implications. {For fixed} knots $\bft_{j, M_j}$'s and {the} tuning parameters $k_j$'s and $h_j$'s, one can obtain a sequence of estimators  $\widehat{f}$ by using a sequence of increasing $\lambda_{NT}$'s and these estimators forms a solution path. For sufficiently large $\lambda_{NT}$, all the nonlinear components $\widehat{f}_{j,\sim}$'s will vanish and result in a model consisting of all linear components.  On the other hand, when $\lambda_{NT}$ is close to zero, all the $\widehat{f}_j$'s will be nonlinear. {Consequently, we may} obtain $p+1$ different models in the solution path by increasing $\lambda_{NT}$ from zero to infinity. The following corollary is a direct consequence of Theorem \ref{thm:selection:consistency:together}.
\begin{corollary}\label{corollary:selection:consistency:soluton:path}
Suppose $\lambda_{NT}\to 0$ and either one of the following conditions is satisfied:
\begin{enumerate}[label=(\alph*)]
\item Assumptions \ref{Assumption:A1}, \ref{Assumption:common} hold and  $\;\sum_{j=d+1}^p h_j^{2m_j}=o(\lambda_{NT}^2)$, $\sum_{j=1}^p h_j^{-1}=o(N\lambda_{NT}^2)$, $\sum_{j=1}^ph_j^{-2}=o(N)$;
\item Assumptions \ref{Assumption:A2}, \ref{Assumption:common} hold and $\;\sum_{j=d+1}^ph_j^{2m_j}=o(\lambda_{NT}^2)$, $\sum_{j=1}^p h_j^{-1}=o(NT\lambda_{NT}^2)$, $\sum_{j=1}^ph_j^{-2}=o(N),  \sum_{j=1}^ph_j^{-1}=o(T)$.
\end{enumerate}
Then with probability approaching one, one model contained in the solution path will correctly specify all the linear components.
\end{corollary}
Corollary \ref{corollary:selection:consistency:soluton:path} indicates that the solution path is consistent  in the sense that, one in the $p+1$ models will correctly identify both the  linear and the  nonlinear parts. Notice that for linearity detection problem, {one essentially needs}  to identify the correct model out of $2^p$ candidates. Another {immediate}  implication from Corollary \ref{corollary:selection:consistency:soluton:path}  is that in practice, any model selection method, e.g., Akaike Information Criterion (AIC) or Bayesian Information Criterion (BIC) criteria,  based on these $p+1$ models is valid, reliable and is equivalent to that based on $2^p$ models, which is a significant reduction on  model complexity.

\subsection{Joint Asymptotic Distribution}
In this section, we will present the limit distribution of proposed estimator $\widehat{f}$. To proceed further, recall that $\bfPsi_{j,\sim}(z)=( \psi_{j, 2}(z), \ldots, \psi_{j, r_j}(z), \widetilde{\psi}_{j, 1}(z), \ldots, \widetilde{\psi}_{j, M_j-1}(z))^\top$  is the basis of $\Theta_{NT,j,\sim}$, for $j=1,\ldots, p$. We  further define $\bfPsi_{j,-}(z)=\psi_{j,1}(z)=z-1/2$, $\bfPsi_j(z)=(\bfPsi_{j,-}(z), \bfPsi{j,\sim}^\top(z))^\top$, and $\bfPsi^0(\bfx)=(\bfPsi_{1,-}(z_1),\ldots, \bfPsi_{d,-}(z_d), \bfPsi_{d+1}^\top(z_{d+1}),\ldots, \bfPsi^\top_{p}(z_p))^\top$. By this definition, we know $\bfPsi_{j}(z)$ is the basis of $\Theta_{NT,j}$. If we define the space of correctly specified model
\begin{align}
	\Theta_{NT}^0=\bigg\{f(\bfx)=\sum_{j=1}^pf_j(z_j)\in \Theta_{NT}\;\bigg|\; f_j(z)=\beta_j(z-1/2) \textrm{ for } \beta_j\in \mathbb{R} \textrm{ and } j=1,\ldots,d \bigg\},\nonumber
\end{align}
then $\bfPsi^0(z)$ will be its basis. By Theorems \ref{thm:selection:consistency:together}, it follows that $\widehat{f}\in \Theta_{NT}^0$ with probability approaching one, and thus we have the following expression for the proposed estimator:
\begin{align*}
	\widehat{f}(\bfx)=\sum_{j=1}^d\widehat{\beta}_j(z_j-1/2)+\sum_{j=d+1}^p\widehat{f}_{j}(z_j), \;\textrm{ with }\; \widehat{f}_j \in \Theta_{NT,j}\; \textrm{ for }\;j=d+1,\ldots, p.
\end{align*}
Therefore, it is natural for us to study the asymptotic distributions of $\widehat{\beta}_j$'s and $\widehat{f}_j(z_{j,0})$'s  with $z_{d+1,0}, \ldots, z_{p, 0}\in [0, 1]$ being some fixed constants. We consider following elements in $\Theta_{NT}^0$:
\begin{align}
	v_{NT,j}^*(\bfx)=c_{j}^{*\top}V^{-1}\bfPsi^0(\bfx)\quad\textrm{ and }\quad \widehat{v}_{NT,j}^*(\bfx)&=c_{j}^{*\top}V_{NT}^{-1}\bfPsi^0(\bfx), \quad\textrm{ for } j=1,\ldots, p,\label{eq:definition:hat:vnt}
\end{align} 
with 
\begin{align*}
V_{NT}&=\frac{1}{NT}\sum_{i=1}^N\sum_{t=1}^T \bigg(\bfPsi^0(\bfX_{it})-\frac{1}{T}\sum_{s=1}^T\bfPsi^0(\bfX_{is})\bigg)\bigg(\bfPsi^0(\bfX_{it})-\frac{1}{T}\sum_{s=1}^T\bfPsi^0(\bfX_{is})\bigg)^\top,\quad V=\ev(V_{NT}),\nonumber\\
	c_j^*&=(\underbrace{0,\ldots,0}_{j-1}, 1, 0,\ldots, 0)^\top \in \mathbb{R}^{d+\sum_{k=d+1}^p(M_k+r_k-1)},\quad\textrm{ for }\quad j=1,\ldots, d,\nonumber\\
	c_j^*&=(\underbrace{\!\! 0, \ldots, 0 \!\!}_{\text{\makebox[36pt]{$d+\sum_{k=d+1}^{j-1}(M_k+r_k-1)$}}}, \bfPsi_j^\top(z_{j,0}),0, \ldots, 0)^\top \in \mathbb{R}^{d+\sum_{k=d+1}^p(M_k+r_k-1)},\quad\textrm{ for }\quad j=d+1,\ldots, p.
\end{align*}
It can be shown that, for any $f(\bfx)=\sum_{j=1}^d\beta_j(z_j-1/2)+\sum_{j=d+1}^p{f}_{j}(z_j)\in \Theta_{NT}^0$ with $f_j\in \Theta_{NT,j}$, $j=d+1,\ldots, p$, the following equality holds:
\begin{align*}
	\langle v_{NT,j}^*, f\rangle=\begin{cases}
	\beta_j & \textrm{ for } j=1, \ldots, d;\\
	f_j(z_{j,0}) & \textrm{ for  } j=d+1,\ldots, p.
	\end{cases}
\end{align*}
If we define linear functionals from $\Theta_{NT}^0$ to $\mathbb{R}$ such that $\mathcal{L}_j(f)=\beta_{j}$ for $j=1,\ldots, d$ and $\mathcal{L}_j(f)=f_j(z_{j,0})$ for $j=d+1,\ldots, p$, then $v_{NT,j}^*$'s are essentially the Riesz representatives of $\mathcal{L}_j$'s.  

In order to establish the asymptotic distribution, more regular assumptions on the error terms $\epsilon_{it}$'s are needed. Thus, in the following, we define the standard deviation inner product and norm in $\Theta_{NT}$, which contains the information of $\bfepsilon$. For $g, f \in \Theta_{NT}$, we define
\begin{align}
	\langle g, f\rangle_\textrm{sd}=\frac{1}{NT}\ev\bigg(\bfg^\top M_H\bfepsilon \bfepsilon^\top M_H \bff \bigg)\quad \textrm{ and }\quad \|g\|^2_\textrm{sd}=\langle g, g\rangle_\textrm{sd},\nonumber
\end{align}
where $\bfg=(g(\bfX_{11}), g(\bfX_{12}),\ldots, g(\bfX_{it}), g(\bfX_{NT}))^\top, \bff=(f(\bfX_{11}), f(\bfX_{12}),\ldots, f(\bfX_{it}), f(\bfX_{NT}))^\top \in \mathbb{R}^{NT}$. In addition, denoting $\bfepsilon_i=(\epsilon_{i1},\ldots, \epsilon_{iT})$ for $i \in [N]$, we propose Assumption \ref{Assumption:A4} on the error terms $\bfepsilon_i$'s and $v_{NT,j}^*$ for statistical inference. 
\begin{Assumption}\label{Assumption:A4}
\begin{enumerate}[label={(\roman*}),ref={(\roman*})]
\item\label{A4:a}  There exists $a_8>0$, such that 
$\sup_{i\in [N]}\sup_{t\in [T]}\ev(\epsilon_{it}^4|\mathbb{Z})\leq a_8$.
\item\label{A4:b}  $(\bfW_i, \bfepsilon_i)$'s are independent across $i$.
\item \label{A4:c} In the case of diverging $T$, for each $i$, $\{(\bfX_{it}, \epsilon_{it}), t\in [T]\}$ is an alpha-mixing  sequence with mixing coefficient $\widetilde{\alpha}_{[i]}(t)\leq a_9^t$ for all $t\geq 0$ and some $0<a_9<1$. 
\end{enumerate}
\end{Assumption}
\begin{Assumption}\label{Assumption:A5}
There exist constants $\sigma_j>0$ and $r_{j,k}$ for $j,k \in \{1,\ldots, p\}$ such that the following convergence conditions hold:
\begin{eqnarray*}
	&&\|v_{{NT},j}^*\|_{\textrm{sd}}^2 \to \sigma_j^2>0, \textrm{ for } j=1,\ldots, d,\quad\quad \|v_{{NT},j}^*\|_{\textrm{sd}}^2h_j \to \sigma_j^2>0 \textrm{ for } j=d+1,\ldots, p,\\
	&&\frac{\langle v_{{NT},j}^*, v_{{NT},k}^*\rangle_{\textrm{sd}}}{	\|v_{{NT},j}^*\|_{\textrm{sd}} 	\|v_{{NT},k}^*\|_{\textrm{sd}}} \to r_{j,k}, \textrm{ for } 1\leq j, k \leq p, \nonumber\\
	&&\Sigma=\begin{pmatrix}
	\sigma_1^2& r_{1,2}\sigma_1\sigma_2 &r_{1,3}\sigma_1\sigma_3&\ldots &r_{1,p}\sigma_1\sigma_p\\ 
	r_{1,2}\sigma_1\sigma_2 & \sigma_2^2 & r_{2, 3}\sigma_2\sigma_3 &\ldots & r_{2,p}\sigma_2\sigma_p\\
	\vdots&\vdots&\vdots&\vdots&\vdots\\
	r_{1,p}\sigma_1\sigma_p &r_{2,p}\sigma_2\sigma_p&r_{3,p}\sigma_3\sigma_p&\ldots & \sigma_p^2
	\end{pmatrix}\in \mathbb{R}^{p\times p} \textrm{ is positive definite}.
\end{eqnarray*}
\end{Assumption}
\begin{Remark}
Assumption \ref{Assumption:A4}.\ref{A4:a} is a stronger  moment condition on the error terms to verify Lyapunov condition.
Assumption \ref{Assumption:A4}.\ref{A4:b}  is the condition for cross-sectional independence, which can be relaxed to be conditional independence given the fixed effects $\alpha_1, \ldots, \alpha_N$, see  \cite{sc13}. Assumption \ref{Assumption:A4}.\ref{A4:c} requires that each individual time series $\{(W_{it}, \epsilon_{it}), t=1,\ldots, T\}$ is alpha-mixing and the level of dependence is controlled by a factor of $a_9$. Assumptions \ref{Assumption:A4}.\ref{A4:a}-\ref{A4:c} are standard conditions in literature, which, e.g., can be found in  \cite{sj12} , \cite{sc13}, and \cite{ls16}. 
\end{Remark}
\begin{Remark}
Assumption \ref{Assumption:A5} is a regular condition to express the covariance matrix of joint asymptotic distribution for $(\widehat{\beta}_1, \ldots, \widehat{\beta}_d, \widehat{f}_{d+1}(z_{j,0}),\ldots, \widehat{f}_{p}(z_{j,0}))$. The marginal asymptotic distribution of each component is still valid without this assumption. Nevertheless, it is verified in Lemmas \ref{lemma:verify:condition:c1:parametric} and  \ref{lemma:verify:condition:c1:parametric}  that $\|v_{NT,j}^*\|_\textrm{sd}^2\asymp 1$ for $j=1,\ldots, d$ and $\|v_{NT,j}^*\|_\textrm{sd}^2\asymp h_j^{-1}$ for $j=d+1,\ldots, p$. Similar conditions also imposed in  \cite{sc13aos} and \cite{cs15} to obtain the joint distribution of parametric and nonparametric components.
\end{Remark}

For presentation purpose, we choose  $h_1=h_2=\ldots=h_p=h$ and define $m_*=\min_{d+1\leq j \leq p}m_j$. Theorem \ref{thm:asymptotic:normal:multiple:together} below states that, with suitable choice of $h$ and $\lambda_{NT}$, we can obtain the asymptotic distribution of  $(\widehat{\beta}_1, \ldots, \widehat{\beta}_d, \widehat{f}_{d+1}(z_{j,0}),\ldots, \widehat{f}_{p}(z_{j,0}))$.

\begin{theorem}\label{thm:asymptotic:normal:multiple:together}
Suppos $\lambda_{NT}\to 0$ and one of the following conditions is satisfied:
\begin{enumerate}
\item Assumptions \ref{Assumption:A1}, \ref{Assumption:common}, \ref{Assumption:A4}, \ref{Assumption:A5} are valid and  $h^{-1}=o(N\lambda_{NT}^2)$, $h^{2m_*}=o(\lambda_{NT}^2)$, $h^{-3}=o(N)$, $h^{2m_*-2}=o(1)$, $Nh^{2m_*}=o(1)$;
\item  Assumptions \ref{Assumption:A2}, \ref{Assumption:common}, \ref{Assumption:A4}, \ref{Assumption:A5} are valid and  $h^{-1}=o(T)$, $h^{-1}=o(NT\lambda_{NT}^2)$, $h^{2m_*}=o(\lambda_{NT}^2)$, $h^{-4}=o(NT)$, $h^{2m_*-3}=o(1)$, $h^{-5}=o(N^2)$, $h^{2m_*-4}T=o(N)$, $NTh^{2m_*}=o(1)$.
\end{enumerate}
Then with probability approaching one, the following holds:
\begin{align*}
	\begin{pmatrix}
	\sqrt{NT}(\widehat{\beta}_1-\beta_{1,0})\\
	\vdots\\
	\sqrt{NT}(\widehat{\beta}_d-\beta_{d,0})\\
	\sqrt{NT}(\widehat{f}_{d+1}(z_{d+1,0})-f_{d+1,0}(z_{d+1,0}))\\
	\vdots\\
	\sqrt{NT}(\widehat{f}_{p}(z_{p,0})-f_{p,0}(z_{p,0}))\\
	\end{pmatrix}\cid \textrm{N}(0, \Sigma),
\end{align*}
where $z_{d+1,0},\ldots, z_{p,0}\in [0, 1]$.
\end{theorem}

Theorem \ref{thm:asymptotic:normal:multiple:together} establishes the joint asymptotic distribution of both the linear and nonlinear components of $\widehat{f}$, which includes estimators with different rate of convergence. \cite{sc13aos}, \cite{cs15} and \cite{dl18} also established similar joint asymptotic results in partially linear model. However, compared with their results, Theorem \ref{thm:asymptotic:normal:multiple:together} does not require the prior knowledge of linearity. The constant $m_*$ is the smallest degree of smoothness among all the $f_{j,0}$'s, which represents the effective smoothness of $f_0$. From Theorem \ref{thm:asymptotic:normal:multiple:together}, a necessary condition is $m_*>1.5$ for short panel and $m_*>2$ for large panel, which requires the underlying regression function needs to be enough smooth. If one is of more interest in the marginal distribution of each $\widehat{f}_j$,  Theorem \ref{thm:asymptotic:normal:marginal:together} below establishes the limit distribution of $\widehat{f}_{j}(z_{j,0})$ without Assumption \ref{Assumption:A5}, where $z_{j,0}\in [0,1]$ is fixed constant for $j=1,\ldots,p$.
\begin{theorem}\label{thm:asymptotic:normal:marginal:together}
Suppos $\lambda_{NT}\to 0$ and one of the following conditions is satisfied:
\begin{enumerate}
\item Assumptions \ref{Assumption:A1}, \ref{Assumption:common}, \ref{Assumption:A4} are valid and  $h^{-1}=o(N\lambda_{NT}^2)$, $h^{2m_*}=o(\lambda_{NT}^2)$, $h^{-3}=o(N)$, $h^{2m_*-2}=o(1)$, $Nh^{2m_*}=o(1)$;
\item  Assumptions \ref{Assumption:A2}, \ref{Assumption:common}, \ref{Assumption:A4} are valid and  $h^{-1}=o(T)$, $h^{-1}=o(NT\lambda_{NT}^2)$, $h^{2m_*}=o(\lambda_{NT}^2)$, $h^{-4}=o(NT)$, $h^{2m_*-3}=o(1)$, $h^{-5}=o(N^2)$, $h^{2m_*-4}T=o(N)$, $NTh^{2m_*}=o(1)$.
\end{enumerate}
Then with probability approaching one, the following holds:
\begin{align*}
	\frac{\sqrt{NT}(\widehat{f}_j(z_{j,0})-f_{j,0}(z_{j,0}))}{\|v_{NT,j}^*\|_{\textrm{sd}}}\cid \textrm{N}(0, (z_{j,0}-1/2)^2), \;\;\textrm{ for }\;\; j=1,\ldots, d,
\end{align*}
and 
\begin{align*}
	\frac{\sqrt{NT}(\widehat{f}_j(z_{j,0})-f_{j,0}(z_{j,0}))}{\|v_{NT,j}^*\|_{\textrm{sd}}}\cid \textrm{N}(0, 1),\;\;\textrm{ for }\;\; j=d+1,\ldots, p.
\end{align*}
where $z_{d+1,0},\ldots, z_{p,0}\in [0, 1]$.
\end{theorem}

The choice of homogeneous $h_j$'s in Theorems \ref{thm:asymptotic:normal:multiple:together}  and \ref{thm:asymptotic:normal:marginal:together}  is not only simple for presentation, but also it is  practically convenient. As discussed in Section \ref{sec:choice:turning:parameter}, homogeneous $h_j$'s will reduce the complexity of tuning parameter selection. For theoretical interest, we include the case of heterogeneous $h_j$'s in Appendix. 

\begin{Remark}
To apply Theorems \ref{thm:asymptotic:normal:multiple:together} and \ref{thm:asymptotic:normal:marginal:together}, one needs to estimate the unknown variance. We use the estimator proposed in \cite{sj12} to estimate the variance in the presence of heteroskedasticity and autocorrelation. To be more specific, for functions $u$ and $v$,  we define
\begin{eqnarray*}
	\widehat{\bfepsilon}_{i}=H(\bfY_i-\widehat{\bff}_i),\quad S_{ij}=\frac{1}{T}\sum_{t=j+1}^T u(\bfZ_{it})v(\bfZ_{i,t-j})\widehat{\epsilon}_{it}\widehat{\epsilon}_{i,t-j}\quad\textrm{ and }\quad S_i=S_{i0}+2\sum_{j=1}^{l_T}k_{Tj}S_{ij}.
\end{eqnarray*}
where $\widehat{\bff}_i=(\widehat{f}(\bfZ_{i1}),\ldots, \widehat{f}(\bfZ_{iT}))^\top$, $l_T$ is the window size, $k_{Tj}$ is a weight function such that $\sup_{j}|k_{Tj}|<\infty$ and $\lim_{T\to \infty}|k_{Tj}|=1$ for each $j$, and $\widehat{\epsilon}_{it}$ is the $t$-th element of $\widehat{\bfepsilon}_{i}$. By above notation, $\langle u, v\rangle_{\textrm{sd}}$ can be estimated by 
\begin{equation*}
	\widehat{\langle u, v\rangle}_\textrm{sd}=\frac{1}{N}\sum_{i=1}^N S_i.
\end{equation*}
Therefore, the unknown quantity $\langle v_{NT,j}^*, v_{NT,k}^*\rangle_{\textrm{sd}}$ can be estimated by
\begin{equation*}
	\widehat{\langle \widehat{v}_{NT,j}^*, \widehat{v}_{NT,k}^*\rangle}_{\textrm{sd}},
\end{equation*}
where $\widehat{v}_{NT,j}^*$ and  $\widehat{v}_{NT,k}^*$ are defined in (\ref{eq:definition:hat:vnt}).
\end{Remark}

\section{Practical Choice of Tuning Parameters}\label{sec:choice:turning:parameter}
{In this section we discuss}  how to determine tuning parameters $r_j$'s, $h_j$'s and $\lambda_{NT}$. {Motivated by two different objectives,}  we propose two distinct strategies to select $\lambda_{NT}$ for estimation and {for} linearity detection. For convenience, we simply choose each of the knots $\bft_{j, M_j}$, to be an uniform partition of $[0, 1]$ in practice when $h_j$'s are determined.



\subsection{Cross Validation}
Before proceeding further, we formally define $k$-fold cross validation procedure in the framework of panel data. Given positive integer $k\geq 2$, $N$ individuals are randomly separated into $k$ disjointed groups and let $I_1, \ldots, I_k$ be the corresponding sets of indexes with $N_1,\ldots, N_k$ elements, respectively. By this notation, it follows that $I_1,\ldots, I_k$ is a partition of $\{1,\ldots, N\}$. Moreover, we denote $I_s^c$ as the compliment of $I_s$ for $s=1,\ldots, k$ and $\theta=(r_1,\ldots, r_p, h_1,\ldots, h_p, \lambda_{NT})$ as the tuning parameters. Given $\theta$, we set $\widehat{f}_{I_s^c,\theta}$ to be the penalized estimator based on observations $\{(\bfY_i, \bfW_i), i\in I_s^c\}$ and tuning parameter $\theta$. The cross validation value is defined as follows,
\begin{align}
	\textrm{CV}(\theta)&=\textrm{CV}(r_1,\ldots, r_p, h_1,\ldots, h_p, \lambda_{NT})\nonumber\\
	&=\sum_{s=1}^k\frac{1}{N_sT}\sum_{i\in I_s}\sum_{t=1}^T\bigg[Y_{it}-\widehat{f}_{I_s^c,\theta}(\bfX_{it})-\frac{1}{T}\sum_{l=1}^T\bigg(Y_{il}-\widehat{f}_{I_s^c,\theta}(\bfX_{il})\bigg)\bigg]^2.\label{eq:definition:cv}
\end{align}
Based on (\ref{eq:definition:cv}), the optimal tuning parameter ${\theta}_{\textrm{opt}}$ is defined as the minimizer of $\textrm{CV}(\theta)$  among several candidates, i.e., 
\begin{align}
	{\theta}_{\textrm{opt}}=\argmin_{\theta}\textrm{CV}(\theta),\label{eq:theta:opt}
\end{align}
where the minimum is taken over some pre-specified values. The procedure in (\ref{eq:theta:opt}) is called $k$-fold cross validation, which {provides a}  powerful tool {for choosing the}  tuning parameters with solid theoretical {justifications} , see \cite{a91}, \cite{h14b} and \cite{y07}. Other methods for empirical {choices}  of $k_j$'s and $h_j$'s in the framework of sieve estimator can be found in  \cite{ho14} and \cite{cc18}.


\subsection{Determination of $k_j$'s and $h_j$'s}
 In sieve estimation, the choices of $k_j$'s and $h_j$'s  {play}  essential roles in the estimation accuracy.  For example,  Assumption \ref{Assumption:common}.\ref{Ac:d} specifies lower bounds on $r_j$'s, while Theorem \ref{thm:rate:of:convergence:together} implies that if $h_j\asymp (NT)^{-\frac{1}{2m_j+1}}$ for $j=d+1,\ldots, p$, the rate of convergence for $\widehat{f}$ will be improved and  a more accurate estimation is obtained. The procedure in (\ref{eq:theta:opt}) to determine $k_j$'s and $h_j$'s also needs to specify $\lambda_{NT}$ simultaneously, which is inconvenient in practice {as it involves} too many parameters. To address this concern, we use cross validation criterion based on non-penalized estimator for the choices  of $k_j$'s and $h_j$'s. To be more specific, {their optimal choices,} $k_{j,\textrm{opt}}$'s and $h_{j,\textrm{opt}}$'s are defined as follows,
\begin{align}
	(k_{1,\textrm{opt}},\ldots, k_{p,\textrm{opt}},\ldots, h_{1,\textrm{opt}},\ldots, h_{p,\textrm{opt}})=\argmin_{k_1,\ldots, k_p, h_1,\ldots, h_p}\textrm{CV}(k_1,\ldots, k_p, h_1,\ldots, h_p, 0).\label{eq:kj:hj:opt}
\end{align}
The cross validation procedure in (\ref{eq:kj:hj:opt}) is motivated by Theorem \ref{thm:rate:of:convergence:together}, since the rate of convergence is the same regardless of the penalty.


\subsection{Determination of $\lambda_{NT}$}
After selecting $k_j$'s and $h_j$'s,  we may choose $\lambda_{NT}$ in three distinct ways different purposes.

For estimation, a similar procedure as (\ref{eq:theta:opt}) is recommended. Specifically, given pre-determined $k_j$'s and $h_j$'s, the optimal $\lambda_{NT, \textrm{opt}}$ is selected as follows,
\begin{align}
\lambda_{NT,\textrm{opt}}=\argmin_{\lambda_{NT}}\textrm{CV}(k_1,\ldots, k_p, h_1,\ldots, h_p, \lambda_{NT}),\label{eq:lambda:opt}
\end{align}
with the minimum taken over some pre-determined candidates of $\lambda_{NT}$.

However, for linearity detection, we propose  a practically convenient approach to select a model based on solution path without determining $\lambda_{NT}$. By Corollary \ref{corollary:selection:consistency:soluton:path}, the solution path will select $p+1$ models, {in which one} correctly identifies all the linear components. Therefore, it is natural for us to perform model selection among these $p+1$ candidates. For $\nu=1,\ldots, p+1$, let $J_\nu \subset \{1,\ldots, p\}$ be the set of indexes of linear components selected by $\nu$-th model along the solution path. We further define the function space
\begin{align}
	\widetilde{\Theta}_{J_\nu}=\bigg\{f(\bfx)=\sum_{j=1}^pf_j(z_j)\in \Theta_{NT}\;|\; f_j(z)=\beta_j(z-1/2) \textrm{ for some } \beta_j \in \mathbb{R} \textrm{ and all } j\in J_\nu\bigg\},\nonumber
\end{align}
and non-penalized estimator for $k$-fold cross validation
\begin{align}
	\widehat{f}_{I_s^c,J_\nu}=\argmin_{f\in \widetilde{\Theta}_{J_\nu}}\frac{1}{NT}\sum_{i\in I_s^c}\sum_{t=1}^T\bigg[Y_{it}-f(\bfX_{it})-\frac{1}{T}\sum_{l=1}^T\bigg(Y_{il}-f(\bfX_{il})\bigg)\bigg]^2, \textrm{ for } s\in [k],\; \nu \in [p+1]. \nonumber
\end{align}
Similar to (\ref{eq:theta:opt}), we propose following procedure to identify linearity based on $k$-fold cross validation,
\begin{align}
	\widehat{J}_{\textrm{CV}}=\argmin_{J_\nu}\sum_{s=1}^k\frac{1}{N_sT}\sum_{i\in I_s}\sum_{t=1}^T\bigg[Y_{it}-\widehat{f}_{I_s^c,J_\nu}(\bfX_{it})-\frac{1}{T}\sum_{l=1}^T\bigg(Y_{il}-\widehat{f}_{I_s^c,J_\nu}(\bfX_{il})\bigg)\bigg]^2.\label{eq:opt:J:nu}
\end{align}
The procedure in (\ref{eq:opt:J:nu})  is completely data-driven without {the  need to choose} $\lambda_{NT}$. Based on the solution  path, other information {criteria, such as such AIC or BIC,} also can be applied to conduct model selection, see  \cite{h14b} and \cite{b06}.

To conduct valid statistical inference,  $\lambda_{NT}$ is selected based on the solution path and $\widehat{J}_{CV}$ defined in (\ref{eq:opt:J:nu}). First,  based on the solution path, we find the values of $\lambda_{NT}$ resulting in the model with indexes of linear components being $\widehat{J}_{CV}$. Then the turning parameter $\lambda_{NT,\textrm{inf}}$ is chosen to be the smallest one among these values. 




\section{Simulation}\label{sec:simulation}
To evaluate the finite sample performance of {the proposed estimation and selection procedure}, we consider the following data generating process,
\begin{align}
	y_{it}=f_1(z_{it1})+f_2(z_{it2})+f_3(z_{it3})+f_4(z_{it4})+\alpha_i+\epsilon_{it}.\nonumber
\end{align}
The {functional forms of the underlying} regression functions are specified as follows
\begin{align*}
	f_1(z)=2z,\;\; f_2(z)=3z,\;\;f_3(z)=z+r\sin(6z),\;\; f_4(z)=z+r\beta_{6,9}(z),
\end{align*} 
with the first two $f_j$'s  being linear and the last two being nonlinear functions whose degree of nonlinearity is controlled by a factor of $r$. The function $\beta_{6,9}(z)$ is density of the beta distribution with parameters $(6, 9)$. The fixed effect $\alpha_i$'s and the idiosyncratic error $\epsilon_{it}$'s are i.i.d standard normal random variables across $i$ and $t$. The regressors $z_{itj}, j=1,2,3,4$ are generated as follows, (a)$\{u_{itj}, i\in [N], t\in [T], j\in [p]\}$ are i.i.d uniform random variables on $[0, 1]$;  (b) $z_{it1}=u_{it1}+\alpha_i$; (c) for $j=2,3$, $z_{itj}=u_{itj}+\delta_i$, with $\delta_i$'s being i.i.d  standard normal random variables; (d) $z_{it4}=u_{it4}$.  For the sample size and degree of nonlinearity, we consider all combinations of $(N, T, r)$ with $N=(50, 100, 200)$, $T=(3, 10, 50)$ and $r=(0.01, 0.1, 0.2, 0.5, 1)$, which include both short and large panel settings with weak and strong nonlinearity. The number of replication is set to be $R=500$. For convenience, we choose the degree of polynomial spline $k_j=3$ for $j\in [4]$ and set $h_1=h_2=h_3=h_4=h$ with $h^{-1}$  determined by $5$-fold cross validation among $\{\ceil[\big]{c(NT)^{1/4}}+2, c=0.3, 0.4,\ldots, 2\}$.

In the following, we {consider} three numerical experiments to study the finite sample performance of the proposed procedure.

\noindent \textbf{Experiment 1}: For the estimation, the estimator $\widehat{f}(\bfx)=\sum_{j=1}^T\widehat{f}_j(z_j)$ is evaluated
using the root mean squared error (RMSE) defined as
\begin{align*}
	\textrm{RMSE}=\sqrt{\frac{1}{NT}\sum_{j=1}^p\sum_{i=1}^N\sum_{t=1}^T\bigg[\widehat{f}_j(z_{itj})-\bigg(f_j(z_{itj})-\int_0^1 f_j(z)dz\bigg)\bigg]^2}.
\end{align*}
An integration term is added in above equation, since $\widehat{f}_j$ essentially is the estimator of $f_j-\int_0^1 f_j(z)dz$ (see Assumption \ref{Assumption:common}.\ref{Ac:a}). A sequence of $\lambda_{NT}$'s in $[0, 1]$ are used in the experiment to obtain the estimator. In particular, $\lambda_{NT}=0$ results in a non-penalized estimator.

\noindent \textbf{Experiment 2}:  For linearity detection, we generate the solution path  along a sequence of $\lambda_{NT}$'s with $\log(\lambda_{NT})=\{-6,-5.9,\ldots, 0.9, 1\}$. Four different proportions are calculated among 500 replications, namely, proportion of solution path containing the correct model and proportions of correct linearity detection from solution path based on $5$-fold cross validation score (CV), AIC and BIC, respectively.

\noindent \textbf{Experiment 3}: To study the asymptotic normality of proposal estimator, we consider the setting with $r=1$. For $z_0=(0, 0.25, 0.5, 0.75, 1)$, we construct the point wise confidence intervals for $f_{j,0}(z_0)$ based on Theorem \ref{thm:asymptotic:normal:marginal:together}. We calculate the percentages of the ground truth $f_{j,0}(z_0)$ falling in the 95\% confidence intervals.

Figure \ref{figure:simulation:mse} reports the RMSE of the proposed estimator with different sample sizes and degrees of nonlinearity. Some interesting findings can be observed in Figure \ref{figure:simulation:mse}. Firstly when varying $\lambda_{NT}$, for cases that $r=0.5, 1$ with strong nonlinearity, RMSE decreases and then increases, while for cases $r=0.01, 0.1$ with weak nonlinearity, RMSE decreases and then stays the same regardless of the sample size. For the case with moderately strong linearity, namely $r=0.2$, with small sample size $N=50, T=3$, RMSE follows a similar pattern as that of weak linearity cases, while with other sample sizes, there is a decrease on RMSE when $\lambda_{NT}$ varying from $0$ to $0.17$ and followed by a slight growth when $\lambda_{NT}$ increasing from $0.17$ to $0.2$. With larger $\lambda_{NT}$, the RMSE remains the same. Secondly, with larger $\lambda_{NT}$, RMSE stays at the same level regardless of $\lambda_{NT}$ in each case except $r=1$. Thirdly, for $\lambda_{NT}\geq 0.4$, the RMSE increases as the degree of nonlinearity becomes larger. Finally, with an appropriate choice of $\lambda_{NT}$, a penalized estimator can outperform nonpenalized estimator in terms of  RMSE. For linearity detection, Figure \ref{figure:simulation:proportion} reveals that with stronger nonlinearity, all the procedures are more likely to perform a correct linearity detection. In particular, when $r=1$, the solution path will contain the correct model in all replications except when sample size is small, $N=50, T=3$, which confirms the validity of Corollary \ref{corollary:selection:consistency:soluton:path}. Moreover, among three criteria for model selection, BIC and CV score can effectively choose the true model when $r$ is large, while AIC and CV score work better for small $r$. Figure \ref{figure:simulation:CI} reports the coverage rates of the 95\% confidence intervals for $f_{j,0}(z_0)$. It is worth mentioning that, for the linear components $f_{1,0}$ and $f_{2,0}$, the coverage rates are almost $100\%$ when $z_0=1/2$. This is due to (\ref{eq:linear:nonlinear:definition}) that $\widehat{f}_j(1/2)=0$ if $\widehat{f}_j$ is estimated as a linear function. In general, the coverage rate will approach to 95\% when $N$ becomes larger or both $N$ and $T$ become larger.

\begin{figure}[htp]
\centering
\includegraphics[width=2 in, height=2 in]{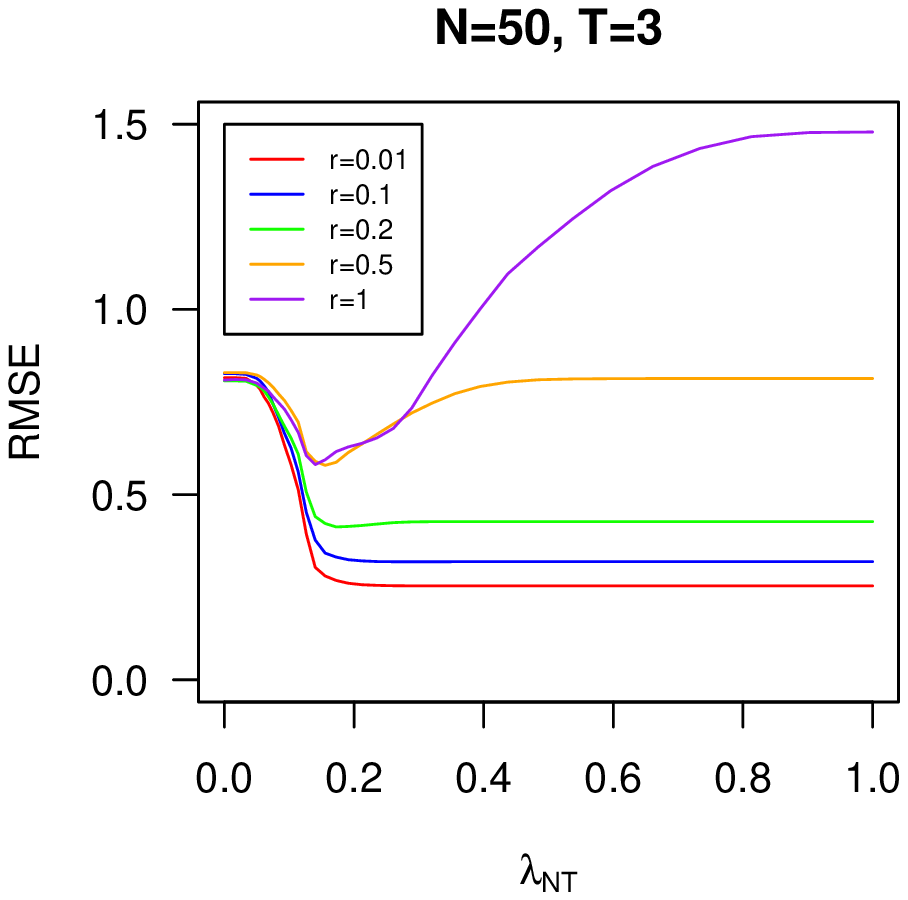}
\includegraphics[width=2 in, height=2 in]{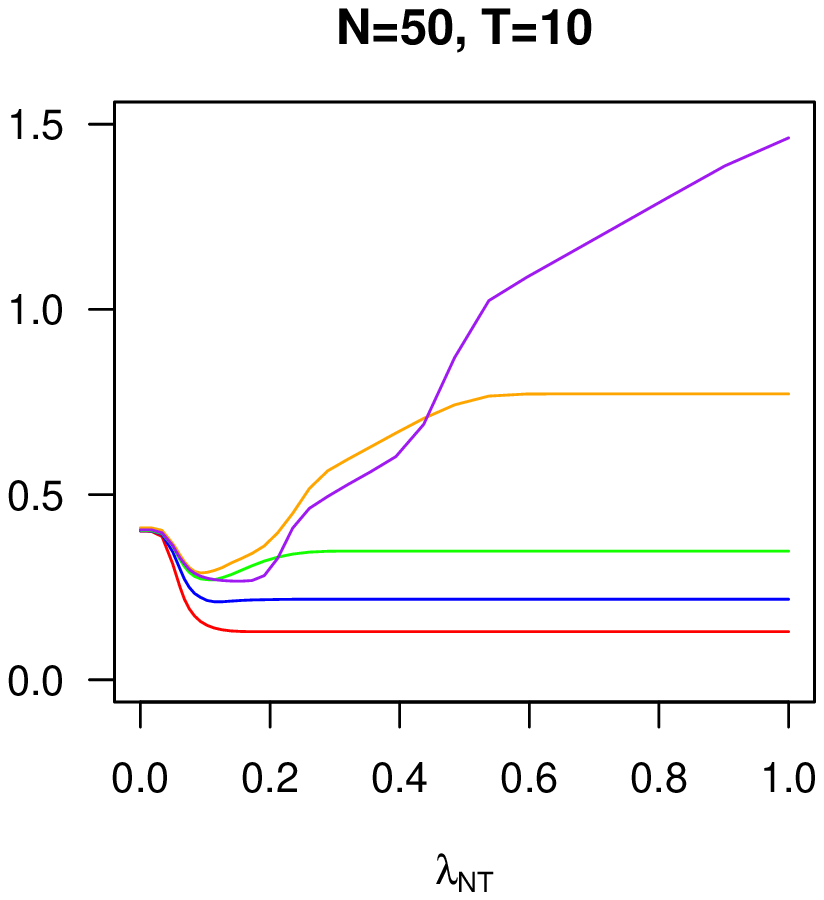}
\includegraphics[width=2 in, height=2 in]{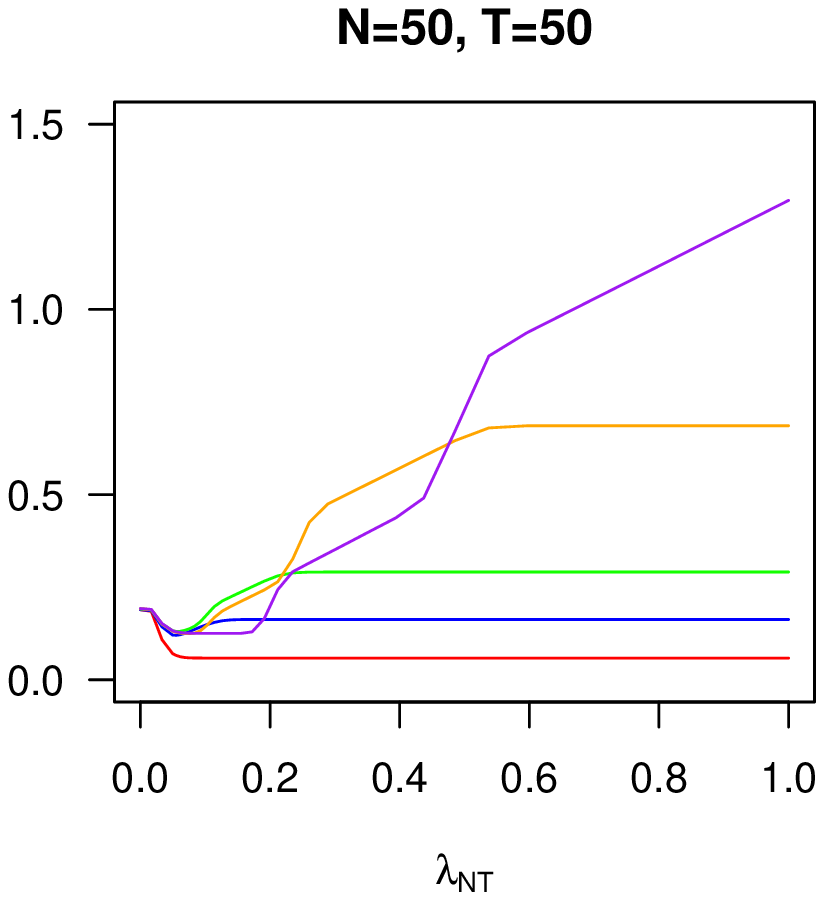}
\includegraphics[width=2 in, height=2 in]{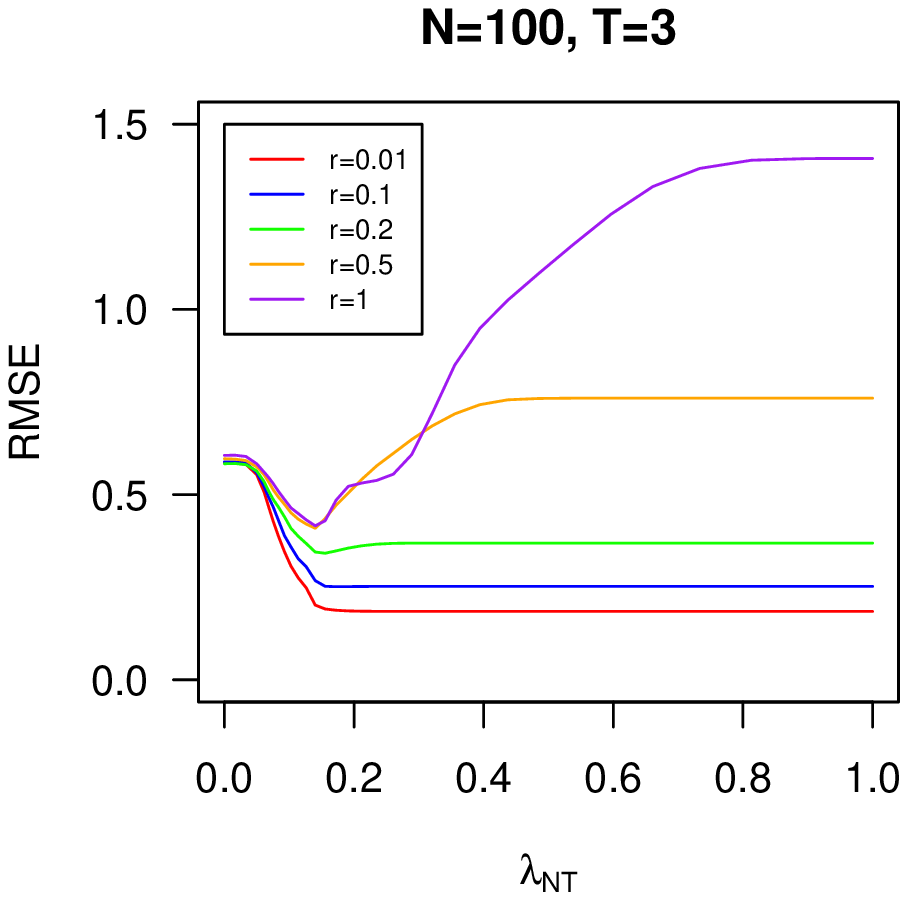}
\includegraphics[width=2 in, height=2 in]{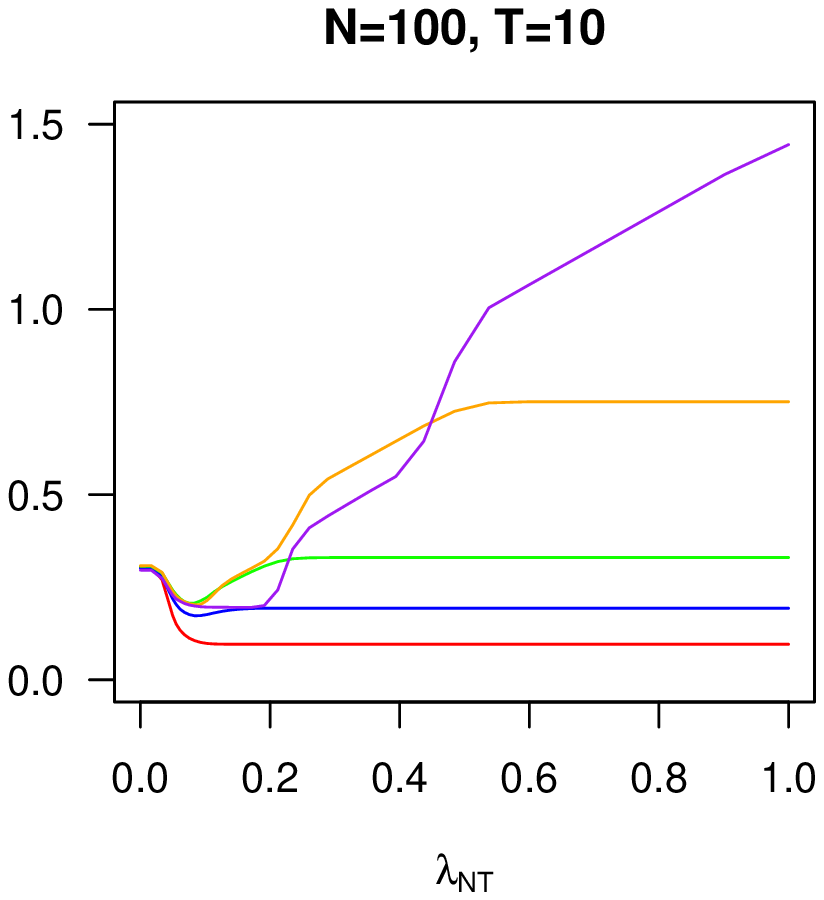}
\includegraphics[width=2 in, height=2 in]{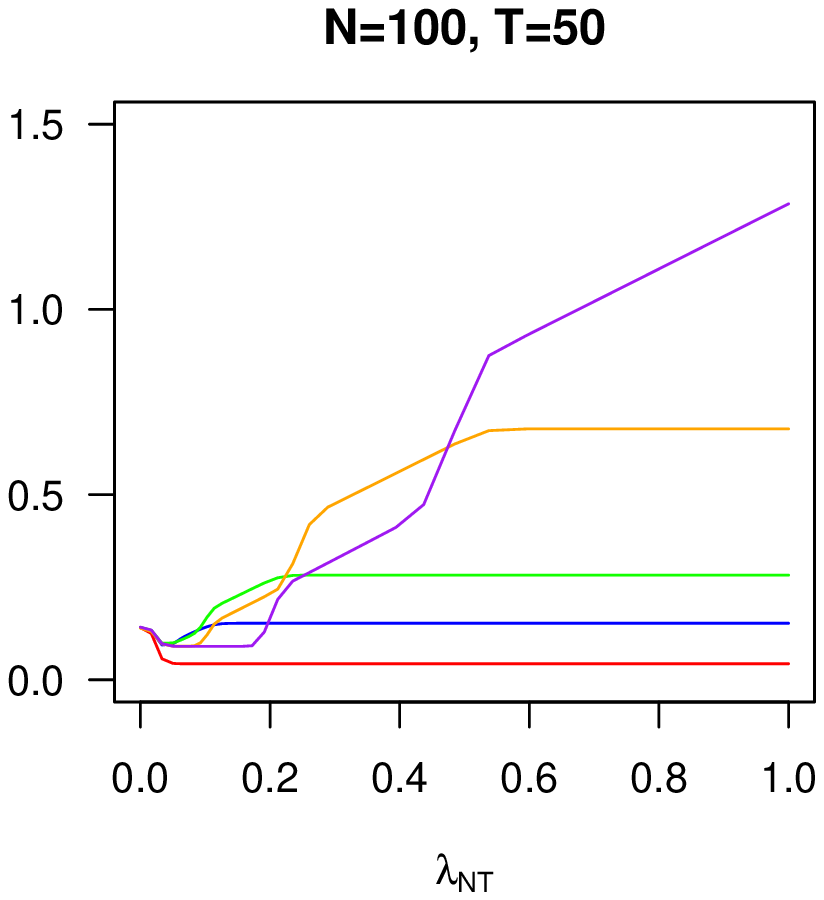}
\includegraphics[width=2 in, height=2 in]{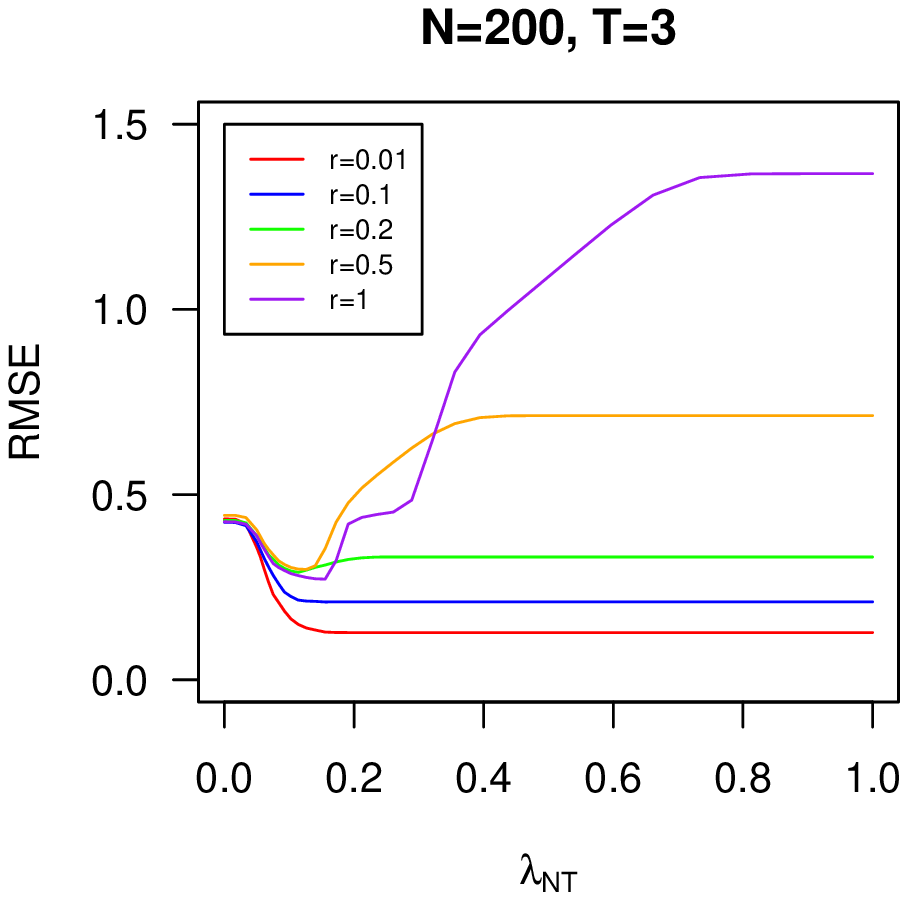}
\includegraphics[width=2 in, height=2 in]{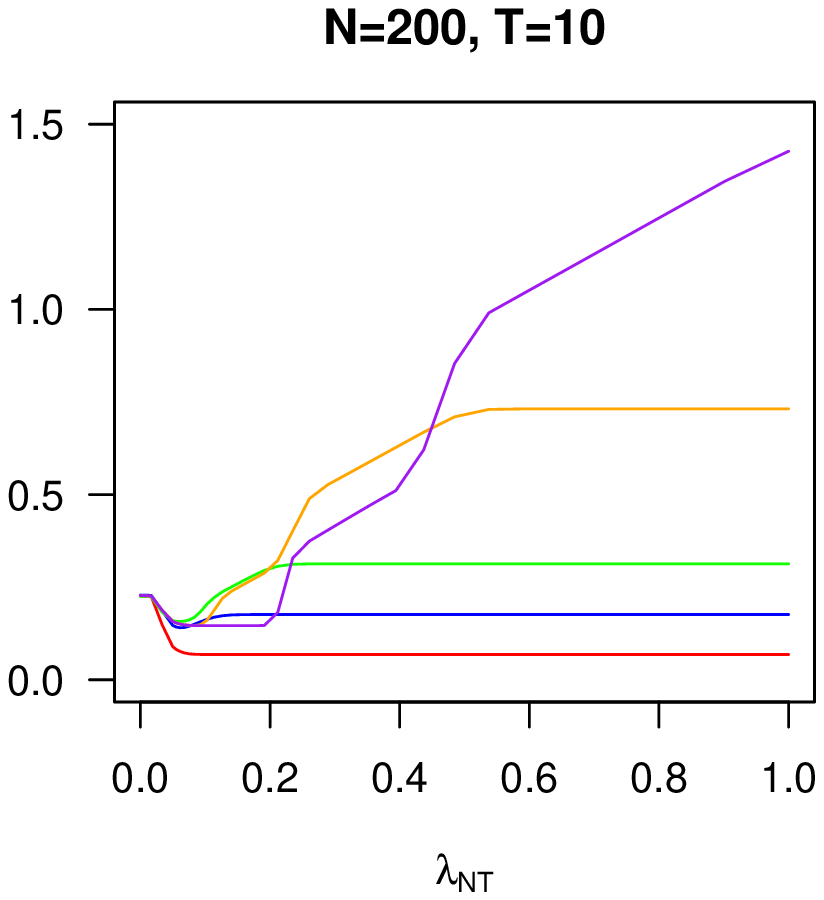}
\includegraphics[width=2 in, height=2 in]{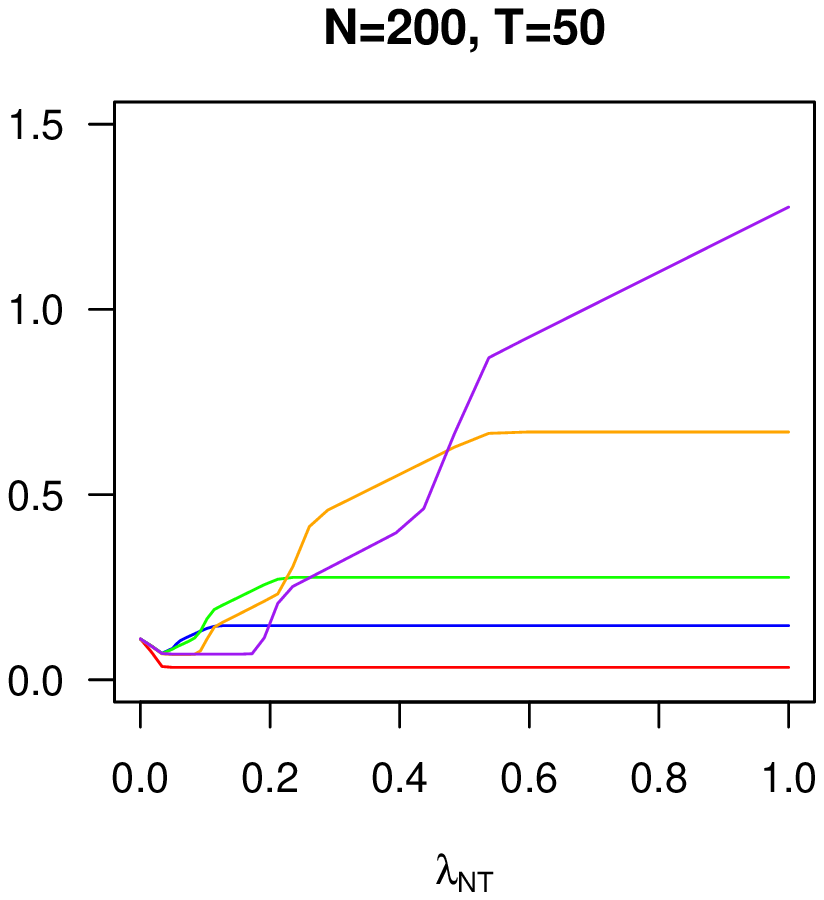}
\caption{RMSE of penalized estimator with different sample sizes, degrees of nonlinearity and tuning parameter $\lambda_{NT}$.}
\label{figure:simulation:mse}
\end{figure}

\begin{figure}[htp]
\centering
\includegraphics[width=2 in, height=2 in]{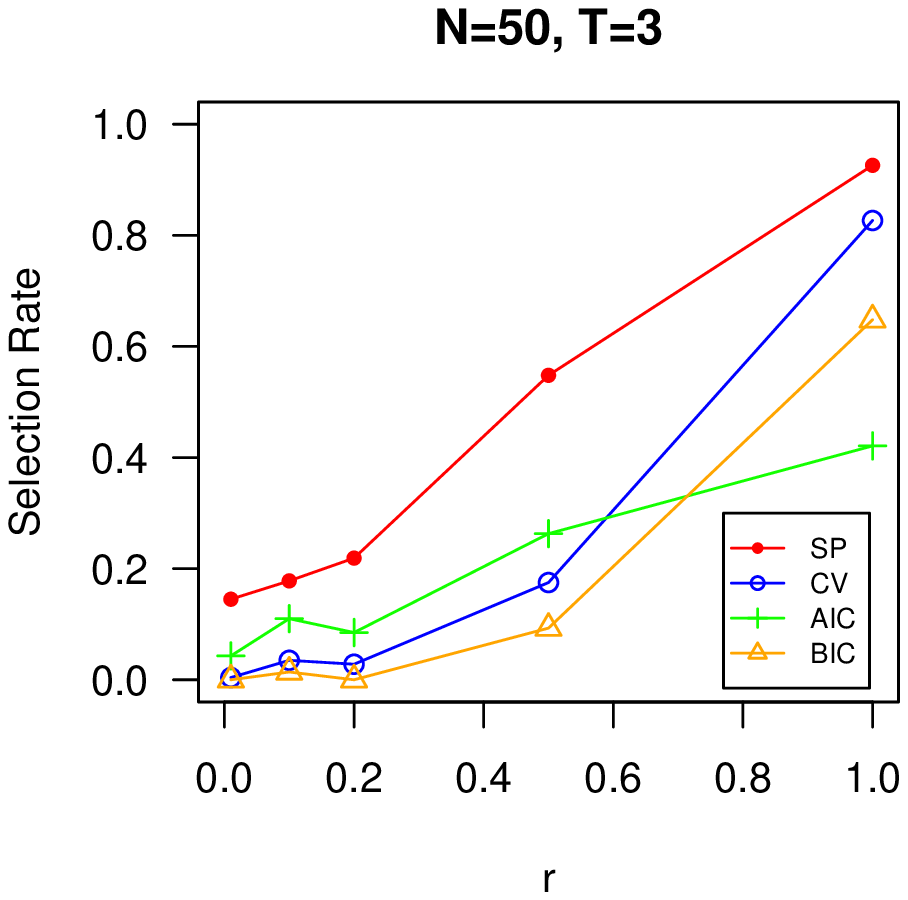}
\includegraphics[width=2 in, height=2 in]{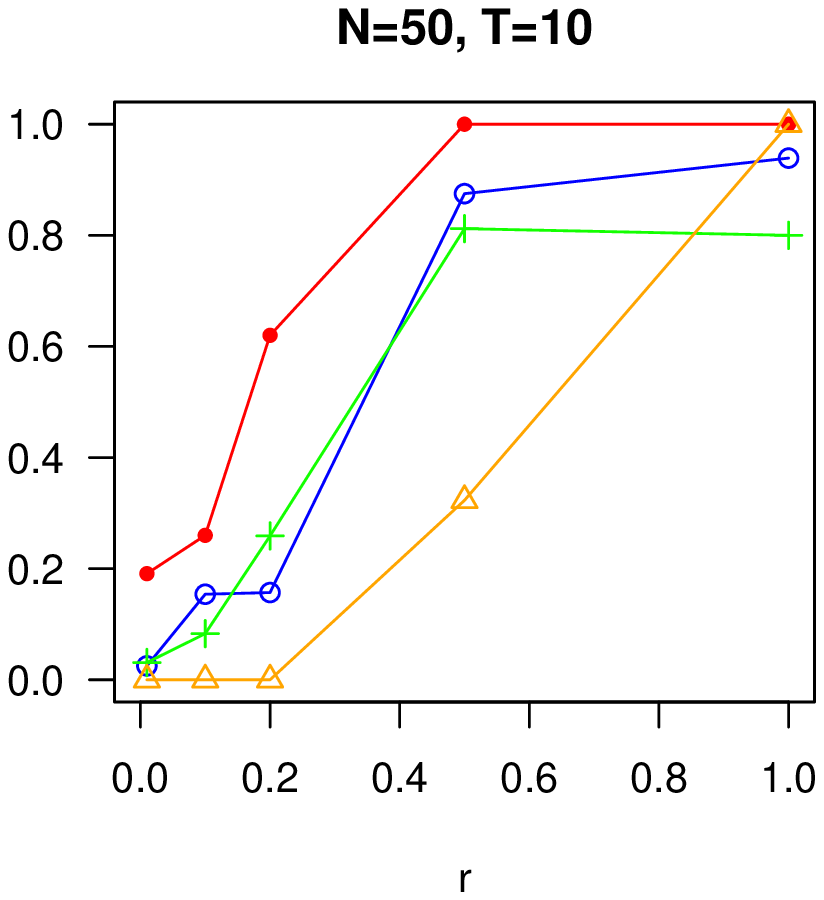}
\includegraphics[width=2 in, height=2 in]{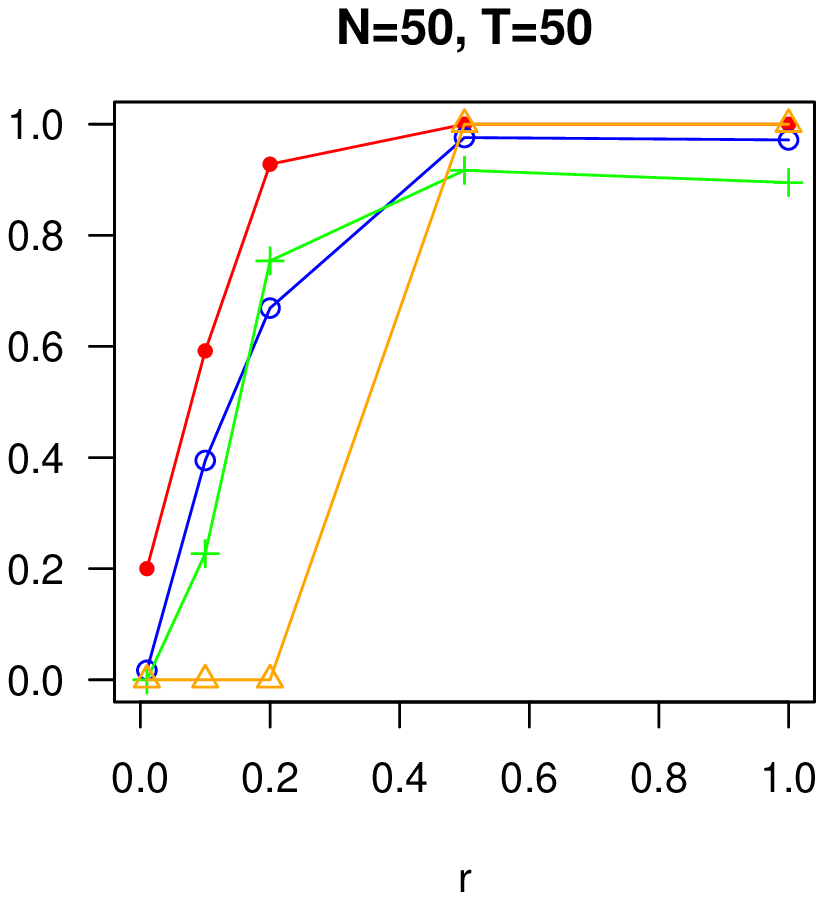}
\includegraphics[width=2 in, height=2 in]{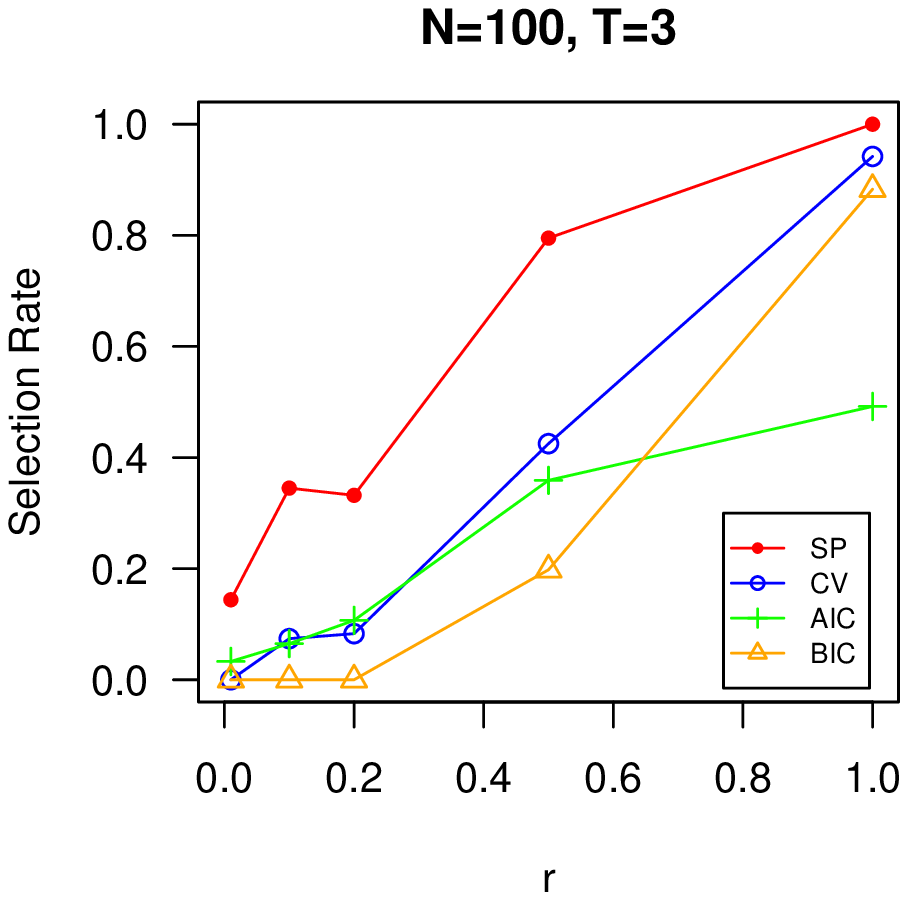}
\includegraphics[width=2 in, height=2 in]{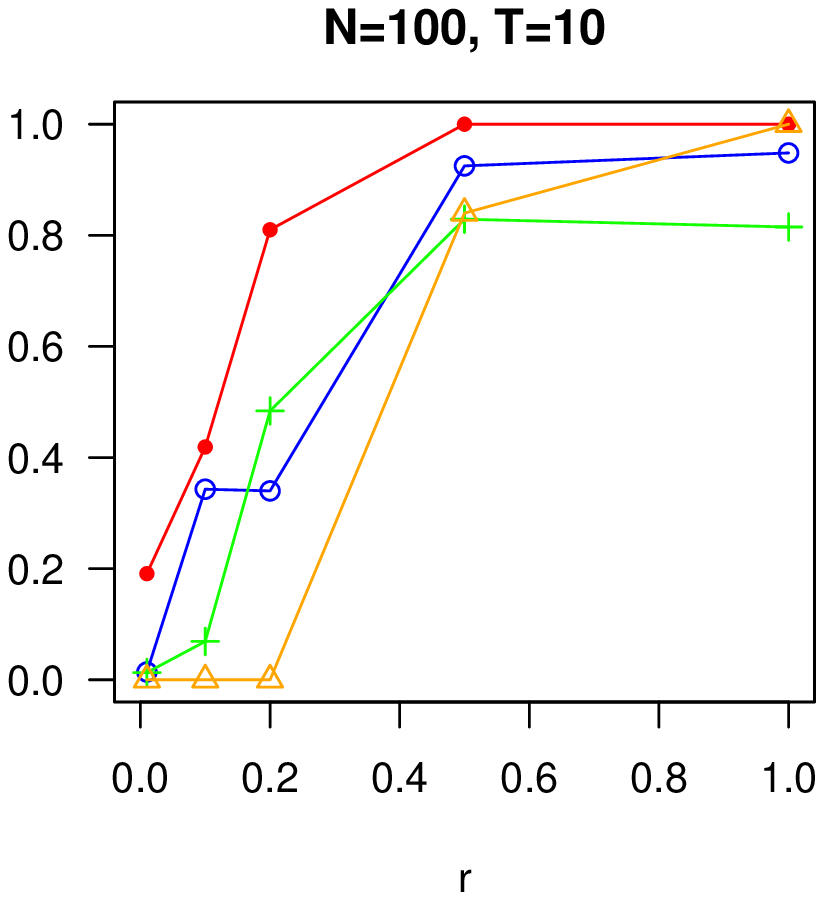}
\includegraphics[width=2 in, height=2 in]{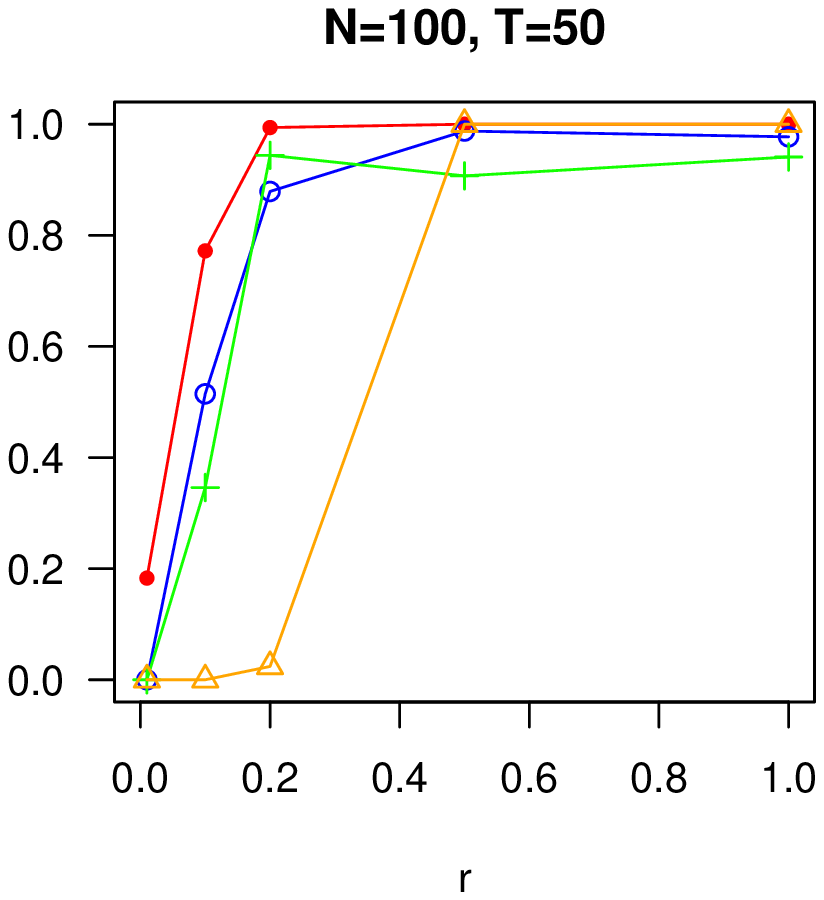}
\includegraphics[width=2 in, height=2 in]{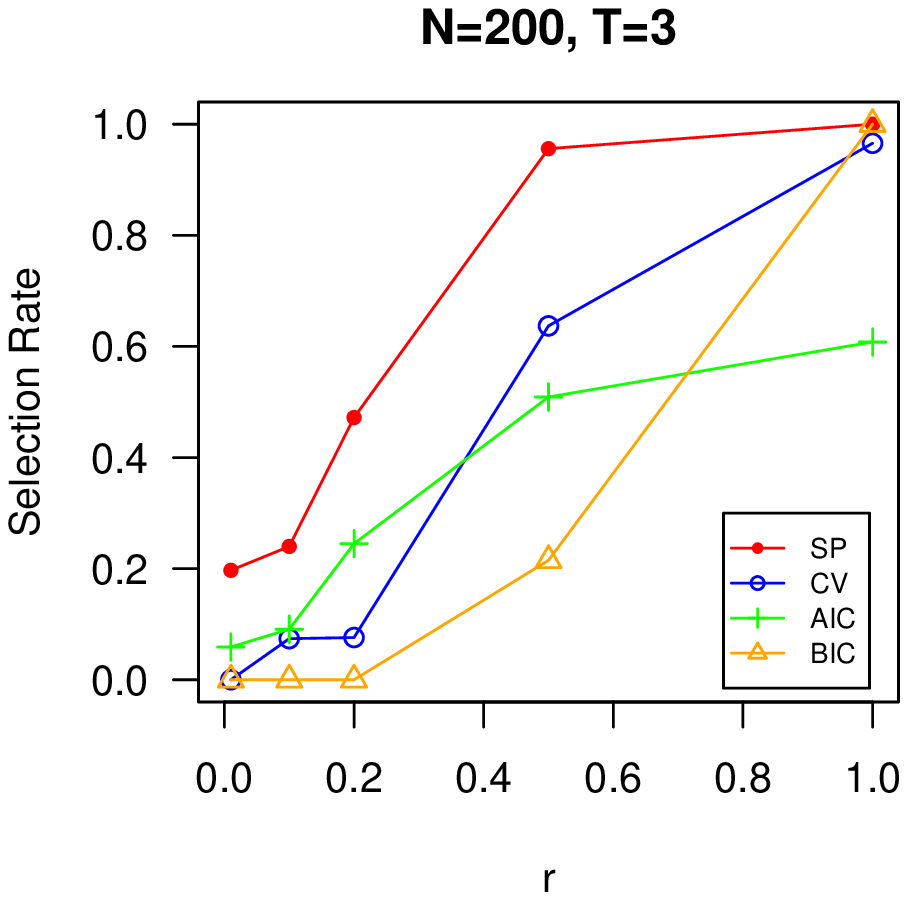}
\includegraphics[width=2 in, height=2 in]{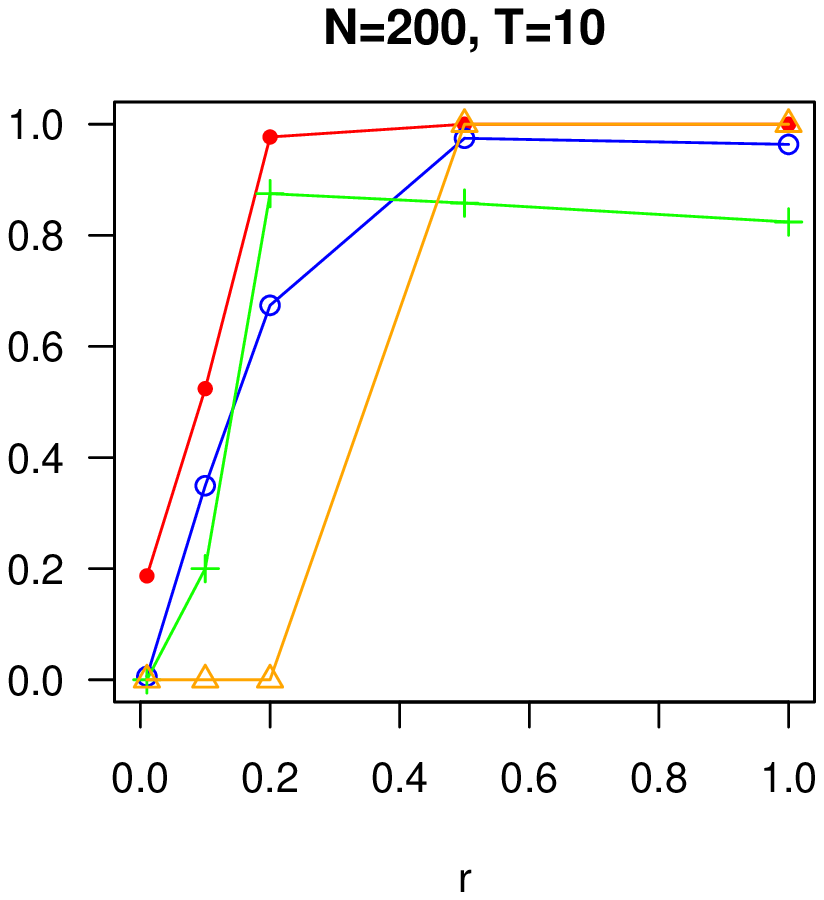}
\includegraphics[width=2 in, height=2 in]{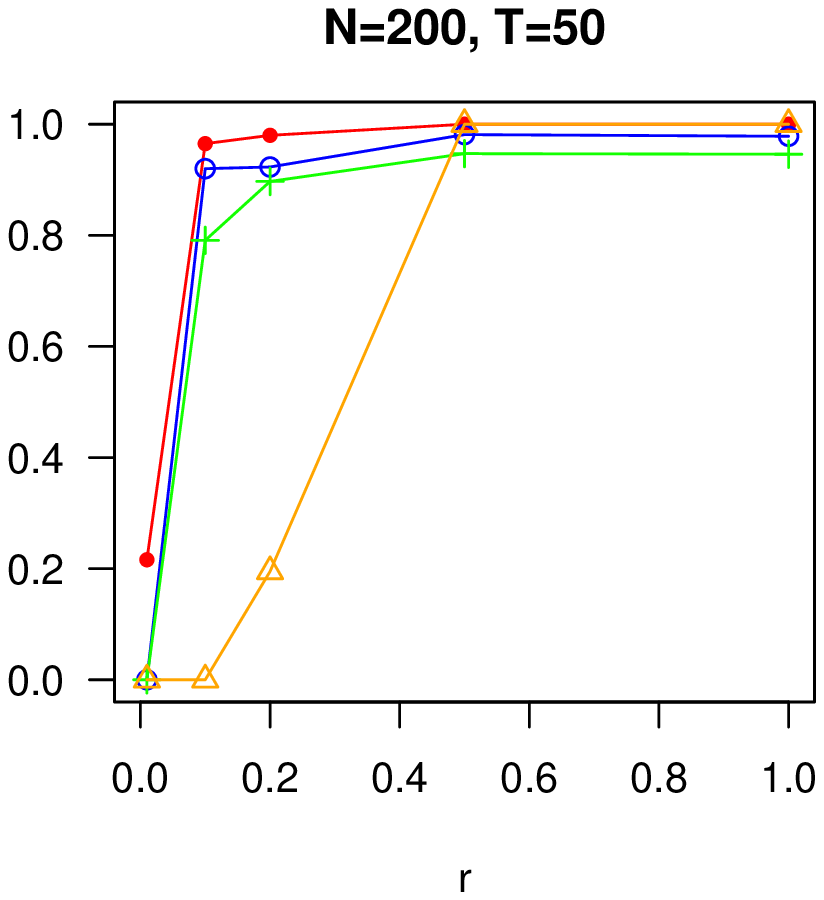}
\caption{Proportion of correct linearity identification with different sample sizes and degrees of nonlinearity. SP stands for the proportion of solution path containing the correct model. CV, AIC and BIC are the proportions of correct model selection from solution path based on 5-fold cross-validation, AIC and BIC, respectively.}
\label{figure:simulation:proportion}
\end{figure}

\begin{figure}[htp]
\centering
\includegraphics[width=2in,height=2in]{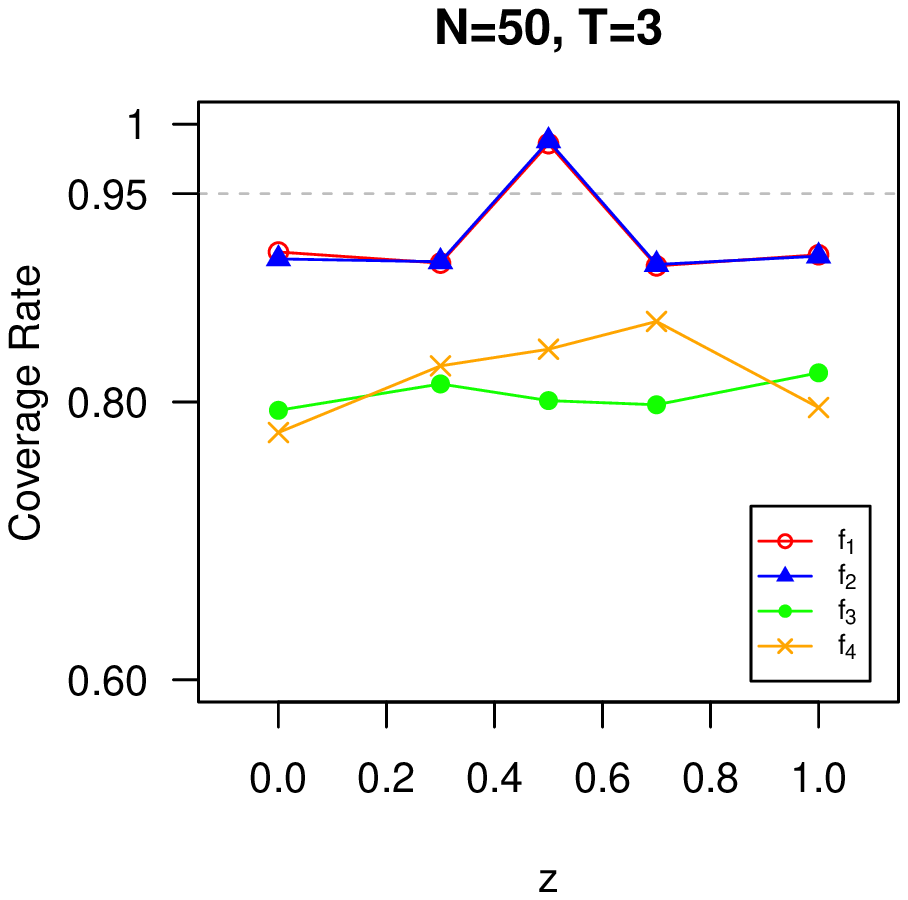}
\includegraphics[width=2in,height=2in]{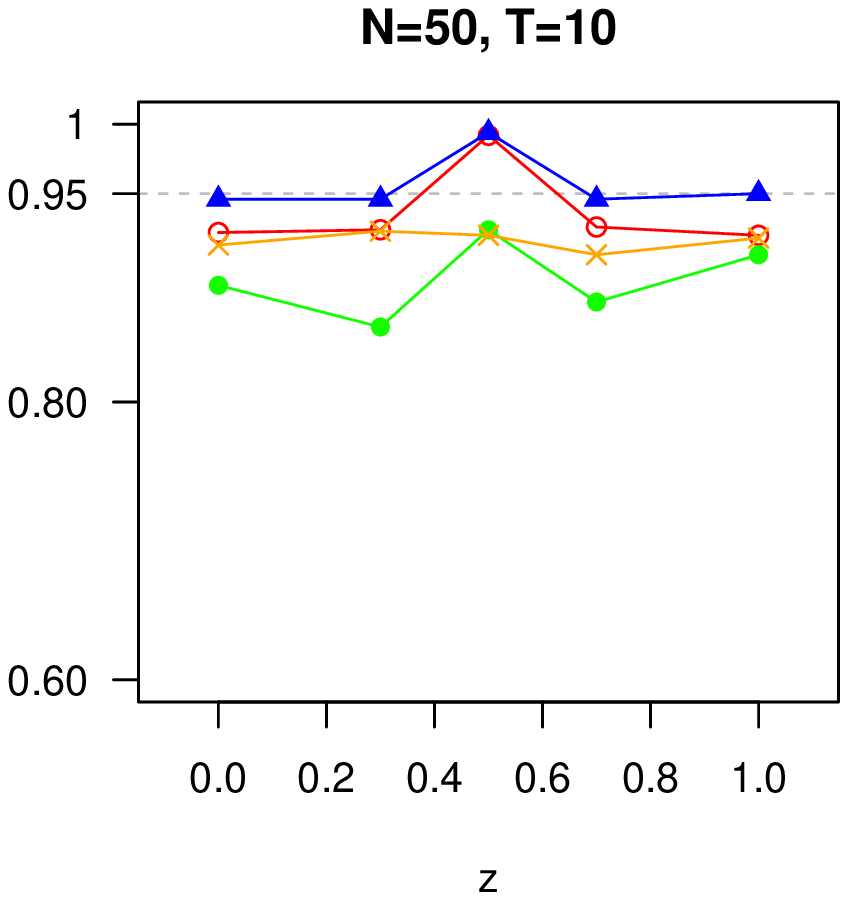}
\includegraphics[width=2in,height=2in]{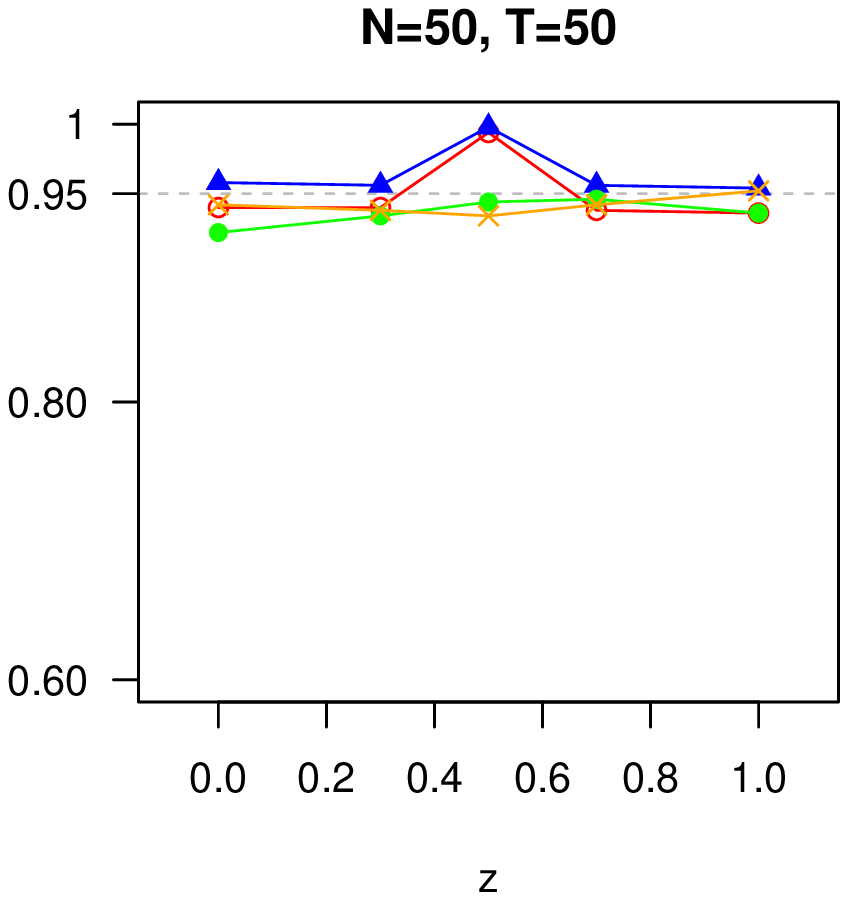}

\includegraphics[width=2in,height=2in]{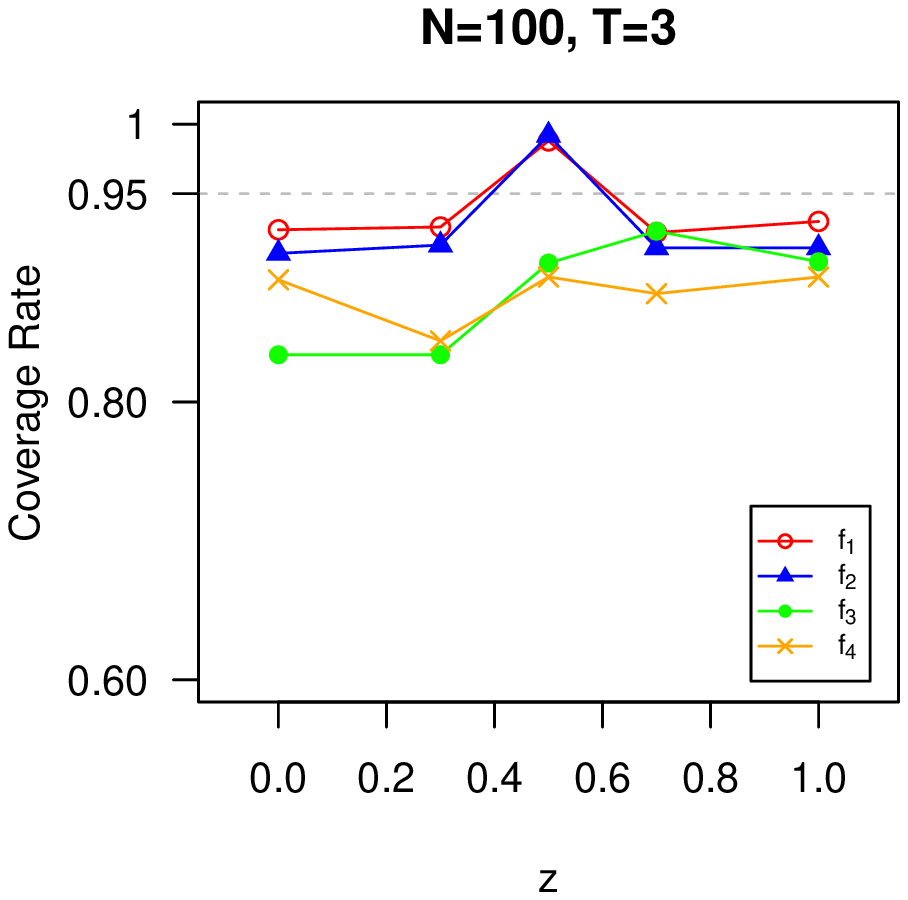}
\includegraphics[width=2in,height=2in]{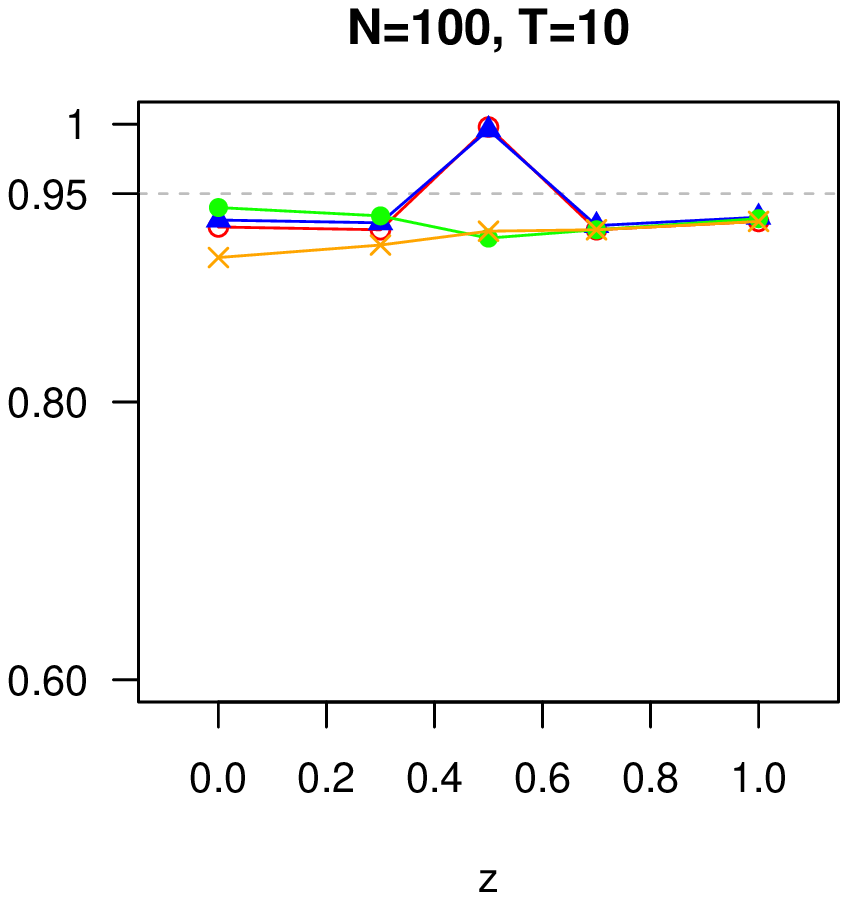}
\includegraphics[width=2in,height=1.7in]{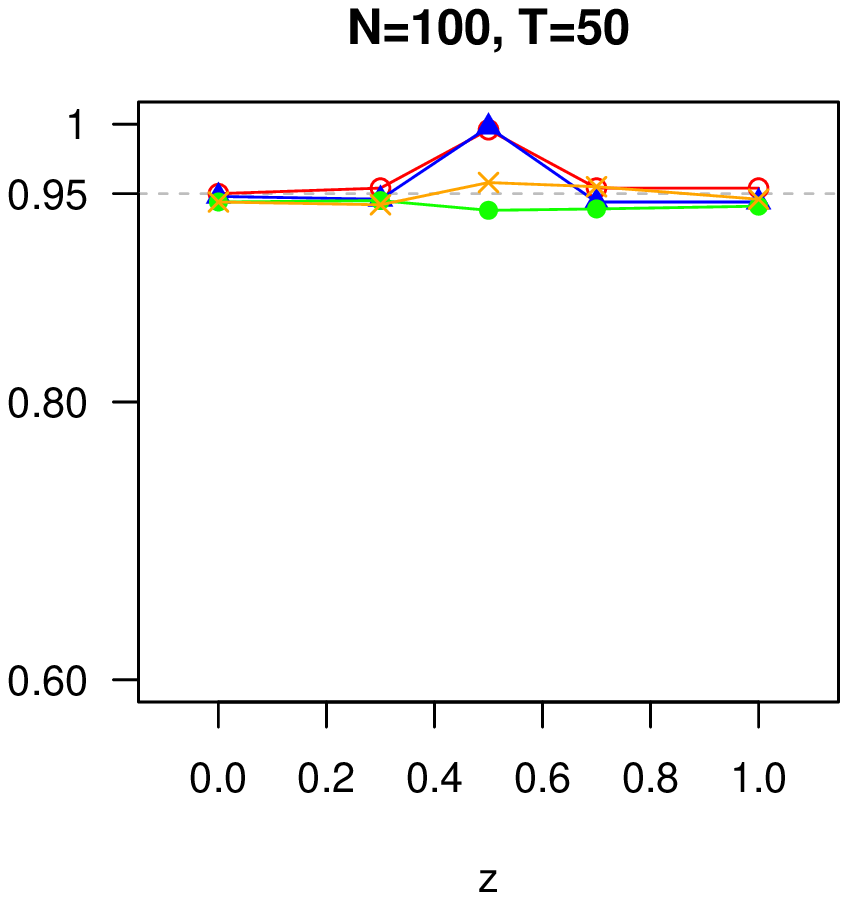}

\includegraphics[width=2in,height=2in]{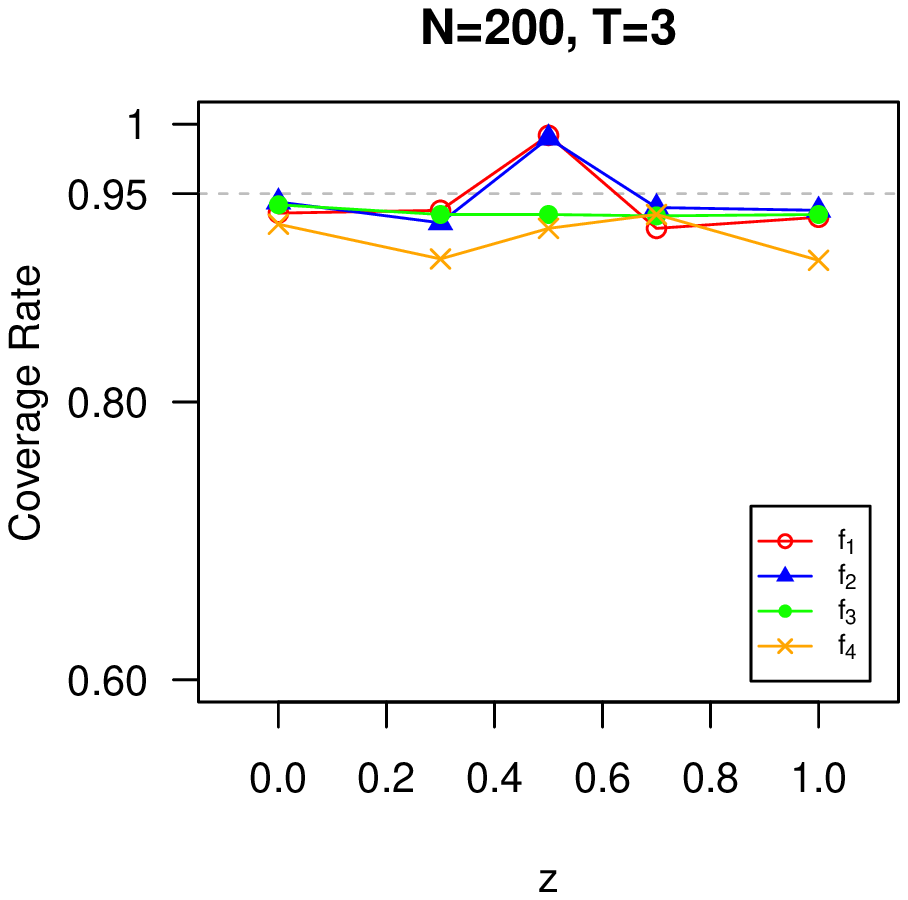}
\includegraphics[width=2in,height=2in]{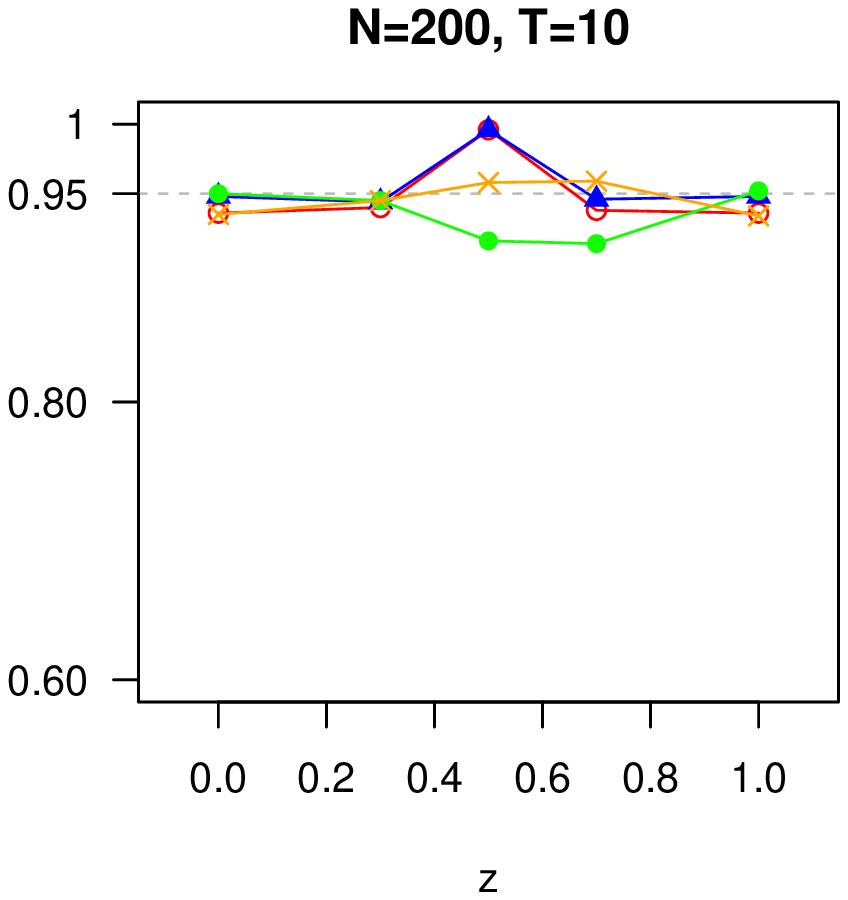}
\includegraphics[width=2in,height=2in]{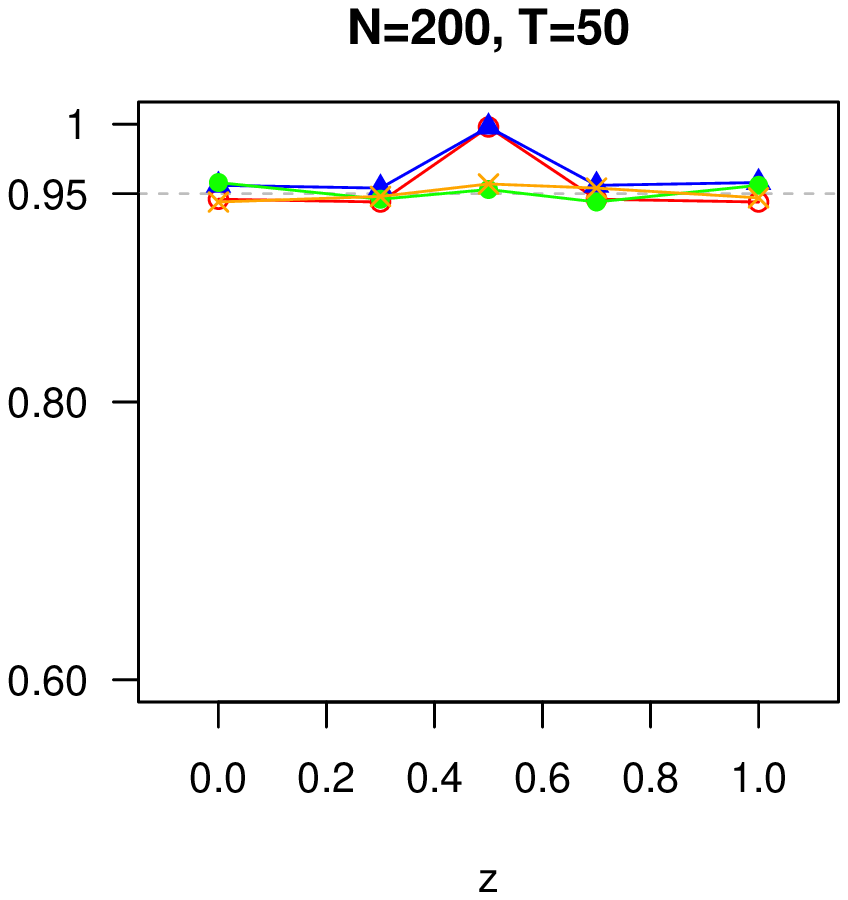}
\caption{Coverage rates of the $95\%$ confidence intervals for $f_{j,0}(z_0)$'s  with different sample sizes and $r=1$.}
\label{figure:simulation:CI}
\end{figure}


\section{Empirical Application}\label{sec:empirical:study}
\subsection{Aggregate Production}
In this section, we apply our linearity detection procedure to Aggregate Production data, which is extracted from version 9.0 of the Penn World Table. We keep a balanced panel dataset for $48$ countries across the world for the period 1950-2014. Following \cite{gks16}, we consider following regression model,
\begin{align}
	y_{it}=f_1(k_{it})+f_2(l_{it})+f_3(\textrm{pub}_{it})+f_4(\textrm{xm}_{it})+\alpha_i+error,\nonumber
\end{align}
where $y_{it}, k_{it}, l_{it}$	are the real log GDP, capital stock, and number of people engaged
of the $i$-th country at time $t$, respectively. Besides, $\textrm{pub}_{it}$ is government/public expenditure, defined as the government spending, and $\textrm{xm}_{it}$ is net trade openness, which equals exports minus imports of merchandise. Using the same criteria as in the simulation study, we choose $h_1=h_2=h_3=h_4=0.2$ by cross-validation. The tuning parameter $\lambda_{NT}$ is selected such that $\log(\lambda_{NT})$ increases from $-4$ to $-3$ with an increment $0.05$. Firstly, the solution path in Figure \ref{figure:realdata:solution:path} indicates five candidate models can be obtained when $\lambda_{NT}$ increasing with models being summarized in Table \ref{table:realdata:selected:model}. Moreover, according to Table \ref{table:realdata:selected:model}, among these five candidates, the model with all linear $f_j$'s is the preferable based on CV. The estimated coefficients of explanatory variables are provided in Table \ref{table:linear:estimator:realdata1}, from which we can see the model is highly significant. Non-penalized estimators of $f_j$'s are constructed and corresponding fitted curves are provided in Figure \ref{figure:realdata:fhat:plot}. The fitted curves of the non-penalized estimators preserve linear patterns if one only looks at the interval between two dashed vertical lines, which is the $2.5\%$ to $97.5\%$ percentile range of the corresponding regressor. This information contained in Figure \ref{figure:realdata:fhat:plot} coincides with the findings based on our proposed linearity detection procedure.

\begin{figure}[htp]
\centering
\includegraphics[width=3 in, height=2.3 in]{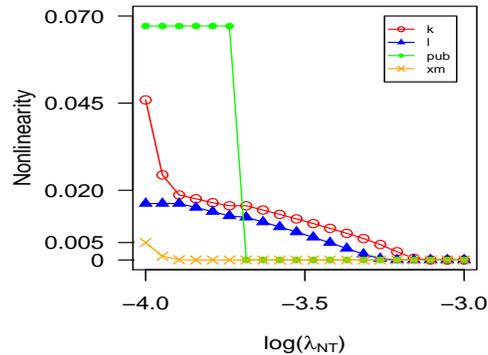}
\caption{Solution Path of nonlinearity.}
\label{figure:realdata:solution:path}
\end{figure}

\begin{table}[htp]
\begin{tabular}{|c|c|c|c|c|c|}
\hline
          & Model 1     & Model 2   & Model 3   & Model 4   & Model 5    \\ \hline
Linearity & $\emptyset$ & xm        & pub,xm    & l,pub,xm  & k,l,pub,xm \\ \hline
CV        & 2.616       & 1.453     & 0.054     & 0.042     & 0.040*      \\ \hline
\end{tabular}
\caption{Models selected by solution path.}
\label{table:realdata:selected:model}
\end{table}
\begin{table}[htp]
\centering
\begin{tabular}{|c|c|c|c|c|}
\hline
     & k     & l     & pub   & xm    \\ \hline
Coef & 7.627$^{***}$ & 2.187$^{***}$ & 0.668$^{***}$ & 0.586$^{***}$ \\ \hline
SE   & 0.178 & 0.426 & 0.230 & 0.122 \\ \hline
\end{tabular}
\caption{Linear Coefficients Estimators for Aggregate Production}
\label{table:linear:estimator:realdata1}
\end{table}
\begin{figure}[htp]
\centering
\includegraphics[width=1.5 in, height=1.4 in]{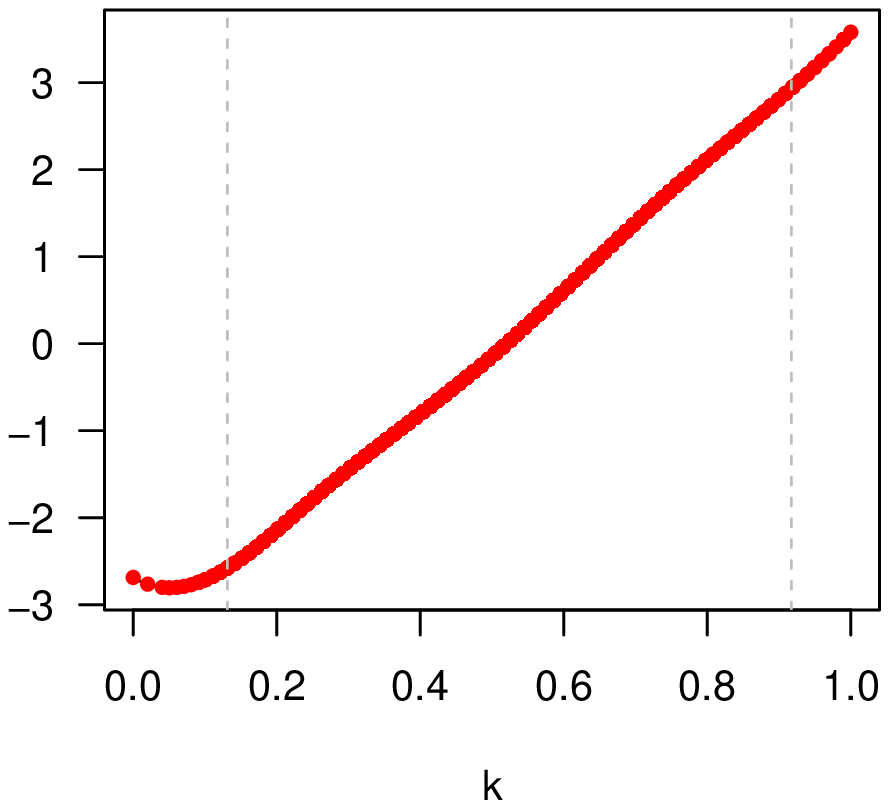}
\includegraphics[width=1.5 in, height=1.4 in]{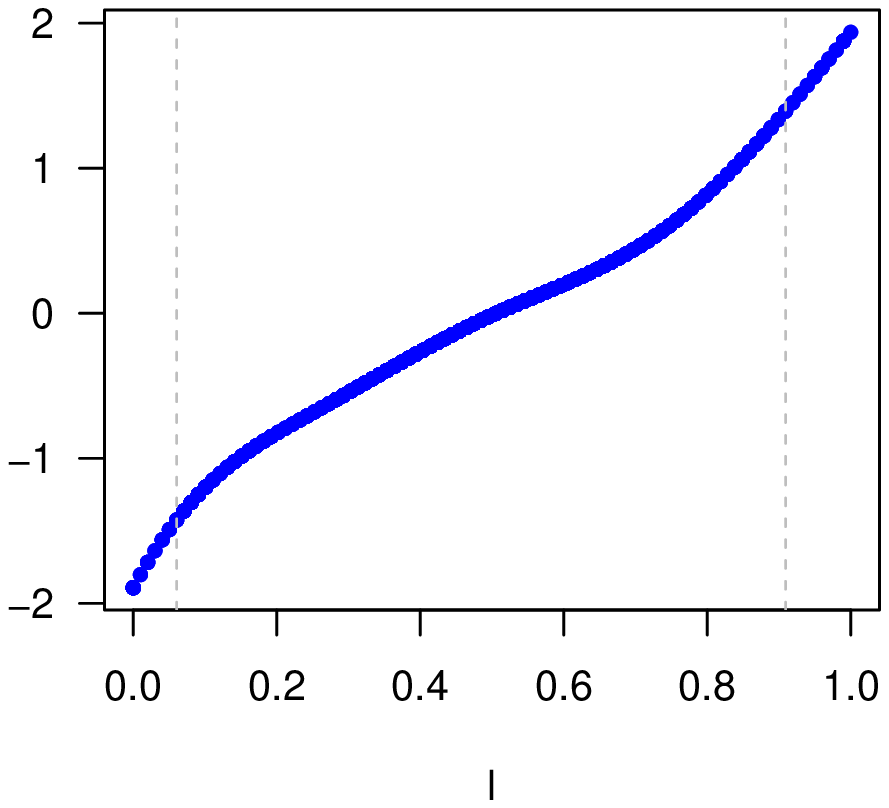}
\includegraphics[width=1.5 in, height=1.4 in]{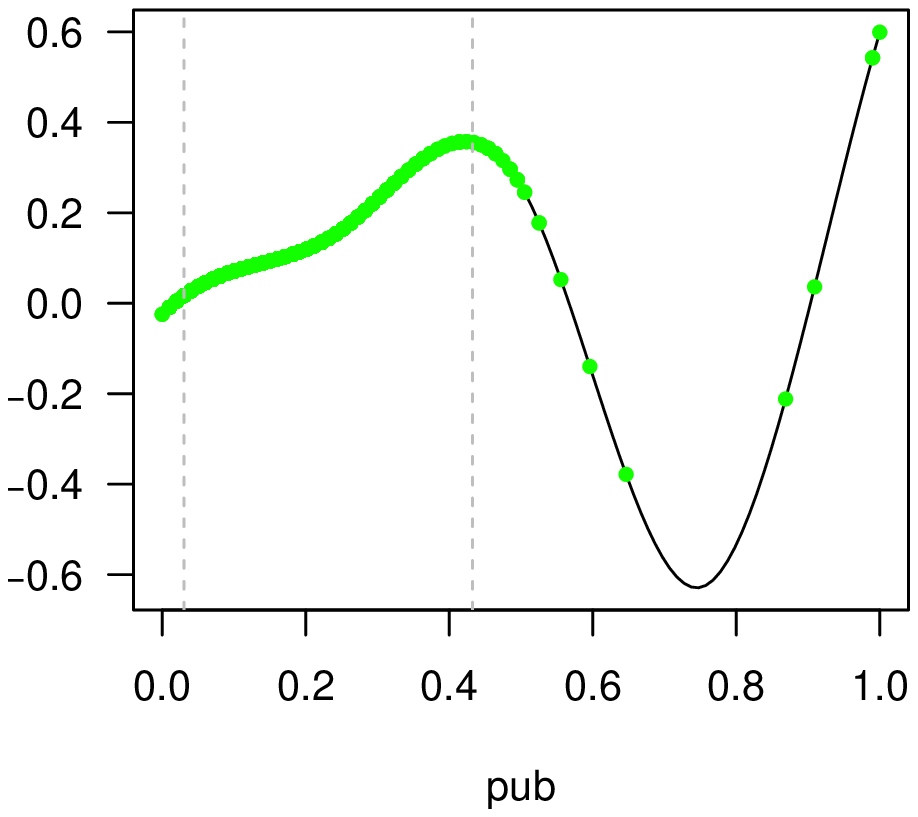}
\includegraphics[width=1.5 in, height=1.4 in]{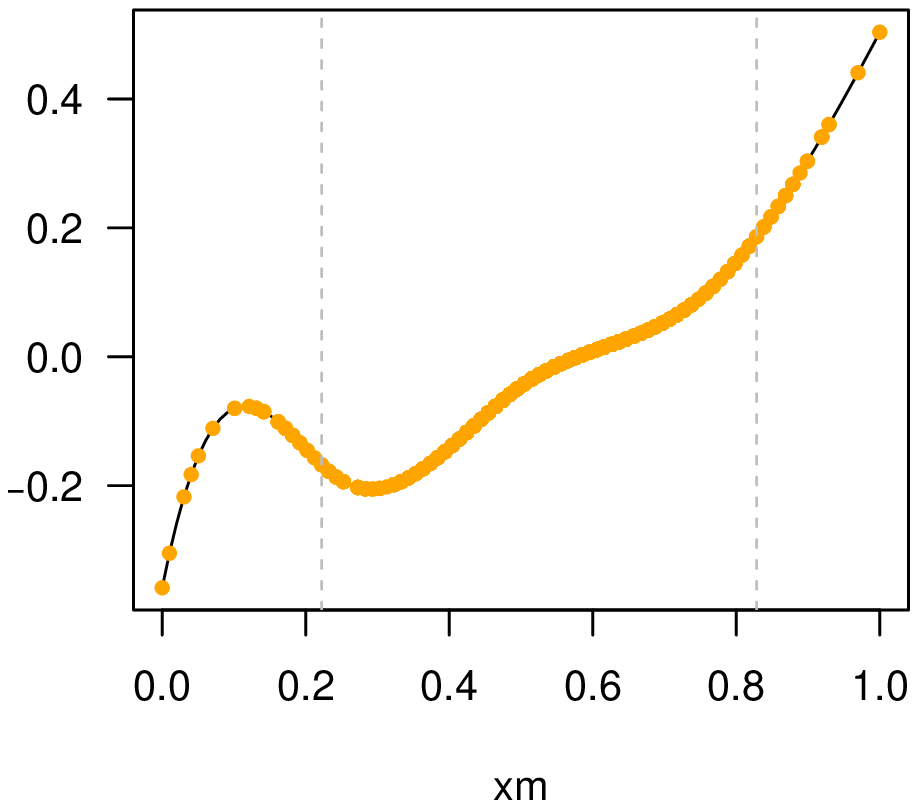}
\caption{Estimated curves of $f_j$'s. Non-penalized estimator for top 4 panels; Penalized estimator for bottom 4 panels. Two vertical dashed lines indicates the $2.5\%$ to $97.5\%$ pencentile range of each regressor.}
\label{figure:realdata:fhat:plot}
\end{figure}

\subsection{Environmental Kuznets Curve}

In the second application, we estimate the environmental Kuznets curve (EKC), which is also studied in \cite{a07}, \cite{ap09} and  \cite{lqs16}. Following \cite{lqs16}, we consider the following nonparametric model:
\begin{equation}
	y_{it}=f_1(\textrm{e}_{it})+f_2(\textrm{gdp}_{it})+f_3(\textrm{trade}_{it})+\alpha_i+\epsilon_{it},\label{eq:model:1:realdata2}
\end{equation}
where $y_{it}$ represents the per capita $\textrm{CO}_2$  emission of country $i$ in year $t$, $\textrm{e}_{it}$ is per capita  energy consumption, $\textrm{gdp}_{it}$ stands for the per capita GDP, and $\textrm{trade}_{it}$ is the per capta trade. All the variables are taken logarithm and all the explanatory variables are scaled to $[0, 1]$. The data is obtained from World Bank Development Indicators and we keep a balanced panel for $N=89$ and $T=40$ after eliminating missing values.

From the solution path in Figure \ref{figure:realdata2:solution:path1}, we extract 4 submodels and calculate their CV scores, which are summarized in Table \ref{table:cv:realdata2:1}. Based on the CV score, the selected linear explanatory variables are trade and e, while the variable gpd will be treated as nonlinear. Based on selected, model, we further estimate the linear and nonlinear components, and the results are reported in Table \ref{table:linear:estimator:realdata2:1} and Figure \ref{figure:realdata2:ci:f2}. Table \ref{table:linear:estimator:realdata2:1}  shows the coefficients of e and trade are both highly significant. Meanwhile,  Figure \ref{figure:realdata2:ci:f2} indicates that as gdp increasing, its effect on $\textrm{CO}_2$ emission will increases first and then begin to fall, which coincides with common hypothesis that the relationship between income and the emission of chemicals like sulfur dioxide ($\textrm{SO}_2$) and carbon dioxide ($\textrm{CO}_2$) or the natural resource usage has an inverted U-shape, see \cite{lqs16}. Finally, we present the nonpenalized estimation curves of the explanatory variables in Figure \ref{figure:realdata2:fhat:plot}. If only screening the fitted curves in $2.5\%$ to $97.5\%$ pencentile range of the regressors in Figure \ref{figure:realdata2:fhat:plot}, we can draw the same conclusion that e and trade are linear, while gdp is nonlinearly correlated with CO$_2$ emission.

\begin{figure}[htp]
\centering
\includegraphics[width=3 in, height=2.3 in]{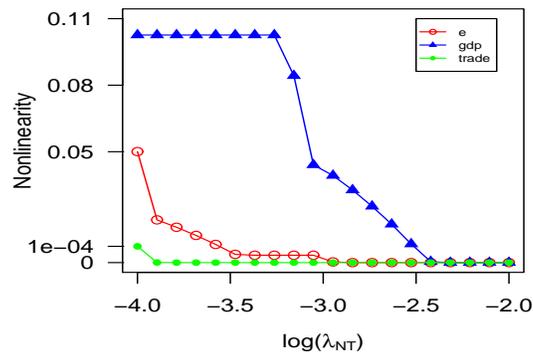}
\caption{Solution Path of nonlinearity.}
\label{figure:realdata2:solution:path1}
\end{figure}

\begin{figure}[htp]
\centering
\includegraphics[width=1.5 in, height=1.4 in]{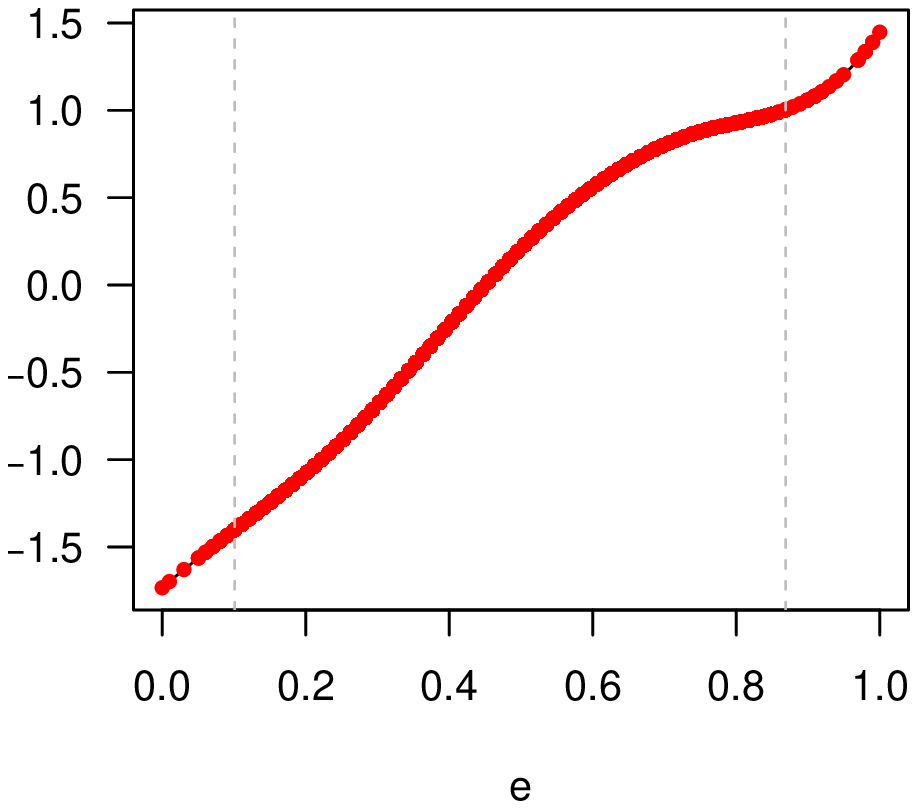}
\includegraphics[width=1.5 in, height=1.4 in]{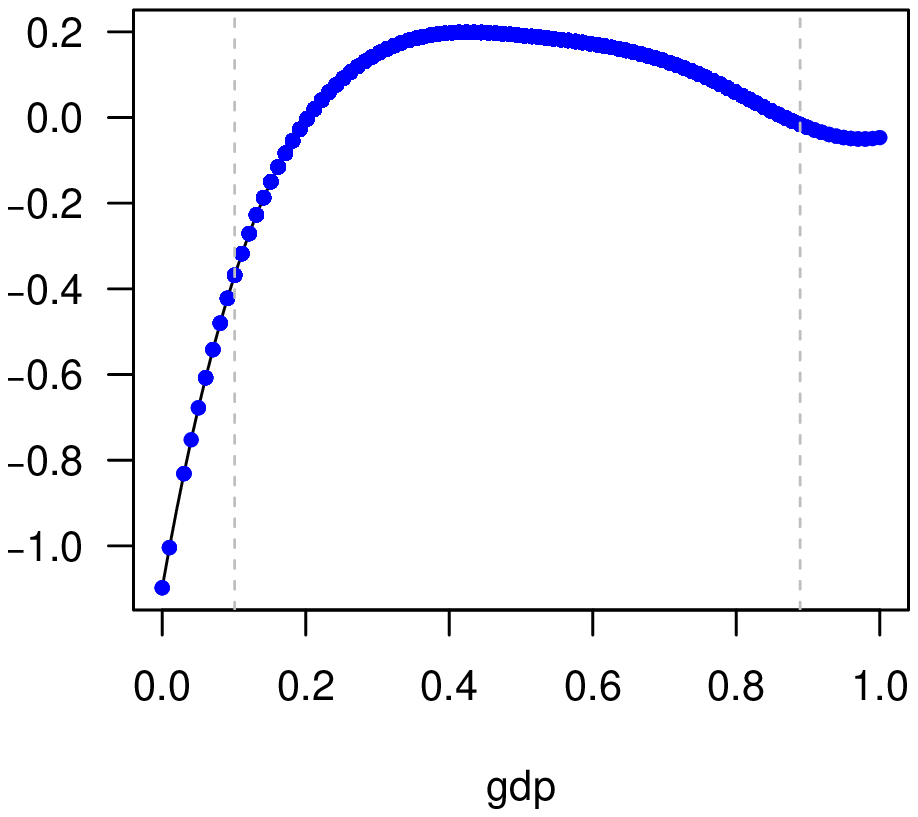}
\includegraphics[width=1.5 in, height=1.4 in]{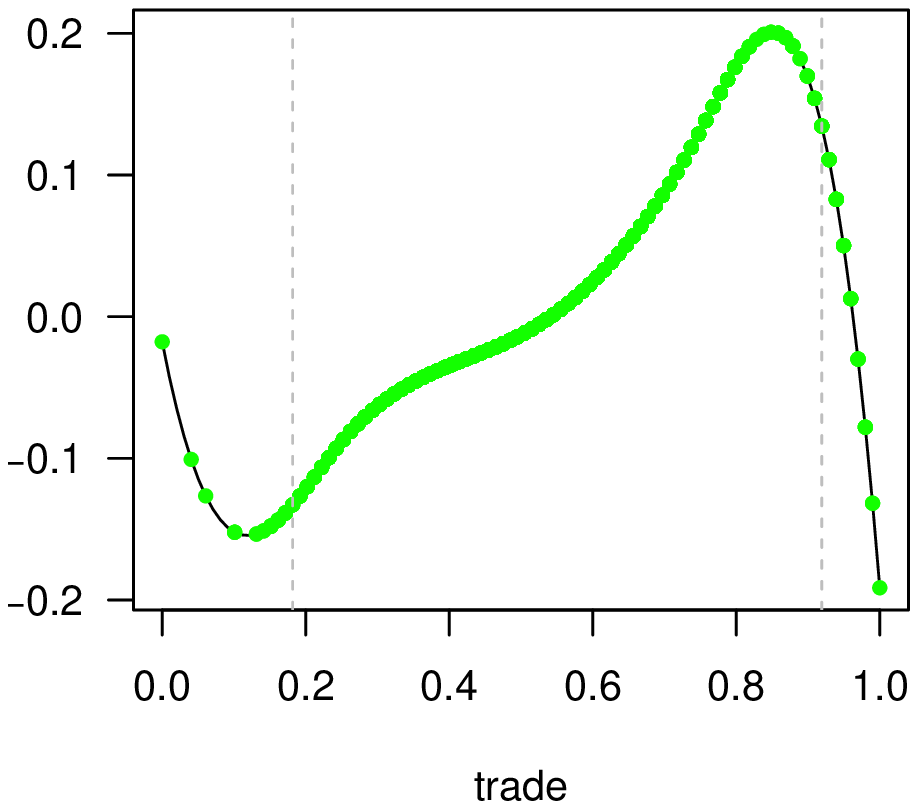}
\caption{Estimated curves of $f_j$'s. Non-penalized estimator for top 4 panels; Penalized estimator for bottom 4 panels. Two vertical dashed lines indicates the $5\%$ to $95\%$ pencentile range of each regressor.}
\label{figure:realdata2:fhat:plot}
\end{figure}

\begin{figure}[htp]
\centering
\includegraphics[width=3 in, height=2.3 in]{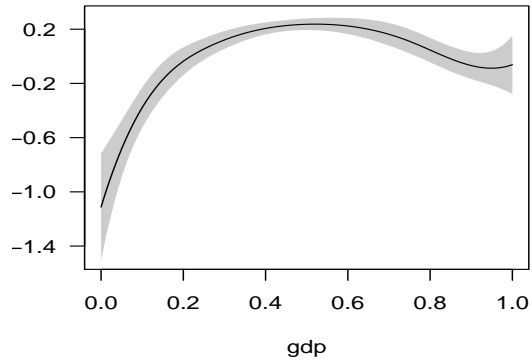}
\caption{Confidence interval of $f_2$}
\label{figure:realdata2:ci:f2}
\end{figure}

\begin{table}[htp]
\begin{tabular}{|c|c|c|c|c|}
\hline
          & Model 1 & Model 2 & Model 3  & Model 4       \\ \hline
Linearity & $\emptyset$      & trade   & trade,e & trade,e,gdp \\ \hline
CV        & 0.0541  & 0.0523  & 0.0489$^*$   & 0.0563        \\ \hline
\end{tabular}
\caption{CV Scores  of Models in Solution Path for EKC}
\label{table:cv:realdata2:1}
\end{table}

\begin{table}[htp]
\centering
\begin{tabular}{|c|c|c|}
\hline
     & e     & trade \\ \hline
Coef & 3.452$^{***}$ & 0.491$^{***}$ \\ \hline
SE   & 0.303 & 0.123 \\ \hline
\end{tabular}
\caption{Linear Coefficients Estimators for EKC}
\label{table:linear:estimator:realdata2:1}
\end{table}

\bibliography{ref}{}
\bibliographystyle{apalike}

\newpage
\setcounter{subsection}{0}
\renewcommand{\thesubsection}{A.\arabic{subsection}}
\setcounter{subsubsection}{0}
\renewcommand{\thesubsubsection}{\textbf{A.\arabic{subsection}.\arabic{subsubsection}}}
\setcounter{equation}{0}
\renewcommand{\theequation}{A.\arabic{equation}}
\setcounter{lemma}{0}
\renewcommand{\thelemma}{A.\arabic{lemma}}
\setcounter{proposition}{0}
\renewcommand{\theproposition}{A.\arabic{proposition}}

\section*{Appendix}
This appendix contains proofs and simulation results which are not included in the main text. For simplicity, we define following notation. For function $g: \mcX \to \mathbb{R}$, define $\ev_i(g)=\ev(g(\bfX_{i1}))$ when ever the expected value exists. Define $\Theta_{NT,j} =\textrm{CSpl}(r_j, \bft_{j, M_j})$, the Centralized Spline Space to approximate $f_j$, for $j\in [p]$, and denote $\Theta_{NT,j,-}, \Theta_{NT,j,\sim}$ as the subspaces of $\Theta_{NT,j}$ for linear and nonlinear component respectively. By this notation and  for any $g \in \Theta_{NT, j}$, we define $g_{-}\in \Theta_{NT, j, -}$ and $g_{\sim}\in \Theta_{NT, j, \sim}$ such that $g=g_{-}+g_{\sim}$, which is the unique decomposition due to (\ref{eq:linear:nonlinear:definition}). Recall that the function space of correct specified model is defined as 
\begin{align}
	\Theta_{NT}^0=\bigg\{f(\bfx)=\sum_{j=1}^pf_j(z_j)\in \Theta_{NT}\;\bigg|\; f_j(z)=\beta_j(z-1/2) \textrm{ for } \beta_j\in \mathbb{R} \textrm{ and } j=1,\ldots,d \bigg\},\nonumber
\end{align}
and we define its complement by the following:
\begin{align*}
	\Theta_{NT}^1=\bigg\{f(\bfx)=\sum_{j=1}^d f_j(z_j)\;\bigg|\: f_j \in \Theta_{NT, j, \sim}\bigg\}.\nonumber
\end{align*}
We also define the non-penalized projection estimators as follows:
\begin{align*}
	\widehat{f}_{*}=\argmin_{f\in \Theta_{NT}} \|Y-f\|_{NT}^2 \textrm{,\;\; }  \widehat{f}_{*,0}=\argmin_{f\in \Theta_{NT}^0} \|Y-f\|_{NT}^2,
\end{align*}
here $Y$ is treated as a function such that $Y(\bfX_{it})=Y_{it}$. Let $\|\cdot\|_\infty$ be the sup-norm of a function, and $\lambda_{\min}(A),\; \lambda_{\max}(A)$ be the smallest, and largest eigenvalues of squared matrix $A$.  

In the following, we need introduce notation for vectors for convenience. We define $\bfepsilon=(\epsilon_{11},\epsilon_{12},\ldots, \epsilon_{it},\ldots, \epsilon_{NT})^\top\in \mathbb{R}^{NT}$,  and for any $f: \mcX \to \mathbb{R}$, we denote the bold-faced $\bff$ as the vector $\bff=(f(\bfX_{11}),f(\bfX_{12}),\ldots, f(\bfX_{it}),\ldots, f(\bfX_{NT}))^\top\in \mathbb{R}^{NT}$. By this definition, we have $\bfepsilon=\bfY-\bff_0$. 

We also define following sequences which will be frequently used in the proof:
\begin{itemize}
\item $d_{NT}$: the dimension of $\Theta_{NT}$ and $d_{NT}\asymp \sum_{j=1}^p h_j^{-1}$;
\item $d_{NT,0}$: the dimension of $\Theta_{NT}^0$  and $d_{NT}\asymp \sum_{j=d+1}^p h_j^{-1}$;
\item $A_{NT}$: a constant such that $\|f\|_\infty \leq A_{NT}\|f\|_2$ for all $f\in \Theta_{NT}$ and $A_{NT}^2\asymp \sum_{j=1}^p h_j^{-1}$, see  Lemma \ref{lemma:l2:norm:sup:norm};
\item $A_{NT, 0}$: a constant such that $\|f\|_\infty \leq A_{NT}\|f\|_2$ for all $f\in \Theta_{NT}^0$  and $A_{NT,0}^2\asymp \sum_{j=d+1}^p h_j^{-1}$, see  Lemma \ref{lemma:l2:norm:sup:norm};
\item $\rho_{NT}$: approximation error bound and $\rho_{NT}\asymp \sum_{j=d+1}^p h_j^{m_j}$, see Lemma \ref{lemma:approximation:error};
\item $\gamma_{NT}$:  $\gamma_{NT}^2=\sum_{j=1}^p(NTh_j)^{-1}+\sum_{j=d+1}^p h_j^{2m_j}$.
\end{itemize}
\subsection{Proof of Theorems \ref{thm:rate:of:convergence:together} and \ref{thm:selection:consistency:together} -- Fixed $T$ Case}
\begin{lemma}\label{lemma:expectation:sample:variance}
Under Assumption \ref{Assumption:A1}, the following holds for all $g\in \mcH$: 
\begin{eqnarray*}
	\frac{1}{2a_1}\|g\|_2^2\leq \ev\bigg(\zeta_i(g,g)\bigg)\leq a_1\|g\|_2^2, \textrm{ and }\;\;\frac{1}{2a_1}\|g\|_2^2\leq \|g\|^2\leq a_1\|g\|_2^2.
\end{eqnarray*}
Moreover, if $g\in \mcH^0$, we have
\begin{eqnarray*}
	\frac{1}{a_1}\|g\|_2^2\leq \ev\bigg(g^2(\bfX_{it})\bigg)\leq a_1\|g\|_2^2.
\end{eqnarray*}
\end{lemma}
\begin{proof}[Proof of Lemma \ref{lemma:expectation:sample:variance}]
We only derive the upper bound, since the lower bounded can be obtained analogically. By Assumption \ref{Assumption:A1} and direct examination, we have
\begin{eqnarray*}
\ev\bigg(\zeta_i(g, g)\bigg)&=&\ev\bigg\{\frac{1}{T}\sum_{t=1}^T\bigg(g(\bfX_{it})-\frac{1}{T}\sum_{s=1}^Tg(\bfX_{is})\bigg)^2\bigg\}\\
&=&\int\frac{1}{T}\sum_{t=1}^T\bigg(g(\bfx_{it})-\frac{1}{T}\sum_{s=1}^Tg(\bfx_{is})\bigg)^2q_i(\bfw_i)d\bfw_i\\
&\leq&a_1\int \frac{1}{T}\sum_{t=1}^T\bigg(g(\bfx_{it})-\frac{1}{T}\sum_{s=1}^Tg(\bfx_{is})\bigg)^2d\bfw_i\\
&=&\frac{a_1(T-1)}{T}\bigg\{\int  g^2(\bfx)d\bfx-\bigg(\int g(\bfx)d\bfx\bigg)\bigg\}\\
&\leq&a_1\|g\|_2^2
\end{eqnarray*}
which is the upper bound. By similar argument, we can show that
\begin{eqnarray*}
\ev\bigg(\zeta_i(g, g)\bigg)&=&\ev\bigg\{\frac{1}{T}\sum_{t=1}^T\bigg(g(\bfX_{it})-\frac{1}{T}\sum_{s=1}^Tg(\bfX_{is})\bigg)^2\bigg\}\\
&=&\int\frac{1}{T}\sum_{t=1}^T\bigg(g(\bfx_{it})-\frac{1}{T}\sum_{s=1}^Tg(\bfx_{is})\bigg)^2q_i(\bfw_i)d\bfw_i\\
&\geq&a_1^{-1}\int \frac{1}{T}\sum_{t=1}^T\bigg(g(\bfx_{it})-\frac{1}{T}\sum_{s=1}^Tg(\bfx_{is})\bigg)^2d\bfw_i\\
&=&\frac{a_1^{-1}(T-1)}{T}\bigg\{\int  g^2(\bfx)d\bfx-\bigg(\int g(\bfx)d\bfx\bigg)\bigg\}\\
&\leq&\frac{1}{2a_1}\|g\|_2^2
\end{eqnarray*}
where we used the fact that $2(T-1)\geq T$. Notice $\|g\|^2=\frac{1}{N}\sum_{i=1}^N\ev(\zeta_{i}(g,g))$, we prove the second inequality. For the  last in equality, the proof is similar and we omit it.
\end{proof}

\begin{lemma}\label{lemma:l2:norm:sup:norm}
Under Assumption \ref{Assumption:A1} and \ref{Assumption:common}, $\|g\|_{\infty}\leq A_{NT}\|g\|_2$, for all $g\in \Theta_{NT}$, and $\|g\|_{\infty}\leq A_{NT,0}\|g\|_2$ for all $g\in \Theta_{NT}^0$, where $A_{NT}\asymp \sum_{j=1}^ph_j^{-1/2}$, and $A_{NT,0}\asymp \sum_{j=d+1}^ph_j^{-1/2}$. 
\end{lemma}
\begin{proof}[Proof of Lemma \ref{lemma:l2:norm:sup:norm}] We only prove the first inequality, as the second one can be proved similarly.
Suppose $g=\sum_{j=1}^pg_j$ with $g_j\in \Theta_{{NT},j}$,  by \cite{dl93}[Theorem 5.1.2], it follows that
\begin{eqnarray*}
	\|g_j\|_{\infty}\leq r_j\|g\|_2, \textrm{ with }  r_j\asymp h_j^{-1/2}, \textrm{ for } j=1,2,\ldots p.
\end{eqnarray*}
By Lemma \ref{lemma:sum:of:l2:norm:bound} and above inequality, we have
\begin{eqnarray*}
	\|g\|_\infty\leq \sum_{j=1}^p\|g_j\|_\infty\leq \sum_{j=1}^pr_j\|g_j\|_2\leq c^{-1}\sum_{j=1}^pr_j\|g\|_2.
\end{eqnarray*}
\end{proof}

\begin{lemma}\label{lemma:sum:of:l2:norm:bound}
Suppose one of the following conditions is satisfied:
\begin{enumerate}
\item  Assumptions \ref{Assumption:A1} and \ref{Assumption:common} are valid;
\item  Assumptions \ref{Assumption:common} and \ref{Assumption:A2}  hold and $A_{NT}^2=o(T)$.
\end{enumerate}
Then there exists $c_1>0$ depending on $a_1,  a_3, p$ such that
\begin{eqnarray*}
	\|g\|_2^2\geq c_1\sum_{j=1}^p\|g_j\|_2^2 \textrm{ and }\; \|g\|^2\geq c_1\sum_{j=1}^p\|g_j\|^2,
\end{eqnarray*}
for all $g=\sum_{j=1}^pg_j\in \Theta_{NT}$ with $g_j \in \Theta_{{NT}, j}, j=1,2,\ldots, p$.
\end{lemma}
\begin{proof}[Proof of Lemma \ref{lemma:sum:of:l2:norm:bound}] 
First inequality is essentially \cite{s94}[Lemma 3.1] and the second one follows from Lemma \ref{lemma:expectation:sample:variance} for fixed $T$ and Lemma \ref{lemma:expectation:sample:variance:large:T} for diverging $T$.
\end{proof}

\begin{lemma}\label{lemma:approximation:error}
Under Assumption \ref{Assumption:common}, there exist $g_*=\sum_{j=1}^pg_{j,*}\in \Theta_{NT}^0$ with $g_{j,*}\in \Theta_{{NT},j}$ such that $g_{j,*}=f_{j,0}$ for $j=1,\ldots, d$, and $\|g_{j,*}-f_{j,0}\|_\infty\leq \rho_{NT,j}$ with $\rho_{NT,j}\asymp h_j^{m_j}$  for $j=d+1,\ldots, p$. As a consequence, it holds that $\|g_*-f_0\|_\infty\leq \rho_{NT}$ with $\rho_{NT}\asymp \sum_{j=d+1}^p h_j^{m_j}$. 
\end{lemma}
\begin{proof}
This a well known result, i.e., see \cite{c07}.
\end{proof}

\begin{lemma}\label{lemma:uniform:equivalence:empirical:population:norm}
Under Assumption \ref{Assumption:A1} and \ref{Assumption:common}, if $d_{NT}A_{NT}^2=o(N)$, then 
\begin{eqnarray*}
	\pr\bigg(\sup_{g,f \in \Theta_{NT}}\bigg| \frac{\langle f, g\rangle_{NT}-\langle f, g\rangle}{\|f\|_2\|g\|_2}\bigg|>x\bigg)&\leq&\frac{2}{e}\bigg\{\frac{128a_1A_{NT}^2}{Nx^2}\sum_{k=1}^\infty\bigg(\frac{9}{4}\bigg)^{k-1}+\frac{128a_1A_{NT}^2}{Nx}\sum_{k=1}^\infty\bigg(\frac{3}{2}\bigg)^{k-1}\bigg\}.
\end{eqnarray*}
As a consequence, it follows that
\begin{eqnarray*}
\sup_{g,f \in \Theta_{NT}}\bigg| \frac{\langle f, g\rangle_{NT}-\langle f, g\rangle}{\|f\|_2\|g\|_2}\bigg|=o_P(1)\; \textrm{ and }\;\sup_{g,f \in \Theta_{NT}}\bigg| \frac{\langle f, g\rangle_{NT}-\langle f, g\rangle}{\|f\|\|g\|}\bigg|=o_P(1).
\end{eqnarray*}
Moreover, we also have
\begin{eqnarray*}
	\sup_{g,f \in \Theta_{NT}}\bigg|\frac{\frac{1}{NT}\sum_{i=1}^N\sum_{t=1}^Tg(\bfX_{it})f(\bfX_{it})-\frac{1}{N}\sum_{i=1}^N\ev_i(gf)}{\|g\|_2\|f\|_2}\bigg|=o_P(1).
\end{eqnarray*}
\end{lemma}
\begin{proof}[Proof of Lemma \ref{lemma:uniform:equivalence:empirical:population:norm}]

Let $\mcF_{\textrm{UB}}^0=\{f \in \Theta_{NT}:  \|f\|_2\leq 1\}$, $\Theta=\mcF_{\textrm{UB}}^0\times \mcF_{\textrm{UB}}^0$ and 
\begin{eqnarray*}
	\xi_i(\theta)&=&\frac{1}{T}\sum_{t=1}^T\bigg(g(\bfX_{it})-\frac{1}{T}\sum_{s=1}^Tg(\bfX_{is})\bigg)\bigg(f(\bfX_{it})-\frac{1}{T}\sum_{s=1}^Tf(\bfX_{is})\bigg), \textrm{ with } \theta=(g, f),\\
	X_i(\theta)&=&{\xi_i(\theta)-\ev(\xi_i(\theta))}, \;\;\;\;S_N(\theta)=\frac{1}{N}\sum_{i=1}^NX_i(\theta)
\end{eqnarray*}
Define metric on $\Theta$ by $d(\theta_1, \theta_2)=\sqrt{\|g_1-g_2\|_2^2/2+\|f_1-f_2\|_2^2/2}$ for $\theta_1=(g_1, f_1), \theta_2=(g_2, f_2)\in \mcF_{\textrm{UB}}^0\times\mcF_{\textrm{UB}}^0$. Moreover, by Lemma \ref{lemma:l2:norm:sup:norm}, we have
\begin{eqnarray*}
&&|\xi_i(\theta_1)-\xi_i(\theta_2)|\\
&\leq&\bigg|\frac{1}{T}\sum_{t=1}^T\bigg(g_1(\bfX_{it})-g_2(\bfX_{it})-\frac{1}{T}\sum_{s=1}^Tg_1(\bfX_{is})+\frac{1}{T}\sum_{s=1}^Tg_2(\bfX_{is})\bigg)\bigg(f_1(\bfX_{it})-\frac{1}{T}\sum_{s=1}^Tf_1(\bfX_{is})\bigg)\bigg|\\
&&+\bigg|\frac{1}{T}\sum_{t=1}^T\bigg(f_1(\bfX_{it})-f_2(\bfX_{it})-\frac{1}{T}\sum_{s=1}^Tf_1(\bfX_{is})+\frac{1}{T}\sum_{s=1}^Tf_2(\bfX_{is})\bigg)\bigg(g_2(\bfX_{it})-\frac{1}{T}\sum_{s=1}^Tg_2(\bfX_{is})\bigg)\bigg|\\
&\leq& 4\|g_1-g_2\|_\infty\|f_1\|_{\infty}+4\|f_1-f_2\|_{\infty}\|g_2\|_\infty\\
&\leq& 4A_{NT}^2\|g_1-g_2\|_2\|f_1\|_{2}+ 4A_{NT}^2\|f_1-f_2\|_{2}\|g_2\|_2\\
&\leq& 8A_{NT}^2d(\theta_1, \theta_2),
\end{eqnarray*}
which further leads to 
\begin{eqnarray*}
	|X_i(\theta_1)-X_i(\theta_2)|\leq 16A_{{NT}}^2d(\theta_1, \theta_2).
\end{eqnarray*}
Similarly, by  Lemma \ref{lemma:expectation:sample:variance} and Lemma \ref{lemma:l2:norm:sup:norm},  we have
\begin{eqnarray*}
	&&\textrm{Var}\bigg(X_i(\theta_1)-X_i(\theta_2)\bigg)\\
	&\leq&\ev\bigg(|\xi_i(\theta_1)-\xi_i(\theta_2)|^2\bigg)\\
	&\leq& 8\|f_1\|_\infty^2\ev\bigg\{\frac{1}{T}\sum_{t=1}^T\bigg(g_1(\bfX_{it})-g_2(\bfX_{it})-\frac{1}{T}\sum_{s=1}^Tg_1(\bfX_{is})+\frac{1}{T}\sum_{s=1}^Tg_2(\bfX_{is})\bigg)^2\bigg\}\\
	&&+8\|g_2\|_\infty^2\ev\bigg\{\frac{1}{T}\sum_{t=1}^T\bigg(f_1(\bfX_{it})-f_2(\bfX_{it})-\frac{1}{T}\sum_{s=1}^Tf_1(\bfX_{is})+\frac{1}{T}\sum_{s=1}^Tf_2(\bfX_{is})\bigg)^2\bigg\}\\
	&=&8\|f_1\|_\infty^2\ev\bigg(\zeta_i(g_1-g_2,g_1-g_2)\bigg)+8\|g_2\|_\infty^2\ev\bigg(\zeta_i(f_1-f_2,f_1-f_2)\bigg)\\
	&\leq&8A_{NT}^2\bigg(\|f_1\|_2^2\|g_1-g_2\|^2+\|g_2\|_2^2\|f_1-f_2\|^2\bigg)\\
	&\leq&8a_1A_{NT}^2\bigg(\|f_1\|_2^2\|g_1-g_2\|_2^2+\|g_2\|_2^2\|f_1-f_2\|_2^2\bigg)\\
	&\leq&16a_1A_{NT}^2d^2(\theta_1,\theta_2).
\end{eqnarray*}
By Bernstein inequality, it follows that
\begin{eqnarray}
	\pr\bigg(\bigg|S_N(\theta_1)-S_N(\theta_2)\bigg)\bigg|>xs\bigg)&\leq&2\exp\bigg(-\frac{Nx^2s^2}{32a_1A_{NT}^2\{d^2(\theta_1, \theta_2)+d(\theta_1,\theta_2)xs\}}\bigg)\nonumber\\
	&\leq&2\exp\bigg(-\frac{Nx^2s^2}{64a_1A_{NT}^2d^2(\theta_1, \theta_2)}\bigg)+2\exp\bigg(-\frac{Nxs}{64a_1A_{NT}^2d(\theta_1,\theta_2)}\bigg).\nonumber\\ \label{eq:lemma:uniform:equivalence:empirical:population:norm:eq1}
\end{eqnarray}

Let $\delta_k=3^{-k}$ for $k\geq 0$. For sufficient large integer $K$, which will be specified later, let $\{0\}=\mcH_0\subset \mcH_1 \ldots \mcH_K$ be a sequence of subsets of $\Theta$ such that $\min_{\theta^*\in \mcH_k}d(\theta^*, \theta)\leq \delta_k$ for all $\theta \in \Theta$. Moreover the subsets $\mcH_K$ is chosen inductively such that two different elements in $\mcH_k$ is at least $\delta_k$ apart.  

By definition, the cardinality $\#(\mcH_k)$ of $\mcH_k$ is bounded by the $\delta_k/2$-covering number $D(\delta_k/2, \Theta, d)$. Moreover, since $d(\theta_1, \theta_2)\leq \|g_1-g_2\|_2+\|f_1-f_2\|_2$ for $\theta_1=(g_1, f_1), \theta_2=(g_2, f_2)\in \Theta$, we have
\begin{eqnarray*}
	D(\delta_k/2, \Theta, d)\leq D^2(\delta_k/4, \mcF_{\textrm{UB}}^0, \|\cdot\|_2)\leq \bigg(\frac{16+\delta_k}{\delta_k}\bigg)^{2d_{NT}},
\end{eqnarray*} 
where the last inequality is due to \cite{van2000}[Corollary 2.6] and the fact that $\mcF_{\textrm{UB}}^0$ is a linear space with dimension $d_{NT}$. For any $\theta \in \Theta$, let $\tau_k(\theta)\in \mcH_k$  be a element such that $d(\tau_k(\theta), \theta)\leq \delta_k$, for $k=1,2,\ldots, K$. Now for fixed $x>0$, choose $K=K(N)>0$, which depends on $N$ and is increasing fast enough such that $x>16A_{NT}^2(2/3)^K$, we see from (\ref{eq:lemma:uniform:equivalence:empirical:population:norm:eq1}) that
\begin{eqnarray}
&&\pr\bigg(\sup_{\theta \in \Theta}\bigg| S_N(\theta)\bigg|>x\bigg)\nonumber\\
&\leq&\pr\bigg(\sup_{\theta \in \Theta}\bigg| S_N(\theta)-S_N\bigg(\tau_K(\theta)\bigg)\bigg|>\frac{x}{2^K}\bigg)\nonumber\\
&&+\sum_{k=1}^K\pr\bigg(\sup_{\theta \in \Theta}\bigg| S_N\bigg(\tau_k\circ\ldots\circ\tau_K(\theta)\bigg)-S_N\bigg(\tau_{k-1}\circ\tau_k\circ\ldots\circ\tau_K(\theta)\bigg)\bigg|>\frac{x}{2^{k-1}}\bigg)\nonumber\\
&\leq&\pr\bigg(16A_{NT}^2d(\theta, \tau_K(\theta))>\frac{x}{2^K}\bigg)\nonumber\\
&&+\sum_{k=1}^K\#(\mcH_k)\sup_{\theta_k \in \mcH_k}\pr\bigg(\sup_{\theta \in \Theta}\bigg| S_N\bigg(\theta_k\bigg)-S_N\bigg(\tau_{k-1}(\theta_k)\bigg)\bigg|>\frac{x}{2^{k-1}}\bigg)\nonumber\\
&\leq&0+\sum_{k=1}^K\bigg(\frac{16+3^{-k}}{3^{-k}}\bigg)^{2d_{NT}}\sup_{\theta_k \in \mcH_k}\pr\bigg(\sup_{\theta \in \Theta}\bigg| S_N\bigg(\theta_k\bigg)-S_N\bigg(\tau_{k-1}(\theta_k)\bigg)\bigg|>\frac{x}{2^{k-1}}\bigg)\nonumber\\
&\leq&\sum_{k=1}^K2\bigg(\frac{16+3^{-k}}{3^{-k}}\bigg)^{2d_{NT}}\exp\bigg(-\frac{Nx^2}{64a_1A_{NT}^22^{2(k-1)}d^2(\theta_k, \tau_{k-1}(\theta_k))}\bigg)\nonumber\\
&&+\sum_{k=1}^K2\bigg(\frac{16+3^{-k}}{3^{-k}}\bigg)^{2d_{NT}}\exp\bigg(-\frac{Nx}{64a_1A_{NT}^22^{k-1}d(\theta_k, \tau_{k-1}(\theta_k))}\bigg)\nonumber\\
&\leq&\sum_{k=1}^K2\bigg(\frac{16+3^{-k}}{3^{-k}}\bigg)^{2d_{NT}}\exp\bigg(-\frac{Nx^2}{64a_1A_{NT}^2}\bigg(\frac{3}{2}\bigg)^{2(k-1)}\bigg)\nonumber\\
&&+\sum_{k=1}^\infty2\bigg(\frac{16+3^{-k}}{3^{-k}}\bigg)^{2d_{NT}}\exp\bigg(-\frac{Nx}{64a_1A_{NT}^2}\bigg(\frac{3}{2}\bigg)^{k-1}\bigg)\nonumber\\
&\leq&2\sum_{k=1}^\infty\exp\bigg(4(3+k)d_{NT}-\frac{Nx^2}{64a_1A_{NT}^2}\bigg(\frac{9}{4}\bigg)^{k-1}\bigg)\nonumber\\
&&+2\sum_{k=1}^\infty\exp\bigg(4(3+k)d_{NT}-\frac{Nx}{64a_1A_{NT}^2}\bigg(\frac{3}{2}\bigg)^{k-1}\bigg),\label{eq:lemma:uniform:equivalence:empirical:population:norm:eq2} 
\end{eqnarray}
where we used the fact that $16\times3^k+1\leq 3^{k+3}$ and $\log 3\leq 2$. Now choose $N$ large enough  such that 
\begin{eqnarray}
	512(3+k)a_1A_{NT}^2d_{NT}<Nx^2(9/4)^{k-1} \; \textrm{ and }\;\;512(3+k)a_1A_{NT}^2d_{NT}<N(3/2)^{k-1}x,\label{eq:lemma:uniform:equivalence:empirical:population:norm:eq3} 
\end{eqnarray}
 for all $k\geq 0$, which is possible, since $d_{NT}A_{NT}^2=o(N)$. For all $N$ large enough satisfying (\ref{eq:lemma:uniform:equivalence:empirical:population:norm:eq2}),  (\ref{eq:lemma:uniform:equivalence:empirical:population:norm:eq3}) further leads to
\begin{eqnarray*}
	\pr\bigg(\sup_{\theta \in \Theta}\bigg| S_N(\theta)\bigg|>x\bigg)&\leq&2\sum_{k=1}^\infty\exp\bigg(-\frac{Nx^2}{128a_1A_{NT}^2}\bigg(\frac{9}{4}\bigg)^{k-1}\bigg)+2\sum_{k=1}^\infty\exp\bigg(-\frac{Nx}{128a_1A_{NT}^2}\bigg(\frac{3}{2}\bigg)^{k-1}\bigg)\\
	&\leq&\frac{2}{e}\bigg\{\frac{128a_1A_{NT}^2}{Nx^2}\sum_{k=1}^\infty\bigg(\frac{9}{4}\bigg)^{k-1}+\frac{128a_1A_{NT}^2}{Nx}\sum_{k=1}^\infty\bigg(\frac{3}{2}\bigg)^{k-1}\bigg\}\to 0, \textrm{ as } N \to \infty,
\end{eqnarray*}
where the fact that $A_{NT}^2=o(N)$ and inequality $e^{-x}\leq e^{-1}/x$ are used. Notice that
\begin{eqnarray*}
	\sup_{\theta \in \Theta}\bigg| S_N(\theta)\bigg|=\sup_{g,f \in \Theta_{NT}}\bigg| \frac{\langle f, g\rangle_{NT}-\langle f, g\rangle}{\|f\|_2\|g\|_2}\bigg|,
\end{eqnarray*}
we prove the first result. The second result is also valid according to Lemma \ref{lemma:expectation:sample:variance}. Since the proof of third result is similar to previous two, we omit the proof.
\end{proof}

\begin{lemma}\label{lemma:bounded:sequence:two:norm:difference}
Under Assumption \ref{Assumption:A1} and \ref{Assumption:common}, for element $v_N\in \mathcal{H}$, if $\|v_N\|_\infty<\infty$, then the following holds:
\begin{eqnarray*}
\ev\bigg(\sup_{g\in \Theta_{NT}}\bigg|\frac{\langle v_N, g \rangle_{NT}-\langle v_N, g \rangle}{\|g\|}\bigg|^2\bigg)\leq 4\|v_N\|_\infty^2\frac{d_{NT}}{N},
\end{eqnarray*}
and
\begin{eqnarray*}
\ev\bigg(\sup_{g\in \Theta_{NT}}\bigg|\frac{\langle v_N, g \rangle_{NT}-\langle v_N, g \rangle}{\|g\|_2}\bigg|^2\bigg)\leq 4a_1^2\|v_N\|_\infty^2\frac{d_{NT}}{N}.
\end{eqnarray*}
\end{lemma}
\begin{proof}[Proof of Lemma \ref{lemma:bounded:sequence:two:norm:difference}]
Let $\{\chi_j\}_{j=1}^{d_{NT}}$ be the orthonormal basis of $\Theta_{NT}$ with respect to $\langle \cdot, \cdot  \rangle$. For $g\in \Theta_{NT}$, we can rewrite $g=\sum_{j=1}^{d_{NT}}b_j\chi_j$ and $\|g\|^2=\sum_{j=1}^{d_{NT}}b_j^2$ with $b_j=\langle g, \chi_j\rangle\in \mathbb{R}$.
\begin{eqnarray*}
	|\langle v_N, g \rangle_{NT}-\langle v_N, g \rangle|&=&\bigg|\sum_{j=1}^{d_{NT}}b_j\bigg(\langle v_N, \chi_j \rangle_{NT}-\langle v_N, \chi_j \rangle\bigg)\bigg|\\
	&\leq&\bigg(\sum_{j=1}^{d_{NT}}b_j^2\bigg)^{1/2}\bigg\{\sum_{j=1}^{d_{NT}}\bigg(\langle v_N, \chi_j \rangle_{NT}-\langle v_N, \chi_j \rangle\bigg)^2\bigg\}^{1/2}\\
	&=&\|g\|\bigg\{\sum_{j=1}^{d_{NT}}\bigg(\langle v_N, \chi_j \rangle_{NT}-\langle v_N, \chi_j \rangle\bigg)^2\bigg\}^{1/2},
\end{eqnarray*}
which leads to 
\begin{eqnarray}
	\sup_{g\in \Theta_{NT}}\bigg|\frac{\langle v_N, g \rangle_{NT}-\langle v_N, g \rangle}{\|g\|}\bigg|\leq \bigg\{\sum_{j=1}^{d_{NT}}\bigg(\langle v_N, \chi_j \rangle_{NT}-\langle v_N, \chi_j \rangle\bigg)^2\bigg\}^{1/2}.\label{eq:lemma:bounded:sequence:two:norm:difference:eq1}
\end{eqnarray}
Notice for each $j\in [d_{NT}]$, it follows that
\begin{eqnarray*}
\ev\bigg\{\bigg(\langle v_N, \chi_j \rangle_{NT}-\langle v_N, \chi_j \rangle\bigg)^2\bigg\}=\textrm{Var}\bigg(\langle v_N, \chi_j \rangle_{NT}\bigg)&=&\textrm{Var}\bigg(\frac{1}{N}\sum_{i=1}^N\zeta_i(v_N, \chi_j)\bigg)\\
&=&\frac{1}{N^2}\sum_{i=1}^N \textrm{Var}\bigg(\zeta_i(v_N, \chi_j)\bigg)\\
&\leq&\frac{1}{N^2}\sum_{i=1}^N \ev\bigg(\zeta_i^2(v_N, \chi_j)\bigg)\\
&\leq&  \frac{4\|v_N\|_\infty^2}{N^2}\sum_{i=1}^N \ev\bigg(\zeta_i(\chi_j, \chi_j)\bigg)\\
&=& \frac{4\|v_N\|_\infty^2}{N} \|\chi_j\|^2=\frac{4\|v_N\|_\infty^2}{N}.
\end{eqnarray*}
As a consequence, by taking expectation on both side of (\ref{eq:lemma:bounded:sequence:two:norm:difference:eq1}) and applying Cauchy–Schwarz inequality, we obtain the first result. The second result follows from the first one and Lemma \ref{lemma:expectation:sample:variance}.
\end{proof}

Define event $\Omega_N=\{\mathbb{X}: 1/2\|g\|\leq \|g\|_{NT}\leq 2\|g\|, \textrm{ for all } g\in \Theta_{NT}\}$ and $\pr(\Omega_N)\to 1$ as $N \to \infty$ by Lemma \ref{lemma:uniform:equivalence:empirical:population:norm}. Then on event $\Omega_N$, $\langle \cdot, \cdot \rangle_{NT}$ is a valid inner product in $\Theta_{NT}$.

\begin{lemma}\label{lemma:expect:value:projection:y:f0}
Under Assumption \ref{Assumption:A1} and \ref{Assumption:common}, if $d_{NT}A_{NT}^2=o(N)$, then on event  $\Omega_N$, the following holds,
$$\frac{a_2^{-1}d_{NT}}{NT}\leq \ev\bigg(\|\widehat{f}_*-\widetilde{f}_*\|_{NT}^2\bigg|\mathbb{Z}\bigg)\leq \frac{a_2d_{NT}}{NT},$$
and
$$\ev\bigg(\|\widehat{f}_*-f_0\|_{NT}^2\bigg| \mathbb{Z}\bigg)\geq \ev\bigg(\|\widehat{f}_*-\widetilde{f}_*\|_{NT}^2\bigg|\mathbb{Z}\bigg)\geq\frac{a_2^{-1}d_{NT}}{NT}.$$
\end{lemma}
\begin{proof}[Proof of Lemma \ref{lemma:expect:value:projection:y:f0}]
On event $\Omega_N$, let $\{\phi_j\}_{j=1}^{d_{NT}}$ be the orthonormal basis of $\Theta_{NT}$ with respect to $\langle \cdot, \cdot \rangle_{NT}$. Recall $\widehat{f}_*=P_{NT}Y$, now we have
	\begin{eqnarray*}
		\widehat{f}_*-	\widetilde{f}_*&=&\sum_{j=1}^{d_{NT}}\langle \widehat{f}_*-	\widetilde{f}_*, \phi_j \rangle_{NT}\phi_j\\
		&=&\sum_{j=1}^{d_{NT}}\langle P_{NT}Y-P_{NT}f_0, \phi_j \rangle_{NT}\phi_j\\
		&=&\sum_{j=1}^{d_{NT}}\langle Y-f_0, \phi_j \rangle_{NT}\phi_j.
	\end{eqnarray*}
By definition, it yields that
\begin{eqnarray*}
\langle Y-f_0, \phi_j \rangle_{NT}&=&\frac{1}{NT}\sum_{i=1}^N\sum_{t=1}^T\bigg\{Y_{it}-f_0(\bfX_{it})-\frac{1}{T}\sum_{s=1}^T\bigg(Y_{is}-f_0(\bfX_{is})\bigg)\bigg\}\bigg\{\phi_j(\bfX_{it})-\frac{1}{T}\sum_{s=1}^T\phi_j(\bfX_{is})\bigg\}\\
&=&\frac{1}{N}\sum_{i=1}^N S_{ji},
\end{eqnarray*}
where
\begin{eqnarray*}
S_{ji}&=&\frac{1}{T}\sum_{t=1}^T\bigg(\epsilon_{it}-\frac{1}{T}\sum_{s=1}^T\epsilon_{is}\bigg)\bigg(\phi_j(\bfX_{it})-\frac{1}{T}\sum_{s=1}^T\phi_j(\bfX_{is})\bigg)\\
&=&\frac{1}{T}\bfepsilon_i^\top H {\Phi}_{ji},
\end{eqnarray*}
where
\begin{eqnarray*}
	\Phi_{ji}&=&\begin{pmatrix}
	\phi(\bfX_{i1}), \phi(\bfX_{i1}),\ldots, \phi(\bfX_{iT}) 
	\end{pmatrix}^\top\in \mathbb{R}^{T},\;\; \bfepsilon_i=(\epsilon_{i1}, \epsilon_{i2}, \ldots, \epsilon_{iT})^\top\in \mathbb{R}^T,\\
H&=&I_T-\frac{1}{T}uu^\top \in \mathbb{R}^{T\times T}, \textrm{ with } u=(1,1,\ldots, 1)^\top\in \mathbb{R}^T.
\end{eqnarray*}
Combining Assumption \ref{Assumption:common}.\ref{Ac:a2} and above equalities, we show that 
\begin{eqnarray*}
	\ev(\widehat{f}_* |\mathbb{Z})=\widetilde{f}_*.	
\end{eqnarray*}
Moreover, Assumption \ref{Assumption:common}.\ref{Ac:a2} also implies $\ev(S_{ji_1}S_{ji_2}| \mathbb{Z})=0$ when $i_1 \neq i_2$ and
\begin{eqnarray*}
	\ev(S_{ji}^2|\mathbb{Z})&=&\frac{1}{T^2}\ev\bigg(\Phi_{ji}^\top H\bfepsilon_i\bfepsilon^\top H\Phi_{ji} \bigg|\mathbb{Z}\bigg)\\
	&=& \frac{1}{T^2}\Phi_{ji}^\top H\ev\bigg(\bfepsilon_i\bfepsilon^\top  \bigg|\mathbb{Z}\bigg)H\Phi_{ji}\\
	&\leq&\frac{a_2}{T^2}\Phi_{ji}^\top H\Phi_{ji}\\
	&=&\frac{a_2}{T}\frac{1}{T}\sum_{t=1}^T\bigg(\phi_j(\bfX_{it})-\frac{1}{T}\sum_{s=1}^T\phi_j(\bfX_{is})\bigg)^2
\end{eqnarray*}
Therefore, it follows that
\begin{eqnarray*}
\ev(\|\widehat{f}_*-	\widetilde{f}_*\|_{NT}^2|\mathbb{Z})&=&\sum_{j=1}^{d_{NT}}\ev(\langle Y-f_0, \phi_j \rangle_{NT}^2|\mathbb{Z})\\
&=&\sum_{j=1}^{d_{NT}}\frac{1}{N^2}\sum_{i=1}^N\ev(S_{ji}^2|\mathbb{Z})\\
&\leq&\frac{a_2}{N}\sum_{j=1}^{d_{NT}}\frac{1}{NT^2}\sum_{i=1}^N\sum_{t=1}^T\bigg(\phi_j(\bfX_{it})-\frac{1}{T}\sum_{s=1}^T\phi_j(\bfX_{is})\bigg)^2\\
&=&\frac{a_2}{NT}\sum_{j=1}^{d_{NT}}\|\phi_j\|_{NT}^2\leq \frac{a_2d_{NT}}{NT}, \textrm{ for all } \mathbb{Z}\in \Omega_N,
\end{eqnarray*}
which is the upper bound. Similarly, utilizing Assumption \ref{Assumption:common}.\ref{Ac:a2} , we also can show that
\begin{eqnarray*}
	\ev(\|\widehat{f}_*-	\widetilde{f}_*\|_{NT}^2|\mathbb{Z})\geq \frac{a_2^{-1}d_{NT}}{NT}, \textrm{ for } \mathbb{Z}\in \Omega_N.
\end{eqnarray*}

For second inequality, by definition of $\widetilde{f}_*$, on event $\Omega_N$, it follows that
\begin{eqnarray*}
	\|\widehat{f}_*-f_0\|_{NT}^2&=&\|\widehat{f}_*-\widetilde{f}_*+\widetilde{f}_*-f_0\|_{NT}^2\\
	&=&\|\widehat{f}_*-\widetilde{f}_*\|_{NT}^2+\|\widetilde{f}_*-f_0\|_{NT}^2+2\langle\widehat{f}_*-\widetilde{f}_*,  \widetilde{f}_*-f_0\rangle_{NT}\\
	&=&\|\widehat{f}_*-\widetilde{f}_*\|_{NT}^2+\|\widetilde{f}_*-f_0\|_{NT}^2+2\langle P_{NT}Y-P_{NT}f_0,  P_{NT}f_0-f_0\rangle_{NT}\\
	&=&\|\widehat{f}_*-\widetilde{f}_*\|_{NT}^2+\|\widetilde{f}_*-f_0\|_{NT}^2\\
	&\geq&\|\widehat{f}_*-\widetilde{f}_*\|_{NT}^2,
\end{eqnarray*}
after taking conditional expectation, we finish the proof.
\end{proof}

\begin{lemma}\label{lemma:rate:of:convergence:non:penalized}
Suppose Assumption \ref{Assumption:A1} and \ref{Assumption:common} hold and  $d_{NT}A_{NT}^2=o(N)$, then on event $\Omega_N$, the following statements are true,
\begin{eqnarray*}
\ev\bigg(\|\widehat{f}_*-f_0\|_{NT}^2\bigg|\mathbb{Z}\bigg)\leq 2500a_2^2\bigg(\frac{d_{NT}}{N}+\rho_{NT}^2\bigg),\; \; \ev\bigg(\|\widehat{f}_*-f_0\|^2\bigg|\mathbb{Z}\bigg)\leq 2500a_2^2\bigg(\frac{d_{NT}}{N}+\rho_{NT}^2\bigg),
\end{eqnarray*}
and
\begin{eqnarray*}
\ev\bigg(\|\widehat{f}_{*,0}-f_0\|_{NT}^2\bigg|\mathbb{Z}\bigg)\leq 2500a_2^2\bigg(\frac{d_{NT,0}}{N}+\rho_{NT}^2\bigg),\; \ev\bigg(\|\widehat{f}_{*,0}-f_0\|^2\bigg|\mathbb{Z}\bigg)\leq 2500a_2^2\bigg(\frac{d_{NT,0}}{N}+\rho_{NT}^2\bigg).
\end{eqnarray*}
\end{lemma}
\begin{proof}[Proof of Lemma \ref{lemma:rate:of:convergence:non:penalized}]
	Let $\widetilde{f}_*=P_{NT}f_0$ and $\bar{f}_*=Pf_0$.  Lemme \ref{lemma:expect:value:projection:y:f0} implies that on event $\Omega_N$, it holds that
\begin{eqnarray}\label{eq:lemma:rate:of:convergence:non:penalized:eq1}
	      \frac{a_2^{-1}d_{NT}}{NT}   \leq  \|\widehat{f}_*-	\widetilde{f}_*\|_{NT}^2\leq \frac{a_2d_{NT}}{NT},
\end{eqnarray}
Now by definition of $\Omega_N$, we have
\begin{eqnarray}\label{eq:lemma:rate:of:convergence:non:penalized:eq1:1}
	 \frac{a_2^{-1}d_{NT}}{4NT}  \leq \|\widehat{f}_*-	\widetilde{f}_*\|^2\leq \frac{4a_2d_{NT}}{NT}.
\end{eqnarray}

Next we will deal with $\widetilde{f}_*-\bar{f}_*$. By definition we have
\begin{eqnarray*}
	\|\widetilde{f}_*-\bar{f}_*\|_{NT}&=&\sup_{g\in \Theta_{NT}}\bigg|\frac{\langle \widetilde{f}_*-\bar{f}_*, g\rangle_{NT}}{\|g\|_{NT}}\bigg|\\
	&=&\sup_{g\in \Theta_{NT}}\bigg|\frac{\langle P_{NT}f_0-Pf_0, g\rangle_{NT}}{\|g\|_{NT}}\bigg|\\
	&=&\sup_{g\in \Theta_{NT}}\bigg|\frac{\langle f_0-Pf_0, g\rangle_{NT}-\langle f_0-Pf_0, g\rangle}{\|g\|_{NT}}\bigg|,
\end{eqnarray*}
where the fact $\langle f_0-Pf_0, g\rangle=0$ is used. Let $g_* \in \Theta_{NT}$ satisfy that $\|g_*-f_0\|_\infty\leq \rho_{NT}$, where the existence of such $g_*$ is guaranteed by Lemma \ref{lemma:approximation:error}.
Since
\begin{eqnarray*}
	\langle f_0-Pf_0, g\rangle_{NT}-\langle f_0-Pf_0, g\rangle=\langle f_0-g_*, g\rangle_{NT}-\langle f_0-g_*, g\rangle+\langle g_*-Pf_0, g\rangle_{NT}-\langle g_*-Pf_0, g\rangle,
\end{eqnarray*}
we have $\|\widetilde{f}_*-\bar{f}_*\|_{NT}\leq R_1+R_2+R_3$, where
\begin{eqnarray*}
R_1&=&\sup_{g\in \Theta_{NT}}\bigg|\frac{ \langle f_0-g_*, g\rangle_{NT}-\langle f_0-g_*, g\rangle}{\|g\|_{NT}}\bigg|,\\
R_2&=&\sup_{g\in \Theta_{NT}}\bigg|\frac{\langle g_*-Pf_0, g\rangle_{NT}}{\|g\|_{NT}}\bigg|,\\
R_3&=&\sup_{g\in \Theta_{NT}}\bigg|\frac{\langle g_*-Pf_0, g\rangle}{\|g\|_{NT}}\bigg|.
\end{eqnarray*}
Since $\|f_0-g_*\|_\infty\leq \rho_{NT}$, by Lemma \ref{lemma:bounded:sequence:two:norm:difference} and definition of $\Omega_N$, we have 
\begin{eqnarray*}
	R_1^2\leq \frac{16\rho_{NT}^2d_{NT}}{N}\leq 16\rho_{NT}^2\;\; \textrm{ on event }\Omega_N.
\end{eqnarray*}

By Lemma \ref{lemma:uniform:equivalence:empirical:population:norm} and triangle inequality, on event $\Omega_N$, we have
\begin{eqnarray*}
	R_2\leq \|g_*-Pf_0\|_{NT}\leq 2\|g_*-Pf_0\|\leq 2\|g_*-f_0\|+2\|Pf_0-f_0\|\leq 4\|g_*-f_0\|\leq 8\rho_{NT},
\end{eqnarray*}
where we use the fact that $\|g\|\leq 2\|g\|_\infty$ and $\|Pf_0-f_0\|\leq \|g_*-f_0\|$.

On event $\Omega_N$, we have
\begin{eqnarray*}
	R_3\leq 2\sup_{g\in \Theta_{NT}}\bigg|\frac{\langle g_*-Pf_0, g\rangle}{\|g\|}\bigg|\leq 2\|g_*-Pf_0\|\leq 8\rho_{NT}.
\end{eqnarray*}
Combining rate of $R_1, R_2, R_3$, we conclude that on event $\Omega_N$,
\begin{eqnarray}\label{eq:lemma:rate:of:convergence:non:penalized:eq2}
	\|\widetilde{f}_*-\bar{f}_*\|_{NT}\leq 20\rho_{NT}.
\end{eqnarray}
and
\begin{eqnarray}\label{eq:lemma:rate:of:convergence:non:penalized:eq2:2}
	\|\widetilde{f}_*-\bar{f}_*\|\leq 40\rho_{NT},
\end{eqnarray}
By definition of projection, on event $\Omega_N$, we have
\begin{eqnarray}
\|\bar{f}_*-f_0\|_{NT}&\leq& \|\bar{f}_*-g_*\|_{NT}+\|g_*-f_0\|_{NT}\nonumber\\
&\leq& 2\|\bar{f}_*-g_*\|+2\|g_*-f_0\|_\infty\nonumber\\
&=&2\|Pf_0-g_*\|+2\|g_*-f_0\|_\infty\nonumber\\
&\leq&2\|Pf_0-f_0\|+2\|f_0-g_*\|+2\|g_*-f_0\|_\infty\nonumber\\
&\leq&4\|f_0-g_*\|+2\|g_*-f_0\|_\infty\nonumber\\ 
&\leq& 10\|f_0-g_*\|_\infty\leq 10\rho_{NT}.\label{eq:lemma:rate:of:convergence:non:penalized:eq3}
\end{eqnarray}
Similarly, we can show
\begin{eqnarray}
	\|\bar{f}_*-f_0\|\leq\|\bar{f}_*-g_*\|+\|g_*-f_0\|\leq 2\|g_*-f_0\|\leq 4\|g_*-f_0\|_\infty\leq 4\rho_{NT}.\label{eq:lemma:rate:of:convergence:non:penalized:eq3:3}
\end{eqnarray}
Combining (\ref{eq:lemma:rate:of:convergence:non:penalized:eq1}), (\ref{eq:lemma:rate:of:convergence:non:penalized:eq2}) and (\ref{eq:lemma:rate:of:convergence:non:penalized:eq3}), we have
\begin{eqnarray*}
\|\widehat{f}_*-f_0\|_{NT}^2\leq 2500a_2^2\bigg(\frac{d_{NT}}{N}+\rho_{NT}^2\bigg)\; \textrm{ on event } \Omega_N.
\end{eqnarray*}
According to (\ref{eq:lemma:rate:of:convergence:non:penalized:eq1:1}), (\ref{eq:lemma:rate:of:convergence:non:penalized:eq2:2}) and (\ref{eq:lemma:rate:of:convergence:non:penalized:eq3:3}), we also obtain
\begin{eqnarray*}
	\|\widehat{f}_*-f_0\|^2\leq 2500a_2^2\bigg(\frac{d_{NT}}{N}+\rho_{NT}^2\bigg)\; \textrm{ on event } \Omega_N.
\end{eqnarray*}
The low bound can be obtained analogically and similar argument can be applied to prove results for $\widehat{f}_{*,0}$.
\end{proof}

\begin{lemma}\label{lemma:rate:non:linear:part}
Under Assumption \ref{Assumption:A1}, it follows that
$$\|g\|^2\geq\frac{1}{2a_1^2}\bigg(\|g_{-}\|^2+\|g_{\sim}\|^2\bigg), \textrm{ for all } g\in \cup_{j=1}^p\Theta_{N,j}.$$
\end{lemma}
\begin{proof}[Proof of Lemma \ref{lemma:rate:non:linear:part}]
This follows from the orthogonal basis and Lemma \ref{lemma:expectation:sample:variance}
\end{proof}

For simplicity, we define the following rate:
\begin{equation}\label{eq:definition:gamma:N}
	\gamma_{NT}^2=\frac{d_{NT}}{NT}+\rho_{NT}^2,
\end{equation}
then $\gamma_{NT}^2\asymp N^{-1}\sum_{j=1}^ph_j^{-1}+\sum_{j=d+1}^p h_j^{2m_j}$ for fixed $T$ and $\gamma_{NT}^2\asymp N^{-1}T^{-1}\sum_{j=1}^ph_j^{-1}+\sum_{j=d+1}^p h_j^{2m_j}$ for diverging $T$.

\begin{lemma}\label{lemma:lower:bound:non:linear}
Let $r_N$ be a real number satisfying $r_N\geq \rho_{NT}$, and $r_N=o(1)$.  Suppose $g=\sum_{j=1}^p g_j\in \Theta_{NT}$ with $g_j \in \Theta_{N, j}$ for $j=1,2,\ldots, p$ such that $\|g-f_0\|\leq Kr_N$, for some constant $K$ not depending on $N$. The following statements are true:
\begin{enumerate}
\item If  Assumptions \ref{Assumption:A1} and \ref{Assumption:common} hold,  then
\begin{eqnarray*}
	\|g_{j, \sim}\|\geq \frac{1}{2}\sqrt{\frac{a_6}{2a_1}}, \textrm{ for } j=d+1,\ldots, p.
\end{eqnarray*}
provided $\sqrt{a_6}\geq2(K+4)\sqrt{2a_1^3c_1^{-1}}r_N$.
\item If Assumptions \ref{Assumption:common}, \ref{Assumption:A2} hold and $A_{NT}^2=o(T)$, then
\begin{eqnarray*}
	\|g_{j, \sim}\|\geq \frac{1}{2}\sqrt{\frac{a_6}{2a_3}}, \textrm{ for } j=d+1,\ldots, p.
\end{eqnarray*}
provided $\sqrt{a_6}>2(K+4)\sqrt{8a_3^3c_1^{-1}}r_N$.
\end{enumerate}
Here $g_{j, \sim}$ is the non linear component of $g_j$ for $j=d+1,\ldots, p$. 
\end{lemma}
\begin{proof}[Proof of Lemma \ref{lemma:lower:bound:non:linear}]
We only prove the result for fixed $T$ under  Assumptions \ref{Assumption:A1} and \ref{Assumption:common}. The case for diverging $T$ under Assumptions \ref{Assumption:common} and \ref{Assumption:A2} can be proved similarly.

Let $Pf_{j,0}=g_{j,0}=g_{j,0,-}+g_{j,0,\sim}\in \Theta_{NT}$ and $g=\sum_{j=1}^p g_j,$ with $g_j \in \Theta_{N, j}$ for $j=1,2,\ldots, p$, then  Lemma \ref{lemma:expectation:sample:variance} (Lemma \ref{lemma:expectation:sample:variance:large:T} for diverging $T$) and Lemma \ref{lemma:sum:of:l2:norm:bound} imply that
\begin{eqnarray*}
	\|g_{j,\sim}-g_{j,0,\sim}\|^2&\leq&a_1\|g_{j,\sim}-g_{j,0,\sim}\|_2^2\\
	&=&a_1\bigg(\|g_{j}-g_{j,0}\|_2^2-\|g_{j,-}-g_{j,0,-}\|_2^2\bigg)\\
	&\leq&a_1\|g_{j}-g_{j,0}\|_2^2\leq \frac{a_1}{c_1}\|Pf_0-g\|_2^2\leq\frac{a_1^2}{c_1}\|Pf_0-g\|^2.
\end{eqnarray*}
Therefore, by triangle inequality and Lemma \ref{lemma:approximation:error}, we have
\begin{eqnarray*}
	\|g_{j,\sim}-g_{j,0,\sim}\|\leq \sqrt{a_1^2c_1^{-1}}\|Pf_0-g\|&\leq& \sqrt{a_1^2c_1^{-1}}(\|Pf_0-f_0\|+\|g-f_0\|)\\
	&\leq&\sqrt{a_1^2c_1^{-1}}(\|g_*-f_0\|+\|g-f_0\|) \\
&\leq&\sqrt{a_1^2c_1^{-1}}(2\|g_*-f_0\|_\infty+5C_\delta r_N)\\
	&\leq&\sqrt{a_1^2c_1^{-1}}(2\rho_{NT}+Kr_N)\\
	&\leq&(K+2)\sqrt{a_1^2c_1^{-1}}r_N,
\end{eqnarray*}
and
\begin{eqnarray*}
	\|Pf_{j,0}-f_{j,0}\|\leq \|g_{j,*}-f_{j,0}\|\leq 2\rho_{NT}.
\end{eqnarray*}
As a consequence,  by Assumption \ref{Assumption:common}.\ref{Ac:a} and Lemma \ref{lemma:expectation:sample:variance} (Lemma \ref{lemma:expectation:sample:variance:large:T} for diverging $T$), on event $U_{N, \delta}$(on event $U_{NT, \delta}$ for diverging $T$), it follows that
\begin{eqnarray*}
	\|g_{j, \sim}\|&=& \|g_{j,\sim}-Pf_{j,0}+Pf_{j,0}-f_{j,0}+f_{j,0}\|\\
	&=& \|g_{j,\sim}-g_{j,0,\sim}-g_{j,0,-}+Pf_{j,0}-f_{j,0}+f_{j,0}\|\\
	&\geq&\|f_{j,0}-g_{j,0,-}\|-\|Pf_{j,0}-f_{j,0}\|-\|g_{j,\sim}-g_{j,0,\sim}\|\\
	&\geq&\sqrt{\frac{1}{2a_1}}\|f_{j,0}-g_{j,0,-}\|_2-2\rho_{NT}-(K+2)\sqrt{a_1^2c_1^{-1}}r_N  \\
	&\geq& \sqrt{\frac{1}{2a_1}}\|f_{j,0}-g_{j,0,-}\|_2-(K+4)\sqrt{a_1^2c_1^{-1}}r_N  \\
	&\geq&\sqrt{\frac{a_6}{2a_1}}-(K+4)\sqrt{a_1^2c_1^{-1}}r_N\\
	&\geq&\frac{1}{2}\sqrt{\frac{a_6}{2a_1}},
\end{eqnarray*}
where the last inequality follows from the rate condition: $\sqrt{a_6}\geq2(K+4)\sqrt{2a_1^3c_1^{-1}}r_N$.
\end{proof}

\begin{proof}[Proof of (a) in Theorem \ref{thm:rate:of:convergence:together}]
We define
\begin{eqnarray*}
	R_{NT}=\sup_{g\in \Theta_{NT}}\bigg|\frac{\|g\|_{NT}^2}{\|g\|^2}-1\bigg|,
\end{eqnarray*}
and by definition, it follows that
\begin{eqnarray*}
	1-R_{NT}\leq \sup_{g\in \Theta_{NT}}\frac{\|g\|_{NT}^2}{\|g\|^2}\leq 1+R_{NT}.
\end{eqnarray*}
For any $0<\delta<1$, define event $$U_{N, \delta}=\bigg\{\|\widehat{f}_*-f_0\|\leq C_\delta\gamma_{NT}, \|\widehat{f}_{*,0}-f_0\|\leq C_\delta\gamma_{NT}, R_{NT}\leq \frac{1}{2}\bigg\}\cap \Omega_N,$$ where $C_\delta>0$ is sufficiently large such that $\pr(U_{N,\delta})\geq 1-\delta$ and this is possible due to Lemma \ref{lemma:rate:of:convergence:non:penalized} and Lemma \ref{lemma:uniform:equivalence:empirical:population:norm}.

By definition of $\widehat{f}$, we have
\begin{eqnarray}
	0&\geq &l_{NT}(\widehat{f})-l_{NT}(\widehat{f}_{*,0})\nonumber\\
	&\geq& \frac{\|Y-\widehat{f} \|_{NT}^2-\|Y-\widehat{f}_{*,0}\|_{NT}^2}{2}+\sum_{j=d+1}^p\bigg\{p_{\lambda_{NT}}(\|\widehat{f} _{j,\sim}\|_{NT})-p_{\lambda_{NT}}(\|\widehat{f}_{j,*,0,\sim}\|_{NT})\bigg\}\nonumber\\
	&=& \frac{\|\widehat{f}_*-\widehat{f}\|_{NT}^2-\|\widehat{f}_*-\widehat{f}_{*,0}\|_{NT}^2}{2}+\sum_{j=d+1}^p\bigg\{p_{\lambda_{NT}}(\|\widehat{f} _{j,\sim}\|_{NT})-p_{\lambda_{NT}}(\|\widehat{f}_{j,*,0,\sim}\|_{NT})\bigg\}\label{thm:rate:of:convergence:eq:1}\\
	&\geq&\frac{1}{2}\bigg(\|\widehat{f}_*-\widehat{f} \|^2(1-R_{NT})-\|\widehat{f}_*-\widehat{f}_{*,0}\|^2(1+R_{NT})\bigg)-\sum_{j=d+1}^pp_{\lambda_{NT}}(\|\widehat{f}_{j,*,0,\sim}\|_{NT}).\label{thm:rate:of:convergence:eq:2}
\end{eqnarray}
By Lemma \ref{lemma:lower:bound:non:linear}, on event $U_{N,\delta}$, we have
\begin{eqnarray*}
	\|\widehat{f}_{j,*,0, \sim}\|_{NT}\geq \frac{1}{2} \|\widehat{f}_{j,*,0, \sim}\|\geq \frac{1}{4}\sqrt{\frac{a_6}{2a_1}}, \textrm{ for } j=d+1,\ldots, p,
\end{eqnarray*}
provided  $\sqrt{a_6}>2(C_\delta+4)\sqrt{2a_1^3c_1^{-1}}\gamma_{NT}$.
As a consequence, it follows that
\begin{eqnarray}
	p_{\lambda_{NT}}(\|\widehat{f}_{j,*,0, \sim}\|_{NT})=(\kappa+1)\lambda_{NT}^2/2, \textrm{ for } j=d+1,\ldots, p,  \textrm{ if }\; \frac{1}{4}\sqrt{\frac{a_6}{2a_1}}\geq \kappa\lambda_{NT}.\label{thm:rate:of:convergence:eq:3}
\end{eqnarray}

Combining  (\ref{thm:rate:of:convergence:eq:2}) and  (\ref{thm:rate:of:convergence:eq:3}), if $\sqrt{a_6}>2(C_\delta+4)\sqrt{2a_1^3c_1^{-1}}\gamma_{NT}$ and $\sqrt{a_6}\geq 4\sqrt{2a_1}\kappa\lambda_{NT}$, on event $U_{N,\delta}$, it follows that
\begin{eqnarray*}
\frac{(p-d)(\kappa+1)\lambda_{NT}^2}{2}+\frac{3}{2}\|\widehat{f}_*-\widehat{f}_{*,0}\|^2\geq \frac{1}{4}\|\widehat{f}_*-\widehat{f} \|^2.
\end{eqnarray*}
Taking square root on both side of above inequality, we have
\begin{eqnarray*}
	\frac{1}{2}\|\widehat{f}_*-\widehat{f} \|\leq \sqrt{{(p-d)(\kappa+1)}}\lambda_{NT}+2\|\widehat{f}_*-\widehat{f}_{*,0}\|,
\end{eqnarray*}
which, by triangle inequality,  further implies that  the following holds on event $U_{N, \delta}$,
\begin{eqnarray*}
	\|\widehat{f}-f_0\|&\leq& 2\sqrt{{(p-d)(\kappa+1)}}\lambda_{NT}+5\|\widehat{f}_*-f_0\|+4\|\widehat{f}_{*,0}-f_0\|\\
	&\leq&2\sqrt{{(p-d)(\kappa+1)}}\lambda_{NT}+9C_\delta\gamma_{NT}.
\end{eqnarray*}
Again by Lemma \ref{lemma:lower:bound:non:linear}, on event $U_{N, \delta}$, it holds that
\begin{eqnarray*}
	\|\widehat{f}_{j, \sim}\|_{NT}\geq \frac{1}{2} \|\widehat{f}_{j,\sim}\|\geq \frac{1}{4}\sqrt{\frac{a_6}{2a_1}}, \textrm{ for } j=d+1,\ldots, p,
\end{eqnarray*}
provided  $\sqrt{a_6}>2(2\sqrt{{(p-d)(\kappa+1)}}+9C_\delta+4)\sqrt{2a_1^3c_1^{-1}}(\lambda_{NT}+\gamma_{NT})$. As a consequence, we have
\begin{eqnarray}
	p_{\lambda_{NT}}(\|\widehat{f}_{j, \sim}\|_{NT})=(\kappa+1)\lambda_{NT}^2/2, \textrm{ for } j=d+1,\ldots, p,  \textrm{ if }\; \frac{1}{4}\sqrt{\frac{a_6}{2a_1}}\geq \kappa\lambda_{NT}.\label{thm:rate:of:convergence:eq:4}
\end{eqnarray}
Now in the view of  (\ref{thm:rate:of:convergence:eq:1}),  (\ref{thm:rate:of:convergence:eq:3}) and (\ref{thm:rate:of:convergence:eq:4}), the following holds on event $U_{N, \delta}$,
\begin{eqnarray*}
	\|\widehat{f}_*-\widehat{f}\|\leq 2\|\widehat{f}_*-\widehat{f}\|_{NT}\leq 2\|\widehat{f}_*-\widehat{f}_{*,0}\|_{NT}\leq 4\|\widehat{f}_*-\widehat{f}_{*,0}\|,
\end{eqnarray*}
which further implies 
\begin{eqnarray*}
	\|\widehat{f}-f_0\|&\leq&\|\widehat{f}-\widehat{f}_*\|+\|\widehat{f}_*-f_0\|\\
	&\leq&4\|\widehat{f}_*-\widehat{f}_{*,0}\|+\|\widehat{f}_*-f_0\|\\
	&\leq&5\|\widehat{f}_*-f_0\|+4\|\widehat{f}_{*,0}-f_0\|\\
	 &\leq& 9C_\delta\gamma_{NT}.
\end{eqnarray*}
Since $\delta$ can be arbitrary small, we finish the proof.
\end{proof}

\begin{proof}[Proof of (a) Theorem \ref{thm:selection:consistency:together}]
Fixing $C_\delta>0$ large enough, we need to show that for any $g=\sum_{j=1}^pg_j\in \Theta_{NT}^0$ with $\|g-f_0\|\leq C_\delta\gamma_{NT}$  and any $C>0$ one has $l_{NT}(g)=\min_{r\in \Theta_{NT}^1, \|r\|\leq C\gamma_{NT}}l_{NT}(g+r)$. 

Since $r=\sum_{j=1}^d{r_{j,\sim}}$ with $r_{j, \sim}\in \Theta_{NT,j,\sim}$ for $j=1,\ldots, d$, by Lemma \ref{lemma:sum:of:l2:norm:bound}, we have
\begin{eqnarray}
	\|r_{j,\sim}\|\leq c_1^{-1}\|r\|\leq c_1^{-1}C\gamma_{NT}.\label{eq:thm:selection:consistency:eq1}
\end{eqnarray}
Define event
\begin{eqnarray*}
	E_{N,\delta}=\{\|\widehat{f}_*-f_0\|\leq C_\delta\gamma_{NT}\}\cap \Omega_N,
\end{eqnarray*}
so on event $E_{N,\delta}$, we have $\|r_{j,\sim}\|_{NT}\leq 2c_1^{-1}C\gamma_{NT}.$ Moreover, we can select $C_\delta$ sufficient large such that $\pr(E_{N,\delta})\geq 1-\delta$, which is feasible by Lemma \ref{lemma:rate:of:convergence:non:penalized} and Lemma \ref{lemma:uniform:equivalence:empirical:population:norm}.
Direct calculation shows
\begin{eqnarray*}
	l_{NT}(g)-l_{NT}(g+r)&=&\frac{\|Y-g\|_{NT}^2-\|Y-g-r\|_{NT}^2}{2}-\sum_{j=1}^dp_{\lambda_{NT}}(\|r_{j,\sim}\|_{NT})\\
	&=&\frac{\|\widehat{f}_*-g\|_{NT}^2-\|\widehat{f}_*-g-r\|_{NT}^2}{2}-\sum_{j=1}^dp_{\lambda_{NT}}(\|r_{j,\sim}\|_{NT})\\
	&\leq& \frac{1}{2}\|r\|_{NT}\bigg(\|\widehat{f}_*-g\|_{NT}+\|\widehat{f}_*-g-r\|_{NT}\bigg)-\sum_{j=1}^dp_{\lambda_{NT}}(\|r_{j,\sim}\|_{NT})\\
	&\equiv&\frac{1}{2}S_1+S_2.
\end{eqnarray*}
Notice on event $E_{N,\delta}$, we have
\begin{eqnarray*}
	S_1&\leq&2\|r\|_{NT}\bigg(\|\widehat{f}_*-g\|+\|\widehat{f}_*-g-r\|\bigg)\\
	&\leq&2\|r\|_{NT}\bigg(2\|\widehat{f}_*-g\|+\|r\|\bigg)\\
	&\leq& 2\|r\|_{NT}\bigg(2\|\widehat{f}_*-f_0\|+2\|g-f_0\|+\|r\|\bigg)\\
	&\leq&2(4C_\delta+C)\gamma_{NT}\|r\|_{NT}\\
	&\leq&2(4C_\delta+C)\gamma_{NT}\sum_{j=1}^d\|r_{j,\sim}\|_{NT}.
\end{eqnarray*}

By (\ref{eq:thm:selection:consistency:eq1}), on event $E_{N, \delta}$, if $2c_1^{-1}C\gamma_{NT}\leq \lambda_{NT}$, we have 
\begin{eqnarray*}
	p_{\lambda_{NT}}(\|r_{j,\sim}\|_{NT})=\lambda_{NT}\|r_{j,\sim}\|_{NT}.
\end{eqnarray*}
Therefore, it follows that
\begin{eqnarray*}
	l_{NT}(g)-l_{NT}(g+r)&\leq&2(4C_\delta+C)\gamma_{NT}\sum_{j=1}^d\|r_{j,\sim}\|_{NT}-\sum_{j=1}^dp_{\lambda_{NT}}(\|r_{j,\sim}\|_{NT})\\
	&=&2(4C_\delta+C)\gamma_{NT}\sum_{j=1}^d\|r_{j,\sim}\|_{NT}-\sum_{j=1}^d\lambda_{NT}\|r_{j,\sim}\|_{NT}\\
	&=&\bigg(2(4C_\delta+C)\gamma_{NT}-\lambda_{NT}\bigg)\sum_{j=1}^d\|r_{j,\sim}\|_{NT}<0,
\end{eqnarray*}
provided $2(4C_\delta+C)\gamma_{NT}<\lambda_{NT}$.  

Now we prove that for each on event $E_{N,\delta}$, for all $g\in \Theta_{NT}^0$ with $\|g-f_0\|\leq C_{\delta} \gamma_{NT}$ and all  $r\in \Theta_{NT}^1$ with $\|r\|\leq C\gamma_{NT}$ for any $C>0$, we have
\begin{eqnarray}
	l_{NT}(g)=\min_{r\in \Theta_{NT}^1, \|r\|\leq C\gamma_{NT}}l_{NT}(g+r),\label{eq:thm:selection:consistency:eq2}
\end{eqnarray} 
provided $2(4C_\delta+c_1^{-1}C)\gamma_{NT}<\lambda_{NT}$.  Furthermore let event $F_{N,\delta}=\{\|\widehat{f}-f_0\|\leq C_\delta \gamma_{NT}\}$ and choose $C_\delta$ large such that $\pr(F_{N,\delta})\geq 1-\delta$. Then we have $\pr(E_{N,\delta}\cap F_{N,\delta})\geq 1-2\delta$. By (\ref{eq:thm:selection:consistency:eq2}), we have with probability at least $1-2\delta$,
\begin{eqnarray*}
	l_{NT}(\widehat{f})= \min_{r\in \Theta_{NT}^1, \|r\|\leq C\gamma_{NT}}l_{NT}(\widehat{f}+r),
\end{eqnarray*}
provided $2(4C_\delta+c_1^{-1}C)\gamma_{NT}<\lambda_{NT}$, which proves the first conclusion.

By Lemma \ref{lemma:lower:bound:non:linear}, we can see that with probability at least $1-2\delta$, 
\begin{eqnarray*}
		\|\widehat{f}_{j,\sim}\|\geq \frac{1}{2}\sqrt{\frac{a_6}{2a_1}}>0, \textrm{ for } j=d+1,\ldots, p,
\end{eqnarray*}
provided $\sqrt{a_6}>2(C_\delta+4)\sqrt{2a_1^3c_1^{-1}}\gamma_{NT}$, which is the second conclusion.
\end{proof}

\subsection{Proof of Theorems \ref{thm:rate:of:convergence:together} and \ref{thm:selection:consistency:together} -- Diverging $T$ Case}
\begin{proposition}\label{proposition:calculation:sub:exponential:mgf}
Let $\xi$ be a random variable with zero mean. If there exist positive constants $A, B$ such that $E(e^{\lambda \xi})\leq \exp(\frac{A\lambda^2}{1-B\lambda})$ for all $0\leq \lambda< 1/B$, then we have
\begin{eqnarray}
	\pr(\xi>x)\leq \exp\bigg(-\frac{x^2}{4A+2Bx}\bigg), \textrm{ for all } x\geq 0.\nonumber
\end{eqnarray}
\end{proposition}
\begin{proof}[Proof of Proposition \ref{proposition:calculation:sub:exponential:mgf}]
Markov's inequality and the bound of $E(e^{\lambda \xi})$ yield 
\begin{eqnarray}
		\pr(\xi>x)&\leq& \exp\bigg(\frac{A\lambda^2}{1-B\lambda}-\lambda x\bigg).\nonumber
\end{eqnarray}
To finish the proof, we evaluate right side of above inequality at $\lambda=\frac{x}{Bx+2A}$.
\end{proof}
\begin{lemma}[Bernstein Inequality under Strong Mixing]\label{theorem:FME2009:theorem:1}
Let $\{\xi_t, t\geq 1\}$ be a sequence of centered real-valued random variables. Suppose that the sequence satisfies that alpha mixing coefficients $\alpha(t)\leq c\rho^t$ for some $\rho \in (0, 1)$, all $t\geq 0$ and $\sup_{t \geq 1}|\xi_t|\leq M$ for some $M>0$. Then there are positive constant $L_1, L_2$ depending only on $c$ and $\rho$ such that for all $T\geq 4$ and $x$ satisfying $0\leq \lambda <\frac{1}{L_1M (\log T) (\log \log T)}$, we have
\begin{eqnarray}
	\log \ev\bigg(\exp(\lambda\sum_{t=1}^T \zeta_t )\bigg)\leq \frac{L_2M^2T\lambda^2}{1-L_1M (\log T) (\log \log T)\lambda}.\nonumber
\end{eqnarray}
\end{lemma}
\begin{proof}[Proof of Lemma \ref{theorem:FME2009:theorem:1}]
This is Theorem 1 from \cite{mpr09}. 
\end{proof}
\begin{proposition}\label{proposition:panel:sub:exponential:inequality}
Let $\{\xi_{t},  t\geq 1\}$ be a sequence of centered real-valued random variables. Suppose the sequence $\{\zeta_{t},  t\geq 1\}$ satisfies that alpha mixing    coefficients $\alpha(t)\leq c\rho^t$ for some $\rho \in (0, 1)$ and all $t\geq 0$. Moreover, $\sup_{t\geq 1}|\zeta_{t}|\leq M$, for some $M>0$. Then there are positive constants $L_3$ relying only on $c$ such that  for all $T\geq 4$ and $x\geq 0$
\begin{eqnarray}
	\pr\bigg(\bigg|\frac{1}{T}\sum_{t=1}^T \zeta_{t}\bigg|>x\bigg)\leq 2\exp\bigg(-\frac{L_3Tx^2}{M^2+M  (\log T) (\log \log T)x}\bigg).\nonumber
\end{eqnarray}
\end{proposition}
\begin{proof}[Proof of Proposition \ref{proposition:panel:sub:exponential:inequality}]
Lemma \ref{theorem:FME2009:theorem:1} implies that
\begin{eqnarray}
	\log \ev\bigg(\exp(\lambda\sum_{t=1}^T \zeta_{t} )\bigg)\leq \frac{L_2M^2T\lambda^2}{1-L_1M (\log T) (\log \log T)\lambda},\nonumber
\end{eqnarray}
where $L_1, L_2$ are positive constants relying on $c$ only. Applying Proposition \ref{proposition:calculation:sub:exponential:mgf} with $\xi=\sum_{t=1}^T \zeta_{it}$, we have
\begin{eqnarray}
	\pr\bigg(\sum_{t=1}^T \zeta_{t}>a\bigg)\leq \exp\bigg(-\frac{a^2}{4L_2M^2T+2L_1M  (\log T) (\log \log T) a}\bigg), \textrm{ for all } a \geq 0.\nonumber
\end{eqnarray}
Finally, evaluating $a=Tx$, we prove the inequality with $L_3=(4L_2+2L_2)^{-1}$.
\end{proof}

\begin{lemma}\label{lemma:expectation:sample:variance:large:T}
For $g\in \mcH^0$, it follows that 
\begin{eqnarray*}
	a_3^{-1}\|g\|_2^2\leq \ev\bigg(g^2(\bfX_{i1})\bigg)\leq a_3\|g\|_2^2, \textrm{ for all } i\in [N],
\end{eqnarray*}
and
\begin{eqnarray}
a_3^{-1}\|g\|_2^2-\frac{8a_3}{\sqrt{L_3}}\bigg(\frac{\|g\|_\infty^2}{T}+\frac{\|g\|_\infty\|g\|_2}{\sqrt{T}}\bigg)\leq \|g\|^2\leq a_3\|g\|_2^2+\frac{8a_3}{\sqrt{L_3}}\bigg(\frac{\|g\|_\infty^2}{T}+\frac{\|g\|_\infty\|g\|_2}{\sqrt{T}}\bigg).\nonumber
\end{eqnarray}
Furthermore, if $A_{NT}^2=o(T)$, then for all $g\in \Theta_{NT}$, it also holds that
\begin{eqnarray*}
	\frac{1}{2a_3}\|g\|_2^2\leq \|g\|^2\leq 2a_3\|g\|_2^2, 
\end{eqnarray*}
provided $A_{NT}^2/T\leq L_3/256$ and $A_{NT}^2/T\leq L_3a_3^{-4}/1024$.
\end{lemma}
\begin{proof}[Proof of Lemma \ref{lemma:expectation:sample:variance}]
By Assumption \ref{Assumption:A2} and direct examination, we have
\begin{eqnarray*}
	\ev_i(g^2)={\int g^2(\bfx)\pi_i(\bfx)}d\bfx\leq a_3\int g^2(\bfx)d\bfx=a_3\|g\|_2^2
\end{eqnarray*}
where we use the fact that $\int g(\bfx)d\bfx=0$. Similar argument can obtain the lower bound and finish the proof of first inequality.

Now we will prove the second inequality. Direct examination yields 
\begin{eqnarray}
	\|g\|^2=\ev(\|g\|_{NT}^2)&=&\frac{1}{N}\sum_{i=1}^N\ev_i(g^2)-\frac{1}{N}\sum_{i=1}^N\ev(\bar{g}_i^2),\label{eq:lemma:expectation:sample:variance:large:T:eq:1}
\end{eqnarray}
where $\bar{g}_i=T^{-1}\sum_{t=1}^Tg(\bfX_{it})$. Next simple algebra leads to 
\begin{eqnarray*}
	\bigg|\bar{g}_i^2-\ev_i^2(g)\bigg|\leq \bigg|\bar{g}_i-\ev_i(g)\bigg|^2+2\bigg|\bar{g}_i-\ev_i(g)\bigg|\bigg|\ev_i(g)\bigg|.
\end{eqnarray*}
As a consequence, we have
\begin{eqnarray}
	\bigg|\ev(\bar{g}_i^2)-\ev_i^2(g)\bigg|&\leq&\ev\bigg(\bigg|\bar{g}_i^2-\ev_i^2(g)\bigg|\bigg)\nonumber\\
	&\leq&\textrm{Var}(\bar{g}_i)+2\sqrt{\textrm{Var}(\bar{g}_i)}\sqrt{\ev_i(g^2)}\nonumber\\
	&\leq& \textrm{Var}(\bar{g}_i)+2a_3\sqrt{\textrm{Var}(\bar{g}_i)}\|g\|_2\label{eq:lemma:expectation:sample:variance:large:T:eq:2}.
\end{eqnarray}
Let $\gamma_T=(\log T)(\log \log T)$, by Proposition \ref{proposition:panel:sub:exponential:inequality}, for all $x>0$, we have
\begin{eqnarray*}
	\pr\bigg(\bigg|\bar{g}_i-\ev_i(g)\bigg|>x\bigg)&\leq&2\exp\bigg(-\frac{L_3Tx^2}{\|g\|_\infty^2+\|g\|_\infty\gamma_T x}\bigg)\\
	&\leq&2\exp\bigg(-\frac{L_3Tx^2}{2\|g\|_\infty^2}\bigg)+2\exp\bigg(-\frac{L_3Tx}{2\|g\|_\infty\gamma_T}\bigg),
\end{eqnarray*}
which further implies that
\begin{eqnarray*}
	\textrm{Var}(\bar{g}_i)&=&\ev\bigg(\bigg|\bar{g}_i-\ev_i(g)\bigg|^2\bigg)\\
	&=&\int_0^\infty\pr\bigg(\bigg|\bar{g}_i-\ev_i(g)\bigg|^2>x\bigg)dx\\
	&\leq&2\int_0^\infty\exp\bigg(-\frac{L_3Tx}{2\|g\|_\infty^2}\bigg)dx+2\int_0^\infty\exp\bigg(-\frac{L_3T\sqrt{x}}{2\|g\|_\infty\gamma_T}\bigg)dx\\
	&\leq&\frac{4\|g\|_\infty^2}{L_3T}+\frac{16\|g\|_\infty^2\gamma_T^2}{L_3^2T^2}\\
	&\leq& \frac{8\|g\|_\infty^2}{L_3T},
\end{eqnarray*}
where the fact that $\gamma_T^2=o(T)$ is used. Therefore, by \ref{eq:lemma:expectation:sample:variance:large:T:eq:2}, we have
\begin{eqnarray*}
	\bigg|\ev(\bar{g}_i^2)-\ev_i^2(g)\bigg|\leq \frac{8\|g\|_\infty^2}{L_3T}+\frac{6a_3\|g\|_\infty\|g\|_2}{\sqrt{L_3T}}\leq \frac{8a_3}{\sqrt{L_3}}\bigg(\frac{\|g\|_\infty^2}{T}+\frac{\|g\|_\infty\|g\|_2}{\sqrt{T}}\bigg).
\end{eqnarray*}
Above inequality and (\ref{eq:lemma:expectation:sample:variance:large:T:eq:1}) together imply that
\begin{eqnarray}
	\bigg|\|g\|^2-\frac{1}{N}\sum_{i=1}^N\ev_i(g^2)+\frac{1}{N}\sum_{i=1}^N\ev_i^2(g)\bigg|&\leq& \frac{1}{N}\sum_{i=1}^N\bigg|\ev(\bar{g}_i^2)-\ev_i^2(g)\bigg|\nonumber\\
	&\leq&\frac{8a_3}{\sqrt{L_3}}\bigg(\frac{\|g\|_\infty^2}{T}+\frac{\|g\|_\infty\|g\|_2}{\sqrt{T}}\bigg).\label{eq:lemma:expectation:sample:variance:large:T:eq:3}
\end{eqnarray}
By Assumption \ref{Assumption:A2} and direct examination, we have
\begin{eqnarray*}
	\ev_i(g^2)-\ev_i^2(g)&=&{\int \bigg(g(\bfx)-\int g(\bfx)\pi_i(\bfx)d\bfx\bigg)^2\pi_i(\bfx)}d\bfx\\
	&\leq&{\int \bigg(g(\bfx)-\int g(\bfx)d\bfx\bigg)^2\pi_i(\bfx)}d\bfx\\
	&\leq&a_3{\int \bigg(g(\bfx)-\int g(\bfx)d\bfx\bigg)^2d\bfx}\\
	&=&a_3\|g\|_2^2,
\end{eqnarray*}
and
\begin{eqnarray*}
	\ev_i(g^2)-\ev_i^2(g)&=&{\int \bigg(g(\bfx)-\int g(\bfx)\pi_i(\bfx)d\bfx\bigg)^2\pi_i(\bfx)}d\bfx\\
	&\geq& a_3^{-1}{\int \bigg(g(\bfx)-\int g(\bfx)\pi_i(\bfx)d\bfx\bigg)^2}d\bfx\\
	&\geq&a_3^{-1}{\int \bigg(g(\bfx)-\int g(\bfx)d\bfx\bigg)^2d\bfx}\\
	&=&a_3^{-1}\|g\|_2^2.
\end{eqnarray*}
Combining two bounds above and (\ref{eq:lemma:expectation:sample:variance:large:T:eq:3}), we obtain
\begin{eqnarray}
a_3^{-1}\|g\|_2^2-\frac{8a_3}{\sqrt{L_3}}\bigg(\frac{\|g\|_\infty^2}{T}+\frac{\|g\|_\infty\|g\|_2}{\sqrt{T}}\bigg)\leq \|g\|^2\leq a_3\|g\|_2^2+\frac{8a_3}{\sqrt{L_3}}\bigg(\frac{\|g\|_\infty^2}{T}+\frac{\|g\|_\infty\|g\|_2}{\sqrt{T}}\bigg),\nonumber
\end{eqnarray}
which is the second inequality. When $g\in \Theta_{NT}$, by Lemma \ref{lemma:l2:norm:sup:norm}, it follow that
\begin{eqnarray*}
	\bigg\{ a_3^{-1}-\frac{8a_3}{\sqrt{L_3}}\bigg(\frac{A_{NT}^2}{T}+\sqrt{\frac{A_{NT}^2}{T}}\bigg)\bigg\}\|g\|_2^2\leq \|g\|^2\leq\bigg\{ a_3+\frac{8a_3}{\sqrt{L_3}}\bigg(\frac{A_{NT}^2}{T}+\sqrt{\frac{A_{NT}^2}{T}}\bigg)\bigg\}\|g\|_2^2,
\end{eqnarray*}
which is the third inequality by noticing $A_{NT}^2=o(T)$.
\end{proof}

\begin{lemma}\label{lemma:expectation:sample:variance:mixing}
Under Assumption \ref{Assumption:common} and \ref{Assumption:A2}, there exists constant  $c_2>0$ free of $T, N, i$ such that
\begin{eqnarray*}
	\textrm{Var}\bigg(\zeta_i(g, f)\bigg)\leq \frac{c_2}{T}\|g\|_\infty^2\|f\|_\infty^2, \textrm{ for all bounded } g, f\in \mcH. 
\end{eqnarray*}
Furthermore, $\textrm{ for all } g_1, g_2, f_1, f_2\in \Theta_{NT}$, the following also holds:
\begin{eqnarray*}
	\textrm{Var}\bigg(\zeta_i(g_1, f_1)-\zeta_i(g_2, f_2)\bigg)\leq c_2\frac{A_{NT}^4}{T}\bigg(\|g_1-g_2\|_2^2\|f_1\|_2^2+\|f_1-f_2\|_2^2\|g_2\|_2^2\bigg),
\end{eqnarray*}
and
\begin{eqnarray*}
\textrm{Var}\bigg(\frac{1}{T}\sum_{t=1}^Tg_1(\bfX_{it})f_1(\bfX_{it})-\frac{1}{T}\sum_{t=1}^Tg_2(\bfX_{it})f_2(\bfX_{it})\bigg) \leq c_2\frac{A_{NT}^4}{T}\bigg(\|g_1-g_2\|_2^2\|f_1\|_2^2+\|f_1-f_2\|_2^2\|g_2\|_2^2\bigg).
\end{eqnarray*}
\end{lemma}
\begin{proof}[Proof of lemma \ref{lemma:expectation:sample:variance:mixing}]
For general bounded $g, f \in \mcH$, by simple algebra, it follows that
\begin{eqnarray}
	&&\bigg|\zeta_{i}(g, f)-\ev_i(gf)+\ev_i(g)\ev_i(f)\bigg|\nonumber\\
	&=&\bigg|\frac{1}{T}\sum_{t=1}^Tg(\bfX_{it})f(\bfX_{it})-\ev_i(gf)-\frac{1}{T}\sum_{t=1}^Tg(\bfX_{it})\frac{1}{T}\sum_{t=1}^Tf(\bfX_{it})+\ev_i(g)\ev_i(f)\bigg|\nonumber\\
	&\leq&\bigg|\frac{1}{T}\sum_{t=1}^Tg(\bfX_{it})f(\bfX_{it})-\ev_i(gf)\bigg|+\bigg|\frac{1}{T}\sum_{t=1}^Tg(\bfX_{it})\frac{1}{T}\sum_{t=1}^Tf(\bfX_{it})-\ev_i(g)\ev_i(f)\bigg|\nonumber\\
	&\leq&R_1+R_2+R_3,\label{eq:lemma:expectation:sample:variance:mixing:eq:1}
\end{eqnarray}
where
\begin{eqnarray*}
	R_{1}&=&\bigg|\frac{1}{T}\sum_{t=1}^Tg(\bfX_{it})f(\bfX_{it})-\ev_i(gf)\bigg|,\\
	R_{2}&=&\bigg|\frac{1}{T}\sum_{t=1}^Tg(\bfX_{it})-\ev_i(g)\bigg|\|f\|_\infty,\;\;\;R_{3}=\bigg|\frac{1}{T}\sum_{t=1}^Tf(\bfX_{it})-\ev_i(f)\bigg| \ev\bigg(|g|\bigg),
\end{eqnarray*}
Let $\gamma_T=(\log T)(\log T \log T)$ and by Proposition \ref{proposition:panel:sub:exponential:inequality} and Lemma \ref{lemma:l2:norm:sup:norm}, we have
\begin{eqnarray*}
	\pr\bigg(R_1>x\bigg)&\leq&2\exp\bigg(-\frac{L_3Tx^2}{\|g\|_\infty^2\|f\|_\infty^2+\|g\|_\infty\|f\|_\infty\gamma_T x}\bigg)\\
	&\leq&2\exp\bigg(-\frac{L_3Tx^2}{2\|g\|_\infty^2\|f\|_\infty^2}\bigg)+2\exp\bigg(-\frac{L_3Tx}{2\|g\|_\infty\|f\|_\infty\gamma_T}\bigg).
\end{eqnarray*}
By simple calculus, $\int_0^\infty \exp(-ax^{1/k})dx=k!a^{-k}$ for integer $k\geq 1$. As a consequence, it follows that
\begin{eqnarray*}
	\ev(R_1^2)&=&\int_0^\infty \pr\bigg(R_1>\sqrt{x}\bigg)dx\\
	&\leq&2\int_0^\infty\exp\bigg(-\frac{L_3Tx}{2\|g\|_\infty^2\|f\|_\infty^2}\bigg)dx+2\int_0^\infty \exp\bigg(-\frac{L_3Tx^{1/2}}{2\|g\|_\infty\|f\|_\infty\gamma_T}\bigg)dx\\
	&\leq&\frac{4\|g\|_\infty^2\|f\|_\infty^2}{L_3T}+\frac{16\|g\|_\infty^2\|f\|_\infty^2\gamma_T^2}{L_3^2T^2}.
\end{eqnarray*}
Similarly we can show
\begin{eqnarray*}
	\ev(R_2^2)&\leq&\frac{4\|g\|_\infty^2\|f\|_\infty^2}{L_3T}+\frac{8\|g\|_\infty^2\|f\|_\infty^2\gamma_T^2}{L_3^2T^2}\\
	\ev(R_3^2)&\leq&\frac{4a_3\|g\|_2^2\|f\|_\infty^2}{L_3T}+\frac{8a_3\|g\|_2^2\|f\|_\infty^2\gamma_T^2}{L_3^2T^2},
\end{eqnarray*}
where Lemma \ref{lemma:expectation:sample:variance:large:T} is used.
Therefore, by property of variance and (\ref{eq:lemma:expectation:sample:variance:mixing:eq:1}), it follows that
\begin{eqnarray}
	\textrm{Var}\bigg(\zeta_i(g, f)\bigg)&\leq&\ev\bigg(\bigg|\zeta_{i}(g, f)-\ev_i(gf)+\ev_i(g)\ev_i(f)\bigg|^2\bigg)\nonumber\\
	&\leq&4\sum_{j=1}^3\ev(R_j^2)\leq\frac{C}{T}\|g\|_\infty^2\|f\|_\infty^2,\label{eq:lemma:expectation:sample:variance:mixing:eq:2}
\end{eqnarray}
and
\begin{eqnarray}
	\textrm{Var}\bigg(\frac{1}{T}\sum_{t=1}^Tg(\bfX_{it})f(\bfX_{it})\bigg)=\ev(R_1^2)\leq \frac{C}{T}\|g\|_\infty^2\|f\|_\infty^2\;\;\textrm{ with }\;\; C=128a_3L_3^{-1},\label{eq:lemma:expectation:sample:variance:mixing:eq:3}
\end{eqnarray}
which is the first result. Combining Lemma \ref{lemma:l2:norm:sup:norm},  (\ref{eq:lemma:expectation:sample:variance:mixing:eq:2}) and (\ref{eq:lemma:expectation:sample:variance:mixing:eq:3}),  for $g_1, g_2, f_1, f_2 \in \Theta_{NT}$, we have
\begin{eqnarray*}
	\textrm{Var}\bigg(\zeta_i(g_1, f_1)-\zeta_i(g_2, f_2)\bigg)&=&\textrm{Var}\bigg(\zeta_i(g_1-g_2, f_1)+\zeta_i(g_2, f_1-f_2)\bigg)\nonumber\\
	&\leq&2\textrm{Var}\bigg(\zeta_i(g_1-g_2, f_1)\bigg)+2\textrm{Var}\bigg(\zeta_i(g_2, f_1-f_2)\bigg)\\
	&\leq&  2\frac{C}{T}\bigg(\|g_1-g_2\|_\infty^2\|f_1\|_\infty^2+\|f_1-f_2\|_\infty^2\|g_2\|_\infty^2\bigg)\\
	&\leq&  c_2\frac{A_{NT}^4}{T}\bigg(\|g_1-g_2\|_2^2\|f_1\|_2^2+\|f_1-f_2\|_2^2\|g_2\|_2^2\bigg),
\end{eqnarray*}
and
\begin{eqnarray*}
	&&\textrm{Var}\bigg(\frac{1}{T}\sum_{t=1}^Tg_1(\bfX_{it})f_1(\bfX_{it})-\frac{1}{T}\sum_{t=1}^Tg_2(\bfX_{it})f_2(\bfX_{it})\bigg)\\
	&=&\textrm{Var}\bigg\{\frac{1}{T}\sum_{t=1}^Tf_1(\bfX_{it})\bigg(g_1(\bfX_{it})-g_2(\bfX_{it})\bigg)+\frac{1}{T}\sum_{t=1}^T\bigg(f_1(\bfX_{it})-f_2(\bfX_{it})\bigg)g_2(\bfX_{it})\bigg\}\\
	&\leq&2\textrm{Var}\bigg\{\frac{1}{T}\sum_{t=1}^Tf_1(\bfX_{it})\bigg(g_1(\bfX_{it})-g_2(\bfX_{it})\bigg)\bigg\}+2\textrm{Var}\bigg\{\frac{1}{T}\sum_{t=1}^T\bigg(f_1(\bfX_{it})-f_2(\bfX_{it})\bigg)g_2(\bfX_{it})\bigg\}\\
	&\leq&  2\frac{C}{T}\bigg(\|g_1-g_2\|_\infty^2\|f_1\|_\infty^2+\|f_1-f_2\|_\infty^2\|g_2\|_\infty^2\bigg)\\
	&\leq&  c_2\frac{A_{NT}^4}{T}\bigg(\|g_1-g_2\|_2^2\|f_1\|_2^2+\|f_1-f_2\|_2^2\|g_2\|_2^2\bigg),
\end{eqnarray*}
with $c_2=2C$, which is the second inequality.
\end{proof}

\begin{lemma}\label{lemma:uniform:equivalence:empirical:population:norm:large:T}
Under Assumption \ref{Assumption:A2}, if $A_{NT}^4d_{NT}=o(NT)$, $A_{NT}^2d_{NT}=o(N)$ and $A_{NT}^2=o(T)$, then the following holds,
\begin{eqnarray*}
\sup_{g,f \in \Theta_{NT}}\bigg| \frac{\langle f, g\rangle_{NT}-\langle f, g\rangle}{\|f\|_2\|g\|_2}\bigg|=o_P(1),\;\;\;\;\sup_{g,f \in \Theta_{NT}}\bigg| \frac{\langle f, g\rangle_{NT}-\langle f, g\rangle}{\|f\|\|g\|}\bigg|=o_P(1),
\end{eqnarray*}
and
\begin{eqnarray*}
	\sup_{g,f \in \Theta_{NT}}\bigg| \frac{\frac{1}{NT}\sum_{i=1}^N\sum_{t=1}^Tg(\bfX_{it})f(\bfX_{it})-\frac{1}{N}\sum_{t=1}^N\ev_i(gf)}{\|f\|_2\|g\|_2}\bigg|&=&o_P(1),\\
	\sup_{g,f \in \Theta_{NT}}\bigg| \frac{\frac{1}{NT}\sum_{i=1}^N\sum_{t=1}^Tg(\bfX_{it})f(\bfX_{it})-\frac{1}{N}\sum_{t=1}^N\ev_i(gf)}{\|f\|\|g\|}\bigg|&=&o_P(1).
\end{eqnarray*}
\end{lemma}
\begin{proof}[Proof of Lemma \ref{lemma:uniform:equivalence:empirical:population:norm:large:T}]
Let $\mcF_{\textrm{UB}}^0=\{f \in \Theta_{NT}:  \|f\|_2\leq 1\}$ and $\Theta=\mcF_{\textrm{UB}}^0\times \mcF_{\textrm{UB}}^0$. Define metric on $\Theta$ by $d(\theta_1, \theta_2)=\sqrt{\|g_1-g_2\|_2^2/2+\|f_1-f_2\|_2^2/2}$ for for any $\theta_1=(g_1, f_1), \theta_2=(g_2, f_2)\in \Theta$.

By definition, it follows that
\begin{eqnarray}
\sup_{g,f \in \mcF_{\textrm{UB}}^0}\bigg|\langle g, f\rangle_{NT}-\langle g, f\rangle\bigg|&=&\sup_{g,f \in \mcF_{\textrm{UB}}^0}\bigg|\langle g, f\rangle_{NT}-\ev\bigg(\langle g, f\rangle_{NT}\bigg)\bigg|\nonumber\\
&=&\sup_{g,f \in \mcF_{\textrm{UB}}^0}\bigg|\frac{1}{N}\sum_{i=1}^N\bigg\{\zeta_{i}(g, f)-\ev\bigg(\zeta_i(g, f)\bigg)\bigg\}\bigg|\nonumber.
\end{eqnarray}
For any $\theta=(g, f)\in \Theta,$ define
\begin{eqnarray*}
	X_i(\theta)=\zeta_i(g, f)-\ev\bigg(\zeta_i(g, f)\bigg)\;\; \textrm{ and }\;\; S_N(\theta)=\frac{1}{N}X_i(\theta).
\end{eqnarray*}
Therefore, for any $\theta_1=(g_1, f_1), \theta_2=(g_2, f_2)\in \Theta$, we have
\begin{eqnarray*}
	|\zeta_i(g_1, f_1)-\zeta_i(g_2, f_2)|&\leq&|\zeta_i(g_1-g_2, f_1)|+|\zeta_i(g_2, f_1-f_2)|\\
	&\leq&4\|g_1-g_2\|_\infty\|f_1\|_\infty+4\|f_1-f_2\|_\infty\|g_2\|_\infty\\
	&\leq&4A_{NT}^2\bigg(\|g_1-g_2\|_2+\|f_1-f_2\|_2\bigg)\\
	&\leq& 8A_{NT}^2d(\theta_1, \theta_2),
\end{eqnarray*}
which further leads to
\begin{eqnarray}
	|X_i(\theta_1)-X_i(\theta_2)|\leq 16A_{NT}^2d(\theta_1, \theta_2).\label{eq:lemma:uniform:equivalence:empirical:population:norm:large:T:eq2}
\end{eqnarray}
Moreover, by Lemma \ref{lemma:l2:norm:sup:norm} and Lemma \ref{lemma:expectation:sample:variance:mixing}, for any $\delta>0$, we have
\begin{eqnarray}
	\textrm{Var}\bigg(X_i(\theta_1)-X_i(\theta_2)\bigg)&\leq& 2c_2\frac{A_{NT}^4}{T}d^2(\theta_1, \theta_2).\label{eq:lemma:uniform:equivalence:empirical:population:norm:large:T:eq2.5}
\end{eqnarray}
Bernstein inequality, (\ref{eq:lemma:uniform:equivalence:empirical:population:norm:large:T:eq2}) and (\ref{eq:lemma:uniform:equivalence:empirical:population:norm:large:T:eq2.5}) together imply that
\begin{eqnarray}
	\pr\bigg\{\bigg|S_N(\theta_1)-S_N(\theta_2)\bigg|>xs\bigg\}&\leq&2\exp\bigg(-\frac{Nx^2s^2}{4c_2A_{NT}^4T^{-1}d^2(\theta_1, \theta_2)+16A_{NT}^2d(\theta_1, \theta_2)xs}\bigg)\nonumber\\
	&\leq&2\exp\bigg(-\frac{Nx^2s^2}{8c_2A_{NT}^4T^{-1}d^2(\theta_1, \theta_2)}\bigg)\nonumber\\
	&&+2\exp\bigg(-\frac{Nxs}{32A_{NT}^2d(\theta_1, \theta_2)}\bigg), \textrm{ for all } x,s>0.\label{eq:lemma:uniform:equivalence:empirical:population:norm:large:T:eq3}
\end{eqnarray}

Similar proof in that of Lemma \ref{lemma:uniform:equivalence:empirical:population:norm}, Let $\delta_k=3^{-k}$ for $k\geq 0$. For sufficient large integer $K$, which will be specified later, let $\{0\}=\mcH_0\subset \mcH_1 \ldots \mcH_K$ be a sequence of subsets of $\Theta$ such that $\min_{\theta^*\in \mcH_k}d(\theta^*, \theta)\leq \delta_k$ for all $\theta \in \Theta$. Moreover the subsets $\mcH_K$ is chosen inductively such that two different elements in $\mcH_k$ is at least $\delta_k$ apart.  By definition, the cardinality $\#(\mcH_k)$ of $\mcH_k$ is bounded by the $\delta_k/2$-covering number $D(\delta_k/2, \Theta, d)$. Moreover, since $d(\theta_1, \theta_2)\leq \|g_1-g_2\|_2+\|f_1-f_2\|_2$ for $\theta_1=(g_1, f_1), \theta_2=(g_2, f_2)\in \Theta$, we have
\begin{eqnarray*}
	D(\delta_k/2, \Theta, d)\leq D^2(\delta_k/4, \mcF_{\textrm{UB}}^0, \|\cdot\|_2)\leq \bigg(\frac{16+\delta_k}{\delta_k}\bigg)^{2d_{NT}},
\end{eqnarray*} 
where the last inequality is due to \cite{van2000}[Corollary 2.6] and the fact that $\mcF_{\textrm{UB}}^0$ is a linear space with dimension $d_{NT}$. For any $\theta \in \Theta$, let $\tau_k(\theta)\in \mcH_k$  be a element such that $d(\tau_k(\theta), \theta)\leq \delta_k$, for $k=1,2,\ldots, K$. Now for fixed $x>0$, choose $K=K(N)>0$, which depends on $N$ and is increasing fast enough such that $x>16A_{NT}^2(2/3)^K$, we see from (\ref{eq:lemma:uniform:equivalence:empirical:population:norm:large:T:eq3}) that
\begin{eqnarray}
&&\pr\bigg(\sup_{\theta \in \Theta}\bigg| S_N(\theta)\bigg|>x\bigg)\nonumber\\
&\leq&\pr\bigg(\sup_{\theta \in \Theta}\bigg| S_N(\theta)-S_N\bigg(\tau_K(\theta)\bigg)\bigg|>\frac{x}{2^K}\bigg)\nonumber\\
&&+\sum_{k=1}^K\pr\bigg(\sup_{\theta \in \Theta}\bigg| S_N\bigg(\tau_k\circ\ldots\circ\tau_K(\theta)\bigg)-S_N\bigg(\tau_{k-1}\circ\tau_k\circ\ldots\circ\tau_K(\theta)\bigg)\bigg|>\frac{x}{2^{k-1}}\bigg)\nonumber\\
&\leq&\pr\bigg(16A_{NT}^2d(\theta, \tau_K(\theta))>\frac{x}{2^K}\bigg)\nonumber\\
&&+\sum_{k=1}^K\#(\mcH_k)\sup_{\theta_k \in \mcH_k}\pr\bigg(\sup_{\theta \in \Theta}\bigg| S_N\bigg(\theta_k\bigg)-S_N\bigg(\tau_{k-1}(\theta_k)\bigg)\bigg|>\frac{x}{2^{k-1}}\bigg)\nonumber\\
&\leq&0+\sum_{k=1}^K\bigg(\frac{16+3^{-k}}{3^{-k}}\bigg)^{2d_{NT}}\sup_{\theta_k \in \mcH_k}\pr\bigg(\sup_{\theta \in \Theta}\bigg| S_N\bigg(\theta_k\bigg)-S_N\bigg(\tau_{k-1}(\theta_k)\bigg)\bigg|>\frac{x}{2^{k-1}}\bigg)\nonumber\\
&\leq&\sum_{k=1}^K2\bigg(\frac{16+3^{-k}}{3^{-k}}\bigg)^{2d_{NT}}\exp\bigg(-\frac{Nx^2}{8c_2A_{NT}^4T^{-1}2^{2(k-1)}d^2(\theta_k, \tau_{k-1}(\theta_k))}\bigg)\nonumber\\
&&+\sum_{k=1}^K2\bigg(\frac{16+3^{-k}}{3^{-k}}\bigg)^{2d_{NT}}\exp\bigg(-\frac{Nx}{32A_{NT}^22^{k-1}d(\theta_k, \tau_{k-1}(\theta_k))}\bigg)\nonumber\\
&\leq&\sum_{k=1}^K2\bigg(\frac{16+3^{-k}}{3^{-k}}\bigg)^{2d_{NT}}\exp\bigg(-\frac{TNx^2}{8c_2A_{NT}^4}\bigg(\frac{3}{2}\bigg)^{2(k-1)}\bigg)\nonumber\\
&&+\sum_{k=1}^\infty2\bigg(\frac{16+3^{-k}}{3^{-k}}\bigg)^{2d_{NT}}\exp\bigg(-\frac{Nx}{32A_{NT}^2}\bigg(\frac{3}{2}\bigg)^{k-1}\bigg)\nonumber\\
&\leq&2\sum_{k=1}^\infty\exp\bigg(4(3+k)d_{NT}-\frac{TNx^2}{8c_2A_{NT}^4}\bigg(\frac{9}{4}\bigg)^{k-1}\bigg)+2\sum_{k=1}^\infty\exp\bigg(4(3+k)d_{NT}-\frac{Nx}{32A_{NT}^2}\bigg(\frac{3}{2}\bigg)^{k-1}\bigg),\nonumber\\\label{eq:lemma:uniform:equivalence:empirical:population:norm:large:T:eq4}
\end{eqnarray} 
where we used the fact that $16\times3^k+1\leq 3^{k+3}$ and $\log 3\leq 2$. Now choose $N$ large enough  such that 
\begin{eqnarray}
	64c_2(3+k)a_1A_{NT}^4d_{NT}<NT(9/4)^{k-1}x^2  \;\;\textrm{ and }\;\;256(3+k)a_1A_{NT}^2d_{NT}<N(3/2)^{k-1}x,\label{eq:lemma:uniform:equivalence:empirical:population:norm:large:T:eq5}
\end{eqnarray}
 for all $k\geq 0$, which is possible, since $d_{NT}A_{NT}^4=o(TN)$ and $d_{NT}A_{NT}^2=o(N)$. Therefore, for all $N$ satisfying (\ref{eq:lemma:uniform:equivalence:empirical:population:norm:large:T:eq5}), (\ref{eq:lemma:uniform:equivalence:empirical:population:norm:large:T:eq4}) further leads to
\begin{eqnarray*}
	\pr\bigg(\sup_{\theta \in \Theta}\bigg| S_N(\theta)\bigg|>x\bigg)&\leq&2\sum_{k=1}^\infty\exp\bigg(-\frac{TNx^2}{16c_2A_{NT}^4}\bigg(\frac{9}{4}\bigg)^{k-1}\bigg)+2\sum_{k=1}^\infty\exp\bigg(-\frac{Nx}{64A_{NT}^2}\bigg(\frac{3}{2}\bigg)^{k-1}\bigg)\\
	&\leq&\frac{2}{e}\bigg\{\frac{16c_2A_{NT}^4}{TNx^2}\sum_{k=1}^\infty\bigg(\frac{9}{4}\bigg)^{k-1}+\frac{64A_{NT}^2}{Nx}\sum_{k=1}^\infty\bigg(\frac{3}{2}\bigg)^{k-1}\bigg\}\to 0, \textrm{ as } N \to \infty,
\end{eqnarray*}
where the fact that $A_{NT}^4=o(NT)$, $A_{NT}^2=o(N)$ and inequality $e^{-x}\leq e^{-1}/x$ are used. Therefore, we have
\begin{eqnarray}
	\sup_{g, f \in \Theta_{NT}}\bigg|\frac{\langle g, f\rangle_{NT}-\langle g, f\rangle}{\|g\|_2\|f\|_2}\bigg|= \sup_{g,f \in \mcF_{\textrm{UB}}^0}\bigg|\langle g, f\rangle_{NT}-\langle g, f\rangle\bigg|=\sup_{\theta \in \Theta}\bigg| S_N(\theta)\bigg|=o_P(1),\nonumber
\end{eqnarray}
which is the first inequality. The second one follows from \ref{lemma:expectation:sample:variance:large:T}. Similar argument can be applied to prove the third and fourth results.
\end{proof}

\begin{lemma}\label{lemma:bounded:sequence:two:norm:difference:large:T}
Under Assumption \ref{Assumption:common} and \ref{Assumption:A2}, for element $v_{NT}\in \mathcal{H}$, if $\|v_{N}\|_\infty\leq \rho_{NT}$, it holds that
\begin{eqnarray*}
\sup_{g\in \Theta_{NT}}\bigg|\frac{\langle v_N, g \rangle_{NT}-\langle v_N, g \rangle}{\|g\|_2}\bigg|^2=O_P\bigg(\frac{A_{NT}^2\rho_{NT}^2d_{NT}}{NT}\bigg).
\end{eqnarray*}
Furthermore, if $A_{NT}^2=o(T)$, it also holds that
\begin{eqnarray*}
\sup_{g\in \Theta_{NT}}\bigg|\frac{\langle v_N, g \rangle_{NT}-\langle v_N, g \rangle}{\|g\|}\bigg|^2=O_P\bigg(\frac{A_{NT}^2\rho_{NT}^2d_{NT}}{NT}\bigg).
\end{eqnarray*}
\end{lemma}
\begin{proof}[Proof of Lemma \ref{lemma:bounded:sequence:two:norm:difference}]
Let $\{\chi_j\}_{j=1}^{d_{NT}}$ be the orthonormal basis of $\Theta_{NT}$ with respect to $\langle \cdot, \cdot  \rangle$. For $g\in \Theta_{NT}$, we can rewrite $g=\sum_{j=1}^{d_{NT}}b_j\chi_j$ and $\|g\|^2=\sum_{j=1}^{d_{NT}}b_j^2$ for some $b_j\in \mathbb{R}$.
\begin{eqnarray*}
	|\langle v_{N}, g \rangle_{NT}-\langle v_{N}, g \rangle|&=&\bigg|\sum_{j=1}^{d_{NT}}b_j\bigg(\langle v_{N}, \chi_j \rangle_{NT}-\langle v_{N}, \chi_j \rangle\bigg)\bigg|\\
	&\leq&\bigg(\sum_{j=1}^{d_{NT}}b_j^2\bigg)^{1/2}\bigg\{\sum_{j=1}^{d_{NT}}\bigg(\langle v_{N}, \chi_j \rangle_{NT}-\langle v_{N}, \chi_j \rangle\bigg)^2\bigg\}^{1/2}\\
	&=&\|g\|\bigg\{\sum_{j=1}^{d_{NT}}\bigg(\langle v_{N}, \chi_j \rangle_{NT}-\langle v_{N}, \chi_j \rangle\bigg)^2\bigg\}^{1/2},
\end{eqnarray*}
which leads to 
\begin{eqnarray}
	\sup_{g\in \Theta_{NT}}\bigg|\frac{\langle v_{N}, g \rangle_{NT}-\langle v_{N}, g \rangle}{\|g\|}\bigg|\leq \bigg\{\sum_{j=1}^{d_{NT}}\bigg(\langle v_{N}, \chi_j \rangle_{NT}-\langle v_{N}, \chi_j \rangle\bigg)^2\bigg\}^{1/2}.\label{eq:lemma:bounded:sequence:two:norm:difference:large:T:eq1}
\end{eqnarray}

For each $j\in [d_{NT}]$, decomposing mean square error to variance and bias, it follows that
\begin{eqnarray}
	\ev\bigg(|\langle v_{N}, \chi_j \rangle_{NT}-\langle v_{N}, \chi_j \rangle|^2\bigg)=\textrm{Var}\bigg(\langle v_{N}, \chi_j \rangle_{NT}\bigg)\nonumber&=&\frac{1}{N^2}\sum_{i=1}^N\textrm{Var}\bigg(\zeta_i(v_N, \chi_j)\bigg)\nonumber\\
	&&\leq \frac{c_2}{NT}\|v_N\|_\infty^2\|\chi_j\|_\infty^2\leq c_2\frac{\rho_{NT}^2A_{NT}^2}{NT},\nonumber
\end{eqnarray}
where  Lemma \ref{lemma:expectation:sample:variance:mixing}, Lemma \ref{lemma:l2:norm:sup:norm}  and Assumption \ref{Assumption:A2}.\ref{A2:a} are used and we prove the first equation. The second one follows from Lemma \ref{lemma:expectation:sample:variance:large:T}.
\end{proof}

To proceed further, we define event 
\begin{eqnarray*}
	\Omega_{NT}=\bigg\{\mathbb{Z}: \frac{1}{2}\|g\|\leq \|g\|_{NT}\leq 2\|g\|, \frac{1}{2a_3}\leq \frac{1}{NT}\sum_{i=1}^N\sum_{t=1}^Tg^2(\bfX_{it})\leq 2a_3\|g\|_2 \textrm{ for all } g\in \Theta_{NT}\bigg\}
\end{eqnarray*}
and $\pr(\Omega_{NT})\to 1$ as $(N, T) \to \infty$ by Lemmas \ref{lemma:uniform:equivalence:empirical:population:norm:large:T}, and \ref{lemma:expectation:sample:variance:large:T}. Then on event $\Omega_{NT}$, $\langle \cdot, \cdot \rangle_{NT}$ is a valid inner product in $\Theta_{NT}$.

\begin{proposition}\label{proposition:bound:of:trace}
Let $A, B\in \mathbb{R}^{k\times k}$ be symmetric and positive definite matrices. If $c^\top A c\geq c^\top B c$ for all $c\in \mathbb{R}^k$, then $$\textrm{Tr}(A^{-1}B)\leq k.$$
\end{proposition}
\begin{proof}[Proof of Proposition \ref{proposition:bound:of:trace}]
For fixed $b\in \mathbb{R}^k$, let $c=A^{-1/2}b$, by conditions given, we have $b^\top b\geq b^\top  A^{-1/2}BA^{-1/2} b$. This implies all the eigenvalues $\lambda_1, \lambda_2, \ldots \lambda_k$ of $A^{-1/2}BA^{-1/2}$ are bounded by $1$. Therefore, by property of trace operator, we have $\textrm{Tr}(A^{-1}B)=\textrm{Tr}(A^{-1/2}BA^{-1/2})=\sum_{i=1}^k\lambda_i\leq k.$
\end{proof}

\begin{lemma}\label{lemma:expect:value:projection:y:f0:large:T}
Under Assumption \ref{Assumption:common} and \ref{Assumption:A2}, if $A_{NT}^4d_{NT}=o(NT)$, $A_{NT}^2d_{NT}=o(N)$ and $A_{NT}^2=o(T)$, then
$$\ev\bigg(\|\widehat{f}_*-\widetilde{f}_*\|_{NT}^2\bigg|\mathbb{Z}\bigg)=O_P\bigg(\frac{d_{NT}}{NT}\bigg).$$
\end{lemma}
\begin{proof}[Proof of Lemma \ref{lemma:expect:value:projection:y:f0:large:T}]
Let $\Psi(\bfx)=(B_1(\bfx), B_2(\bfx), \ldots, B_{d_{NT}}(\bfx))^\top\in \mathbb{R}^{d_{NT}}$ and 
\begin{eqnarray*}
	\Psi_i&=&\begin{pmatrix}
	\Psi(\bfX_{i1}), \Psi(\bfX_{i1}),\ldots, \Psi(\bfX_{iT}) 
	\end{pmatrix}^\top\in \mathbb{R}^{T \times d_{NT}},\;\; \bfY_i=(Y_{i1}, Y_{i2}, \ldots, Y_{it})^\top\in \mathbb{R}^T,\\
H&=&I_T-\frac{1}{T}uu^\top \in \mathbb{R}^{T\times T}, \textrm{ with } u=(1,1,\ldots, 1)^\top\in \mathbb{R}^T,\\
\bfY&=&(\bfY_1^\top, \bfY_2^\top,\ldots, \bfY_N^\top)^\top \in \mathbb{R}^{NT}, \;\;M_H=I_N\otimes H\in \mathbb{R}^{NT \times NT}, \\
\mathbf{\Psi}&=&(\Psi_1^\top, \Psi_2^\top, \ldots, \Psi_N^\top)^\top \in \mathbb{R}^{NT\times d_{NT}}.
\end{eqnarray*}
For any $c\in \mathbb{R}^{d_{NT}}$, let $g(\bfx)=c^\top \Psi(\bfx)\in \Theta_{NT}$. Therefore, by Lemma \ref{lemma:expectation:sample:variance:large:T}, one event $\Omega_{NT}$, it follows that
\begin{eqnarray}
	\frac{1}{NT}c^\top \bfPsi^\top M_H \bfPsi c&=&\frac{1}{NT}\sum_{i=1}^N c^\top\Psi_i^\top H\Psi_i c\nonumber\\
	&=&\frac{1}{N}\sum_{i=1}^T\zeta_i(g, g)=\|g\|_{NT}^2\geq \frac{1}{2}\|g\|^2\geq \frac{1}{4a_3}\|g\|_2^2\geq 0,\label{eq:lemma:expect:value:projection:y:f0:large:T:eq:0}
\end{eqnarray}
where all equalities hold if and only if $g=0$ or, equivalently, $c=0$. Moreover, on event $\Omega_{NT}$, direct examination leads to
\begin{eqnarray}
	\frac{1}{NT}c^\top \bfPsi^\top \bfPsi c&=&\frac{1}{NT}\sum_{i=1}^N c^\top\Psi_i^\top\Psi_i c\nonumber\\
	&=&\frac{1}{NT}\sum_{i=1}^N\sum_{t=1}^Tg^2(\bfX_{it}),\nonumber
\end{eqnarray}
which further implies that
\begin{eqnarray}
	\frac{1}{2a_3}\|g\|_2^2\leq\frac{1}{NT}c^\top \bfPsi^\top \bfPsi c\leq 2a_3\|g\|_2^2\label{eq:lemma:expect:value:projection:y:f0:large:T:eq:1}.
\end{eqnarray}
In the view of (\ref{eq:lemma:expect:value:projection:y:f0:large:T:eq:0}) and (\ref{eq:lemma:expect:value:projection:y:f0:large:T:eq:1}), we conclude that, on event $\Omega_{NT}$, both $\bfPsi^\top M_H \bfPsi$ and $\bfPsi^\top \bfPsi$ are invertible and 
\begin{eqnarray*}
	8a_3^2c^\top \bfPsi^\top M_H \bfPsi c\geq c^\top \bfPsi^\top\bfPsi c, \textrm{ for all } c\in \mathbb{R}^{d_{NT}},
\end{eqnarray*}
which, by Proposition \ref{proposition:bound:of:trace}, further implies that
\begin{eqnarray}
\textrm{Tr}\bigg\{\bigg(\bfPsi^\top M_H \bfPsi c\geq\bfPsi^\top\bfPsi\bigg)^{-1}\bfPsi^\top \bfPsi\bigg\}\leq 8a_3^2d_{NT}.\label{eq:lemma:expect:value:projection:y:f0:large:T:eq:2}
\end{eqnarray}

By definition, we have
\begin{eqnarray*}
	\widehat{a}=\argmin_{a\in \mathbb{R}^{d_{NT}}}\|Y-a^T\Psi\|_{NT}^2&=&\bigg(\sum_{i=1}^N\Psi_i^\top H\Psi_i\bigg)^{-1}\sum_{i=1}^N\Psi_i^\top H\bfY_i\\
	&=&\bigg(\bfPsi^\top M_H \bfPsi\bigg)^{-1}\bfPsi^\top M_H \bfY.
\end{eqnarray*}
As a consequence, it follows that
\begin{eqnarray*}
	\widehat{\bff}_*&=&\bigg(\widehat{f}(\bfX_{11}), \widehat{f}(\bfX_{i1}),\ldots, \widehat{f}(\bfX_{1T}),\ldots, \widehat{f}(\bfX_{N1}), \widehat{f}(\bfX_{N2})\ldots, \widehat{f}(\bfX_{NT})\bigg)^\top\\
	&=&\bfPsi \widehat{a}=\bfPsi\bigg(\bfPsi^\top M_H \bfPsi\bigg)^{-1}\bfPsi^\top M_H \bfY.
\end{eqnarray*}
Since $\widetilde{f}_*=\argmin_{f\in \Theta_{NT}}\|f_0-f\|_{NT}^2$, if we define
\begin{eqnarray*}
	\widetilde{a}=\argmin_{a\in \mathbb{R}^{d_{NT}}}\|f_0-a^T\Psi\|_{NT}^2=\bigg(\bfPsi^\top M_H \bfPsi\bigg)^{-1}\bfPsi^\top M_H \bff_0,
\end{eqnarray*}
then we have
\begin{eqnarray*}
	\widetilde{\bff}_*=\bfPsi \widetilde{a}=\bfPsi \bigg(\bfPsi^\top M_H \bfPsi\bigg)^{-1}\bfPsi^\top M_H \bff_0.
\end{eqnarray*}
Therefore, it follows that
\begin{eqnarray*}
	\|\widehat{f}_*-\widetilde{f}_*\|_{NT}^2&=&\frac{1}{NT}(\widehat{\bff}_*-\widetilde{\bff}_*)^\top (\widehat{\bff}_*-\widetilde{\bff}_*)\\
	&=&\frac{1}{NT}\bigg\{\bfPsi\bigg(\bfPsi^\top M_H \bfPsi\bigg)^{-1}\bfPsi^\top M_H\bfepsilon\bigg\}^\top\bigg\{\bfPsi\bigg(\bfPsi^\top M_H \bfPsi\bigg)^{-1}\bfPsi^\top M_H \bf\bfepsilon\bigg\}\\
	&=&\frac{1}{NT}\bigg\{\bfepsilon^\top M_H\bfPsi \bigg(\bfPsi^\top M_H \bfPsi\bigg)^{-1}\bfPsi^\top\bfPsi\bigg(\bfPsi^\top M_H \bfPsi\bigg)^{-1}\bfPsi^\top M_H\bfepsilon\bigg\}\\
	&=&\frac{1}{NT}\textrm{Tr}\bigg\{\bfepsilon^\top M_H\bfPsi \bigg(\bfPsi^\top M_H \bfPsi\bigg)^{-1}\bfPsi^\top\bfPsi\bigg(\bfPsi^\top M_H \bfPsi\bigg)^{-1}\bfPsi^\top M_H\bfepsilon\bigg\}\\
	&=&\frac{1}{NT}\textrm{Tr}\bigg\{\bfPsi\bigg(\bfPsi^\top M_H \bfPsi\bigg)^{-1}\bfPsi^\top M_H\bfepsilon \bfepsilon^\top M_H\bfPsi \bigg(\bfPsi^\top M_H \bfPsi\bigg)^{-1}\bfPsi^\top\bigg\},
\end{eqnarray*}
where $\bfepsilon=(\epsilon_{11}, \epsilon_{12},\ldots, \epsilon_{1T},\ldots, \epsilon_{N1}, \epsilon_{N2}, \ldots, \epsilon_{NT})\in \mathbb{R}^{NT}$. Now taking conditional expectation and by Assumption \ref{Assumption:common}.\ref{Ac:a2}, on event $\Omega_{NT}$, we have
\begin{eqnarray}
	\ev\bigg(\|\widehat{f}_*-\widetilde{f}_*\|_{NT}^2\bigg|\mathbb{Z}\bigg)&=&\frac{1}{NT}\textrm{Tr}\bigg\{\bfPsi\bigg(\bfPsi^\top M_H \bfPsi\bigg)^{-1}\bfPsi^\top M_H\ev\bigg(\bfepsilon \bfepsilon^\top\bigg| \mathbb{Z}\bigg) M_H\bfPsi \bigg(\bfPsi^\top M_H \bfPsi\bigg)^{-1}\bfPsi^\top\bigg\}\nonumber\\
	&\leq&\frac{a_2}{NT}\textrm{Tr}\bigg\{\bfPsi\bigg(\bfPsi^\top M_H \bfPsi\bigg)^{-1}\bfPsi^\top M_H M_H\bfPsi \bigg(\bfPsi^\top M_H \bfPsi\bigg)^{-1}\bfPsi^\top\bigg\}\nonumber\\
	&=&\frac{a_2}{NT}\textrm{Tr}\bigg\{\bfPsi\bigg(\bfPsi^\top M_H \bfPsi\bigg)^{-1}\bfPsi^\top\bigg\}\nonumber\\
	&=&\frac{a_2}{NT}\textrm{Tr}\bigg\{\bigg(\bfPsi^\top M_H \bfPsi\bigg)^{-1}\bfPsi^\top\bfPsi\bigg\},\nonumber\\
	&\leq&\frac{8a_3^2a_2d_{NT}}{NT}.\nonumber
\end{eqnarray}
Since $\pr(\Omega_{NT})\to 1$, we finish the proof.
\end{proof}

\begin{lemma}\label{lemma:rate:of:convergence:non:penalized:large:T}
Under Assumption \ref{Assumption:common} and \ref{Assumption:A2}, if $A_{NT}^4d_{NT}=o(NT)$, $A_{NT}^2d_{NT}=o(N)$ and $A_{NT}^2=o(T)$, then
\begin{eqnarray*}
\|\widehat{f}_*-f_0\|_{NT}^2=O_P\bigg(\frac{d_{NT}}{NT}+\rho_{NT}^2\bigg) \textrm{ and }\; \|\widehat{f}_*-f_0\|^2=O_P\bigg(\frac{d_{NT}}{NT}+\rho_{NT}^2\bigg).
\end{eqnarray*}
\begin{eqnarray*}
\|\widehat{f}_{*,0}-f_0\|_{NT}^2=O_P\bigg(\frac{d_{N,0}}{NT}+\rho_{N,0}^2\bigg) \textrm{ and }\; \|\widehat{f}_{*,0}-f_0\|^2=O_P\bigg(\frac{d_{N,0}}{NT}+\rho_{N,0}^2\bigg).
\end{eqnarray*}
\end{lemma}
\begin{proof}[Proof of Lemma \ref{lemma:rate:of:convergence:non:penalized}]
	Let $\widetilde{f}_*=P_{NT}f_0$ and $\bar{f}_*=Pf_0$. By Lemma \ref{lemma:expect:value:projection:y:f0:large:T}, it follows that
\begin{eqnarray}\label{eq:lemma:rate:of:convergence:non:penalized:large:T:eq1}
	\|\widehat{f}_*-	\widetilde{f}_*\|_{NT}^2=O_P\bigg(\frac{d_{NT}}{NT}\bigg).
\end{eqnarray}
Now by definition of $\Omega_{NT}$, we have
\begin{eqnarray}\label{eq:lemma:rate:of:convergence:non:penalized:large:T:eq1:1}
	\|\widehat{f}_*-	\widetilde{f}_*\|^2=O_P\bigg(\frac{d_{NT}}{NT}\bigg).
\end{eqnarray}
Next we will deal with $\widetilde{f}_*-\bar{f}_*$. By definition we have
\begin{eqnarray*}
	\|\widetilde{f}_*-\bar{f}_*\|_{NT}&=&\sup_{g\in \Theta_{NT}}\bigg|\frac{\langle \widetilde{f}_*-\bar{f}_*, g\rangle_{NT}}{\|g\|_{NT}}\bigg|\\
	&=&\sup_{g\in \Theta_{NT}}\bigg|\frac{\langle P_{NT}f_0-Pf_0, g\rangle_{NT}}{\|g\|_{NT}}\bigg|\\
	&=&\sup_{g\in \Theta_{NT}}\bigg|\frac{\langle f_0-Pf_0, g\rangle_{NT}-\langle f_0-Pf_0, g\rangle}{\|g\|_{NT}}\bigg|,
\end{eqnarray*}
where the fact $\langle f_0-Pf_0, g\rangle=0$ is used. Let $g_* \in \Theta_{NT}$ satisfy that $\|g_*-f_0\|_\infty\leq \rho_{NT}$, where the existence of such $g_*$ is guaranteed by Lemma \ref{lemma:approximation:error}.
Since
\begin{eqnarray*}
	\langle f_0-Pf_0, g\rangle_{NT}-\langle f_0-Pf_0, g\rangle&=&\langle f_0-g_*, g\rangle_{NT}-\langle f_0-g_*, g\rangle\\
	&&+\langle g_*-Pf_0, g\rangle_{NT}-\langle g_*-Pf_0, g\rangle,
\end{eqnarray*}
we have $\|\widetilde{f}_*-\bar{f}_*\|_{NT}\leq R_1+R_2+R_3$, where
\begin{eqnarray*}
R_1&=&\sup_{g\in \Theta_{NT}}\bigg|\frac{ \langle f_0-g_*, g\rangle_{NT}-\langle f_0-g_*, g\rangle}{\|g\|_{NT}}\bigg|,\\
R_2&=&\sup_{g\in \Theta_{NT}}\bigg|\frac{\langle g_*-Pf_0, g\rangle_{NT}}{\|g\|_{NT}}\bigg|,\\
R_3&=&\sup_{g\in \Theta_{NT}}\bigg|\frac{\langle g_*-Pf_0, g\rangle}{\|g\|_{NT}}\bigg|.
\end{eqnarray*}
Since $\|f_0-g_*\|_\infty\leq \rho_{NT}$, by Lemma \ref{lemma:bounded:sequence:two:norm:difference:large:T}, on event $\Omega_{NT}$, we have 
\begin{eqnarray*}
	R_1^2\leq 4\sup_{g\in \Theta_{NT}}\bigg|\frac{ \langle f_0-g_*, g\rangle_{NT}-\langle f_0-g_*, g\rangle}{\|g\|}\bigg|^2=O_P\bigg(\frac{A_{NT}^2\rho_{NT}^2d_{NT}}{NT}\bigg)=O_P(\rho_{NT}^2).
\end{eqnarray*}

By Lemma \ref{lemma:uniform:equivalence:empirical:population:norm} and triangle inequality, on event $\Omega_N$, we have
\begin{eqnarray*}
	R_2\leq \|g_*-Pf_0\|_{NT}\leq 2\|g_*-Pf_0\|\leq 2\|g_*-f_0\|+2\|Pf_0-f_0\|\leq 4\|g_*-f_0\|\leq 8\rho_{NT},
\end{eqnarray*}
where we use the fact that $\|g\|\leq 2\|g\|_\infty$ and $\|Pf_0-f_0\|\leq \|g_*-f_0\|$.

On event $\Omega_N$, we have
\begin{eqnarray*}
	R_3\leq 2\sup_{g\in \Theta_{NT}}\bigg|\frac{\langle g_*-Pf_0, g\rangle}{\|g\|}\bigg|\leq 2\|g_*-Pf_0\|\leq 8\rho_{NT}.
\end{eqnarray*}
Combining rates of $R_1, R_2, R_3$, we conclude that 
\begin{eqnarray}\label{eq:lemma:rate:of:convergence:non:penalized:large:T:eq2}
	\|\widetilde{f}_*-\bar{f}_*\|_{NT}^2=O_P(\rho_{NT}^2),
\end{eqnarray}
and
\begin{eqnarray}\label{eq:lemma:rate:of:convergence:non:penalized:large:T:eq2:2}
	\|\widetilde{f}_*-\bar{f}_*\|^2=O_P(\rho_{NT}^2).
\end{eqnarray}
By definition of projection, on event $\Omega_N$, we have
\begin{eqnarray*}
\|\bar{f}_*-f_0\|_{NT}&\leq& \|\bar{f}_*-g_*\|_{NT}+\|g_*-f_0\|_{NT}\\
&\leq& 2\|\bar{f}_*-g_*\|+2\|g_*-f_0\|_\infty\\
&=&2\|Pf_0-g_*\|+2\|g_*-f_0\|_\infty\\
&\leq&2\|Pf_0-f_0\|+2\|f_0-g_*\|+2\|g_*-f_0\|_\infty\\
&\leq&4\|f_0-g_*\|+2\|g_*-f_0\|_\infty\\ 
&\leq& 10\|f_0-g_*\|_\infty\leq 10\rho_{NT}
\end{eqnarray*}
which further implies
\begin{eqnarray}\label{eq:lemma:rate:of:convergence:non:penalized:large:T:eq3}
	\|\bar{f}_*-f_0\|_{NT}=O_P(\rho_{NT}).
\end{eqnarray}
Similarly, we can show
\begin{eqnarray}
	\|\bar{f}_*-f_0\|\leq\|\bar{f}_*-g_*\|+\|g_*-f_0\|\leq 2\|g_*-f_0\|\leq 4\|g_*-f_0\|_\infty\leq 4\rho_{NT}.\label{eq:lemma:rate:of:convergence:non:penalized:large:T:eq3:3}
\end{eqnarray}
Combining (\ref{eq:lemma:rate:of:convergence:non:penalized:eq1}), (\ref{eq:lemma:rate:of:convergence:non:penalized:eq2}) and (\ref{eq:lemma:rate:of:convergence:non:penalized:eq3}), we have
\begin{eqnarray*}
\|\widehat{f}_*-f_0\|_{NT}^2=O_P\bigg(\frac{d_{NT}}{NT}+\rho_{NT}^2\bigg).
\end{eqnarray*}
According to (\ref{eq:lemma:rate:of:convergence:non:penalized:eq1:1}), (\ref{eq:lemma:rate:of:convergence:non:penalized:eq2:2}) and (\ref{eq:lemma:rate:of:convergence:non:penalized:eq3:3}), we also obtain
\begin{eqnarray*}
	\|\widehat{f}_*-f_0\|^2=O_P\bigg(\frac{d_{NT}}{NT}+\rho_{NT}^2\bigg).
\end{eqnarray*}
Similar argument can be applied to prove the rate of convergence of $\widehat{f}_{*,0}$.
\end{proof}

In the following, we define  $$\gamma_{NT}^2=\frac{d_{NT}}{NT}+\rho_{NT}^2,$$ and it follows that $\gamma_{NT}^2\asymp (NT)^{-1}\sum_{j=1}^ph_j^{-1}+\sum_{j=d+1}^p h_j^{2m_j}$.

\begin{proof}[Proof of (b) in Theorem \ref{thm:rate:of:convergence:together}]
Define
\begin{eqnarray*}
	R_{NT}=\sup_{g\in \Theta_{NT}}\bigg|\frac{\|g\|_{NT}^2}{\|g\|^2}-1\bigg|.
\end{eqnarray*}
By definition, we have
 Moreover, we have
\begin{eqnarray*}
	1-R_{NT}\leq \sup_{g\in \Theta_{NT}}\frac{\|g\|_{NT}^2}{\|g\|^2}\leq 1+R_{NT}.
\end{eqnarray*}
For any $0<\delta<1$, define event $$U_{{NT}, \delta}=\bigg\{\|\widehat{f}_*-f_0\|\leq C_\delta\gamma_{NT}, \|\widehat{f}_{*,0}-f_0\|\leq C_\delta\gamma_{NT}, R_{NT}\leq \frac{1}{2}\bigg\}\cap \Omega_{NT},$$ where $C_\delta>0$ is sufficiently large such that $\pr(U_{{NT},\delta})\geq 1-\delta$ and this is possible due to Lemma \ref{lemma:rate:of:convergence:non:penalized:large:T} and Lemma \ref{lemma:uniform:equivalence:empirical:population:norm:large:T}.

By definition of $\widehat{f}$, we have
\begin{align}
	0&\geq l_{NT}(\widehat{f})-l_{NT}(\widehat{f}_{*,0})\nonumber\\
	&\geq \frac{\|Y-\widehat{f} \|_{NT}^2-\|Y-\widehat{f}_{*,0}\|_{NT}^2}{2}+\sum_{j=d+1}^p\bigg\{p_{\lambda_{NT}}(\|\widehat{f} _{j,\sim}\|_{NT})-p_{\lambda_{NT}}(\|\widehat{f}_{j,*,0,\sim}\|_{NT})\bigg\}\nonumber\\
	&= \frac{\|\widehat{f}_*-\widehat{f}\|_{NT}^2-\|\widehat{f}_*-\widehat{f}_{*,0}\|_{NT}^2}{2}+\sum_{j=d+1}^p\bigg\{p_{\lambda_{NT}}(\|\widehat{f} _{j,\sim}\|_{NT})-p_{\lambda_{NT}}(\|\widehat{f}_{j,*,0,\sim}\|_{NT})\bigg\}\label{thm:rate:of:convergence:large:T:eq:1}\\
	&\geq\frac{1}{2}\bigg(\|\widehat{f}_*-\widehat{f} \|^2(1-R_{NT})-\|\widehat{f}_*-\widehat{f}_{*,0}\|^2(1+R_{NT})\bigg)-\sum_{j=d+1}^pp_{\lambda_{NT}}(\|\widehat{f}_{j,*,0,\sim}\|_{NT}).\label{thm:rate:of:convergence:large:T:eq:2}
\end{align}
By Lemma \ref{lemma:lower:bound:non:linear}, on event $U_{N,\delta}$, we have
\begin{eqnarray*}
	\|\widehat{f}_{j,*,0, \sim}\|_{NT}\geq \frac{1}{2} \|\widehat{f}_{j,*,0, \sim}\|\geq \frac{1}{4}\sqrt{\frac{a_6}{2a_3}}, \textrm{ for } j=d+1,\ldots, p,
\end{eqnarray*}
provided  $\sqrt{a_6}>2(C_\delta+4)\sqrt{8a_3^3c_1^{-1}}\gamma_{NT}$.
As a consequence, it follows that
\begin{eqnarray}
	p_{\lambda_{NT}}(\|\widehat{f}_{j,*,0, \sim}\|_{NT})=(\kappa+1)\lambda_{NT}^2/2, \textrm{ for } j=d+1,\ldots, p,  \textrm{ if }\; \frac{1}{4}\sqrt{\frac{a_6}{2a_3}}\geq \kappa\lambda_{NT}.\label{thm:rate:of:convergence:large:T:eq:3}
\end{eqnarray}

Combining  (\ref{thm:rate:of:convergence:large:T:eq:2}) and  (\ref{thm:rate:of:convergence:large:T:eq:3}), if $\sqrt{a_6}>2(C_\delta+4)\sqrt{8a_3^3c_1^{-1}}\gamma_{NT}$ and $\sqrt{a_6}\geq 4\sqrt{2a_3}\kappa\lambda_{NT}$, on event $U_{{NT},\delta}$, it follows that
\begin{eqnarray*}
\frac{(p-d)(\kappa+1)\lambda_{NT}^2}{2}+\frac{3}{2}\|\widehat{f}_*-\widehat{f}_{*,0}\|^2\geq \frac{1}{4}\|\widehat{f}_*-\widehat{f} \|^2.
\end{eqnarray*}
Taking square root on both side of above inequality, we have
\begin{eqnarray*}
	\frac{1}{2}\|\widehat{f}_*-\widehat{f} \|\leq \sqrt{{(p-d)(\kappa+1)}}\lambda_{NT}+2\|\widehat{f}_*-\widehat{f}_{*,0}\|,
\end{eqnarray*}
which, by triangle inequality,  further implies that  the following holds on event $U_{{NT}, \delta}$,
\begin{eqnarray*}
	\|\widehat{f}-f_0\|&\leq& 2\sqrt{{(p-d)(\kappa+1)}}\lambda_{NT}+5\|\widehat{f}_*-f_0\|+4\|\widehat{f}_{*,0}-f_0\|\\
	&\leq&2\sqrt{{(p-d)(\kappa+1)}}\lambda_{NT}+9C_\delta\gamma_{NT}.
\end{eqnarray*}
Again by Lemma \ref{lemma:lower:bound:non:linear}, on event $U_{{NT}, \delta}$, it holds that
\begin{eqnarray*}
	\|\widehat{f}_{j, \sim}\|_{NT}\geq \frac{1}{2} \|\widehat{f}_{j,\sim}\|\geq \frac{1}{4}\sqrt{\frac{a_6}{2a_3}}, \textrm{ for } j=d+1,\ldots, p,
\end{eqnarray*}
provided  $\sqrt{a_6}>2(2\sqrt{{(p-d)(\kappa+1)}}+9C_\delta+4)\sqrt{8a_3^3c_1^{-1}}(\lambda_{NT}+\gamma_{NT})$. As a consequence, we have
\begin{eqnarray}
	p_{\lambda_{NT}}(\|\widehat{f}_{j, \sim}\|_{NT})=(\kappa+1)\lambda_{NT}^2/2, \textrm{ for } j=d+1,\ldots, p,  \textrm{ if }\; \frac{1}{4}\sqrt{\frac{a_6}{2a_3}}\geq \kappa\lambda_{NT}.\label{thm:rate:of:convergence:large:T:eq:4}
\end{eqnarray}
Now in the view of  (\ref{thm:rate:of:convergence:large:T:eq:1}),  (\ref{thm:rate:of:convergence:large:T:eq:3}) and (\ref{thm:rate:of:convergence:large:T:eq:4}), the following holds on event $U_{NT, \delta}$,
\begin{eqnarray*}
	\|\widehat{f}_*-\widehat{f}\|\leq \|\widehat{f}_*-\widehat{f}\|_{NT}\leq \|\widehat{f}_*-\widehat{f}_{*,0}\|_{NT}\leq \|\widehat{f}_*-\widehat{f}_{*,0}\|,
\end{eqnarray*}
which further implies 
\begin{eqnarray*}
	\|\widehat{f}-f_0\|\leq 2\|\widehat{f}_*-f_0\|+\|\widehat{f}_{*,0}-f_0\|\leq 3C_\delta\gamma_{NT}.
\end{eqnarray*}
Since $\delta$ can be arbitrary small, we finish the proof.
\end{proof}

\begin{proof}[Proof of (b) in Theorem \ref{thm:selection:consistency:together}]
Fixing $C_\delta>0$ large enough, we need to show that for any $g=\sum_{j=1}^pg_j\in \Theta_{NT}^0$ with $\|g-f_0\|\leq C_\delta\gamma_{NT}$  and any $C>0$ one has $l_{NT}(g)=\min_{r\in \Theta_{NT}^1, \|r\|\leq C\gamma_{NT}}l_{NT}(g+r)$. 

Since $r=\sum_{j=1}^d{r_{j,\sim}}$ and $r_{j, \sim}\in \Theta_{NT}$, by Lemma \ref{lemma:sum:of:l2:norm:bound}, we have
\begin{eqnarray}
	\|r_{j,\sim}\|\leq c_1^{-1}\|r\|\leq c_1^{-1}C\gamma_{NT}.\label{eq:thm:selection:consistency:large:T:eq1}
\end{eqnarray}
Define event
\begin{eqnarray*}
	E_{{NT},\delta}=\{\|\widehat{f}_*-f_0\|\leq C_\delta\gamma_{NT}\}\cap \Omega_{NT},
\end{eqnarray*}
so on event $E_{{NT},\delta}$, we have $\|r_{j,\sim}\|_{NT}\leq 2c_1^{-1}C\gamma_{NT}.$ Moreover, we can select $C_\delta$ sufficient large such that $\pr(E_{N,\delta})\geq 1-\delta$, which is feasible by Lemma \ref{lemma:rate:of:convergence:non:penalized:large:T} and Lemma \ref{lemma:uniform:equivalence:empirical:population:norm:large:T}.
Direct calculation shows
\begin{eqnarray*}
	l_{NT}(g)-l_{NT}(g+r)&=&\frac{\|Y-g\|_{NT}^2-\|Y-g-r\|_{NT}^2}{2}-\sum_{j=1}^dp_{\lambda_{NT}}(\|r_{j,\sim}\|_{NT})\\
	&=&\frac{\|\widehat{f}_*-g\|_{NT}^2-\|\widehat{f}_*-g-r\|_{NT}^2}{2}-\sum_{j=1}^dp_{\lambda_{NT}}(\|r_{j,\sim}\|_{NT})\\
	&\leq& \frac{1}{2}\|r\|_{NT}\bigg(\|\widehat{f}_*-g\|_{NT}+\|\widehat{f}_*-g-r\|_{NT}\bigg)-\sum_{j=1}^dp_{\lambda_{NT}}(\|r_{j,\sim}\|_{NT})\\
	&\equiv&\frac{1}{2}S_1+S_2.
\end{eqnarray*}
Notice on event $E_{NT,\delta}$, we have
\begin{eqnarray*}
	S_1&\leq&2\|r\|_{NT}\bigg(\|\widehat{f}_*-g\|+\|\widehat{f}_*-g-r\|\bigg)\\
	&\leq&2\|r\|_{NT}\bigg(2\|\widehat{f}_*-g\|+\|r\|\bigg)\\
	&\leq& 2\|r\|_{NT}\bigg(2\|\widehat{f}_*-f_0\|+2\|g-f_0\|+\|r\|\bigg)\\
	&\leq&2(4C_\delta+C)\gamma_{NT}\|r\|_{NT}\\
	&\leq&2(4C_\delta+C)\gamma_{NT}\sum_{j=1}^d\|r_{j,\sim}\|_{NT}.
\end{eqnarray*}
By (\ref{eq:thm:selection:consistency:eq1}), on event $E_{NT, \delta}$, if $2c_1^{-1}C\gamma_{NT}\leq \lambda_{NT}$, we have 
\begin{eqnarray*}
	p_{\lambda_{NT}}(\|r_{j,\sim}\|_{NT})=\lambda_{NT}\|r_{j,\sim}\|_{NT}.
\end{eqnarray*}
Therefore, it follows that
\begin{eqnarray*}
	l_{NT}(g)-l_{NT}(g+r)&\leq&2(4C_\delta+C)\gamma_{NT}\sum_{j=1}^d\|r_{j,\sim}\|_{NT}-\sum_{j=1}^dp_{\lambda_{NT}}(\|r_{j,\sim}\|_{NT})\\
	&=&2(4C_\delta+C)\gamma_{NT}\sum_{j=1}^d\|r_{j,\sim}\|_{NT}-\sum_{j=1}^d\lambda_{NT}\|r_{j,\sim}\|_{NT}\\
	&=&\bigg(2(4C_\delta+C)\gamma_{NT}-\lambda_{NT}\bigg)\sum_{j=1}^d\|r_{j,\sim}\|_{NT}<0,
\end{eqnarray*}
provided $2(4C_\delta+C)\gamma_{NT}<\lambda_{NT}$.  

Now we prove that for each on event $E_{{NT},\delta}$, for all $g\in \Theta_{NT}^0$ with $\|g-f_0\|\leq C_{\delta} \gamma_{NT}$ and all  $r\in \Theta_{NT}^1$ with $\|r\|\leq C\gamma_{NT}$ for any $C>0$, we have
\begin{eqnarray}
	l_{NT}(g)=\min_{r\in \Theta_{NT}^1, \|r\|\leq C\gamma_{NT}}l_{NT}(g+r),\label{eq:thm:selection:consistency:large:T:eq2}
\end{eqnarray} 
provided $2(4C_\delta+c_1^{-1}C)\gamma_{NT}<\lambda_{NT}$.  Furthermore let event $F_{{NT},\delta}=\{\|\widehat{f}-f_0\|\leq C_\delta \gamma_{NT}\}$ and choose $C_\delta$ large such that $\pr(F_{{NT},\delta})\geq 1-\delta$. Then we have $\pr(E_{{NT},\delta}\cap F_{{NT},\delta})\geq 1-2\delta$. By (\ref{eq:thm:selection:consistency:eq2}), we have with probability at least $1-2\delta$,
\begin{eqnarray*}
	l_{NT}(\widehat{f})= \min_{r\in \Theta_{NT}^1, \|r\|\leq C\gamma_{NT}}l_{NT}(\widehat{f}+r),
\end{eqnarray*}
provided $2(4C_\delta+c_1^{-1}C)\gamma_{NT}<\lambda_{NT}$, which proves the first conclusion.

By Lemma \ref{lemma:lower:bound:non:linear}, we can see that with probability at least $1-2\delta$, 
\begin{eqnarray*}
		\|\widehat{f}_{j,\sim}\|\geq \frac{1}{2}\sqrt{\frac{a_6}{2a_3}}, \textrm{ for } j=d+1, \ldots, p.
\end{eqnarray*}
provided $\sqrt{a_6}>2(C_\delta+4)\sqrt{8a_3^3c_1^{-1}}\gamma_{NT}$, which is the second conclusion.
\end{proof}

\subsection{Proof of Theorems \ref{thm:asymptotic:normal:multiple} and \ref{thm:asymptotic:normal:multiple:large:T}}
\subsubsection{General Functional}
To estimate $F(f_0)$ for some known functional $F(\cdot)$, the plug-in estimator $F(\widehat{f})$ is used. We follow \cite{cl14} to prove its limit distribution. To proceed further, we need introduce the following notation:
\begin{align*}
Q_{NT}(f)&=\frac{1}{2NT}\sum_{i=1}^N\sum_{t=1}^T\bigg[Y_{it}-f(\bfX_{it})-\frac{1}{T}\sum_{s=1}^T\bigg(Y_{is}-f(\bfX_{is})\bigg)\bigg]^2, &\textrm{ and }\quad Q(f)&=\ev[Q_{NT}(f)].
\end{align*}
Recalling the vector representation of function:
\begin{align*}
	\bfY&=(Y_{11}, Y_{12},\ldots, Y_{it},\ldots, Y_{NT})^\top \in \mathbb{R}^{NT},\\
	\bff&=(f(\bfX_{11}),f(\bfX_{12}),\ldots, f(\bfX_{it}),\ldots, f(\bfX_{NT}))^\top\in \mathbb{R}^{NT}, \textrm{ for } f: \mcX \to \mathbb{R},\\
	\bfepsilon&=(\epsilon_{11},\epsilon_{12},\ldots, \epsilon_{it},\ldots, \epsilon_{NT})^\top=\bfY-\bff_0\in \mathbb{R}^{NT},
\end{align*}
we have 
\begin{align}
	Q_{NT}(f)-Q_{NT}(f_0)&=\frac{1}{2NT}(\bfY-\bff)^\top M_H(\bfY-\bff)-\frac{1}{2NT}(\bfY-\bff_0)^\top M_H(\bfY-\bff_0)\nonumber\\
	&=\frac{1}{2NT}(\bff-\bff_0)^\top M_H(\bff-\bff_0)-\frac{1}{NT}(\bff-\bff_0)^\top M_H\bfepsilon \nonumber\\
	&=\frac{1}{2}\|f-f_0\|_{NT}^2-\frac{1}{NT}(\bff-\bff_0)^\top M_H\bfepsilon,\label{eq:QNT:f:minus:QNT:f0}
\end{align}
and
\begin{equation}
		Q(f)-Q(f_0)=\frac{1}{2}\|f-f_0\|^2.\label{eq:Q:f:minus:Q:f0}
\end{equation}
Furthermore, we define a neighbourhood of $f_0$ as follows
\begin{align*}
	\mcN_{NT}=\{f\in \mcH_0 \;\;|\;\; \|f-f_0\|\leq \gamma_{NT} \log(NT) \}, \quad\textrm{ and }\quad \mcB_{NT}=\mcN_{NT}\cap \Theta_{NT}^0,
\end{align*}
where $\gamma_{NT}$ is defined in (\ref{eq:definition:gamma:N}). By Theorems \ref{thm:rate:of:convergence:together}, \ref{thm:selection:consistency:together},  we have $\widehat{f}\in \mcB_{NT}$ with probability approaching one. Suppose $Q_{NT}(f)-Q_{NT}(f_0)$ can  be approximated by some linear functional for $f\in \mcN_{NT}$. To be more specific, define linear functional $\Delta_{NT}(f_0)[\cdot]$ as follows:
\begin{align}
	\Delta_{NT}(f_0)[f-f_0]=\lim_{\tau \to 0}\frac{Q_{NT}(f_0+\tau(f-f_0))-Q_{NT}(f_0)}{\tau}=-\frac{1}{NT}(\bff-\bff_0)^\top M_H(\bfY-\bff_0),\label{eq:Delta:NT}
\end{align}
which is called the pathwise derivative of $Q_{NT}$ at $f_0$ in the direction $[f-f_0]$. The population counterpart of $\Delta_{NT}$ is defined as $\Delta(f_0)[\cdot]=\ev(\Delta_{NT}(f_0)[\cdot])$. Direct examination shows that
\begin{align*}
\frac{\partial \ev(\Delta_{NT}(f_0+\tau(f-f_0))[f-f_0])}{\partial \tau}|_{\tau=0}=\|f-f_0\|^2.
\end{align*}
Let $\mcV$ be the closed  linear span of $\{f-f_0 \;|\; f\in \mcN_{NT}\}$ under norm $\|\cdot\|$, and it can be verified that $\mcV$ is a Hilbert space with inner product $\langle \cdot, \cdot \rangle$ and
\begin{align*}
	\frac{\partial \ev(\Delta_{NT}(f_0+\tau v_2)[v_1])}{\partial \tau}|_{\tau=0}=\langle v_1, v_2 \rangle, \textrm{ for all } v_1, v_2 \in \mcV.
\end{align*}
Define the best approximation of $f_0$ in $\mcB_{NT}$ as follows: $$f_{0,N}=\argmin_{f\in \mcB_{NT}}\|f-f_0\|_\infty.$$ Now, let $\mcV_{NT}$ be the closed linear span of $\{f-f_{0,N}\; |\; f\in \mcB_{NT}\}$ under $\|\cdot\|$. For any $v\in \mcV$, we further define the following pathwise derivative of $F(\cdot)$ at $f_0$ in the direction of $v\in \mcV$:
\begin{align*}
	\frac{\partial F(f_0)}{\partial f}[v]=\frac{\partial F(f_0+\tau v)}{\partial \tau}|_{\tau=0}, \textrm{ for any } v \in \mcV.
\end{align*}
We assume $\frac{\partial F(f_0)}{\partial f}[v]$ is a linear continuous functional on $\mcV$ and extend $\frac{\partial F(f_0)}{\partial f}$ to  the subset $ \{f-f_{0,N} \;|\; f\in \mcB_{NT}\}$ as follows:
\begin{align}
	\frac{\partial F(f_0)}{\partial f}[f-f_{0,N}]=\frac{\partial F(f_0)}{\partial f}[f-f_0]-\frac{\partial F(f_0)}{\partial f}[f_{0,N}-f_0], \label{eq:extension:functional}
\end{align}
which is still a linear functional. By Extension Theorem of linear continuous functional, (\ref{eq:extension:functional}) can be extended to $\mcV_{NT}$ and by Riesz Representation Theorem, there exists an unique $v_{NT}^*\in \mcV_{NT}$ such that
\begin{align}
	\frac{\partial F(f_0)}{\partial f}[v_{NT}]=\langle v_{NT}, v_{NT}^* \rangle, \textrm{ for all } v_{NT}\in \mcV_{NT}.\label{eq:v:N:star}
\end{align}
To proceed further, for $f, g\in \Theta_{NT}^0$, we define standard deviation inner product as follows:
\begin{align}
\langle g, f \rangle_{\textrm{sd}}=\frac{1}{NT}\ev\bigg(\bff^\top M_H \bfepsilon \bfepsilon^\top M_H \textbf{g}\bigg), \;\;\textrm{ and }\;\; \|f\|_{\textrm{sd}}^2=\langle f, f \rangle_{\textrm{sd}}=\frac{1}{NT}\ev\bigg(\bff^\top M_H \bfepsilon \bfepsilon^\top M_H \textbf{f}\bigg).\label{eq:definition:sd:norm}
\end{align}
It is not difficult to verify that  $\|v_{NT}^*\|_{\textrm{sd}}=\textrm{Var}(\sqrt{NT}\Delta_{NT}(f_0)[v_{NT}^*])$. Let $b_{NT}=o(N^{-1/2}T^{-1/2})$, $u_{NT}^*=v_{NT}^*/\|v_{NT}^*\|_{\textrm{sd}}$.

\begin{Condition}\label{Condition:C1}
\begin{enumerate}[label={(\roman*}),ref={(\roman*})]
\item \label{C1:a}  $\frac{\partial F(f_0)}{\partial f}[v]$ is a linear continuous functional from $\mcV$ to $\mathbb{R}$.
\item \label{C1:b}  $$\sup_{f \in \mcB_{NT}}\frac{\bigg|F(f)-F(f_0)-\frac{\partial F(f_0)}{\partial f}[f-f_0]\bigg|}{\|v_{NT}^*\|}=o(N^{-1/2}T^{-1/2}).$$
\item  \label{C1:c} $$\frac{\bigg|\frac{\partial F(f_0)}{\partial f}[f_{0,N}-f_0]\bigg|}{\|v_{NT}^*\|}=o(N^{-1/2}T^{-1/2})$$
\item \label{C1:d} $\|v_{NT}^*\|/\|v_{NT}^*\|_{\textrm{sd}}=O(1)$.
\end{enumerate}
\end{Condition}

\begin{lemma}\label{lemma:general:asymptotoic:expansion}
Suppose Condition \ref{Condition:C1} and one of following conditions are satisfied:
\begin{enumerate}
\item Assumptions \ref{Assumption:A1}, \ref{Assumption:common} hold and $d_{NT}A_{NT}^2=o(N)$, $\gamma_{NT}=o(\lambda_{NT})$, $d_{NT,0}A_{NT,0}^2\gamma_{NT}^2=o(1)$, $N\rho_{NT}^2=o(1)$;
\item Assumptions \ref{Assumption:A2}, \ref{Assumption:common} hold and $d_{NT}A_{NT}^4=o(NT)$, $d_{NT}A_{NT}^2=o(N)$, $A_{NT}^2=o(T)$, $\gamma_{NT}=o(\lambda_{NT})$, $d_{NT,0}A_{NT,0}^4\gamma_{NT}^2=o(1)$, $d_{NT,0}^2A_{NT,0}^4\gamma_{NT}^2T=o(N)$, $NT\rho_{NT}^2=o(1)$.
\end{enumerate}
Then it follows that
\begin{align*}
	\frac{\sqrt{NT}\bigg(F(\widehat{f})-F(f_{0})\bigg)}{\|v_{NT}^*\|_{\textrm{sd}}}=-\sqrt{NT}\Delta_{NT}(f_0)[u_{NT}^*]+o_P(1).
\end{align*}
\end{lemma}
\begin{proof}[Proof of Lemma \ref{lemma:general:asymptotoic:expansion}]
Let $\widehat{f}_{u_{NT}^*}=\widehat{f}+b_{NT}u_{NT}^*$.  Since $\|\widehat{f}\|=O_P(\gamma_{NT})$, $\|u_{NT}^*\|=O(1)$, and $b_{NT}=o(N^{-1/2}T^{-1/2})$, we have $\widehat{f},\; \widehat{f}_{u_{NT}^*}\in \mcB_{NT}$ with probability approaching one by Theorems \ref{thm:rate:of:convergence:together} and \ref{thm:selection:consistency:together}. By definition of $\widehat{f}$, it follows that
\begin{align}
	0&\leq l_{NT}(\widehat{f}_{u_{NT}^*})-l_{NT}(\widehat{f})\nonumber\\
	&=2Q_{NT}(\widehat{f}_{u_{NT}^*})-2Q_{NT}(\widehat{f})+\sum_{j=1}^pp_{\lambda_{NT}}\bigg(\|\widehat{f}_{j,u_{NT}^*,\sim}\|_{NT}\bigg)-\sum_{j=1}^pp_{\lambda_{NT}}\bigg(\|\widehat{f}_{j,\sim}\|_{NT}\bigg)\nonumber\\
	&=2Q_{NT}(\widehat{f}_{u_{NT}^*})-2Q_{NT}(\widehat{f}),\label{eq:lemma:general:asymptotoic:expansion:eq1}
\end{align}
where the last equality follows from Lemma \ref{lemma:SNT:rate:3}. Direct examination yields
\begin{align}
	Q_{NT}(\widehat{f}_{u_{NT}^*})-Q_{NT}(\widehat{f})&=Q(\widehat{f}_{u_{NT}^*})-Q(\widehat{f})+\Delta_{NT}(f_0)[\widehat{f}_{u_{NT}^*}-\widehat{f}]-\Delta(f_0)[\widehat{f}_{u_{NT}^*}-\widehat{f}]+S_{NT},\label{eq:lemma:general:asymptotoic:expansion:eq2}
\end{align}
where
\begin{align*}
	S_{NT}&=\bigg(Q_{NT}(\widehat{f}_{u_{NT}^*})-Q_{NT}(\widehat{f})\bigg)-\bigg(Q(\widehat{f}_{u_{NT}^*})-Q(\widehat{f})\bigg)-\bigg(\Delta_{NT}(f_0)[\widehat{f}_{u_{NT}^*}-\widehat{f}]-\Delta(f_0)[\widehat{f}_{u_{NT}^*}-\widehat{f}]\bigg)\nonumber\\
	&\equiv S_1-S_2-S_3.
\end{align*}
Therefore, by (\ref{eq:QNT:f:minus:QNT:f0}), and the fact that $\widehat{f}_{u_{NT}^*}=\widehat{f}+b_{NT}u_{NT}^*$, we have
\begin{align*}
	S_1&=Q_{NT}(\widehat{f}_{u_{NT}^*})-Q_{NT}(\widehat{f})\nonumber\\
	&=\frac{1}{2NT}(\bfY-\widehat{\bff}_{u_{NT}^*})^\top M_H (\bfY-\widehat{\bff}_{u_{NT}^*})-\frac{1}{2NT}(\bfY-\widehat{\bff})^\top M_H (\bfY-\widehat{\bff})\nonumber\\
	&=\frac{b_{NT}^2}{2NT}\bfu_{NT}^{*\top}M_H \bfu_{NT}^*-\frac{b_{NT}}{NT}\bfu_{NT}^{*\top}M_H (\bfY-\widehat{\bff}),
\end{align*}
and 
\begin{align*}
S_2&=\ev\bigg(Q_{NT}(\widehat{f}_{u_{NT}^*})-Q_{NT}(\widehat{f})\bigg)=\ev(S_1).
\end{align*}
Furthermore, by (\ref{eq:Delta:NT}), it follows that
\begin{align*}
	\Delta_{NT}(f_0)[\widehat{f}_{u_{NT}^*}-\widehat{f}]&=b_{NT}\Delta_{NT}(f_0)[u_{NT}^*]=-\frac{b_{NT}}{NT}\bfu_{NT}^{*\top}M_H(\bfY-\bff_0)=-\frac{b_{NT}}{NT}\bfu_{NT}^{*\top}M_H\bfepsilon
\end{align*}
which, by Assumption \ref{Assumption:common}.\ref{Ac:a2}  further implies
\begin{align*}
	S_3&=-\frac{b_{NT}}{NT}\bfu_{NT}^{*\top}M_H\bfepsilon+\ev\bigg(\frac{b_{NT}}{NT}\bfu_{NT}^{*\top}M_H\bfepsilon\bigg)=-\frac{b_{NT}}{NT}\bfu_{NT}^{*\top}M_H\bfepsilon.
\end{align*}
Combining above equations, we have
\begin{align*}
	S_{NT}&=\frac{b_{NT}^2}{2NT}\bfu_{NT}^{*\top}M_H \bfu_{NT}^*-\ev\bigg(\frac{b_{NT}^2}{2NT}\bfu_{NT}^{*\top}M_H \bfu_{NT}^*\bigg)+\frac{b_{NT}}{NT}\bfu_{NT}^{*\top}M_H(\widehat{\bff}-\bff_0)-\ev\bigg(\frac{b_{NT}}{NT}\bfu_{NT}^{*\top}M_H(\widehat{\bff}-\bff_0)\bigg)\nonumber\\
	&=O_P(b_{NT}^2)+o_P(b_{NT}N^{-1/2}T^{-1/2}),
\end{align*}
here the last equation is due to Lemmas \ref{lemma:SNT:rate:1} and \ref{lemma:SNT:rate:2} (Lemmas \ref{lemma:SNT:rate:1:large:T} and \ref{lemma:SNT:rate:2:large:T} for diverging $T$). Furthermore, by (\ref{eq:Q:f:minus:Q:f0}), it follows that
\begin{align*}
	Q(\widehat{f}_{u_{NT}^*})-Q(\widehat{f})&=Q(\widehat{f}_{u_{NT}^*})-Q(f_0)+Q(f_0)-Q(\widehat{f})\nonumber\\
	&=\frac{1}{2}\|\widehat{f}_{u_{NT}^*}-f_0\|^2-\frac{1}{2}\|\widehat{f}-f_0\|^2\nonumber\\
	&=\frac{1}{2}\|\widehat{f}+b_{NT}u_{NT}^*-f_0\|^2-\frac{1}{2}\|\widehat{f}-f_0\|^2\nonumber\\
	&=b_{NT}^2\|u_{NT}^*\|^2+b_{NT}\langle u_{NT}^*, \widehat{f}-f_0 \rangle \nonumber\\
	&=b_{NT}\langle u_{NT}^*, \widehat{f}-f_0 \rangle+b_{NT}^2.
\end{align*}
Above equations, (\ref{eq:lemma:general:asymptotoic:expansion:eq1}), and (\ref{eq:lemma:general:asymptotoic:expansion:eq2}) together imply that
\begin{align}
	0&\leq  b_{NT}\langle u_{NT}^*, \widehat{f}-f_0 \rangle+b_{NT}^2+\Delta_{NT}(f_0)[\widehat{f}_{u_{NT}^*}-\widehat{f}]+O_P(b_{NT}^2)+o_P(b_{NT}N^{-1/2}T^{-1/2})\nonumber\\
	&= b_{NT}\langle u_{NT}^*, \widehat{f}-f_0 \rangle+b_{NT} \Delta_{NT}(f_0)[u_{NT}^*]+O_P(b_{NT}^2)+o_P(b_{NT}N^{-1/2}T^{-1/2}).\label{eq:lemma:general:asymptotoic:expansion:eq3}
\end{align}
Replacing $u_{NT}^*$ by $-u_{NT}^*$, we can obtain that
\begin{align}
		0&\leq - b_{NT}\langle u_{NT}^*, \widehat{f}-f_0 \rangle-b_{NT}\Delta_{NT}(f_0)[u_{NT}^*]+O_P(b_{NT}^2)+o_P(b_{NT}N^{-1/2}T^{-1/2}).\label{eq:lemma:general:asymptotoic:expansion:eq4}
\end{align}
Combining (\ref{eq:lemma:general:asymptotoic:expansion:eq3}), and (\ref{eq:lemma:general:asymptotoic:expansion:eq4}), we conclude that
\begin{align}
	\langle u_{NT}^*, \widehat{f}-f_0 \rangle&=-\Delta_{NT}(f_0)[u_{NT}^*]+O_P(b_{NT})+o_P(N^{-1/2}T^{-1/2})\nonumber\\
	&=-\Delta_{NT}(f_0)[u_{NT}^*]+o_P(N^{-1/2}T^{-1/2}).\label{eq:lemma:general:asymptotoic:expansion:eq5}
\end{align}
By definition of $f_{0,N}$ and Lemma \ref{lemma:approximation:error}, it follows that
\begin{align*}
|\langle f_{0,N}-f_0, u_{NT}^* \rangle| &\leq \|f_{0,N}-f_0\| \|u_{NT}^*\| \leq \|f_{0,N}-f_0\|_\infty\|u_{NT}^*\| \leq \|g_*-f_0\|\|u_{NT}^*\| = O(\rho_{NT}),
\end{align*}
where we use the fact that $\|u_{NT}^*\|=O(1)$ in Condition \ref{Condition:C1}.\ref{C1:d}.
As a consequence of (\ref{eq:lemma:general:asymptotoic:expansion:eq5}), above inequality, and the fact that $\rho_{NT}=o(N^{-1/2}T^{-1/2})$, we have
\begin{align}
	\langle u_{NT}^*, \widehat{f}-f_{0,N} \rangle&=\langle u_{NT}^*, \widehat{f}-f_0 \rangle+\langle u_{NT}^*,  f_{0,N}-f_0 \rangle\nonumber\\
	&=-\Delta_{NT}(f_0)[u_{NT}^*]+o_P(N^{-1/2}T^{-1/2}).\label{eq:lemma:general:asymptotoic:expansion:eq6}
\end{align}
Moreover, by Conditions \ref{Condition:C1}.\ref{C1:b}, \ref{Condition:C1}.\ref{C1:d}, (\ref{eq:extension:functional}), and (\ref{eq:v:N:star}), it follows that
\begin{align}
	&\frac{F(\widehat{f})-F(f_{0,N})}{\|v_{NT}^*\|_{\textrm{sd}}}\nonumber\\
	=&\frac{F(\widehat{f})-F(f_0)-\frac{\partial F(f_0)}{\partial f}[\widehat{f}-f_0]}{\|v_{NT}^*\|_{\textrm{sd}}}-\frac{F(f_{0,N})-F(f_0)-\frac{\partial F(f_0)}{\partial f}[f_{0,N}-f_0]}{\|v_{NT}^*\|_{\textrm{sd}}}\nonumber\\
	&+\frac{\frac{\partial F(f_0)}{\partial f}[\widehat{f}-f_0]-\frac{\partial F(f_0)}{\partial f}[f_{0,N}-f_0]}{\|v_{NT}^*\|_{\textrm{sd}}}\nonumber\\
	=&o_P(N^{-1/2}T^{-1/2})+o_P(N^{-1/2}T^{-1/2})+\frac{ \langle \widehat{f}-f_{0,N}, v_{NT}^*\rangle}{\|v_{NT}^*\|_{\textrm{sd}}}\nonumber\\
	=& \langle \widehat{f}-f_{0,N}, u_{NT}^*\rangle+o_P(N^{-1/2}T^{-1/2}).\label{eq:lemma:general:asymptotoic:expansion:eq7}
\end{align}
Finally, in the view of (\ref{eq:lemma:general:asymptotoic:expansion:eq6}), and  (\ref{eq:lemma:general:asymptotoic:expansion:eq7}), we have
\begin{align*}
	\frac{\sqrt{NT}\bigg(F(\widehat{f})-F(f_{0,N})\bigg)}{\|v_{NT}^*\|_{\textrm{sd}}}=-\sqrt{N}\Delta_{NT}(f_0)[u_{NT}^*]+o_P(1).
\end{align*}
Moreover, by Conditional \ref{Condition:C1}.\ref{C1:d}, it follows that
\begin{align*}
	\frac{\sqrt{NT}\bigg(F(\widehat{f})-F(f_{0})\bigg)}{\|v_{NT}^*\|_{\textrm{sd}}}=-\sqrt{N}\Delta_{NT}(f_0)[u_{NT}^*]+o_P(1).
\end{align*}
\end{proof}

\begin{lemma}\label{lemma:SNT:rate:3}
Suppose one of the following conditions is satisfied:
\begin{enumerate}
\item Assumptions \ref{Assumption:A1}, \ref{Assumption:common} hold and $d_{NT}A_{NT}^2=o(N)$, $\gamma_{NT}=o(\lambda_{NT})$;
\item Assumptions \ref{Assumption:A2}, \ref{Assumption:common} and $d_{NT}A_{NT}^4=o(NT)$, $d_{NT}A_{NT}^2=o(N)$, $A_{NT}^2=o(T)$, $\gamma_{NT}=o(\lambda_{NT})$.
\end{enumerate}
Then with probability approaching one, it follows that
\begin{align*}
	\sum_{j=1}^pp_{\lambda_{NT}}\bigg(\|\widehat{f}_{j,u_{NT}^*,\sim}\|_{NT}\bigg)-\sum_{j=1}^pp_{\lambda_{NT}}\bigg(\|\widehat{f}_{j,\sim}\|_{NT}\bigg)=0
\end{align*}
\end{lemma}
\begin{proof}[Proof of Lemma \ref{lemma:SNT:rate:3}]
By Theorems \ref{thm:rate:of:convergence:together} and \ref{thm:selection:consistency:together}, and definition of $\widehat{f}_{u_{NT}^*}$, we have that with probability approaching one, $\widehat{f},\; \widehat{f}_{u_{NT}^*} \in \mcB_{NT}$. By Lemma \ref{lemma:uniform:equivalence:empirical:population:norm} (Lemma \ref{lemma:uniform:equivalence:empirical:population:norm:large:T} for diverging $T$) and Lemma \ref{lemma:lower:bound:non:linear} , it follows that with probability approaching one, 
\begin{align*}
	\|\widehat{f}_{j,u_{NT}^*,\sim}\|_{NT}\geq \frac{1}{4}\sqrt{\frac{a_6}{2a_1}},\;\; \textrm{ and }\;\; \|\widehat{f}_{j,\sim}\|_{NT}\geq \frac{1}{4}\sqrt{\frac{a_6}{2a_1}}\;\;\textrm{ for }\;\; j=1,\ldots, d.
\end{align*}
Finally, we finish the proof by noticing $p_{\lambda_{NT}}(z)=(\kappa+1)\lambda_{NT}^2/2$ if $z>\kappa\lambda_{NT}$.
\end{proof}

\begin{lemma}\label{lemma:SNT:rate:1}
Suppose Assumptions \ref{Assumption:A1}, and \ref{Assumption:common} hold. If  $A_{NT,0}^2=o(N)$, then
\begin{align*}
	\frac{1}{2NT}\bfu_{NT}^{*\top}M_H \bfu_{NT}^*-\ev\bigg(\frac{1}{2NT}\bfu_{NT}^{*\top}M_H \bfu_{NT}^*\bigg)=O_P(1)
\end{align*}
\end{lemma}
\begin{proof}[Proof of Lemma \ref{lemma:SNT:rate:1}]
The result can be proved similarly using Bernstein inequality as in Lemma \ref{lemma:SNT:rate:2} and we omit the proof.
\end{proof}

\begin{lemma}\label{lemma:SNT:rate:2}
Suppose Assumptions \ref{Assumption:A1}, and \ref{Assumption:common} hold. If  $d_{NT,0}A_{NT,0}^2\gamma_{NT}^2=o(1)$, then
\begin{align*}
	\frac{1}{NT}\bfu_{NT}^{*\top}M_H(\widehat{\bff}-\bff_0)-\ev\bigg(\frac{1}{NT}\bfu_{NT}^{*\top}M_H(\widehat{\bff}-\bff_0)\bigg)=o_P(N^{-1/2}T^{-1/2})
\end{align*}
\end{lemma}
\begin{proof}[Proof of Lemma \ref{lemma:SNT:rate:2}]
Let $\mcF_N=\{f\in \Theta_{NT}^0 \;|\; \|f-f_0\|_2\leq C\gamma_{NT}\}$. Notice if $f\in \mcF_N$, by Lemmas \ref{lemma:l2:norm:sup:norm} and \ref{lemma:approximation:error}, it follows that
\begin{align*}
	\|f-f_0\|_\infty&\leq \|f-g_*\|_\infty+\|g_*-f_0\|_\infty\nonumber\\
	&\leq A_{NT,0}\|f-g_*\|_2+\rho_{NT}\nonumber\\
	&\leq CA_{NT,0}\gamma_{NT}+\rho_{NT}\nonumber\\
	&\leq 2CA_{NT,0}\gamma_{NT},
\end{align*}
where we use the fact that $\rho_{NT}=o(A_{NT,0}\gamma_{NT})$. Similar to Lemma \ref{lemma:uniform:equivalence:empirical:population:norm}, we define
\begin{align*}
	\xi_i(f)=\zeta_{i}(u_{NT}^*, f-f_0),\;\; X_i(f)={\xi_i(f)-\ev(\xi_i(f))},\textrm{ and }\; S_N(f)=\frac{1}{N}\sum_{i=1}^NX_i(f).
\end{align*}
It is not difficult to verify the following equality:
\begin{align*}
	S_N(f)=\frac{1}{NT}\bfu_{NT}^{*\top}M_H({\bff}-\bff_0)-\ev\bigg(\frac{1}{NT}\bfu_{NT}^{*\top}M_H({\bff}-\bff_0)\bigg).
\end{align*}
Similar to the proof of Lemma \ref{lemma:uniform:equivalence:empirical:population:norm}, we also can show that for $f_1, f_2\in \mcF_N$, the following holds
\begin{align*}
	|\xi_i(f_1)-\xi_i(f_2)|\leq 4A_{NT,0}^2\|u_{NT}^*\|_2\|f_1-f_2\|_2,
\end{align*}
which further leads to 
\begin{align*}
	|X_i(f_1)-X_i(f_2)|\leq 8A_{NT,0}^2\|u_{NT}^*\|_2\|f_1-f_2\|_2.
\end{align*}
Moreover, we also have
\begin{align*}
	\textrm{Var}\bigg(X_i(f_1)-X_i(f_2)\bigg)\leq 8 a_1A_{NT,0}^2\|u_{NT}^*\|_2^2\|f_1-f_2\|_2^2.
\end{align*}
By Bernstein inequality, it follows that
\begin{eqnarray}
	&&\pr\bigg(\bigg|S_N(\theta_1)-S_N(\theta_2\bigg|>\frac{xs}{\sqrt{N}} \bigg)\nonumber\\
	&\leq& 2\exp\bigg(-\frac{x^2s^2}{16 a_1A_{NT,0}^2\|u_{NT}^*\|_2^2\|f_1-f_2\|_2^2+\frac{8A_{NT,0}^2}{3\sqrt{N}}\|u_{NT}^*\|_2\|f_1-f_2\|_2 xs}\bigg)\nonumber\\
	&\leq& 2\exp\bigg(-\frac{x^2s^2}{32a_1A_{NT,0}^2\|u_{NT}^*\|_2^2\|f_1-f_2\|_2^2}\bigg)+2\exp\bigg(-\frac{\sqrt{N}xs}{8A_{NT,0}^2\|u_{NT}^*\|_2\|f_1-f_2\|_2}\bigg).\nonumber\\\label{eq:lemma:SNT:rate:2:eq1}
\end{eqnarray}
Let $\delta_k=3^{-k-k_0}$ for $k\geq 0$ and some $k_0$ such that $C\gamma_{NT} \leq 3^{-k_0}\leq 2C\gamma_{NT}$. For sufficient large integer $K$, which will be specified later, let $\{0\}=\mcH_0\subset \mcH_1 \ldots \mcH_K$ be a sequence of subsets of $\mcF_N$ such that $\min_{f^*\in \mcH_k}\|f^*-f\|_2\leq \delta_k$ for all $f \in \mcF_N$. Moreover the subsets $\mcH_K$ is chosen inductively such that two different elements in $\mcH_k$ is at least $\delta_k$ apart. 

By definition, the cardinality $\#(\mcH_k)$ of $\mcH_k$ is bounded by the $\delta_k/2$-covering number $D(\delta_k/2, \mcF_N, \|\cdot\|_2)$. Therefore, we have
\begin{eqnarray*}
	\#(\mcH_k)\leq D(\delta_k/2, \mcF_N, \|\cdot\|_2)\leq \bigg(\frac{8C\gamma_{NT}+\delta_k}{\delta_k}\bigg)^{d_{NT,0}},
\end{eqnarray*} 
where the last inequality is due to \cite{van2000}[Lemma 2.5] and the fact that $\mcF_N$ can be treated as a ball with radius $C\gamma_{NT}$ in $\mathbb{R}^{d_{NT,0}}$.  For any $f \in \mcF_N$, let $\tau_k(f)\in \mcH_k$  be a element such that $\|\tau_k(f)-f\|\leq \delta_k$, for $k=1,2,\ldots, K$. Now for any fixed $x>0$, by (\ref{eq:lemma:SNT:rate:2:eq1}) and the definition of $\tau_k$, we have
\begin{eqnarray}
	&&\pr\bigg(\sup_{f\in \mcF_N} |S_N(f)|>\frac{x}{\sqrt{N}}\bigg)\nonumber\\
	&=&\pr\bigg(\sup_{f\in \mcF_N} \bigg|S_N(f)-S_N\bigg(\tau_K(f)\bigg)\bigg|>\frac{x}{2^K\sqrt{N}}\bigg)\nonumber\\
	&&+\sum_{k=1}^K\pr\bigg(\sup_{f\in \mcF_N} \bigg|S_N\bigg(\tau_k\circ\ldots \circ \tau_K(f)\bigg)-S_N\bigg(\tau_{k-1}\circ\tau_k\circ\ldots \circ \tau_K(f)\bigg)\bigg|>\frac{x}{2^{k-1}\sqrt{N}}\bigg)\nonumber\\
	&\leq& \pr\bigg(\sup_{f\in \mcF_N}4A_{NT,0}^2\|u_{NT}^*\|_2\|f-\tau_K(f)\|_2>\frac{x}{2^K\sqrt{N}}\bigg)\nonumber\\ 
	&&+\sum_{k=1}^K\#(\mcH_k)\sup_{f\in \mcF_N}\pr\bigg( \bigg|S_N\bigg(\tau_k\circ\ldots \circ \tau_K(f)\bigg)-S_N\bigg(\tau_{k-1}\circ\tau_k\circ\ldots \circ \tau_K(f)\bigg)\bigg|>\frac{x}{2^{k-1}\sqrt{N}}\bigg)\nonumber\\
	&\leq& \pr\bigg(\sup_{f\in \mcF_N}4A_{NT,0}^2\|u_{NT}^*\|_2\delta_K>\frac{x}{2^K\sqrt{N}}\bigg)\nonumber\\ 
	&&+\sum_{k=1}^K\#(\mcH_k)\sup_{f\in \mcH_k}\pr\bigg( \bigg|S_N\bigg(f\bigg)-S_N\bigg(\tau_{k-1}(f)\bigg)\bigg|>\frac{x}{2^{k-1}\sqrt{N}}\bigg)\nonumber\\
	&\leq& \pr\bigg(\sup_{f\in \mcF_N}4A_{NT,0}^2\|u_{NT}^*\|_23^{-K-k_0}>\frac{x}{2^K\sqrt{N}}\bigg)\nonumber\\ 
	&&+\sum_{k=1}^\infty \bigg(\frac{8C\gamma_{NT}+\delta_k}{\delta_k}\bigg)^{d_{NT,0}}\sup_{f\in \mcH_k}2\exp\bigg(-\frac{x^2}{32a_1A_{NT,0}^22^{2(k-1)}\|u_{NT}^*\|_2^2\|f-\tau_{k-1}(f)\|_2^2}\bigg)\nonumber\\
	&&+\sum_{k=1}^\infty \bigg(\frac{8C\gamma_{NT}+\delta_k}{\delta_k}\bigg)^{d_{NT,0}}\sup_{f\in \mcH_k}2\exp\bigg(-\frac{\sqrt{N}x}{8A_{NT,0}^22^{k-1}\|u_{NT}^*\|_2\|f-\tau_{k-1}(f)\|_2}\bigg)\nonumber\\
	&\equiv&S_1+S_2+S_3. \label{eq:lemma:SNT:rate:2:eq2}
\end{eqnarray}

For fixed $x>0$, choose $K=K(N)$ large enough that $8C\gamma_{NT}A_{NT,0}^2\|u_{NT}^*\|_2\sqrt{N}(2/3)^K<x$, which can be done due to Condition \ref{Condition:C1}.\ref{C1:d} and Lemma \ref{lemma:expectation:sample:variance}. So by the fact that $3^{-k_0}\leq 2C\gamma_{NT}$, we have $S_1=0$. Moreover, direct examination leads to
\begin{align*}
	S_2&\leq \sum_{k=1}^\infty \bigg(\frac{8C\gamma_{NT}+\delta_k}{\delta_k}\bigg)^{d_{NT,0}}\sup_{f\in \mcH_k}2\exp\bigg(-\frac{x^2}{32a_1A_{NT,0}^22^{2(k-1)}\|u_{NT}^*\|_2^2\|f-\tau_{k-1}(f)\|_2^2}\bigg)\nonumber\\
	&\leq \sum_{k=1}^\infty \bigg(8\times 3^{k}+1\bigg)^{d_{NT,0}}2\exp\bigg(-\frac{x^2}{32a_1A_{NT,0}^22^{2(k-1)}\|u_{NT}^*\|_2^2\delta_k^2}\bigg)\nonumber\\
	&\leq 2\sum_{k=1}^\infty 3^{(k+2)d_{NT,0}}\exp\bigg(-\frac{x^23^{2k+2k_0}}{32a_1A_{NT,0}^22^{2(k-1)}\|u_{NT}^*\|_2^2}\bigg)\nonumber\\
	&\leq 2\sum_{k=1}^\infty 3^{(k+2)d_{NT,0}}\exp\bigg(-\frac{x^2}{32C^2a_1A_{NT,0}^2\gamma_{NT}^2\|u_{NT}^*\|_2^2}\bigg(\frac{9}{4}\bigg)^k\bigg)\nonumber\\
	&= 2\sum_{k=1}^\infty \exp\bigg((k+2)d_{NT,0}\log 3-\frac{x^2}{32C^2a_1A_{NT,0}^2\gamma_{NT}^2\|u_{NT}^*\|_2^2}\bigg(\frac{9}{4}\bigg)^k\bigg).\nonumber
\end{align*}
Let $N$ is large enough such that $[64C^2a_1\log 3] (k+2)d_{NT,0}A_{NT,0}^2\gamma_{NT}^2\|u_{NT}^*\|_2^2<(9/4)^kx^2$ for all $k\geq 1$. This is possible due to $d_{NT,0}A_{NT,0}^2\gamma_{NT}^2=o(1)$, and $\|u_{NT}^*\|_2=O(1)$. So it follows from the equality $e^{-x}\leq e^{-1}/x$ that
\begin{align*}
	S_2&\leq 2\sum_{k=1}^\infty \exp\bigg(-\frac{x^2}{64C^2a_1A_{NT,0}^2\gamma_{NT}^2\|u_{NT}^*\|_2^2}\bigg(\frac{9}{4}\bigg)^k\bigg)\nonumber\\
	&\leq128e^{-1}x^{-2}C^2a_1A_{NT,0}^2\gamma_{NT}^2\|u_{NT}^*\|_2^2\sum_{k=1}^\infty(4/9)^k=O(A_{NT,0}^2\gamma_{NT}^2)=o(1). \nonumber
\end{align*}
By similar technique, if $d_{NT,0}A_{NT,0}^2\gamma_{NT}=o(\sqrt{N})$, then we also can show that
\begin{align*}
	S_3=O\bigg(\frac{A_{NT,0}^2\gamma_{NT}}{\sqrt{N}}\bigg)=o(1).
\end{align*}
Since $x>0$ can be arbitrary, by the bounds of $S_1, S_2, S_3$, and (\ref{eq:lemma:SNT:rate:2:eq2}), we conclude that
\begin{align*}
	\sup_{f\in \mcF_N}\bigg|\frac{1}{NT}\bfu_{NT}^{*\top}M_H({\bff}-\bff_0)-\ev\bigg(\frac{1}{NT}\bfu_{NT}^{*\top}M_H({\bff}-\bff_0)\bigg)\bigg|=o_P(N^{-1/2})
\end{align*}
By Theorems \ref{thm:rate:of:convergence:together}, and \ref{thm:selection:consistency:together}, it follows that $\lim_{C\to \infty}P(\widehat{f}\in \mcF_N)=1$. Therefore, we finish the proof by noticing $T$ is fixed.
\end{proof}

\begin{lemma}\label{lemma:existence:good:basis}
 Under Assumption \ref{Assumption:common}, there exists a basis $\bfP_j(z)=(P_{j1}(z), \ldots, P_{j,M_j+r_j-1}(z))^\top \in \mathbb{R}^{M_j+r_j-1}$ of $\textrm{CSpl}(r_j, \bft_{M_j})$ such that
 \begin{align*}
 	c_3^{-1}h_j\leq \lambda_{\min}\bigg(\int_0^1 \bfP_j(z)\bfP_j^\top(z)dz\bigg)\leq \lambda_{\max}\bigg(\int_0^1 \bfP_j(z)\bfP_j^\top(z)dz\bigg)\leq c_3h_j,
 \end{align*}
 for some $c_3>1$ and all $j=d+1,\ldots, p$. Moreover, for fixed $z_0\in [0,1]$, there exists $c_4>1$, which is only relying on $z_0$ such that
\begin{align*}
c_4^{-1}\leq \bfP_j^\top(z_0)\bfP_j(z_0)\leq  c_4, \textrm{ for all } j=d+1, \ldots, p.	
\end{align*} 
\end{lemma}
\begin{proof}[Proof of Lemma \ref{lemma:existence:good:basis}]
Let $\bfR_j(z)=(R_{j,1}(z),\ldots, R_{j,M_j+r_j})^\top \in \mathbb{R}^{M_j+r_j}$ be the  B-spline basis on knots $\bft_{M_j}$ with degree $r_j$. By \cite{d78}, \cite{gkkw06}[Lemma 14.4], and Assumption \ref{Assumption:common}.\ref{Ac:b}, there exists $a>1$ such that
\begin{align}
a^{-1}h_jc^\top c \leq c^T\bigg(\int_0^1\bfR_j(z)\bfR_j^\top(z)dz\bigg) c\leq ah_jc^\top c ,  \textrm{ for all }  c\in \mathbb{R}^{M_j+r_j}\label{eq:lemma:existence:good:basis:eq00}
\end{align}
and
\begin{align}
	\sum_{k=1}^{M_j+r_j}R_{j,k}(z)=1, \textrm{ for all } z\in [0,1].\label{eq:lemma:existence:good:basis:eq0}
\end{align}
Let $\overline{\bfR}_j=(\overline{R}_{j,1},\ldots, \overline{R}_{j,M_j+r_j})=\int_0^1\bfR_j(z)dz\in \mathbb{R}^{M_j+r_j}$, $u=(1, \ldots, 1)^\top \in \mathbb{R}^{M_j+r_j}$, and we try to find matrices $D\in \mathbb{R}^{(M_j+r_j-1)\times (M_j+r_j)}$ and $K\in \mathbb{R}^{(M_j+r_j)\times (M_j+r_j-1)}$ such that
\begin{align}
	DK=I\quad \textrm{ and } \quad KD=I-\overline{\bfR}_ju^\top.\label{eq:lemma:existence:good:basis:eq1}
\end{align}
In the following, we will show such matrices $D, K$ exist. It is not difficult to verify the following eigen-decomposition:
\begin{align}
	I-\overline{\bfR}_ju^\top=U\Gamma U^\top,\label{eq:lemma:existence:good:basis:eq2}
\end{align}
where $U=(u_1, u_2, \ldots, u_{M_j+r_j})\in \mathbb{R}^{(M_j+r_j)\times (M_j+r_j)}$ and $\Gamma=\textrm{diag}(1,1,\ldots, 1, 0)\in \mathbb{R}^{(M_j+r_j)\times (M_j+r_j)}$, with $u_k$'s be the standardized eigenvectors.  Let $K=(u_1, \ldots, u_{M_j+r_j-1})$ and $D=K^\top$, and by the definition of eigenvectors and eigenvalues, it is not difficult to verify $D$, $K$ indeed satisfy (\ref{eq:lemma:existence:good:basis:eq1}). 

Now we define $\bfP_j(z)=D\bfR_j(z)\in  \mathbb{R}^{M_j+r_j-1}$, and we will show $\bfP_j(z)$ is a basis of $\textrm{CSpl}(r_j, \bft_{M_j})$. Notice that for any $f\in \textrm{CSpl}(r_j, \bft_{M_j})$, there exist $c\in \mathbb{R}^{M_j+r_j}$ such that
\begin{align*}
	f(z)-\int_0^1f(s)ds=f(z)=c^\top \bfR_j(z)=c^\top (\bfR_j(z)-\overline{\bfR}_j)=c^\top(I-\overline{\bfR}_ju\top)\bfR_j(z)=c^\top K\bfP_j(z),
\end{align*}
where (\ref{eq:lemma:existence:good:basis:eq0}), (\ref{eq:lemma:existence:good:basis:eq1}), (\ref{eq:lemma:existence:good:basis:eq2}) and the fact that $\int_0^1f(s)ds=0$ are used. Therefore, we verify that  $\bfP_j(z)$ is a basis of $\textrm{CSpl}(r_j, \bft_{M_j})$. Moreover, for any $c\in \mathbb{R}^{M_j+r_j-1}$, it follows that
\begin{align*}
	c^\top \int_0^1\bfP_j(z)\bfP_j^\top(z)dz c=c^\top K^\top \int_0^1\bfR_j(z)\bfR_j^\top(z)dz   Kc\leq ah_jc^\top c, 
\end{align*}
and
\begin{align*}
	c^\top \int_0^1\bfP_j(z)\bfP_j^\top(z)dz c=c^\top K^\top \int_0^1\bfR_j(z)\bfR_j^\top(z)dz   Kc\geq a^{-1}h_jc^\top c, 
\end{align*}
where (\ref{eq:lemma:existence:good:basis:eq00}) and (\ref{eq:lemma:existence:good:basis:eq1}) are used. Therefore, the first result follows with $c_3=a$. Moreover, for any fixed $z_0\in [0, 1]$, define collection of indexes $K(z_0)=\{k \;\;|\;\; R_{j,k}(z_0)>0\}$, then by the properties of B-spline, there are $r_j+1$ elements in $K(z_0)$. Moreover, by the properties of B-spline,  it follows that $s^{-1}\leq \bfR_j^\top(z_0)\bfR_j(z_0)\leq s$ for some $s>1$ ($s$ relies on $z_0$). As a consequence, it follows that
\begin{align*}
\bfP_j^\top(z_0)\bfP_j(z_0)= \bfR_j^\top(z_0)D^\top D \bfR_j(z_0)\leq \bfR_j^\top(z_0) \bfR_j(z_0)\leq s.
\end{align*}
Moreover, we have
\begin{align*}
	\bfP_j^\top(z_0)\bfP_j(z_0)&= \bfR_j^\top(z_0)D^\top D \bfR_j(z_0)\nonumber\\
	&=\bfR_j^\top(z_0)K D \bfR_j(z_0)\nonumber\\
	&=\bfR_j^\top(z_0)K D KD \bfR_j(z_0)\nonumber\\
	&=\bfR_j^\top(z_0)(I-\overline{\bfR}_ju^\top)^2\bfR_j(z_0)\nonumber\\
	&=(\bfR_j(z_0)-\overline{\bfR}_j)^\top(\bfR_j(z_0)-\overline{\bfR}_j)\nonumber\\
	&\geq \sum_{k\in K(z_0)}(R_{j,k}(z_0)-\overline{R}_{j,k})^2
\end{align*}
Notice $R_{j,k}\geq 0$ for all $k=1,\ldots, M_j+r_j$ and $\sup_{1\leq k \leq M_j+r_j}\overline{R}_{j,k}=O(h_j)$, so we conclude that $\bfP_j^\top(z_0)\bfP_j(z_0)\geq s^{-1}-O(h_j)\geq s^{-1}/2$. As a consequence, the second result follows with $c_4=2s$.
\end{proof} 

Let $\bfP_j(z)=z-1/2$ for $j=1,\ldots, d$, and $\bfP(\bfx)=(\bfP_1^\top(z_1), \ldots, \bfP_p^\top(z_d))^\top \in \mathbb{R}^{d+\sum_{j=d+1}^p(M_j+r_j-1)}$  for $\bfx=(z_1,\ldots, z_p)^\top \in \mcX$, with $\bfP_j$ being the basis in Lemma \ref{lemma:existence:good:basis} for $j=d+1,\ldots, p$. Define matrices
\begin{align*}
	\Lambda_{NT}&=\frac{1}{NT}\sum_{i=1}^N\sum_{t=1}^T\bigg(\bfP(\bfX_{it})-\frac{1}{T}\sum_{s=1}^T\bfP(\bfX_{is})\bigg)\bigg(\bfP(\bfX_{it})-\frac{1}{T}\sum_{s=1}^T\bfP(\bfX_{is})\bigg)^\top,\nonumber\\
	\Lambda_{NT,j}&=\frac{1}{NT}\sum_{i=1}^N\sum_{t=1}^T\bigg(\bfP_j(Z_{itj})-\frac{1}{T}\sum_{s=1}^T\bfP_j(Z_{isj})\bigg)\bigg(\bfP_j(Z_{itj})-\frac{1}{T}\sum_{s=1}^T\bfP_j(Z_{isj})\bigg)^\top,\textrm{ for } j\in [p],\nonumber\\
	\Lambda&=\ev(\Lambda_{NT}), \textrm{ and } \Lambda_j=\ev(\Lambda_{NT,j}), \textrm{ for } j\in [p].
\end{align*}
By above definition, if $g(\bfx)=u^\top\bfP(\bfx)$ and $f(\bfx)=v^\top \bfP(\bfx)$, then 
\begin{align}
\langle g, f \rangle_{NT}=u^\top \Lambda_{NT} v, \quad \textrm{ and }\quad\;\; \langle g, f \rangle=u^\top \Lambda v. \label{eq:expression:quadratic:norm}
\end{align}
Similarly, if $g_j(z)=u_j^\top \bfP_j(z)$ and $f_j(z)=v_j^\top \bfP(z)$, then
\begin{align}
\langle g_j, f_j \rangle_{NT}=u_j^\top \Lambda_{NT,j} v_j, \quad \textrm{ and }\quad\;\; \langle g_j, f_j \rangle=u^\top \Lambda_j v. \label{eq:expression:quadratic:norm:j}
\end{align}
\begin{lemma}\label{lemma:eigen:value:Lambda:j}
Suppose one of the following conditions is satisfied:
\begin{enumerate}
\item Assumptions \ref{Assumption:A1}, and \ref{Assumption:common} hold;
\item Assumptions \ref{Assumption:A2},  \ref{Assumption:common} hold, and $A_{NT}^2=o(T)$.
\end{enumerate}
Then there exists $c_5>1$ such that the following holds:
\begin{align*}
	c_5^{-1}h_j\leq \lambda_{\min}(\Lambda_{j})\leq \lambda_{\min}(\Lambda_{j})\leq c_5h_j \textrm{ for } j=d+1,\ldots, p,
\end{align*}
and
\begin{align*}
	c_5^{-1}\bigg(\sum_{j=1}^du_j^2+ \sum_{j=d+1}^p u_j^\top u_jh_j\bigg)\leq u^\top\Lambda u \leq c_5\bigg(\sum_{j=1}^du_j^2+ \sum_{j=d+1}^p u_j^\top u_jh_j\bigg),
\end{align*}
for any $u=(u_1, \ldots, u_d, u_{d+1}^\top, \ldots, u_p^\top)^\top\in \mathbb{R}^{d+\sum_{j=d+1}^p(M_j+r_j-1)}$ with $u_j\in \mathbb{R}$ for $j=1,\ldots, d$, and $u_j \in \mathbb{R}^{M_j+r_j-1}$ for $j=d+1,\ldots, p$.
\end{lemma}
\begin{proof}[Proof of Lemma \ref{lemma:eigen:value:Lambda:j}]
The first inequality follows from Lemmas \ref{lemma:expectation:sample:variance} (Lemma \ref{lemma:expectation:sample:variance:large:T} for diverging $T$), \ref{lemma:existence:good:basis}, and (\ref{eq:expression:quadratic:norm:j}). The second inequality follows from Lemma \ref{lemma:sum:of:l2:norm:bound} and triangular inequality.
\end{proof}
\subsubsection{\textbf{Nonparametric Part}}
In this section, we consider to estimate functional 
\begin{equation}\label{eq:definition:irregular:functional}
	F(f_0)=f_{j,0}(z_0), \textrm{ for some $j=d+1,\ldots, p$},
\end{equation}
where $z_0\in [0,1]$ is a pre-specified fixed point.  Next, we will find $v_{NT}^*$ in defined (\ref{eq:v:N:star}). It is not difficult to verify that for any $v_{NT}\in \mcV_{NT}\subset \Theta_{NT}^0$, we have $\frac{\partial F(f_0)}{\partial f}[v_{NT}]=v_{j,{NT}}(z_0)$, where $v_{NT}$ has the decomposition $v_{NT}(\bfx)=\sum_{j=1}^pv_{j,{NT}}(z_j)$ with $v_{j, {NT}}\in \Theta_{NT,j}$ for $j=1,\ldots, p]$. Since  $v_{NT}\in \mcV_{NT}\subset \Theta_{NT}^0$, it follows that
\begin{align*}
	v_{NT}(\bfx)=u^\top \bfP(\bfx)=\sum_{j=1}^p u_j^\top \bfP_j(z_j),
\end{align*} 
for some  $u=(u_1, \ldots, u_d, u_{d+1}^\top, \ldots, u_p^\top)^\top\in \mathbb{R}^{d+\sum_{j=d+1}^p(M_j+r_j-1)}$ with $u_j\in \mathbb{R}$ for $j=1,\ldots, d$, and $u_j \in \mathbb{R}^{M_j+r_j-1}$ for $j=d+1,\ldots, p$. Furthermore, it is not difficult to verify that  $\frac{\partial F(f_0)}{\partial f}[v_{NT}]=v_{j,{NT}}(z_0)=\langle v_{NT}, v_{NT}^*\rangle$, where
\begin{align}
	v_{NT}^*(\bfx)=u^{*\top}\Lambda^{-1}\bfP(\bfx), \textrm{ with } u^*=(0,\ldots, 0, \bfP_j^\top(z_0), 0, \ldots, 0)^\top \in \mathbb{R}^{d+\sum_{j=d+1}^p(M_j+r_j-1)}.\label{eq:expression:vN:star:nonparametric}
\end{align}
It can be verified that $v_{NT}^*(\bfx)$ defined above is the same as $v_{NT,j}^*(\bfx)$ defined in (\ref{eq:definition:hat:vnt}) for $j=d+1,\ldots, p$.
\begin{proposition}\label{proposition:inverse:quadratic:form}
Let $A \in \mathbb{R}^{k\times k}$ be a symmetric and positive definite matrix. Suppose there exist $a>1$ and positive constants $h_j$'s such that for all $u_i\in \mathbb{R}^{k_i}, i=1,\ldots, p$ with $\sum_{i=1}^p k_i=k$, the following inequality is satisfied:
\begin{align*}
	a^{-1}\sum_{i=1}^ph_i u_i^\top u_i\leq u^\top A u\leq a\sum_{i=1}^ph_i u_i^\top u_i,
\end{align*} 
where $u=(u_1^\top,\ldots, u_p^\top)^\top \in \mathbb{R}^k$ . Then it follows that
\begin{align*}
	a^{-1}\sum_{i=1}^ph_i^{-1} u_i^\top u_i\leq u^\top A^{-1} u\leq a\sum_{i=1}^ph_i^{-1} u_i^\top u_i.
\end{align*}
\end{proposition}
\begin{proof}[Proof of Proposition \ref{proposition:inverse:quadratic:form}]
By conditions given, for all $u\in \mathbb{R}^k$, the following holds:
\begin{align*}
	a^{-1}u^{\top}Du \leq u^\top A u \leq au^{\top}Du,
\end{align*}
where 
\begin{align*}
	D=\textrm{Diag}(\underbrace{h_1,\ldots,h_1}_{k_1}, \underbrace{h_2,\ldots, h_2}_{k_2}, \ldots, \underbrace{h_{p-1}, \ldots, h_{p-1}}_{k_{p-1}}, \underbrace{h_p, \ldots, h_p}_{k_p}).
\end{align*}
Let $u=D^{-1/2}v$, we have
\begin{align*}
	a^{-1}v^\top v\leq v^\top D^{-1/2}AD^{-1/2} v\leq av^\top v,
\end{align*}
which, by definition of eigenvalues, further implies
\begin{align*}
	a^{-1}v^\top v\leq v^\top D^{1/2}A^{-1}D^{1/2} v\leq av^\top v.
\end{align*}
Let $v=D^{-1/2}b$, it follows that
\begin{align*}
a^{-1}b^\top D^{-1}b\leq b^\top A^{-1}b\leq ab^\top D^{-1}b.
\end{align*} 
Notice that $u, v , b$ can be arbitrary and the fact that
$$D^{-1}=\textrm{Diag}(\underbrace{h_1^{-1},\ldots,h_1^{-1}}_{k_1}, \underbrace{h_2^{-1},\ldots, h_2^{-1}}_{k_2}, \ldots, \underbrace{h_{p-1}^{-1}, \ldots, h_{p-1}^{-1}}_{k_{p-1}}, \underbrace{h_p^{-1}, \ldots, h_p^{-1}}_{k_p}),$$
we finish the proof.
\end{proof}

\begin{lemma}\label{lemma:verify:condition:c1:nonparametric}
Suppose one of following conditions is satisfied:
\begin{enumerate}
\item Assumptions \ref{Assumption:A1}, \ref{Assumption:common} hold and $d_{NT}A_{NT}^2=o(N)$, $\gamma_{NT}=o(\lambda_{NT})$, $d_{NT,0}A_{NT,0}^2\gamma_{NT}^2=o(1)$, $N\rho_{NT}^2=o(1)$;
\item Assumptions \ref{Assumption:A2}, \ref{Assumption:common} hold and $d_{NT}A_{NT}^4=o(NT)$, $d_{NT}A_{NT}^2=o(N)$, $A_{NT}^2=o(T)$, $\gamma_{NT}=o(\lambda_{NT})$, $d_{NT,0}A_{NT,0}^4\gamma_{NT}^2=o(1)$, $d_{NT,0}^2A_{NT,0}^4\gamma_{NT}^2T=o(N)$, $NT\rho_{NT}^2=o(1)$.
\end{enumerate}
Then There exists $c_6>1$ such that
\begin{align*}
	c_6^{-1}h_j^{-1}\leq \|v_{NT}^*\|^2\leq c_6h_j^{-1},\quad \quad\;\; c_6^{-1}h_j^{-1}\leq \|v_{NT}^*\|^2_{\textrm{sd}}\leq c_6h_j^{-1},
\end{align*}
and Condition \ref{Condition:C1} is satisfied. As a consequence, the following expansion holds:
\begin{align*}
	\frac{\sqrt{NT}(\widehat{f}_{j}(z_0)-f_{j,0}(z_0))}{\|v_{NT}^*\|_{\textrm{sd}}}=-\frac{1}{\sqrt{NT}\|v_{NT}^*\|_{\textrm{sd}}}\bfv_{NT}^{*\top}M_H\bfepsilon+o_P(1),
\end{align*}
where $\bfv_{NT}^{*}=(v_{NT}^*(\bfX_{11}), \ldots, v_{NT}^*(\bfX_{NT}))^\top$ with $v_{NT}^*$ defined in (\ref{eq:expression:vN:star:nonparametric}).
\end{lemma}
\begin{proof}[Proof of Lemma \ref{lemma:verify:condition:c1:nonparametric}]
It is trivial to show that Conditions  \ref{Condition:C1}.\ref{C1:a}, and \ref{Condition:C1}.\ref{C1:b} are valid. In the following, we will verify \ref{Condition:C1}.\ref{C1:c}, and \ref{Condition:C1}.\ref{C1:d}.

By the definition of $u^*$ in (\ref{eq:expression:vN:star:nonparametric}), $\bfP_j(z_0)$, Lemmas \ref{lemma:existence:good:basis}, \ref{lemma:eigen:value:Lambda:j}, and Proposition \ref{proposition:inverse:quadratic:form} we have
\begin{align*}
	\|v_{NT}^*\|^2=\langle v_{NT}^*, v_{NT}^* \rangle=u^{*\top}\Lambda^{-1}u^{*\top}\leq c_5h_j^{-1} \bfP_j^\top(z_0)\bfP_j(z_0)\leq c_4c_5h_j^{-1}.
\end{align*}
By similar argument, we can show find the lower bound that $\|v_{NT}^*\|^2\geq c_4^{-1}c_5^{-1}h_j^{-1}$. Thus, with $c_6=c_4c_5$, we conclude that
\begin{equation}\label{eq:lemma:verify:condition:c1:nonparametric:eq1}
c_6^{-1}h_j^{-1}\leq \|v_{NT}^*\|^2\leq c_6 h_j^{-1}.
\end{equation}

By  Lemma \ref{lemma:approximation:error} and (\ref{eq:lemma:verify:condition:c1:nonparametric:eq1}), we have
\begin{align*}
	\frac{\bigg|\frac{\partial F(f_0)}{\partial f}[f_{0,N}-f_0]\bigg|}{\|v_{NT}^*\|}=\frac{|f_{0, N}(z_0)-f_0(z_0)|}{\|v_{NT}^*\|}&\leq \frac{\|f_{0,N}-f_0\|_\infty}{\|v_{NT}^*\|}\leq   \frac{\|g_*-f_0\|_\infty}{\|v_{NT}^*\|}=O(\rho_{NT}h_j),
\end{align*}
which further implies Condition \ref{Condition:C1}.\ref{C1:c} due to rate condition $\rho_{NT}h_j=o(N^{-1/2})$ for short panel and $\rho_{NT}h_j=o(N^{-1/2}T^{-1/2})$ for large panel.

Moreover, direct examination shows that
\begin{align*}
\|v_{NT}^*\|_{\textrm{sd}}^2=\textrm{Var}(\sqrt{NT}\Delta_{NT}(f_0)[v_{NT}^*])&=\frac{1}{NT}\ev\bigg(\bfv_{NT}^{*\top} M_H\bfepsilon\bfepsilon^\top M_H \bfv_{NT}^*\bigg) \nonumber\\
&=\frac{1}{NT}\ev\bigg(\bfv_{NT}^{*\top} M_H\ev(\bfepsilon\bfepsilon^\top| \mathbb{Z}) M_H \bfv_{NT}^*\bigg).
\end{align*}
By Assumption \ref{Assumption:common}.\ref{Ac:a2}  and above equation, we have
\begin{align*}
	\|v_{NT}^*\|_{\textrm{sd}}^2\leq \frac{a_2}{NT}\ev\bigg(\bfv_{NT}^{*\top} M_H \bfv_{NT}^*\bigg)=a_2 \|v_{NT}^*\|^2.
\end{align*}
Similarly, we can establish the lower bound that  $\|v_{NT}^*\|_{\textrm{sd}}^2\geq a_2^{-1} \|v_{NT}^*\|^2.$
Therefore, above inequalities and (\ref{eq:lemma:verify:condition:c1:nonparametric:eq1}) lead to \ref{Condition:C1}.\ref{C1:d}. Finally, by Lemma \ref{lemma:general:asymptotoic:expansion}, we prove the second result.
\end{proof}

\subsubsection{\textbf{Parametric  Part}}
In this section, we consider to estimate functional 
\begin{equation}\label{eq:definition:regular:functional}
	F(f_0)=\beta_{j,0}, \textrm{ for some $j=1,\ldots, d$}.
\end{equation}
Next, we will find $v_{NT}^*$  defined in (\ref{eq:v:N:star}). It is not difficult to verify that for any $v_{NT}\in \mcV_{NT}\subset \Theta_{NT}^0$, we have $\frac{\partial F(f_0)}{\partial f}[v_{NT}]=v_{j,{NT}}(z_0)$, where $v_{NT}$ has the decomposition $v_{NT}(\bfx)=\sum_{j=1}^Nv_{j,{NT}}(z_j)$ with $v_{j, {NT}}\in \Theta_{NT,j}$ for $j=1,\ldots, p$, and $v_{j, {NT}}(z)=\beta_j(z-1/2)$ for $j=1,\ldots, d$. Since  $v_{NT}\in \mcV_{NT}\subset \Theta_{NT}^0$, it follows that $v_{NT}(\bfx)=u^\top \bfP(\bfx)=\sum_{j=1}^p u_j^\top \bfP_j(z_j)$, for some  $u=(u_1, \ldots, u_d, u_{d+1}^\top, \ldots, u_p^\top)^\top\in \mathbb{R}^{d+\sum_{j=d+1}^p(M_j+r_j-1)}$ with $u_j=\beta_j\in \mathbb{R}$ for $j=1,\ldots, d$, and $u_j \in \mathbb{R}^{M_j+r_j-1}$ for $j=d+1,\ldots, p$. Furthermore, it is not difficult to verify that  $\frac{\partial F(f_0)}{\partial f}[v_{NT}]=\beta_{j}=\langle v_{NT}, v_{NT}^*\rangle$, where
\begin{align}
	v_{NT}^*(\bfx)=u^{*\top}\Lambda^{-1}\bfP(\bfx), \textrm{ with } u^*=(\underbrace{0,\ldots, 0}_{j-1}, 1, 0, \ldots, 0)^\top \in \mathbb{R}^{d+\sum_{j=d+1}^p(M_j+r_j-1)}.\label{eq:expression:vN:star:parametric}
\end{align}
It can be verified that $v_{NT}^*(\bfx)$ defined above is the same as $v_{NT,j}^*(\bfx)$ defined in (\ref{eq:definition:hat:vnt}) for $j=1,\ldots, d$.
\begin{lemma}\label{lemma:verify:condition:c1:parametric}
Suppose one of following conditions is satisfied:
\begin{enumerate}
\item Assumptions \ref{Assumption:A1}, \ref{Assumption:common} hold and $d_{NT}A_{NT}^2=o(N)$, $\gamma_{NT}=o(\lambda_{NT})$, $d_{NT,0}A_{NT,0}^2\gamma_{NT}^2=o(1)$, $N\rho_{NT}^2=o(1)$;
\item Assumptions \ref{Assumption:A2}, \ref{Assumption:common} hold and $d_{NT}A_{NT}^4=o(NT)$, $d_{NT}A_{NT}^2=o(N)$, $A_{NT}^2=o(T)$, $\gamma_{NT}=o(\lambda_{NT})$, $d_{NT,0}A_{NT,0}^4\gamma_{NT}^2=o(1)$, $d_{NT,0}^2A_{NT,0}^4\gamma_{NT}^2T=o(N)$, $NT\rho_{NT}^2=o(1)$.
\end{enumerate}
 Then there exists $c_7>1$ such that
\begin{align*}
	c_7^{-1}h_j^{-1}\leq \|v_{NT}^*\|^2\leq c_6h_j^{-1},\quad\quad\;\; c_7^{-1}h_j^{-1}\leq \|v_{NT}^*\|^2_{\textrm{sd}}\leq c_6h_j^{-1},
\end{align*}
and Condition \ref{Condition:C1} is satisfied. As a consequence, the following expansion holds:
\begin{align*}
	\sqrt{NT}(\widehat{\beta}_j-\beta_{j,0})=-\frac{1}{\sqrt{NT}}\bfv_{NT}^{*\top}M_H\bfepsilon+o_P(1),
\end{align*}
where $\bfv_{NT}^{*}=(v_{NT}^*(\bfX_{11}), \ldots, v_{NT}^*(\bfX_{NT}))^\top$ with $v_{NT}^*$ defined in (\ref{eq:expression:vN:star:parametric}).
\end{lemma}
\begin{proof}[Proof of Lemma \ref{lemma:verify:condition:c1:parametric}]
It is trivial to show that Conditions  \ref{Condition:C1}.\ref{C1:a}, and \ref{Condition:C1}.\ref{C1:b} are valid. In the following, we will verify \ref{Condition:C1}.\ref{C1:c}, and \ref{Condition:C1}.\ref{C1:d}.

By the definition of $u^*$ in (\ref{eq:expression:vN:star:parametric}), $\bfP_j(z_0)$ , Lemmas \ref{lemma:existence:good:basis}, \ref{lemma:eigen:value:Lambda:j}, and Proposition \ref{proposition:inverse:quadratic:form} we have
\begin{align*}
	\|v_{NT}^*\|^2=\langle v_{NT}^*, v_{NT}^* \rangle=u^{*\top}\Lambda^{-1}u^{*\top}\leq c_5.
\end{align*}
By similar technique, we can establish the lower bound that $\|v_{NT}^*\|^2\geq c_5^{-1}.$
Combining above two inequalities and with $c_7=c_5$, we obtain that
\begin{equation}\label{eq:lemma:verify:condition:c1:parametric:eq1}
	c_7^{-1}\leq \|v_{NT}^*\|^2\leq c_7,
\end{equation}

By definition of $f_{0, N}$, the linear components of $f_{0,N}$ and $f_0$ should be the same, and we can conclude that 
\begin{align*}
	\frac{\bigg|\frac{\partial F(f_0)}{\partial f}[f_{0,N}-f_0]\bigg|}{\|v_{NT}^*\|}=\frac{|\beta_{j,0,N}-\beta_{j,0}|}{\|v_{NT}^*\|}=0,
\end{align*}
which is Condition \ref{Condition:C1}.\ref{C1:c}.

Moreover,  direct examination leads to follow equation:
\begin{align*}
\|v_{NT}^*\|_{\textrm{sd}}^2=\textrm{Var}(\sqrt{NT}\Delta_{NT}(f_0)[v_{NT}^*])&=\frac{1}{NT}\ev\bigg(\bfv_{NT}^{*\top} M_H\bfepsilon\bfepsilon^\top M_H \bfv_{NT}^*\bigg) \nonumber\\
&=\frac{1}{NT}\ev\bigg(\bfv_{NT}^{*\top} M_H\ev(\bfepsilon\bfepsilon^\top| \mathbb{Z}) M_H \bfv_{NT}^*\bigg).
\end{align*}
Next, by Assumption \ref{Assumption:common}.\ref{Ac:a2}, we have
\begin{align*}
	\|v_{NT}^*\|_{\textrm{sd}}^2\leq \frac{a_2}{NT}\ev\bigg(\bfv_{NT}^{*\top} M_H \bfv_{NT}^*\bigg)=a_2 \|v_{NT}^*\|^2.
\end{align*}
Similarly, we can obtain the lower bound that $\|v_{NT}^*\|_{\textrm{sd}}^2\geq a_2^{-1} \|v_{NT}^*\|^2.$
Therefore, above inequalities and (\ref{eq:lemma:verify:condition:c1:parametric:eq1}) lead to \ref{Condition:C1}.\ref{C1:d}. Finally, by Lemma \ref{lemma:general:asymptotoic:expansion}, we prove the second result.
\end{proof}

\begin{lemma}\label{thm:asymptotic:normal:multiple}
Suppose Assumptions \ref{Assumption:A1}, \ref{Assumption:common}, and \ref{Assumption:A4} hold. If $d_{NT}A_{NT}^2=o(N)$, $\gamma_{NT}=o(\lambda_{NT})$, $d_{NT,0}A_{NT,0}^2\gamma_{NT}^2=o(1)$, $N\rho_{NT}^2=o(1)$, and $T$ is fixed, then
\begin{align*}
	\begin{pmatrix}
	\sqrt{NT}(\widehat{\beta}_1-\beta_{1,0})\\
	\vdots\\
	\sqrt{NT}(\widehat{\beta}_d-\beta_{d,0})\\
	\sqrt{NTh_j}(\widehat{f}_{d+1}(z_{d+1,0})-f_{d+1,0}(z_{d+1,0}))\\
	\vdots\\
	\sqrt{NTh_j}(\widehat{f}_{p}(z_{p,0})-f_{p,0}(z_{p,0}))\\
	\end{pmatrix}\cid \textrm{N}(0, \Sigma).
\end{align*}
\end{lemma}
\begin{proof}[Proof of Lemma \ref{thm:asymptotic:normal:multiple}]

For $u\in \Theta_{NT}^0$ with $\|u\|=O(1)$, and $\lim_{N \to \infty}\|u\|_{\textrm{sd}}^2>0$, we define
\begin{align*}
	w_i=\frac{1}{T}\sum_{t=1}^Tv_{it}\epsilon_{it},\;\; \textrm{ with }\;\; v_{it}=u(\bfX_{it})-\frac{1}{T}\sum_{s=1}^Tu(\bfX_{is}).
\end{align*}
Then, it follows that $\bfu^\top M_H \bfepsilon/T=\sum_{i=1}^N w_i.$
By Assumption \ref{Assumption:A4}.\ref{A4:a}, it follows that
\begin{align*}
	\ev\bigg(\bigg[\frac{1}{T}\sum_{t=1}^T\epsilon_{it}^2\bigg]^{2}\bigg| \mathbb{Z}\bigg)=\frac{1}{T^2}\sum_{t=1}^T\ev(\epsilon_{it}^4 | \mathbb{Z})+\frac{2}{T^2}\sum_{1\leq t<s\leq T}\ev(\epsilon_{it}^2\epsilon_{is}^2| \mathbb{Z})\leq 3\sup_{1\leq t \leq T}\ev(\epsilon_{it}^4 | \mathbb{Z})\leq 3a_8.
\end{align*}
Furthermore, notice $|v_{it}|\leq 2\|u\|_\infty\leq 2A_{NT,0}\|u\|$, we have
\begin{align*}
	\ev\bigg(\bigg[\frac{1}{T}\sum_{t=1}^Tv_{it}^2\bigg]^{2}\bigg)&=\frac{1}{T^2}\sum_{t=1}^T\ev(v_{it}^4)+\frac{2}{T^2}\sum_{1\leq t<s\leq T}\ev(v_{it}^2v_{is}^2)\nonumber\\
	&\leq  \frac{4\|u\|_\infty^2}{T^2}\sum_{t=1}^T\ev(v_{it}^2)+\frac{8\|u\|_\infty^2}{T}\sum_{t=1}^T\ev(v_{it}^2)\nonumber\\
	&\leq \frac{12A_{NT,0}^2\|u\|^2}{T}\sum_{t=1}^T\ev(v_{it}^2)\nonumber\\
	&= \frac{12A_{NT,0}^2\|u\|^2}{T}\sum_{t=1}^T\ev\bigg(\bigg|u(\bfX_{it})-\frac{1}{T}\sum_{s=1}^Tu(\bfX_{is})\bigg|^2\bigg)\nonumber\\
	&= 12A_{NT,0}^2\|u\|^2\ev\bigg(\zeta_i(u,u)\bigg),
\end{align*}
where $\zeta_i(\cdot, \cdot)$ is defined in Section \ref{sec:estimation}.
Abov inequalites together lead to
\begin{align}
	\ev(|w_i|^4)\leq \ev\bigg(\bigg[\frac{1}{T}\sum_{t=1}^Tv_{it}^2\bigg]^{2}\bigg[\frac{1}{T}\sum_{t=1}^T\epsilon_{it}^2\bigg]^{2}\bigg)&=\ev\bigg\{\bigg[\frac{1}{T}\sum_{t=1}^Tv_{it}^2\bigg]^{2}\ev\bigg(\bigg[\frac{1}{T}\sum_{t=1}^T\epsilon_{it}^2\bigg]^{2}\bigg| \mathbb{Z}\bigg)\bigg\}\nonumber\\
	&\leq 12a_2A_{NT,0}^2\|u\|^2\ev\bigg(\zeta_i(u,u)\bigg),\nonumber
\end{align}
which further implies that
\begin{align}
	\sum_{i=1}^N \ev(|w_i|^4)\leq 12a_2NA_{NT,0}^2\|u\|^4.\label{eq:lemma:asymptotic:normal:multiple:eq1}
\end{align}
By independence of $w_i$'s, it follows from Assumption \ref{Assumption:common}.\ref{Ac:a2}  that
\begin{align}
	\textrm{Var}(\sum_{i=1}^Nw_i)&=\sum_{i=1}^N \ev(w_i^2)=\frac{1}{T^2}\ev\bigg(\bfu^\top M_H\bfepsilon\bfepsilon^\top M_H \bfu\bigg)\geq \frac{1}{T^2}a_2^{-1}\ev\bigg(\bfu^\top M_H \bfu\bigg)= NT^{-1}a_2^{-1}\|u\|^2.\label{eq:lemma:asymptotic:normal:multiple:eq2}
\end{align}
Since (\ref{eq:lemma:asymptotic:normal:multiple:eq1}) and (\ref{eq:lemma:asymptotic:normal:multiple:eq2}) together imply that  $\sum_{i=1}^N \ev(|w_i|^4)/(\textrm{Var}(\sum_{i=1}^Nw_i))^2\leq 12a_2^3 {A_{NT,0}^2T^2}/N=o(1)$, which, by Lyapunov C.L.T, further leads to
\begin{align}
	-\frac{\sqrt{NT}\Delta_{NT}(f_0)[u]}{\|u\|_{\textrm{sd}}}=\frac{\sum_{i=1}^Nw_i}{\sqrt{\textrm{Var}(\sum_{i=1}^Nw_i)}}\cid \textrm{N}(0, 1).\label{eq:lemma:asymptotic:normal:multiple:eq3}
\end{align}

Consider the functional $F_j(f_{0})=\beta_{j,0}$ for $j=1,\ldots, d$, and $F_j(f_0)=f_{j,0}(z_{j,0})$ with  $z_{j,0} \in [0, 1]$ for $j=d+1,\ldots, p$. Therefore, there are Riesz representatives  $v_{{NT},j}^* \in \Theta_{NT}^0$ for $j=1,\ldots, p$ of functional $F_j$'s and the expression of $v_{{NT},j}^*$'s are given by (\ref{eq:expression:vN:star:nonparametric}), and (\ref{eq:expression:vN:star:parametric}). Define $u_{{NT},j}^*=v_{{NT},j}^*/\|v_{{NT},j}\|_{\textrm{sd}}$ for $j=1,\ldots, p$, and  $u_{NT}^*(\bfx)=\sum_{i=j}^p\eta_j u_{{NT},j}^*(\bfx) \in \Theta_{NT}^0$.   By Lemmas \ref{lemma:verify:condition:c1:nonparametric} and \ref{lemma:verify:condition:c1:parametric}, it follows that
\begin{align*}
	\frac{\sqrt{NT}(\widehat{\beta}_j-\beta_{j,0})}{\|v_{{NT},j}^*\|_{\textrm{sd}}}=-\sqrt{NT}\Delta_{NT}(f_0)[u_{{NT},j}^*]+o_P(1), \textrm{ for } j=1,\ldots, d,\nonumber
\end{align*}
and
\begin{align*}
	\frac{\sqrt{NT}(\widehat{f}_{j}(z_{j,0})-f_{j,0}(z_{j,0}))}{\|v_{{NT},j}^*\|_{\textrm{sd}}}=-\sqrt{NT}\Delta_{NT}(f_0)[u_{{NT},j}^*]+o_P(1), \textrm{ for } j=d+1,\ldots, p.\nonumber
\end{align*}
As a consequence, we have
\begin{align*}
	\sum_{j=1}^d \eta_j\frac{\sqrt{NT}(\widehat{\beta}_j-\beta_{j,0})}{\|v_{{NT},j}^*\|_{\textrm{sd}}}+\sum_{j=d+1}^p  \eta_j\frac{\sqrt{NT}(\widehat{f}_{j}(z_{j,0})-f_{j,0}(z_{j,0}))}{\|v_{{NT},j}^*\|_{\textrm{sd}}}=-\sqrt{NT}\Delta_{NT}(f_0)[u_{NT}^*]+o_P(1).
\end{align*}

By Assumption \ref{Assumption:A4}.\ref{A4:d}, we have
\begin{align}
	\|u_{NT}^*\|_{\textrm{sd}}^2&=\sum_{j=1}^p \eta_j^2\|u_{{NT},j}^*\|_{\textrm{sd}}^2+2\sum_{1\leq j<k\leq p}\eta_j\eta_k \langle u_{{NT},j}^*, u_{{NT},k}^*\rangle_{\textrm{sd}}\nonumber\\
	&=\sum_{j=1}^p \eta_j^2+2\sum_{1\leq j<k\leq p}\eta_j\eta_k \frac{\langle v_{{NT},j}^*, v_{{NT},k}^*\rangle_{\textrm{sd}}}{\|v_{{NT},j}^*\|_{\textrm{sd}} 	\|v_{{NT},k}^*\|_{\textrm{sd}}} \nonumber\\
	&\to \sum_{j=1}^p \eta_j^2+2\sum_{1\leq j<k\leq p}\eta_j\eta_kr_{j,k}\nonumber\\
	&=(\eta_1, \ldots, \eta_p) R \;(\eta_1, \ldots, \eta_p)^\top,\label{eq:lemma:asymptotic:normal:multiple:eq4}
\end{align}
with 
\begin{align*}
	R=\begin{pmatrix}
	1& r_{1,2}&r_{1,3}&\ldots &r_{1,p}\\ 
	r_{1,2} & 1 & r_{2, 3} &\ldots & r_{2,p}\\
	\vdots&\vdots&\vdots&\vdots&\vdots\\
	r_{1,p} &r_{2,p}&r_{3,p}&\ldots & 1
	\end{pmatrix}=\begin{pmatrix}
	\sigma_1^{-1}&0&\ldots&0\\
	0&\sigma_2^{-1}&\ldots&0\\
	\vdots&\vdots&\vdots&\vdots\\
	0&0&\ldots&0&\sigma_p^{-1}
	\end{pmatrix}\Sigma \begin{pmatrix}
	\sigma_1^{-1}&0&\ldots&0\\
	0&\sigma_2^{-1}&\ldots&0\\
	\vdots&\vdots&\vdots&\vdots\\
	0&0&\ldots&0&\sigma_p^{-1}
	\end{pmatrix}.
\end{align*}
Since $\sigma_j>0$, $R$ is also positive definite and $\|u_{NT}^*\|_{\textrm{sd}}^2>0$. As a consequence of (\ref{eq:lemma:asymptotic:normal:multiple:eq3}),  we conclude that
\begin{align}
	\sum_{j=1}^d \eta_j\frac{\sqrt{NT}(\widehat{\beta}_j-\beta_{j,0})}{\|v_{{NT},j}^*\|_{\textrm{sd}}\|u_{NT}^*\|_{\textrm{sd}}}+\sum_{j=d+1}^p  \eta_j\frac{\sqrt{NT}(\widehat{f}_{j}(z_{j,0})-f_{j,0}(z_{j,0}))}{\|v_{{NT},j}^*\|_{\textrm{sd}}\|u_{NT}^*\|_{\textrm{sd}}}\cid \textrm{N}(0, 1).\nonumber
\end{align}
which, together with (\ref{eq:lemma:asymptotic:normal:multiple:eq4}), lead to
\begin{align*}
	\sum_{j=1}^d \eta_j{\sqrt{NT}(\widehat{\beta}_j-\beta_{j,0})}+\sum_{j=d+1}^p  \eta_j{\sqrt{NTh_j}(\widehat{f}_{j}(z_{j,0})-f_{j,0}(z_{j,0}))}\cid \textrm{N}(0, (\eta_1, \ldots, \eta_p)^\top \Sigma  (\eta_1, \ldots, \eta_p)).
\end{align*}
Finally, by Cram\'{e}r-Wold device and above equation, we prove the result.
\end{proof}

\begin{lemma}\label{lemma:SNT:rate:1:large:T}
Suppose Assumptions \ref{Assumption:A2} and \ref{Assumption:common} hold. If  $A_{NT,0}^2=o(N)$ and $A_{NT,0}^4=o(NT)$, then it follows that
\begin{align*}
	\frac{1}{2NT}\bfu_{NT}^{*\top}M_H \bfu_{NT}^*-\ev\bigg(\frac{1}{2NT}\bfu_{NT}^{*\top}M_H \bfu_{NT}^*\bigg)=O_P(1).
\end{align*}
\end{lemma}
\begin{proof}[Proof of Lemma \ref{lemma:SNT:rate:1:large:T}]
The result can be proved similarly using Bernstein inequality as in Lemma \ref{lemma:SNT:rate:2:large:T} and we omit the proof.
\end{proof}

\begin{lemma}\label{lemma:SNT:rate:2:large:T}
Suppose Assumptions \ref{Assumption:A2} and \ref{Assumption:common} hold. Furthermore, if $d_{NT,0}A_{NT,0}^4\gamma_{NT}^2=o(1)$ and $d_{NT,0}^2A_{NT,0}^4\gamma_{NT}^2 T=o(N)$, then
\begin{align*}
	\frac{1}{NT}\bfu_{NT}^{*\top}M_H(\widehat{\bff}-\bff_0)-\ev\bigg(\frac{1}{NT}\bfu_{NT}^{*\top}M_H(\widehat{\bff}-\bff_0)\bigg)=o_P(N^{-1/2}T^{-1/2})
\end{align*}
\end{lemma}
\begin{proof}[Proof of Lemma \ref{lemma:SNT:rate:2:large:T}]
Let $\mcF_{NT}=\{f\in \Theta_{NT}^0 \;|\; \|f-f_0\|_2\leq C\gamma_{NT}\}$. Notice if $f\in \mcF_{NT}$, by Lemmas \ref{lemma:l2:norm:sup:norm} and \ref{lemma:approximation:error}, it follows that
\begin{align*}
	\|f-f_0\|_\infty&\leq \|f-g_*\|_\infty+\|g_*-f_0\|_\infty\nonumber\\
	&\leq A_{NT,0}\|f-g_*\|_2+\rho_{NT}\nonumber\\
	&\leq CA_{NT,0}\gamma_{NT}+\rho_{NT}\nonumber\\
	&\leq 2CA_{NT,0}\gamma_{NT},
\end{align*}
where we use the fact that $\rho_{NT}=o(A_{NT,0}\gamma_{NT})$. Similar to Lemma \ref{lemma:uniform:equivalence:empirical:population:norm:large:T}, we define
\begin{align*}
	\xi_i(f)=\zeta_{i}(u_{NT}^*, f-f_0),\;\; X_i(f)={\xi_i(f)-\ev(\xi_i(f))},\textrm{ and }\; S_{NT}(f)=\frac{1}{N}\sum_{i=1}^NX_i(f).
\end{align*}
It is not difficult to verify the following equality:
\begin{align*}
	S_N(f)=\frac{1}{NT}\bfu_{NT}^{*\top}M_H({\bff}-\bff_0)-\ev\bigg(\frac{1}{NT}\bfu_{NT}^{*\top}M_H({\bff}-\bff_0)\bigg).
\end{align*}

Similar to the proof of Lemma \ref{lemma:uniform:equivalence:empirical:population:norm:large:T}, we also can show that for $f_1, f_2\in \mcF_{NT}$, the following holds:
\begin{align*}
	|\xi_i(f_1)-\xi_i(f_2)|\leq 4A_{NT,0}^2\|u_{NT}^*\|_2\|f_1-f_2\|_2,
\end{align*}
which further leads to 
\begin{align*}
	|X_i(f_1)-X_i(f_2)|\leq 8A_{NT,0}^2\|u_{NT}^*\|_2\|f_1-f_2\|_2.
\end{align*}
Moreover,  by Lemma \ref{lemma:l2:norm:sup:norm} and Lemma \ref{lemma:expectation:sample:variance:mixing}, we have
\begin{align*}
	\textrm{Var}\bigg(X_i(f_1)-X_i(f_2)\bigg)\leq \frac{c_2}{T}\|u_{NT}^*\|_\infty^2\|f_1-f_2\|_\infty^2\leq \frac{c_2A_{NT,0}^4}{T}\|u_{NT}^*\|_2^2\|f_1-f_2\|_2^2.
\end{align*}
By Bernstein inequality, it follows that
\begin{eqnarray}
&&\pr\bigg(\bigg|S_N(f_1)-S_N(f_2)\bigg|>\frac{xs}{\sqrt{NT}}\bigg)\nonumber\\
&\leq& 2\exp\bigg(-\frac{x^2s^2/T}{\frac{2c_2A_{NT,0}^4}{T}\|u_{NT}^*\|_2^2\|f_1-f_2\|_2^2+\frac{8A_{NT,0}^2}{\sqrt{NT}}\|u_{NT}^*\|_2\|f_1-f_2\|_2xs}\bigg)\nonumber\\
&\leq& 2\exp\bigg(-\frac{x^2s^2}{{4c_2A_{NT,0}^4}\|u_{NT}^*\|_2^2\|f_1-f_2\|_2^2}\bigg)+2\exp\bigg(-\frac{\sqrt{N/T}xs}{16A_{NT,0}^2\|u_{NT}^*\|_2\|f_1-f_2\|_2}\bigg).\nonumber\\\label{eq:lemma:SNT:rate:2:large:T:eq1}
\end{eqnarray}
Let $\delta_k=3^{-k-k_0}$ for $k\geq 0$ and some $k_0$ such that $C\gamma_N \leq 3^{-k_0}\leq 2C\gamma_{NT}$. For sufficient large integer $K$, which will be specified later, let $\{0\}=\mcH_0\subset \mcH_1 \ldots \mcH_K$ be a sequence of subsets of $\mcF_N$ such that $\min_{f^*\in \mcH_k}\|f^*-f\|_2\leq \delta_k$ for all $f \in \mcF_{NT}$. Moreover the subsets $\mcH_K$ is chosen inductively such that two different elements in $\mcH_k$ is at least $\delta_k$ apart. 

By definition, the cardinality $\#(\mcH_k)$ of $\mcH_k$ is bounded by the $\delta_k/2$-covering number $D(\delta_k/2, \mcF_N, \|\cdot\|_2)$. Therefore, we have
\begin{eqnarray*}
	\#(\mcH_k)\leq D(\delta_k/2, \mcF_{NT}, \|\cdot\|_2)\leq \bigg(\frac{8C\gamma_{NT}+\delta_k}{\delta_k}\bigg)^{d_{NT,0}},
\end{eqnarray*} 
where the last inequality is due to \cite{van2000}[Lemma 2.5] and the fact that $\mcF_{NT}$ can be treated as a ball with radius $C\gamma_{NT}$ in $\mathbb{R}^{d_{NT,0}}$.  For any $f \in \mcF_{NT}$, let $\tau_k(f)\in \mcH_k$  be a element such that $\|\tau_k(f)-f\|\leq \delta_k$, for $k=1,2,\ldots, K$. Now for any fixed $x>0$, by (\ref{eq:lemma:SNT:rate:2:large:T:eq1}) and the definition of $\tau_k$, we have
\begin{eqnarray}
&&\pr\bigg(\sup_{f\in \mcF_{NT}}|S_N(f)|>\frac{x}{\sqrt{NT}}\bigg)\nonumber\\
&\leq &\pr\bigg(\sup_{f\in \mcF_{NT}} \bigg|S_N(f)-S_N\bigg(\tau_K(f)\bigg)\bigg|>\frac{x}{2^K\sqrt{NT}}\bigg)\nonumber\\
	&&+\sum_{k=1}^K\pr\bigg(\sup_{f\in \mcF_{NT}} \bigg|S_N\bigg(\tau_k\circ\ldots \circ \tau_K(f)\bigg)-S_N\bigg(\tau_{k-1}\circ\tau_k\circ\ldots \circ \tau_K(f)\bigg)\bigg|>\frac{x}{2^{k-1}\sqrt{NT}}\bigg)\nonumber\\
	&\leq& \pr\bigg(\sup_{f\in \mcF_N} 8A_{NT,0}^2\|u_{NT}^*\|_2\|f-\tau_K(f)\|_2>\frac{x}{2^K\sqrt{NT}}\bigg)\nonumber\\
	&&+\sum_{k=1}^K\#(\mcH_k)\sup_{f\in \mcF_{NT}}\pr\bigg( \bigg|S_N\bigg(\tau_k\circ\ldots \circ \tau_K(f)\bigg)-S_N\bigg(\tau_{k-1}\circ\tau_k\circ\ldots \circ \tau_K(f)\bigg)\bigg|>\frac{x}{2^{k-1}\sqrt{NT}}\bigg)\nonumber\\
	&\leq& \pr\bigg(\sup_{f\in \mcF_N} 8A_{NT,0}^2\|u_{NT}^*\|_2\delta_K>\frac{x}{2^K\sqrt{NT}}\bigg)\nonumber\\
	&&+\sum_{k=1}^K\#(\mcH_k)\sup_{f\in \mcH_k}\pr\bigg( \bigg|S_N\bigg(f\bigg)-S_N\bigg(\tau_{k-1}(f)\bigg)\bigg|>\frac{x}{2^{k-1}\sqrt{NT}}\bigg)\nonumber\\
	&\leq& \pr\bigg(\sup_{f\in \mcF_N} 8A_{NT,0}^2\|u_{NT}^*\|_23^{-K-k_0}>\frac{x}{2^K\sqrt{NT}}\bigg)\nonumber\\
	&&+\sum_{k=1}^\infty \bigg(\frac{8C\gamma_{NT}+\delta_k}{\delta_k}\bigg)^{d_{NT,0}} \sup_{f\in \mcH_k}2\exp\bigg(-\frac{x^2}{{4c_2A_{NT,0}^4}2^{2(k-1)}\|u_{NT}^*\|_2^2\|f-\tau_{k-1}(f)\|_2^2}\bigg)\nonumber\\
	&&+\sum_{k=1}^\infty \bigg(\frac{8C\gamma_{NT}+\delta_k}{\delta_k}\bigg)^{d_{NT,0}} \sup_{f\in \mcH_k}2\exp\bigg(-\frac{\sqrt{N/T}x}{16A_{NT,0}^22^{k-1}\|u_{NT}^*\|_2\|f-\tau_{k-1}(f)\|_2}\bigg)\nonumber\\
	&=&S_1+S_2+S_3.\label{eq:lemma:SNT:rate:2:large:T:eq2}
\end{eqnarray}

For fixed $x>0$, choose $K=K(NT)$ large enough such that $16C\gamma_{NT}A_{NT,0}^2\|u_{NT}^*\|_2(2/3)^K<x$, which can be done due to Condition \ref{Condition:C1}.\ref{C1:d} and Lemma \ref{lemma:expectation:sample:variance:large:T}. So by the fact that $3^{-k_0}\leq 2 \gamma_{NT}$, we have $S_1=0$.  Moreover, direct examination leads to
\begin{align*}
	S_2&\leq \sum_{k=1}^\infty \bigg(\frac{8C\gamma_N+3^{-k-k_0}}{3^{-k-k_0}}\bigg)^{d_{NT,0}}\sup_{f\in \mcH_k}2\exp\bigg(-\frac{x^2}{4c_2A_{NT,0}^42^{2(k-1)}\|u_{NT}^*\|_2^2\|f-\tau_{k-1}(f)\|_2^2}\bigg)\nonumber\\
	&\leq \sum_{k=1}^\infty \bigg(8\times 3^{k}+1\bigg)^{d_{NT,0}}2\exp\bigg(-\frac{x^2}{4c_2A_{NT,0}^42^{2(k-1)}\|u_{NT}^*\|_2^2\delta_k^2}\bigg)\nonumber\\
	&\leq 2\sum_{k=1}^\infty 3^{(k+2)d_{NT,0}}\exp\bigg(-\frac{x^23^{2k+2k_0}}{4c_2A_{NT,0}^42^{2(k-1)}\|u_{NT}^*\|_2^2}\bigg)\nonumber\\
	&\leq 2\sum_{k=1}^\infty 3^{(k+2)d_{NT,0}}\exp\bigg(-\frac{x^23^{2k_0}}{c_2A_{NT,0}^4\|u_{NT}^*\|_2^2}\bigg(\frac{9}{4}\bigg)^k\bigg)\nonumber\\
	&= 2\sum_{k=1}^\infty 3^{(k+2)d_{NT,0}}\exp\bigg(-\frac{x^23^{2k_0}}{c_2A_{NT,0}^4\|u_{NT}^*\|_2^2}\bigg(\frac{9}{4}\bigg)^k\bigg)\nonumber\\
	&\leq 2\sum_{k=1}^\infty 3^{(k+2)d_{NT,0}}\exp\bigg(-\frac{x^2}{4c_2C^2\gamma_{NT}^2A_{NT,0}^4\|u_{NT}^*\|_2^2}\bigg(\frac{9}{4}\bigg)^k\bigg)\nonumber\\
	&=2\sum_{k=1}^\infty \exp\bigg((k+2)d_{NT,0}\log 3-\frac{x^2}{4c_2C^2\gamma_{NT}^2A_{NT,0}^4\|u_{NT}^*\|_2^2}\bigg(\frac{9}{4}\bigg)^k\bigg).
\end{align*}
Let $N,T$ are large enough such that $[8c_2C^2\log 3](k+2)d_{NT,0}A_{NT,0}^4\gamma_{NT}^2\|u_{NT}^*\|_2^2<(9/4)^kx^2$ for all $k\geq 1$. This is possible, as $d_{NT,0}A_{NT,0}^4\gamma_{NT}^2=o(1)$, and $\|u_{NT}^*\|_2=O(1)$. So it follows from the inequality $e^{-x}\leq e^{-1}/x$ that
\begin{align*}
	S_2&\leq 2\sum_{k=1}^\infty \exp\bigg(-\frac{x^2}{8c_2C^2\gamma_{NT}^2A_{NT,0}^4\|u_{NT}^*\|_2^2}\bigg(\frac{9}{4}\bigg)^k\bigg)\nonumber\\
	&\leq 16e^{-1}x^{-2}c_2C^2\gamma_{NT}^2A_{NT,0}^4\|u_{NT}^*\|_2^2\sum_{k=1}^\infty(4/9)^k=O(\gamma_{NT}^2A_{NT,0}^4)=o(1).
\end{align*}
Similarly, we can show that if $d_{NT,0}A_{NT,0}^2\sqrt{T/N}\gamma_{NT}=o(1)$, then
\begin{align*}
	S_3=O\bigg(\frac{A_{NT,0}^2\gamma_{NT}\sqrt{T}}{\sqrt{N}}\bigg)=o(1).
\end{align*}
Since $x>0$ can be arbitrary, by the bounds of $S_1$, $S_2$, $S_3$, and (\ref{eq:lemma:SNT:rate:2:large:T:eq2}), we conclude that
\begin{align*}
\sup_{f\in \mcF_{NT}}\bigg|\frac{1}{NT}\bfu_{NT}^{*\top}M_H({\bff}-\bff_0)-\ev\bigg(\frac{1}{NT}\bfu_{NT}^{*\top}M_H({\bff}-\bff_0)\bigg)\bigg|=o(N^{-1/2}T^{-1/2}).
\end{align*}
Finally, by Theorems \ref{thm:rate:of:convergence:together} and \ref{thm:selection:consistency:together}, it follows that $\lim_{C\to \infty}\lim_{(N, T)\to \infty}\pr(\widehat{f}\in \mcF_{NT})=1$, which, together with above equation, complete the proof.
\end{proof}


\begin{lemma}\label{thm:asymptotic:normal:multiple:large:T}
Suppose Assumptions \ref{Assumption:A2}, \ref{Assumption:common}, and \ref{Assumption:A4} hold. If  $d_{NT}A_{NT}^2=o(N)$, $A_{NT}^2=o(T)$, $\gamma_{NT}=o(\lambda_{NT})$, $d_{NT,0}A_{NT,0}^4\gamma_{NT}^2=o(1)$, $d_{NT,0}^2A_{NT,0}^4\gamma_{NT}^2T=o(N)$, $NT\rho_{NT}^2=o(1)$, then
\begin{align*}
	\begin{pmatrix}
	\sqrt{NT}(\widehat{\beta}_1-\beta_{1,0})\\
	\vdots\\
	\sqrt{NT}(\widehat{\beta}_d-\beta_{d,0})\\
	\sqrt{NTh_j}(\widehat{f}_{d+1}(z_{d+1,0})-f_{d+1,0}(z_{d+1,0}))\\
	\vdots\\
	\sqrt{NTh_j}(\widehat{f}_{p}(z_{p,0})-f_{p,0}(z_{p,0}))\\
	\end{pmatrix}\cid \textrm{N}(0, \Sigma).
\end{align*}
\end{lemma}
\begin{proof}[Proof of Lemma \ref{thm:asymptotic:normal:multiple:large:T}]
For $u\in \Theta_{NT}^0$ with $\|u\|=O(1)$, and $\lim_{N \to \infty}\|u\|_{\textrm{sd}}^2>0$, we define
\begin{align*}
	w_i=\sum_{t=1}^Tv_{it}\epsilon_{it},\;\; \textrm{ with }\;\; v_{it}=u(\bfX_{it})-\frac{1}{T}\sum_{s=1}^Tu(\bfX_{is}).
\end{align*}

By simple inequality $(a+b)^3\leq 4|a|^3+4|b|^3$, we have
\begin{align}
	\ev(|w_i|^3)=\ev(|\sum_{t=1}^Tv_{it}\epsilon_{it}|^3)&\leq 6\ev\bigg(\bigg|\sum_{t=1}^Tu(\bfX_{it})\epsilon_{it}\bigg|^3\bigg)+6\ev\bigg(\bigg|\frac{1}{T}\sum_{s=1}^Tu(\bfX_{is})\bigg|^3|\sum_{t=1}^T\epsilon_{it}|^3\bigg)\nonumber\\
	&=6S_1+6S_2. \label{eq:thm:asymptotic:normal:multiple:large:T:eq1}
\end{align}
Since Lemmas \ref{lemma:l2:norm:sup:norm}, \ref{lemma:expectation:sample:variance:large:T}, and Assumption \ref{Assumption:A4}.\ref{A4:a} together lead to following inequality:
\begin{align*}
	\ev(|u(\bfX_{it})\epsilon_{it}|^4)=\ev\bigg(|u(\bfX_{it})|^4\ev(\epsilon_{it}^4|\mathbb{Z})\bigg)\leq a_9\|u\|_\infty^2 \ev(|u(\bfX_{it})|^2)&\leq a_3a_9A_{NT,0}^2\|u\|_2^4 \nonumber\\
	&\leq 4a_3^3a_9A_{NT,0}^2\|u\|^4,
\end{align*}
we can apply  \cite{y78}[Theorem 3] using Assumptions \ref{Assumption:A4}.\ref{A4:a} and \ref{Assumption:A4}.\ref{A4:c} to obtain that $S_1\leq CT^{3/2}A_{NT,0}^{3/2}\|u\|^3$, for some $C>0$ which if free of $N, T$. By similar technique, we can show that $\ev(|\sum_{t=1}^T\epsilon_{it}|^3)\leq \widetilde{C}T^{3/2}$, which, by Lemmas \ref{lemma:l2:norm:sup:norm} and \ref{lemma:expectation:sample:variance:large:T}, further implies $S_2 \leq 8\widetilde{C}a_3^3A_{NT,0}^3T^{3/2}\|u\|^3$, where $\widetilde{C}>0$ is a constant free of $N, T$. Using the bounds of $S_1$ and $S_2$ in (\ref{eq:thm:asymptotic:normal:multiple:large:T:eq1}), we conclude that
\begin{align}
	\sum_{i=1}^N\ev(|w_i|^3)=O(NT^{3/2}A_{NT,0}^3\|u\|^3).\label{eq:thm:asymptotic:normal:multiple:large:T:eq2}
\end{align}
By independence in Assumption \ref{Assumption:A4}.\ref{A4:b} and Assumption \ref{Assumption:common}.\ref{Ac:a2}, it follows that
\begin{align}
	\textrm{Var}(\sum_{i=1}^Nw_i)=\ev\bigg(\bfu^\top M_H \bfepsilon \bfepsilon^\top M_H \bfu\bigg)=\ev\bigg(\bfu^\top M_H \ev(\bfepsilon \bfepsilon^\top|\mathbb{Z}) M_H \bfu\bigg)&\geq a_2^{-1}\ev\bigg(\bfu^\top  M_H \bfu\bigg)\nonumber\\
	&= a_2^{-1}NT\|u\|^2.\label{eq:thm:asymptotic:normal:multiple:large:T:eq3}
\end{align}
Combining (\ref{eq:thm:asymptotic:normal:multiple:large:T:eq2}) and (\ref{eq:thm:asymptotic:normal:multiple:large:T:eq3}), we can verify the following condition for Lyapunov C.L.T is satisfied: $\sum_{i=1}^N\ev(|w_i|^3)/(\textrm{Var}(\sum_{i=1}^Nw_i))^{3/2}=O(A_{NT,0}^3N^{-1/2})=o(1).$ Therefore, the following Lyapunov C.L.T holds:
\begin{align}
	-\frac{\sqrt{NT}\Delta_{NT}(f_0)[u]}{\|u\|_{\textrm{sd}}}=\frac{\sum_{i=1}^N w_i}{\sqrt{\textrm{Var}(\sum_{i=1}^Nw_i)}}\cid \textrm{N}(0, 1).\nonumber
\end{align}
The rest proof is the same as Lemma \ref{thm:asymptotic:normal:multiple}.
\end{proof}

\begin{proof}[Proof of Theorems \ref{thm:asymptotic:normal:multiple:together} and \ref{thm:asymptotic:normal:marginal:together}]
It follows from Lemmas \ref{thm:asymptotic:normal:multiple} and \ref{thm:asymptotic:normal:multiple:large:T}.
\end{proof}
\end{document}